\documentclass[12pt]{report}   %
\usepackage{graphicx}  %
\usepackage[letterpaper, left=1in, right=1in, top=1in, bottom=1in]{geometry}
\usepackage{setspace}  %
\usepackage[explicit]{titlesec}  %
\usepackage[titles]{tocloft}  %
\usepackage[
    backend=bibtex,
    style=alphabetic,
    sorting=none,
]{biblatex} %
\usepackage{xcolor}
\usepackage[bookmarks=true]{hyperref}
\hypersetup{
    colorlinks=true,%
    citecolor=black,%
    filecolor=black,%
    linkcolor=black,%
    urlcolor=black
}
\usepackage[page]{appendix}  %
\usepackage{rotating}  %
\usepackage[normalem]{ulem}  %
\usepackage{amsmath, amsthm, amssymb, mathrsfs, amscd}

\theoremstyle{plain}
\newtheorem{thm}{Theorem}
\newtheorem{lem}{Lemma}

\usepackage[english]{babel}
\usepackage[autostyle, english = american]{csquotes}
\MakeOuterQuote{"}

\newcommand\numberthis{\addtocounter{equation}{1}\tag{\theequation}}

\renewcommand{\v}[1]{\ensuremath{\mathbf{#1}}} 
\newcommand{\gv}[1]{\pmb{#1}} 
 
\newcommand{\mc}[1]{\mathcal{#1}} 
\newcommand{\mb}[1]{\mathbb{#1}} 
\newcommand{\ms}[1]{\mathscr{#1}} 
\newcommand{\mr}[1]{\mathrm{#1}} 
\newcommand{\mf}[1]{\mathfrak{#1}}

\newcommand{\fC}{\ms{C}}

\newcommand{\cL}{\mc{L}}
\newcommand{\fL}{\ms{L}}
\newcommand{\fD}{\ms{D}}
\newcommand{\sR}{\ms{R}}
\newcommand{\fS}{\ms{S}}
\newcommand{\fF}{\ms{F}}

\newcommand{\R}{\mb{R}}
\newcommand{\C}{\mb{C}}
\newcommand{\N}{\mb{N}}
\newcommand{\F}{\mb{F}}

\newcommand{\sZ}{\ms{Z}}

\newcommand{\J}{\gv{\mc{J}}}
\newcommand{\T}{\mc{T}}
\newcommand{\I}{\mc{I}}
\newcommand{\Z}{\mc{Z}}
\newcommand{\E}{\mc{E}}
\newcommand{\X}{\mc{X}}
\newcommand{\V}{\mc{V}}
\newcommand{\U}{\mc{U}}
\newcommand{\B}{\mc{B}}
\newcommand{\D}{\mc{D}}
\newcommand{\cF}{\mc{F}}
\newcommand{\Ham}{\mc{H}}

\newcommand{\s}{\mf{s}}

\newcommand{\dd}{\mr{d}}

\newcommand{\gvsig}{\gv{\sigma}}
\newcommand{\gvalp}{\gv{\alpha}}

\newcommand{\ul}[1]{\underline{#1}}
\newcommand{\ol}[1]{\overline{#1}}
\newcommand{\BWN}[1]{\bigwedge\nolimits^{\!\!N}{\!#1}}

\DeclareMathOperator*{\esssup}{ess\,sup}
\DeclareMathOperator{\curl}{curl}
\DeclareMathOperator{\diver}{div}
\DeclareMathOperator{\re}{Re}
\DeclareMathOperator{\im}{Im}
\DeclareMathOperator{\Hproj}{\mc{P}}

\DeclareMathOperator{\dist}{dist}
\DeclareMathOperator*{\slim}{s-lim}
\DeclareMathOperator*{\supp}{supp}
\bibliography{references}

\AtEveryBibitem{\clearfield{issn}}
\AtEveryBibitem{\clearlist{issn}}

\AtEveryBibitem{\clearfield{language}}
\AtEveryBibitem{\clearlist{language}}

\AtEveryBibitem{\clearfield{doi}}
\AtEveryBibitem{\clearlist{doi}}

\AtEveryBibitem{\clearfield{url}}
\AtEveryBibitem{\clearlist{url}}

\AtEveryBibitem{%
  \ifentrytype{online}
    {}
    {\clearfield{urlyear}\clearfield{urlmonth}\clearfield{urlday}}}
    
\DeclareFieldFormat{labelalpha}{\mkbibbold{#1}}
\DeclareFieldFormat{extraalpha}{\mkbibbold{\mknumalph{#1}}}

\begin{document}

\newcommand{\thesisTitle}{The Maxwell-Pauli Equations}
\newcommand{\yourName}{Thomas Forrest Kieffer}
\newcommand{\yourSchool}{Mathematics}
\newcommand{\yourMonth}{May}
\newcommand{\yourYear}{2020}

\begin{titlepage}
\begin{center}

\begin{singlespacing}

\textbf{\MakeUppercase{\thesisTitle}}\\
\vspace{10\baselineskip}
A Dissertation\\
Presented to\\
The Academic Faculty\\
\vspace{3\baselineskip}
By\\
\vspace{3\baselineskip}
\yourName\\
\vspace{3\baselineskip}
In Partial Fulfillment\\
of the Requirements for the Degree\\
Doctor of Philosophy in the\\
School of \yourSchool\\
\vspace{3\baselineskip}
Georgia Institute of Technology\\
\vspace{\baselineskip}
\yourMonth{} \yourYear{}
\vfill
Copyright \copyright{} \yourName{} \yourYear{}

\end{singlespacing}

\end{center}
\end{titlepage}

\currentpdfbookmark{Title Page}{titlePage}  %

\newcommand{\committeeMemberOne}{Professor Michael Loss, Advisor}
\newcommand{\committeeMemberOneDepartment}{School of Mathematics}
\newcommand{\committeeMemberOneAffiliation}{Georgia Institute of Technology}

\newcommand{\committeeMemberTwo}{Professor Evans Harrell}
\newcommand{\committeeMemberTwoDepartment}{School of Mathematics}
\newcommand{\committeeMemberTwoAffiliation}{Georgia Institute of Technology}

\newcommand{\committeeMemberThree}{Professor Federico Bonetto}
\newcommand{\committeeMemberThreeDepartment}{School of Mathematics}
\newcommand{\committeeMemberThreeAffiliation}{Georgia Institute of Technology}

\newcommand{\committeeMemberFour}{Professor Brian Kennedy}
\newcommand{\committeeMemberFourDepartment}{School of Physics}
\newcommand{\committeeMemberFourAffiliation}{Georgia Institute of Technology}

\newcommand{\committeeMemberFive}{Professor Chongchun Zeng}
\newcommand{\committeeMemberFiveDepartment}{School of Mathematics}
\newcommand{\committeeMemberFiveAffiliation}{Georgia Institute of Technology}

\newcommand{\approvalDay}{10}
\newcommand{\approvalMonth}{March}
\newcommand{\approvalYear}{2020}

\begin{titlepage}
\begin{singlespacing}
\begin{center}

\textbf{\MakeUppercase{\thesisTitle}}\\
\vspace{10\baselineskip}

\end{center}
\vfill

\ifdefined\committeeMemberFour

Approved by:
\vspace{2\baselineskip}		%

\begin{minipage}[b]{0.4\textwidth}
	
	\committeeMemberOne\\
	\committeeMemberOneDepartment\\
	\textit{\committeeMemberOneAffiliation}\\
	
	\committeeMemberTwo\\
	\committeeMemberTwoDepartment\\
	\textit{\committeeMemberTwoAffiliation}\\
	
	\committeeMemberThree\\
	\committeeMemberThreeDepartment\\
	\textit{\committeeMemberThreeAffiliation}\\
	
	\vspace{2\baselineskip}		%
	
\end{minipage}
\hspace{0.1\textwidth}
\begin{minipage}[b]{0.4\textwidth}
	
	\committeeMemberFour\\
	\committeeMemberFourDepartment\\
	\textit{\committeeMemberFourAffiliation}\\
	
	\ifdefined\committeeMemberSix
	\committeeMemberFive\\
	\committeeMemberFiveDepartment\\
	\textit{\committeeMemberFiveAffiliation}\\
	
	\committeeMemberSix\\
	\committeeMemberSixDepartment\\
	\textit{\committeeMemberSixAffiliation}\\
	
	Date Approved: \approvalMonth{} \approvalDay, \approvalYear
	\vspace{1\baselineskip}		%
	
	\else
	
	\committeeMemberFive\\
	\committeeMemberFiveDepartment\\
	\textit{\committeeMemberFiveAffiliation}\\
		
	Date Approved: \approvalMonth{} \approvalDay, \approvalYear
	\vspace{5\baselineskip}		%
	
	\fi
	
\end{minipage}

\else

\hspace{0.6\textwidth}
\begin{minipage}[b]{0.4\textwidth}
	
	Approved by:
	\vspace{2\baselineskip}		%
	
	\committeeMemberOne\\
	\committeeMemberOneDepartment\\
	\textit{\committeeMemberOneAffiliation}\\
	
	\committeeMemberTwo\\
	\committeeMemberTwoDepartment\\
	\textit{\committeeMemberTwoAffiliation}\\
	
	\committeeMemberThree\\
	\committeeMemberThreeDepartment\\
	\textit{\committeeMemberThreeAffiliation}\\
	
	\vspace{2\baselineskip}		%
	
	Date Approved: \approvalMonth{} \approvalDay, \approvalYear
	\vspace{\baselineskip}		%
	
\end{minipage}

\fi

\end{singlespacing}
\end{titlepage}

\newcommand{\yourQuote}{"We take a piece of metal. Or a stone. When we think about it, we are astonished that this quantity
of matter should occupy so large a volume. Admittedly, the molecules are packed tightly together, and likewise the atoms within each molecule. But why are the atoms themselves so big.... Answer: only the Pauli principle, 'No two electrons in the same state.' That is why atoms are so unnecessarily big, and why metal and stone are so bulky."}
\newcommand{\yourAuthor}{- P.~Ehrenfest addressing W.~E.~Pauli in 1931 on the occasion of the Lorentz medal.}

\begin{titlepage}

\vspace*{\fill}
\textit{\yourQuote}\\

\yourAuthor
\vspace*{\fill}

\end{titlepage}

\newcommand{\yourDedication}{For Catherine}

\begin{titlepage}
\begin{center}

\vspace*{\fill}
\yourDedication\\
\vspace*{\fill}

\end{center}
\end{titlepage}

\pagenumbering{roman}
\addcontentsline{toc}{chapter}{Acknowledgments}
\setcounter{page}{5} %
\clearpage
\begin{centering}
\textbf{ACKNOWLEDGEMENTS}\\
\vspace{\baselineskip}
\end{centering}

I would like to express my deepest appreciation and gratitude to Michael Loss, without whom this work would not have been possible. I could not have wished for a better mentor and PhD advisor. Our discussions about physics, mathematics, and philosophy have been a source of great happiness during this period of my life. 

A special thank you goes to my parents, Lauren and Tom Kieffer, my sister, Courtney McKean, and my wife, Catherine Chen Kieffer. They have been a source of enduring support, financially and emotionally, during my time at Georgia Tech. I would not have succeeded in completing this PhD without them.  

I would like to thank Brian Kennedy, Federico Bonetto, Evans Harrell, and Chongchun Zeng for serving on my thesis committee. I am also grateful to the many Professors at Georgia Tech who have inspired and encouraged me. Lastly, I am grateful to the National Science Foundation for support.

\clearpage
\renewcommand{\cftdot}{}
\renewcommand{\cftchapfont}{\bfseries}  %
\renewcommand{\cftchappagefont}{}  %
\renewcommand{\cftchappresnum}{Chapter }
\renewcommand{\cftchapaftersnum}{.}
\renewcommand{\cftchapnumwidth}{6em}
\renewcommand{\cftchapafterpnum}{\vskip\baselineskip} %
\renewcommand{\cftsecafterpnum}{\vskip\baselineskip}  %
\renewcommand{\cftsubsecafterpnum}{\vskip\baselineskip} %
\renewcommand{\cftsubsubsecafterpnum}{\vskip\baselineskip} %

\titleformat{\chapter}[display]
{\normalfont\bfseries\filcenter}{\chaptertitlename\ \thechapter}{0pt}{\MakeUppercase{#1}}

\renewcommand\contentsname{Table of Contents}

\begin{singlespace}
\tableofcontents
\end{singlespace}

\currentpdfbookmark{Table of Contents}{TOC}

\clearpage

\addcontentsline{toc}{chapter}{Notation}
\clearpage
\begin{centering}
\textbf{NOTATION}\\
\vspace{\baselineskip}
\end{centering}

This thesis will often deal with the mixed time and space Sobolev spaces $L^q (\I ; W^{m,p} ( \R^d ; \F^n) )$ for some $q , p \in [ 1, \infty ]$, $m \in \R$, and $d,n \in \N$, with $\I \subset \R$ a (possibly infinite) time interval and $\F$ is either $\R$ or $\C$. For this reason we introduce the following special notations when the underlying sets $\I$, $\R^d$, $\F^n$ are easily understood from the given context.
\begin{itemize}
\item $\| \cdot \|_p \equiv \| \cdot \|_{L^p}$
\item $\| \cdot \|_{s , p} \equiv \| \cdot \|_{W^{s,p}}$
\item $\| \cdot \|_{q; s , p} \equiv \| \cdot \|_{L^q W^{s,p}}$
\item $\| \cdot \|_{q_1 ; s_1 , p_1 \oplus q_2 ; s_2 , p_2} \equiv \| \cdot \|_{L^{q_1} W^{s_1 , p_1} \oplus L^{q_2} W^{s_2 , p_2}}$
\end{itemize}
This notation comes with the understanding that $\| \cdot \|_{p} \equiv \| \cdot \|_{0 , p} \equiv \| \cdot \|_{0 ; 0 , p}$ and $\| \cdot \|_{q ; p} \equiv \| \cdot \|_{q ; 0 , p}$. For more detail regarding these and other notations see \S\ref{sec:notation}.

\addcontentsline{toc}{chapter}{Units}
\clearpage
\begin{centering}
\textbf{UNITS}\\
\vspace{\baselineskip}
\end{centering}

Let $e_0$, $m_e$, $\hbar$, and $c$ be the electron charge, electron rest mass, the reduced Plank's constant, and the speed of light, respectively. Let $\alpha = e_0^2 / (\hbar c)$ is Sommerfeld's dimensionless fine structure constant. In our universe $\alpha \simeq 1/137$, but in this thesis we will think of $\alpha$ as a parameter that can take any positive real value. In this thesis we will use the following set of "atomic units" which are very well adapted to the problem we study.
\begin{itemize}
\item The \textit{length unit} is one half the Bohr radius $\ell = \hbar^2 / (2 m_e e_0^2)$.
\item The \textit{energy unit} is $4$ Rydbergs $= 2 m_e e_0^4 / \hbar^2 = 2 m_e \alpha^2 c^2$.
\item The \textit{time unit} is $\tau = \hbar / (4 ~ \mr{Rydbergs}) = \hbar^3 / (2 m_e e_0^4)$.
\end{itemize}
The magnetic $\v{B}$ and electric $\v{E}$ fields are both in units of $e_0 / (\alpha \ell^2)$. It is useful to note that $1/(c \tau) = \alpha / \ell$. 

\clearpage

\clearpage
\begin{centering}
\textbf{SUMMARY}\\
\vspace{\baselineskip}
\end{centering}

Energetic stability of matter in quantum mechanics, which refers to the question of whether the ground state energy of a many-body quantum mechanical system is finite, has long been a deep question of mathematical physics. For a system of many non-relativistic electrons interacting with many nuclei in the absence of electromagnetic fields this question traces back to the seminal work of Freeman Dyson and Andrew Lenard in 1967/68 \cite{DL67, DL68}. In particular, Dyson and Lenard showed the ground state energy of the many-body Schr\"{o}dinger Hamiltonian is bounded below by a constant times the total particle number, regardless of the size of the nuclear charges. This situation changes dramatically when electromagnetic fields and spin interactions are present in the system. Even for a single electron interacting with a single nucleus of charge $Z > 0$ in an external magnetic field with spin-magnetic field interactions included, J\"{u}rg Fr\"{o}hlich, Elliot Lieb, and Michael Loss in 1986 showed that there is no ground state energy if $Z$ exceeds a critical charge $Z_c$ and the ground state energy exists if $Z < Z_c$. In other words, if the nuclear charge is too large, then the one-electron atom is energetically unstable.

Another notion of stability in quantum mechanics is that of dynamic stability, which refers to the question of global well-posedness for a system of partial differential equations modeling the dynamics of many electrons coupled to their self-generated electromagnetic field and interacting with many nuclei. The central motivating question of this PhD thesis is whether energetic stability has any influence over dynamic stability. Concerning this question, we study the quantum mechanical many-body problem of $N \geq 1$ non-relativistic electrons with spin interacting with their self-generated classical electromagnetic field and $K \geq 0$ static nuclei. We model the dynamics of the electrons and their self-generated electromagnetic field using the so-called many-body Maxwell-Pauli equations. The main result is the construction time global, finite-energy, weak solutions to the many-body Maxwell-Pauli equations under the assumption that the fine structure constant $\alpha$ and the nuclear charges are sufficiently small to ensure energetic stability of this system. This result represents an initial step towards understanding the relationship between energetic stability and dynamic stability.

\clearpage
\pagenumbering{arabic}
\setcounter{page}{1} %

\titleformat{\chapter}[display]
{\normalfont\filcenter}{\MakeUppercase\chaptertitlename\ \thechapter}{15pt}{\bfseries{\MakeUppercase{#1}}}  %
\titlespacing*{\chapter}
  {0pt}{0pt}{20pt}	%

\titleformat{\section}{\normalfont\bfseries\filcenter}{\thesection}{1em}{#1}

\titleformat{\subsection}{\normalfont}{\uline{\thesubsection}}{0em}{\uline{\hspace{1em}#1}}

\titleformat{\subsubsection}{\normalfont\itshape}{\thesubsection}{1em}{#1}

\chapter{Introduction}\label{chap:intro}

In the early 20th century, explaining how matter - a soup of point-like, negatively charged electrons interacting with point-like, positively charged nuclei - can exist and give structure to the world around us was one of the many outstanding problems the theory of Quantum Mechanics proposed to resolve. With only the laws of Newtonian physics and classical electromagnetism, even a simple Hydrogen atom cannot exist in a stable configuration. Indeed, if we suppose the electron is moving around the nucleus in a classical orbit, then this electron would constantly give off electromagnetic radiation and lose energy, thereby decreasing its orbital radius until it is sitting on top of the nucleus. From everyday experience we know that this cannot be the case. 
   
In principle the theory Quantum Mechanics resolves this issue. The conclusion obtained by quantum mechanical models of atoms and molecules is that electrons bound to nuclei can only reside in certain discrete energy levels, with the lowest possible energy level configuration referred to as the ground state energy. If this ground state energy exists, then we know the scenario from classical physics cannot hold true. This question of whether there exists a ground state energy for a quantum mechanical model of many electrons and nuclei has become know as the \textit{energetic stability of matter}. 

Energetic stability of matter in quantum mechanics has long been a deep question of mathematical physics. For a system of many non-relativistic electrons interacting with many nuclei in the absence of electromagnetic fields this question goes back to the seminal work of Tosio Kato in 1951 \cite{Kato51} and Freeman Dyson and Andrew Lenard in 1967-1968 \cite{DL67, DL68}. In particular, Kato showed that the ground state energy of the many-body Schr\"{o}dinger Hamiltonian is bounded below, which necessarily implies the ground state energy exists. Furthermore, Dyson and Lenard showed the sharper result that this ground state energy is bounded below by a constant times the total particle number to the first power, regardless of the size of the nuclear charges. However, this picture changes dramatically when electromagnetic fields and spin interactions are included in the problem. Indeed, already in the case of a single electron interacting with a single nucleus of charge $Z > 0$ in an external magnetic field with the spin-magnetic field coupling included there is no ground state energy if $Z$ exceeds a critical charge $Z_c$ and the ground state energy exists if $Z$ is below $Z_c$ \cite{FLL86}. In the case when there is no ground state energy we say the model is \textit{energetically unstable}.

There is another distinct, but seemingly related, notion of stability in quantum mechanics, namely that of dynamic stability. If one considers a system of partial differential equations (PDEs) which models the dynamics of many electrons coupled to their self-generated electromagnetic field and interacting with many nuclei, then a natural question one may ask is: Are these PDEs globally well-posed\footnote{Local well-posedness of a system of PDEs in a certain class of initial data will mean that unique solutions are guaranteed to exist in this class for a finite amount of time and depend continuously on the initial data. If such time local solutions in fact exist for all time, then the term globally well-posed is used.}? We say the model is \textit{dynamically stable} if we have an affirmative answer to this question. If some form of blow up in finite time occurs, then we say the model is \textit{dynamically unstable}. A natural and interesting question to consider is whether dynamic stability depends on energetic stability, and vice versa. That is, does the global well-posedness of such a system of PDEs depend on the existence of the absolute ground state energy? For example, in the case of the one-electron atom in a magnetic field mentioned above, how does the well-posedness of the corresponding dynamical equations (the electromagnetic field now being the self-generated field of the electron) depend on the size of the nuclear charge $Z$? Do we have global well-posedness when $Z < Z_c$? What happens when $Z$ exceeds $Z_c$? These and related questions are the motivation for our main topic of study this thesis, namely the many-body Maxwell-Pauli equations.

\section{The Maxwell-Pauli Equations}\label{sec:MP}

Consider, in three-dimensions, a single non-relativistic electron interacting with a single static\footnote{We consider the nucleus to be static for two reasons. First, the mass of single proton $m_p$ is three orders of magnitude larger than the electron mass $m_e$: $m_p / m_e \simeq 1.83 \times 10^3$. Second, the nuclear radius, $\sim 10^{-14}-10^{-15} ~ \mr{m}$, is much smaller than the Bohr radius, $\sim 10^{-11} ~ \mr{m}$. If the nucleus played an crucial role dynamically, atoms and molecules would look much different with electron orbital radii being much smaller than the Bohr radius.} nucleus of charge $Z > 0$ sitting at the origin in the absence of any electromagnetic fields. The electron has charge $- 1$ in our units and interacts with the positively charged nucleus via the usual Coulomb potential $- Z / |\v{x}|$. Quantum Mechanics dictates that the state of this electron at time $t \in \R$ is a normalized wavefunction $\psi (t) \in L^2 (\R^3 ; \C^2)$ and that the time evolution of this electron is governed by the Schr\"{o}dinger equation:
\begin{align*}
i \partial_t \psi = \left( \v{p}^2 - \frac{Z}{|\v{x}|} \right) \psi .
\end{align*} 
Above $\v{p} = - i \nabla$ is the quantum mechanical canonical momentum and $\v{p}^2 = - \Delta = - \sum_{j = 1}^3 \partial_j^2$ is the quantum mechanical kinetic energy operator. That we assume the wavefunction is normalized, $\| \psi (t) \|_2 = 1$, is to remain consistent with the interpretation of $|\psi (t)|^2$ as a probability density\footnote{Note $\| \psi \|_{2}$ is easily seen to be preserved under the Schr\"{o}dinger time evolution.}.

In the presence of an electromagnetic field $(\v{E} , \v{B})$ generated by electromagnetic potentials $(\varphi , \v{A})$, i.e. $(\v{E} , \v{B}) = ( - \nabla \varphi - \alpha \partial_t \v{A} , \curl{\v{A}} )$, the principal of minimal coupling dictates that the canonical momentum $\v{p}$ of the electron should be replaced by $\v{p} + \v{A}$. This leads us to the magnetic Schr\"{o}dinger equation:
\begin{align*}
i \partial_t \psi = \left( (\v{p} + \v{A})^2 - \frac{1}{\alpha} \varphi - \frac{Z}{|\v{x}|} \right) \psi .
\end{align*}
However, the magnetic Schr\"{o}dinger equation fails to take the coupling of the electron spin and the magnetic field into account. In order to do this, we replace the momentum $\v{p} + \v{A}$ by the \textbf{Pauli operator}: $\gvsig \cdot (\v{p} + \v{A})$, where $\gvsig = (\sigma^1 , \sigma^2 , \sigma^3) \in \R^3 \otimes M_{2 \times 2} (\C)$ is the vector of \textbf{Pauli matrices}\footnote{Both the Pauli operator and the Pauli matrices are named after the Austrian physicist Wolfgang Pauli (1900-1958).}. The Pauli matrices $\sigma^1$, $\sigma^2$, and $\sigma^3$ are $2 \times 2$ Hermitian matrices assumed to satisfy the commutation relations $[\sigma^j , \sigma^k] = 2 i \epsilon_{jk\ell} \sigma^{\ell}$ and anticommutation relations $\left\lbrace \sigma^j , \sigma^k \right\rbrace = 2 \delta_{jk} I$, for $j , k , \ell \in \left\lbrace 1 , 2 ,3 \right\rbrace$. The most common choice for the three Pauli matrices are
\begin{align*}
\sigma^1 = \left( \begin{array}{cc}
0 & 1 \\
1 & 0 
\end{array} \right) , ~~~~~~ \sigma^2 = \left( \begin{array}{cc}
0 & -i \\
i & 0 
\end{array} \right) , ~~~~~~ \sigma^3 = \left( \begin{array}{cc}
1 & 0 \\
0 & -1 
\end{array} \right) .
\end{align*}
However, we will never need to work with an explicit representation of these matrices in this thesis. All together, these replacements lead us to the \textbf{Pauli equation}\footnote{In the Pauli equation we are ignoring the spin-spin interactions of the electron and the nucleus. Here, we would have to add the term $\gvsig^e \cdot \gvsig^n |\v{x}|^{-3} - 3 (\gvsig^e \cdot \v{x})(\gvsig^n \cdot \v{x}) |\v{x}|^{-5}$ where $\gvsig^e$ and $\gvsig^n$ are the Pauli matrices for the electronic and nuclear spins, respectively. The lack of integrability of $|\v{x}|^{-3}$ at the origin always causes energetic instability (see Chapter \ref{chap:stability}), and for this reason we ignore such interactions. See \cite{FLL86, Lieb76} for more detail.}:
\begin{align}\label{eq:Pauli_equation}
i \partial_t \psi = \left( [\gvsig \cdot (\v{p} + \v{A})]^2 - \frac{1}{\alpha} \varphi - \frac{Z}{|\v{x}|} \right) \psi .
\end{align}

Any charge density $\rho (t) : \R^3 \rightarrow \R$ and current density $\v{J} (t) : \R^3 \rightarrow \R$, such as the ones generated by our dynamic electron considered above, will necessarily generate an electromagnetic field which we will describe classically\footnote{This is opposed to the Quantum Electrodynamics point of view where one quantizes the electromagnetic field and describes $\v{E}$ and $\v{B}$ as operators on Fock space.} by the electric field $\v{E} (t) : \R^3 \rightarrow \R^3$ and the magnetic field $\v{B} (t) : \R^3 \rightarrow \R^3$. The dynamics of the electric and magnetic fields are governed by Maxwell's equations, which in our units read
\begin{align}\label{eq:Maxwell_Equations}
\left\lbrace \begin{array}{ll}
\diver{\v{E}} = 4 \pi \rho , & \diver{\v{B}} = 0 \\
\curl{\v{E}} = - \alpha \partial_t \v{B} ,  & \curl{\v{B}} = 4 \pi \alpha \v{J} + \alpha \partial_t \v{E} .
\end{array} \right.
\end{align}
An important consequence of Maxwell's equations is the charge continuity equation: 
\begin{align}\label{eq:charge_continuity}
\diver{\v{J}} + \partial_t \rho = 0 ,
\end{align}
which is obtained by taking the divergence of $\curl{\v{B}} = 4 \pi \alpha \v{J} + \alpha \partial_t \v{E}$ and using that $\diver{\v{E}} = 4 \pi \rho$. Moreover, the total energy in this electromagnetic field $F [\v{B} , \v{E}]$ is given by
\begin{align}\label{def:field_energy_general}
F [ \v{B} , \v{E} ] (t) = \frac{1}{8 \pi \alpha^2} \int_{\R^3} \left( |\v{B} (t)|^2 + | \v{E} (t)|^2 \right) ,
\end{align}  
and the force $\v{F} (t)$ exerted on another charge density $\rho^{\ast} (t)$ and current density $\v{J}^{\ast} (t)$ by this electromagnetic field is determined by the Lorentz force law:
\begin{align}\label{eq:Lorentz_force_law}
\v{F} (t) = \int_{\R^3} \left[ \rho^{\ast} (t) \v{E} (t) + \alpha \v{J}^{\ast} (t) \wedge \v{B} (t) \right] .
\end{align}
It is important to mention that equations (\ref{eq:Maxwell_Equations}), (\ref{eq:charge_continuity}), and (\ref{eq:Lorentz_force_law}) encapsulate all the phenomena of classical electromagnetism. Lastly, using the Helmholtz Theorem for vector fields, the electromagnetic field $(\v{E} , \v{B})$ satisfying Maxwell's equations (\ref{eq:Maxwell_Equations}) may be described in terms of an electrostatic potential $\varphi (t) : \R^3 \rightarrow \R$ and a magnetic vector potential $\v{A} (t) : \R^3 \rightarrow \R^3$ via
\begin{align*}
\v{E} = - \nabla \varphi - \alpha \partial_t \v{A} \hspace{1cm} \text{and} \hspace{1cm} \v{B} = \curl{\v{A}}. 
\end{align*}

As an initial guess for a model of the time evolution of our non-relativistic electron interacting with its self-generated electromagnetic field and a static nucleus of charge $Z > 0$, we consider coupling the Pauli equation (\ref{eq:Pauli_equation}) to the Maxwell's equations (\ref{eq:Maxwell_Equations}). This yields the system
\begin{align}\label{eq:MPC_general}
\left\lbrace \begin{array}{l}
i \partial_t \psi = \left( [\gvsig \cdot (\v{p} + \v{A})]^2 - \dfrac{Z}{|\v{x}|} - \dfrac{1}{\alpha} \varphi + F[\v{B} , \v{E}] \right) \psi  \\[1.2ex]
\diver{\v{E}} = 4 \pi \rho , \hspace{5mm} \diver{\v{B}} = 0  \\
\curl{\v{E}} = - \alpha \partial_t \v{B} , \hspace{5mm} \curl{\v{B}} = 4 \pi \alpha \v{J} + \alpha \partial_t \v{E} \\
\v{E} = - \nabla \varphi - \alpha \partial_t \v{A} , \hspace{5mm} \v{B} = \curl{\v{A}} .
\end{array} \right.
\end{align}
The consideration of the system (\ref{eq:MPC_general}) then begs the question of what precisely is the probability current density generated by our dynamic electron coupled to its self-generated electromagnetic field? One way to reason the precise form of the probability current density is to demand that the total energy is conserved. But what should be the total energy? A reasonable guess would be the quantity
\begin{align*}
\| \gvsig \cdot (\v{p} + \v{A}) \psi \|_2^2 - Z \langle \psi , |\v{x}|^{-1} \psi \rangle_{L^2} - \frac{1}{\alpha} \langle \psi , \varphi \psi \rangle_{L^2} + F[\v{B} , \v{E}] ,
\end{align*}
where $(\psi , \v{E} , \v{B})$ satisfy (\ref{eq:MPC_general}). This expression is simply the expectation value of the Hamiltonian appearing on the right hand side of the first equation in (\ref{eq:MPC_general}) in a normalized state $\psi$. Differentiating the previous guess with respect to time and using the identity $\partial_t u_{\mr{EM}} + \diver{\v{S}} = - \v{E} \cdot \v{J}$, where $u_{\mr{EM}} = (|\v{B}|^2 + |\v{E}|^2)/(8 \pi)$ is the field energy density and $\v{S} = (\v{E} \wedge \v{B})/(4 \pi \alpha)$ is the Poynting vector\footnote{See \cite[Chapter 27]{F11_VolII} for a very enlightening discussion regarding energy conservation in electromagnetism.}, yields
\begin{align*}
& \frac{\dd}{\dd t} \left( \| \gvsig \cdot (\v{p} + \v{A}) \psi \|_2^2 - Z \langle \psi , |\v{x}|^{-1} \psi \rangle_{L^2} - \frac{1}{\alpha} \langle \psi , \varphi \psi \rangle_{L^2} + F[\v{B} , \v{E}]  \right)  \\
& \hspace{1cm} = \int_{\R^3} \left( \partial_t \v{A} \cdot 2 \re{ \langle \gvsig \psi , \gvsig \cdot (\v{p} + \v{A}) \psi \rangle_{\C^2} } - \frac{1}{\alpha} |\psi|^2 \partial_t \varphi - \frac{1}{\alpha^2} \v{E} \cdot \v{J} \right) . 
\end{align*}
Writing $\v{E}$ in term of potentials and using the charge continuity equation (\ref{eq:charge_continuity}) allows us to reduce the previous expression to
\begin{align*}
& \frac{\dd}{\dd t} \left( \| \gvsig \cdot (\v{p} + \v{A}) \psi \|_2^2 - Z \langle \psi , |\v{x}|^{-1} \psi \rangle_{L^2} - \frac{1}{\alpha} \langle \psi , \varphi \psi \rangle_{L^2} + F[\v{B} , \v{E}]  \right)  \\
& \hspace{1cm} = \int_{\R^3} \left( \partial_t \v{A} \cdot \left( 2 \re{ \langle \gvsig \psi , \gvsig \cdot (\v{p} + \v{A}) \psi \rangle_{\C^2} } + \frac{1}{\alpha} \v{J} \right)  - \frac{1}{\alpha} |\psi|^2 \partial_t \varphi + \frac{1}{\alpha^2} \varphi \partial_t \rho \right) .
\end{align*}
From this expression we are encouraged to conclude that the \textbf{Pauli probability current density} for our electron satisfying (\ref{eq:MPC_general}) is 
\begin{align}\label{def:Pauli_prob_current}
\v{J}_{\mr{P}} [\psi , \v{A}]  = - 2 \alpha \re{ \langle \gvsig \psi , \gvsig \cdot (\v{p} + \v{A}) \psi \rangle_{\C^2} } .
\end{align}
From the charge continuity equation (\ref{eq:charge_continuity}) we then conclude the charge density is $\rho = - \alpha |\psi|^2$. Hence, for the choice $\v{J} = \v{J}_{\mr{P}}$, the total conserved energy for the system (\ref{eq:MPC_general}) is\footnote{Note (\ref{eq:total_energy_general}) does not contain the energy corresponding to the electrostatic potential $\varphi$ in the Pauli equation.}
\begin{align}\label{eq:total_energy_general}
E [\psi , \v{B} , \v{E}] = \| \gvsig \cdot (\v{p} + \v{A}) \psi \|_2^2 - Z \langle \psi , |\v{x}|^{-1} \psi \rangle_{L^2} + F[\v{B} , \v{E}] .
\end{align}

Another way of deriving the Pauli probability current density $\v{J}_{\mr{P}} [\psi , \v{A}]$ is to fix $\psi$ and consider minimizing the sum of the Pauli kinetic energy with the magnetic field energy,
\begin{align*}
\| \gvsig \cdot (\v{p} + \v{A}) \psi \|_2^2 + \frac{1}{8 \pi \alpha^2} \int_{\R^3} | \v{B} |^2 ,
\end{align*}
with respect to all $\v{A}$'s which generate a finite-energy magnetic field $\v{B} = \curl{\v{A}}$. This is a straight forward variational calculation that yields $\curl{\v{B}} = 4 \pi \alpha \v{J}_{\mr{P}} [\psi , \v{A}]$. Yet another way to arrive at (\ref{def:Pauli_prob_current}) is by considering the non-relativistic limit of the so-called Maxwell-Dirac equations. We reserve this calculation for Chapter \ref{chap:othermodels}, \S\ref{sec:MD}. Moreover, using the relation $\sigma^i \sigma^j = \delta_{ij} I + i \epsilon_{ijk} \sigma^k$, we can rewrite $\v{J}_{\mr{P}} [\psi , \v{A}]$ as
\begin{align*}
\v{J}_{\mr{P}} [\psi , \v{A}]  = - 2 \alpha \re{ \langle \psi , (\v{p} + \v{A}) \psi \rangle_{\C^2} } - \alpha \curl{ \langle \psi , \gvsig \psi \rangle_{\C^2} } .
\end{align*}
One may recognize that the first term on the right hand side of the previous expression is the usual probability current associated with the magnetic Schr\"{o}dinger equation. The new addition to the probability current in the Pauli case is the \textbf{spin current}: $\curl{\langle \psi , \gvsig \psi \rangle_{\C^2}}$. From the point-of-view of mathematical analysis, the spin current produces unexpected difficulties when studying the well-posedness question for (\ref{eq:MPC_general}). We will elaborate on this observation in great detail in Chapter \ref{chap:othermodels}, \S\ref{sec:MS}.

As it stands, the model (\ref{eq:MPC_general}) does not appear to be physically reasonable due to the presence of the electrostatic potential $\varphi$ in the Pauli equation. Indeed, consider deriving the charge density $\rho$ as we did the Pauli probability current density by fixing $\psi$ and minimizing the sum of the electrostatic energy with the electric field energy, 
\begin{align*}
- \frac{1}{\alpha} \langle \psi , \varphi \psi \rangle_{L^2} + \frac{1}{8 \pi \alpha^2} \int_{\R^3} | \nabla \varphi |^2 , 
\end{align*}
with respect to all electrostatic potentials $\varphi$ which generate a finite-energy electric field. The result gives $- \Delta \varphi = 4 \pi \alpha |\psi |^2$. However, the charge density $\alpha | \psi |^2$ in this equation has the \textit{wrong sign}; it should be $- \alpha | \psi |^2$ for our electron of charge $- 1$. This suggests that it is physically unreasonable to consider the system (\ref{eq:MPC_general}). Another argument against the model (\ref{eq:MPC_general}) is that, from a physical point of view, it is wrong to include self-interactions in this context. Indeed, the Coulomb self-interaction $\varphi$ could be perceived as a mean field originating from the Coulomb interactions between the particles present in a full many-body system. However, when the system consists of only of a single electron there simply are no other particles to interact with. 

To arrive at a more physically reasonable model, we first rewrite Maxwell's equations (\ref{eq:Maxwell_Equations}) in terms of the electromagnetic potentials $(\varphi , \v{A})$. This yields
\begin{align}\label{eq:Maxwell_Equations_potentials}
\left\lbrace \begin{array}{l}
- \Delta \varphi - \alpha \partial_t \diver{\v{A}} = - 4 \pi \alpha |\psi|^2  \\
\square \v{A} + \nabla ( \diver{\v{A}} + \alpha \partial_t \varphi ) = 4 \pi \alpha \v{J}_{\mr{P}} [\psi , \v{A}] ,
\end{array} \right.
\end{align}
where $\square = \alpha^2 \partial_t^2 - \Delta$ is the d'Alembert wave operator. It is clear that the potential functions $(\varphi , \v{A})$ are not unique and they only determine a given electromagnetic field $(\v{E} , \v{B})$ up to a gauge transformation given by 
\begin{align*}
\varphi \longmapsto \varphi - \alpha \partial_t \zeta \hspace{1cm} \text{and} \hspace{1cm} \v{A} \longmapsto \v{A} + \nabla \zeta ,
\end{align*}
where $\zeta : \R^3 \rightarrow \R$ is some gauge function. Typically to solve (\ref{eq:Maxwell_Equations_potentials}) one fixes a gauge for $(\varphi , \v{A})$ where the choice of $\zeta$ is usually implicit by imposing a constraint on the potential functions. Some very common choices for gauge-fixing include the Lorenz gauge\footnote{The Lorenz gauge condition is often mistakenly attributed to the Dutch physicist Hendrik Lorentz (1853-1928) which the Lorentz force is named after. It is in fact named after the Danish physicist Ludvig Lorenz (1829-1891).}: $\diver{\v{A}} + \alpha \partial_t \varphi = 0$, and the \textbf{Coulomb gauge}: $\diver{\v{A}} = 0$. In particular, (\ref{eq:Maxwell_Equations_potentials}) in the Coulomb gauge reads
\begin{align}\label{eq:Maxwell_Equations_potentials_coulombgauge}
\left\lbrace \begin{array}{l}
- \Delta \varphi = - 4 \pi \alpha |\psi|^2 \\
\square \v{A} = 4 \pi \alpha \v{J}_{\mr{P}}[\psi , \v{A}] - \alpha \partial_t \nabla \varphi .
\end{array} \right.
\end{align}
Assuming $\varphi$ has appropriate decay properties at $+ \infty$, the first equation in (\ref{eq:Maxwell_Equations_potentials_coulombgauge}) can be easily solved using the Newtonian potential:
\begin{align*}
\varphi (\v{x}) = - 4 \pi \alpha ( (- \Delta)^{-1} |\psi|^2 ) (\v{x}) = - \alpha \int_{\R^3} \frac{~|\psi (\v{y}) |^2~}{|\v{x} - \v{y}|} \dd \v{y} .  
\end{align*}
Using this together with the charge continuity equation (\ref{eq:charge_continuity}), the second equation in (\ref{eq:Maxwell_Equations_potentials_coulombgauge}) can be reduced further and one finds
\begin{align}\label{eq:waveeqn_A_coulombgauge}
\square \v{A} = 4 \pi \alpha \Hproj{ \v{J}_{\mr{P}} [\psi , \v{A}] } ,
\end{align}
where 
\begin{align}\label{def:projection_op}
\Hproj = \curl{ (- \Delta)^{-1} \curl{} } = 1 + \nabla (- \Delta )^{-1} \diver{}
\end{align}
is the \textbf{Leray-Helmholtz projection} onto divergence-free vector fields. 

With equation (\ref{eq:waveeqn_A_coulombgauge}) in mind, we propose to model the time evolution of our non-relativistic electron interacting with its self-generated radiation field and a static nucleus of charge $Z > 0$ via coupling the Pauli equation (\ref{eq:Pauli_equation}), \textit{without the electrostatic potential} $\varphi$, to the wave equation (\ref{eq:waveeqn_A_coulombgauge}) for $\v{A}$. Such a coupling results in a system of nonlinear PDEs that we call the \textbf{Maxwell-Pauli-Coulomb equations} (abbr.~\textbf{MPC equations}):
\begin{align}\label{eq:MPC}
\left\lbrace \begin{array}{l}
i \partial_t \psi = \left( [\gvsig \cdot (\v{p} + \v{A})]^2 - \dfrac{Z}{|\v{x}|} \right) \psi  \\[1.2ex]
\square \v{A} = 4 \pi \alpha \Hproj{ \v{J}_{\mr{P}} [\psi , \v{A}] }  \\
\diver{\v{A}} = 0 .
\end{array} \right. 
\end{align}
where $\v{J}_{\mr{P}} [ \psi , \v{A} ]$ is given by (\ref{def:Pauli_prob_current}). An important special case of (\ref{eq:MPC}) are the \textbf{Maxwell-Pauli equations} (abbr.~\textbf{MP equations}):
\begin{align}\label{eq:MP}
\left\lbrace \begin{array}{l}
i \partial_t \psi = [\gvsig \cdot (\v{p} + \v{A})]^2 \psi  \\
\square \v{A} = 4 \pi \alpha \Hproj{ \v{J}_{\mr{P}} [\psi , \v{A}] }  \\
\diver{\v{A}} = 0  .
\end{array} \right. 
\end{align}
The MP equations serve as a model for a single non-relativistic electron interacting with its self-generated radiation field without any nuclei present. Note we do not include the electromagnetic field energy in (\ref{eq:MPC}) or (\ref{eq:MP}) since it can be removed via the position-independent gauge transformation $\psi (t) \mapsto \exp{ \left\lbrace -i \int_0^t F [\v{A} , \partial_t \v{A}] (s) \dd s \right\rbrace} \psi (t)$. 

The system that we will be primarily concerned with in this thesis is actually a many-body generalization of (\ref{eq:MPC}). Consider $N \geq 1$ non-relativistic, electrons interacting with their self-generated electromagnetic field and $K \geq 0$ static nuclei with nuclear charges $\Z = (Z_1 , \cdots , Z_K) \in [0 , \infty)^K$. The nuclei are at distinct positions labeled by $\ul{\v{R}} = (\v{R}_1 , \cdots , \v{R}_K) \in \R^{3K}$, $\v{R}_i \neq \v{R}_j$ for $i \neq j$. The $N$ electrons and $K$ nuclei interact via Coulomb forces that are described by the two-body potential interaction 
\begin{align}\label{def:totelectrostatpot}
V ( \ul{\v{R}} , \Z ) (\ul{\v{x}}) = \sum_{1 \leq i < j \leq N} \frac{1}{|\v{x}_i - \v{x}_j|} - \sum_{i = 1}^N \sum_{j = 1}^K \frac{Z_j}{| \v{x}_i - \v{R}_j |} + \sum_{1 \leq i < j \leq K} \frac{Z_i Z_j}{|\v{R}_i - \v{R}_j|} ,
\end{align}
where on the right hand side of (\ref{def:totelectrostatpot}) appears the electron-electron, electron-nuclei, and nuclei-nuclei Coulomb interactions, respectively, and $\ul{\v{x}} = (\v{x}_1 , \cdots , \v{x}_N) \in \R^{3N}$, $\v{x}_i \neq \v{x}_j$ for $i \neq j$, is the collection the electron coordinates. In this situation we are lead to a natural generalization of (\ref{eq:MPC}), referred to as the \textbf{many-body Maxwell-Pauli equations} (abbr.~\textbf{MBMP equations}), which read
\begin{align}\label{eq:MBMP} 
\left\lbrace \begin{array}{l}
i \partial_t \psi = H_{\mr{P}} (\v{A}) \psi  \\
\square \v{A} = 4 \pi \alpha \Hproj{ \J_{\!\mr{P}} [ \psi , \v{A} ] }  \\
\diver{\v{A}} = 0 .
\end{array} \right. 
\end{align}
In the first equation in (\ref{eq:MBMP}) $H_{\mr{P}} (\v{A})$ is the \textbf{many-body Pauli Hamiltonian}:
\begin{align}\label{def:manybody_Pauli_Hamiltonian}
H_{\mr{P}} (\v{A}) = \sum_{j = 1}^N \T_j (\v{A}) + V(\underbar{\v{R}} , \Z) ,
\end{align}
which acts on $\bigotimes^N L^2 (\R^3 ; \C^2)$-functions. In (\ref{def:manybody_Pauli_Hamiltonian}) $\T_j (\v{A}) = [ \gvsig_j \cdot ( \v{p}_j + \v{A}_j ) ]^2$ is the Pauli operator corresponding to the $j^{\mr{th}}$ electron\footnote{The notation for $\T_j$ comes with the understanding that $\T_j \equiv I \otimes \cdots \otimes [\gvsig \cdot (\v{p}_j + \v{A}_j)]^2 \otimes \cdots \otimes I$ where the $[\gvsig \cdot (\v{p}_j + \v{A}_j)]^2$ is appearing in the $j^{\mr{th}}$ position of the tensor product.} and $\v{A}_j (\ul{\v{x}}) = \v{A} (\v{x}_j)$ and $\v{p}_j = - i \nabla_{\v{x}_j}$. In the second equation in (\ref{eq:MBMP}) $\J_{\!\mr{P}} [\psi , \v{A}]$ is the total probability current density of the $N$ electrons and is given by
\begin{align}\label{def:prob_current_compact}
\J_{\!\mr{P}} [\psi , \v{A}] (\v{x}) = - 2 \alpha \sum_{j = 1}^N \re{ \int \langle \gvsig \psi_{\ul{\v{z}}_j'} , \gvsig \cdot (\v{p} + \v{A}) \psi_{\ul{\v{z}}_j'} \rangle_{\C^2} (\v{x}) \dd \ul{\v{z}}_j' } ,
\end{align}
where, for $j \in \left\lbrace 1 , \cdots , N \right\rbrace$, $\v{z}_j = (\v{x}_j , s_j) \in \R^3 \times \left\lbrace \uparrow , \downarrow \right\rbrace$ is the $j^{\mr{th}}$ electron's position coordinate and spin state, and $\psi_{\ul{\v{z}}_j'} : \R^3 \rightarrow \C^2$ is the spinor defined by 
\begin{align*}
\psi_{\ul{\v{z}}_j'} (\v{x} , s) = \psi (\v{z}_1 ; \cdots ; \v{x} , s ; \cdots ; \v{z}_N) ,
\end{align*}
where $s \in \left\lbrace \uparrow , \downarrow \right\rbrace$, $\ul{\v{z}}_j' = (\v{z}_1 , \cdots , \v{z}_{j-1} , \v{z}_{j+1} , \cdots , \v{z}_N)$, and $\dd \v{z}_i \equiv \sum_{s_i \in \{ \uparrow , \downarrow \} } \dd \v{x}_i$. As electrons are Fermions, we will demand that they satisfy the \textbf{Pauli exclusion principle}, meaning the electronic wavefunction $\psi$ satisfies 
\begin{align*}
\psi (\v{z}_1 ; \cdots ; \v{z}_i ; \cdots ; \v{z}_j ; \cdots \v{z}_N) = - \psi (\v{z}_1 ; \cdots ; \v{z}_j ; \cdots ; \v{z}_i ; \cdots \v{z}_N) .
\end{align*}
Loosely speaking, the Pauli exclusion principle introduces a repulsive interaction between electrons in the sense that no two particles are allowed in the same state. Mathematically we state this restriction by writing $\psi (t) \in \BWN{L^2(\R^3;\C^2)}$ where $\BWN{L^2(\R^3;\C^2)}$ is the closed subspace of $\bigotimes^N L^2 (\R^3 ; \C^2)$ consisting of completely antisymmetry functions. 

We conclude this section by highlighting the two important conserved quantities associated with the system (\ref{eq:MBMP}): the $L^2$-norm of $\psi$ and the total energy. To see that $\| \psi \|_2$ is conserved, we simply note that, formally, $H_{\mr{P}} (\v{A})$ is a self-adjoint operator and therefore
\begin{align*}
\frac{\dd }{\dd t} \| \psi \|_{2}^2 = 2 \re{ \langle \psi , - i H_{\mr{P}} (\v{A}) \psi \rangle_{L^2} } = 0 .
\end{align*}
The total Pauli energy $E_{\mr{P}} [\psi , \v{A} , \partial_t \v{A}]$ is defined as
\begin{align}\label{eq:mbmp_energy}
E_{\mr{P}} [\psi , \v{A} , \partial_t \v{A}] = T_{\mr{P}} [\psi , \v{A}] + V [\psi] + F[\v{A} , \partial_t \v{A}] ,
\end{align}
where $T_{\mr{P}} [\psi , \v{A}]$ is the total kinetic energy, given by
\begin{align}\label{def:total_kinetic}
T_{\mr{P}} [\psi , \v{A}] = \sum_{j = 1}^N \| \gvsig_j \cdot (\v{p}_j + \v{A}_{j}) \psi \|_{2}^2 ,
\end{align}
$V[\psi]$ is the total Coulomb potential energy, given by 
\begin{align}\label{def:tot_potential_energy}
V[\psi] = \langle \psi , V ( \ul{\v{R}} , \Z ) \psi \rangle_{L^2} ,
\end{align}
and $F [\v{A} , \partial_t \v{A}]$ is the total electromagnetic field energy, given by 
\begin{align}\label{def:field_energy_potentials}
F [ \v{A} , \partial_t \v{A} ] = \frac{1}{8 \pi \alpha^2} \left( \| \v{B} \|_{2}^2 + \alpha^2 \| \partial_t \v{A} \|_{2}^2 \right) .
\end{align}
We note that $E_{\mr{P}}$ is simply the expectation value of the Hamiltonian $H_{\mr{P}} (\v{A})$ plus the electromagnetic field energy. Checking that $E_{\mr{P}}$ is conserved is also straightforward. Since $\Hproj{}$ is (formally) a self-adjoint operator, we have
\begin{align*}
\frac{\dd E_{\mr{P}}}{ \dd t} & = 2 \sum_{j = 1}^N \re{ \langle \gvsig_j \cdot (\v{p}_j + \v{A}_j) \phi , ( \gvsig_j \cdot \partial_t \v{A}_j ) \phi \rangle_{L^2} } + \partial_t F [\v{A} , \partial_t \v{A}] \\
& = - \frac{1}{\alpha} \langle \J_{\!\mr{P}} [\phi , \v{A}] , \partial_t \v{A} \rangle_{L^2} + 2 \frac{1}{8 \pi \alpha^2} \langle \square \v{A} , \partial_t \v{A} \rangle_{L^2} \\
& =  - \frac{1}{\alpha} \langle \J_{\!\mr{P}} [\phi , \v{A}] , \Hproj{ \partial_t \v{A} } \rangle_{L^2} + \frac{1}{\alpha} \langle \Hproj{ \J_{\!\mr{P}} [\phi , \v{A}] } , \partial_t \v{A} \rangle_{L^2} = 0 .
\end{align*}
It will be important for our study of (\ref{eq:MBMP}) to define the absolute ground state energy associated with $E_{\mr{P}} [\psi , \v{A} , \partial_t \v{A}]$. For this we introduce the function space
\begin{align}\label{def:function_space_C}
\fC_N := \left\lbrace (\psi , \v{A}) \in \BWN{H^1(\R^3;\C^2)} \times \dot{H}^1 (\R^3 ; \R^3) ~ : ~ \| \psi \|_2 = 1 , ~ \diver{\v{A}} = 0 \right\rbrace .
\end{align}
The space $\fC_N$ should be thought of as all pairs $(\psi , \v{A})$ for which $E_{\mr{P}} [ \psi , \v{A} , \v{0} ]$ is finite. The absolute ground state energy $E_{\mr{P}}^{\mr{G}}$ associated with $E_{\mr{P}} [\psi , \v{A} , \partial_t \v{A}]$ is then defined as
\begin{align}\label{def:abs_grnd_state_energy}
E_{\mr{P}}^{\mr{G}}  = E_{\mr{P}}^{\mr{G}} (N , K, \Z , \alpha) = \inf{ \left\lbrace E_{\mr{P}} [\psi , \v{A} , \v{0}] ~ : ~ (\psi , \v{A}) \in \fC_N , ~~ \ul{\v{R}} \right\rbrace } .
\end{align}  

\section{Motivation and Main Result}\label{sec:main_result}

As mentioned in the introduction to this thesis, our study of the MBMP equations (\ref{eq:MBMP}) is primarily motivated by the question of whether there is a relationship between energetic stability of matter and the dynamic stability of matter. For the stability of the ground state energy (\ref{def:abs_grnd_state_energy}), we refer to the developments in \cite{FLL86, LL86, LY86, F95, LLS95}. In particular, Fr\"{o}hlich, Lieb, and Loss in 1986 showed that the absolute ground state energy $E_{\mr{P}}^{\mr{G}}$ in the ($N=K=1$)-case is finite when $Z$ is below a critical charge $Z_c$ and $- \infty$ when $Z$ exceeds $Z_c$ \cite{FLL86}. In other words, the one-electron atom in a magnetic field is not energetically stable when the atomic number $Z$ is too large. More generally, Lieb, Loss, and Solovej in 1995 proved that, if $\alpha \leq 0.06$ and $\alpha^2 \max{\Z} \leq 0.041$, the absolute ground state energy $E_{\mr{P}}^{\mr{G}}$ is bounded below by $- C (N+K)$, where $C$ is a constant depending \textit{only} on $\alpha$ and $\Z$ \cite{LLS95}. We note the range of $\alpha$ includes the actual physical value $\alpha \simeq 1/137 \simeq 0.007$, and for $\alpha \simeq 1/137$, the largest nuclear charge allowed is roughly $769$. (See Chapter \ref{chap:stability} for a more detailed discussion of the energetic stability of matter.)

Considering these results on energetic stability, we find it natural to ask whether finiteness of the absolute ground state energy $E_{\mr{P}}^{\mr{G}}$ has any influence on the well-posedness of the corresponding dynamical equations. More specifically, does the time interval of existence of solutions to (\ref{eq:MBMP}) depend on $Z$ in the $(N=K=1)$-case and, more generally, the size of $\max{\Z}$ and $\alpha$ in the $(\max{\left\lbrace N , K\right\rbrace} > 1)$-case? Is the time interval of existence infinite if $\alpha$ and $\max{\Z}$ are sufficiently small? Does blow-up in finite time occur if $\alpha$ or $\max{\Z}$ are large? A natural starting point to answering such questions is to develop a local well-posedness theory for solutions to (\ref{eq:MP}) (single electron, no nuclei). To our knowledge, there are no such results in the literature and, as we discuss in detail in Chapter \ref{chap:othermodels}, we do not succeed in the endeavor to prove the local well-posedness of (\ref{eq:MP}). Therefore, instead of attempting to prove local well-posedness of (\ref{eq:MP}) and, more generally, (\ref{eq:MBMP}), we've turned our attention to constructing time global \textit{weak} solutions to (\ref{eq:MBMP}). In fact, the main result of this thesis is the global existence of finite-energy weak solutions to (\ref{eq:MBMP}) under the assumption that $\alpha$ and $\alpha^2 \max{\Z}$ are small enough to ensure $E_{\mr{P}}^{\mr{G}} > - \infty$. 

\begin{thm}[Global Finite-Energy Weak Solutions]\label{thm:weak_solns_MBMP}
Suppose $\alpha$ and $\alpha^2 \max{\Z}$ are sufficiently small to ensure $E_{\mr{P}}^{\mr{G}} > - \infty$. Then, given 
\begin{align*}
(\psi_0 ,\v{a}_0 , \dot{\v{a}}_0) \in \BWN{H^1(\R^3;\C^2)} \times  H^1 (\R^3 ; \R^3) \times L^2 (\R^3 ; \R^3)
\end{align*}
with $\| \psi_0 \|_2 = 1$ and $\diver{\v{a}_0} = \diver{\dot{\v{a}}_0} = 0$, there exists at least one finite-energy weak solution
\begin{align*}
(\psi , \v{A} , \partial_t \v{A}) \in C^{\mr{w}} ( \R_+ ; \BWN{H^1(\R^3;\C^2)} \times H^1 (\R^3 ; \R^3)  \times L^2 (\R^3 ; \R^3) )
\end{align*}
to (\ref{eq:MBMP}) such that the initial conditions $(\psi (0) , \v{A} (0) , \partial_t \v{A} (0)) = (\psi_0 , \v{a}_0 , \dot{\v{a}}_0)$ are satisfied.
\end{thm}  

If we are considering the MPC system (\ref{eq:MPC}), namely, the ($N = K = 1$)-case of (\ref{eq:MBMP}), then the hypothesis of Theorem \ref{thm:weak_solns_MBMP} changes to $Z < Z_c$ where $Z_c$ is the critical charge and is defined by (\ref{def:Zc}) (see Chapter \ref{chap:stability} for the context in which $Z_c$ is defined). Moreover, if we are considering just the MP equations (\ref{eq:MP}), then no additional assumptions are needed (there are no nuclear charges present and we do not need to assume $\alpha$ is sufficiently small). The solution obtained in Theorem \ref{thm:weak_solns_MBMP} is a \textit{weak} solution, but does indeed have finite energy, i.e., $(\psi , \v{A})$ belong to the class of functions $\fC_N$ and $\partial_t \v{A} \in L^2 (\R^3 ; \R^3)$ (see Chapter \ref{chap:prelim} for the precise definition of a weak solution to (\ref{eq:MBMP})). We point out that we do not manage to prove uniqueness. This is an artifact of our proof strategy discussed in the following section. 

\section{Proof Strategy and the $\varepsilon$-Modified System}\label{sec:proof_strat}

In order to prove Theorem \ref{thm:weak_solns_MBMP} concerning the existence of global weak solutions to (\ref{eq:MBMP}), we combine the contraction mapping scheme found in \cite{NW05} with ideas from the 1995 work on the so-called Maxwell-Schr\"{o}dinger equations (\ref{eq:MS}) by Guo, Nakamitsu, and Strauss \cite{guo1995}. As we discuss in Chapter \ref{chap:othermodels}, in the latter article, the authors consider an approximate system, the $\varepsilon$-modified Maxwell-Schr\"{o}dinger equations (\ref{eq:MS_epsilon}), and prove that global, finite-energy solutions to the $\varepsilon$-modified system converge, as $\varepsilon \rightarrow 0$, to global, finite-energy, weak solutions of (\ref{eq:MS}).\footnote{A similar $\varepsilon$ approximation argument appears in \cite{ADM17} for a nonlinear extension of the Maxwell-Schr\"{o}dinger system.}

The consideration of \cite{guo1995}, therefore, leads us to study an approximate system to the MBMP equations. Referred to as the \textbf{$\varepsilon$-modified MBMP equations}, this approximate system reads  
\begin{align}\label{eq:MBMP_epsilon} 
\left\lbrace \begin{array}{l}
\partial_t \phi^{\varepsilon} =  - (i + \varepsilon) \Ham^{\varepsilon} (\v{A}^{\!\varepsilon}) \phi^{\varepsilon} + \varepsilon \left( T_{\mr{P}} [ \phi^{\varepsilon} , \tilde{\v{A}}^{\!\varepsilon} ] + V [\phi^{\varepsilon} ] \right) \phi^{\varepsilon}   \\
\square \v{A}^{\!\varepsilon} =  4 \pi \alpha \Lambda^{-1}_{\varepsilon} \Hproj{ \J_{\!\mr{P}} [ \phi^{\varepsilon} , \tilde{\v{A}}^{\!\varepsilon} ] }   \\
\diver{\v{A}^{\!\varepsilon}} = 0 , ~~~~  \tilde{\v{A}}^{\!\varepsilon} = \Lambda_{\varepsilon}^{-1} \v{A}^{\!\varepsilon} ,
\end{array} \right.
\end{align}
where $\Lambda_{\varepsilon} = \sqrt{1 - \varepsilon \Delta}$, $\Ham^{\varepsilon} (\v{A}^{\!\varepsilon})$ is the $\varepsilon$-modified Hamiltonian
\begin{align}\label{def:epsilon_Hamiltonian}
\Ham^{\varepsilon} (\v{A}^{\!\varepsilon}) = \sum_{j = 1}^N \T_j (\tilde{\v{A}}^{\!\varepsilon}) + V( \ul{\v{R}} , \Z ) ,
\end{align} 
$T_{\mr{P}} [ \phi^{\varepsilon} , \tilde{\v{A}}^{\!\varepsilon} ]$ is defined by (\ref{def:total_kinetic}) and $V[\phi^{\varepsilon}]$ is defined by (\ref{def:tot_potential_energy}). We define the total energy of the $\varepsilon$-modified system as 
\begin{align}\label{def:epsilon_tot_energy}
\E [\phi^{\varepsilon} , \v{A}^{\!\varepsilon} , \partial_t \v{A}^{\!\varepsilon}] = T_{\mr{P}} [ \phi^{\varepsilon} , \tilde{\v{A}}^{\!\varepsilon} ] + V [\phi^{\varepsilon} ] + F [ \v{A}^{\!\varepsilon} , \partial_t \v{A}^{\!\varepsilon}] \| \phi^{\varepsilon} \|_2^2 ,
\end{align}
where $F [ \v{A}^{\!\varepsilon} , \partial_t \v{A}^{\!\varepsilon} ]$ is the field energy defined by (\ref{def:field_energy_potentials}). 

For the remainder of this thesis we will drop the dependence on $\varepsilon$ when it is not needed. Note that the Pauli operators $\T_j$ in the definition (\ref{def:epsilon_Hamiltonian}) of $\Ham (\v{A})$ are evaluated at the regularized vector potential $\tilde{\v{A}}$, whereas the field energy $F$ is evaluated at $(\v{A}, \partial_t \v{A})$. Similarly, note that the probability current density $\J_{\!\mr{P}}$ in (\ref{eq:MBMP_epsilon}) is evaluated at $\tilde{\v{A}}$. These choices are made so that the total energy (\ref{def:epsilon_tot_energy}) is dissipative under the time evolution of (\ref{eq:MBMP_epsilon}) (see Theorem \ref{thm:MBMP_epsilon_dissipation-laws}). Moreover, the choice of the right hand side of the first equation in (\ref{eq:MBMP_epsilon}) is made so that normalized wavefunctions remain normalized under the flow of (\ref{eq:MBMP_epsilon}). This point will be crucial for the application of the results concerning the stability of matter in magnetic fields to construct \textit{global} solutions to (\ref{eq:MBMP_epsilon}).  

The space of initial conditions we will consider for the $\varepsilon$-modified MBMP system is
\begin{align}\label{def:initconds}
\X_0^m = \left\lbrace (\psi_0 , \v{a}_0 , \dot{\v{a}}_0) \in [H^m (\R^{3N}) ]^{2^N} \oplus ( H^m \oplus H^{m-1} ) (\R^3 ; \R^3) : \diver{\v{a}_0} = \diver{\dot{\v{a}}_0} = 0 \right\rbrace .
\end{align}
Combining the regularity improving estimates of the heat kernel $e^{\varepsilon t \Delta}$ (see Lemma \ref{lem:Heat-Kernel}) with a contraction mapping scheme similar to the one in \cite{NW05}, we prove the following local well-posedness result for (\ref{eq:MBMP_epsilon}).
\begin{thm}[Local Well-posedness of the $\varepsilon$-Modified System]\label{thm:local_exist_MBMP_epsilon}
Fix $m \in [1 , 2]$ and $\varepsilon > 0$. Given initial data $(\phi_0 , \v{a}_0 , \dot{\v{a}}_0) \in \X_0^m$, there exists a maximal time interval $\I = [0 , T_{\mr{max}})$ and a unique solution
\begin{align*}
(\phi , \v{A}) \in C_{\I} [ H^m (\R^{3N}) ]^{2^N} \times [ C_{\I} H^m (\R^3 ; \R^3) \cap C^1_{\I} H^{m-1} (\R^3 ; \R^3) ]
\end{align*}
to (\ref{eq:MBMP_epsilon}) such that the initial conditions $(\phi (0) , \v{A} (0) , \partial_t \v{A} (0)) = (\phi_0 , \v{a}_0 , \dot{\v{a}}_0)$ are satisfied and the blow-up alternative holds: either $T_{\mr{max}} = \infty$ or $T_{\mr{max}} < \infty$ and 
\begin{align*}
\limsup_{t \rightarrow T_{\mr{max}}}  \| (\phi (t) , \v{A} (t) , \partial_t \v{A} (t) ) \|_{H^m \oplus H^m \oplus H^{m-1}} = \infty .
\end{align*} 
Furthermore, we can approximate lower regularity solutions by higher regularity solutions in the following sense: if $\{ (\phi_0^j , \v{a}_0^j , \dot{\v{a}}_0^j) \}_{j \geq 1} \in \X_0^m$ converges, as $j \rightarrow \infty$, to $(\phi_0 , \v{a}_0 , \dot{\v{a}}_0) \in \X_0^1$ in $H^1 \oplus H^1 \oplus L^2$, then, for each $t \in \I$, the sequence of solutions $\{ (\phi^j (t) , \v{A}^j (t) , \partial_t \v{A}^j (t)) \}_{j \geq 1}$ corresponding to the initial datum $\{ (\phi_0^j , \v{a}_0^j , \dot{\v{a}}_0^j) \}_{j \geq 1}$ converges in $H^1 \oplus H^1 \oplus L^2$ to the solution $(\phi (t) , \v{A} (t) , \partial_t \v{A} (t))$ corresponding to the initial datum $(\phi_0 , \v{a}_0 , \dot{\v{a}}_0)$.
\end{thm}

The limited range of regularity, namely $m \in [1, 2]$, in Theorem \ref{thm:local_exist_MBMP_epsilon} comes from controlling the Coulomb term $V(\ul{\v{R}} , \Z) \phi$ in (\ref{eq:MBMP_epsilon}) (see Lemma \ref{lem:Estimates-Coulomb}). We can, in fact, prove Theorem \ref{thm:local_exist_MBMP_epsilon} for $m$ up to $\frac{5}{2} - \delta$, $\delta > 0$. However, doing so seems to be an unnecessary mathematical generality and has no bearing on the validly of Theorem \ref{thm:weak_solns_MBMP}. However, we do expect this to be the maximum range of regularity for this system. Indeed, already for the Hydrogen ground state $\psi_0 (\v{x}) \propto e^{- |\v{x}| / 2}$ (see Chapter \ref{chap:stability}) one has $\| \psi_0 \|_{s,2} < \infty$ if and only if $s < 5/2$. 

With Theorem \ref{thm:local_exist_MBMP_epsilon} at our disposal, we would then like to consider the limit $\varepsilon \rightarrow 0$ of the low regularity ($m=1$) solutions to (\ref{eq:MBMP_epsilon}). However, one potential obstruction to considering the $\varepsilon \rightarrow 0$ limit is that the local time interval of existence $[0 , T_{\mr{max}})$ in Theorem \ref{thm:local_exist_MBMP_epsilon} might shrink to zero as $\varepsilon \rightarrow 0$. It is therefore necessary to prove that the low regularity $H^1 \oplus H^1 \oplus L^2$ local solutions to (\ref{eq:MBMP_epsilon}) are, in fact, global. A key ingredient that allows us to extend from local to global solutions is to prove apriori $\varepsilon , t$-independent bounds in $H^1 \oplus H^1 \oplus L^2$ on solutions $(\phi^{\varepsilon} , \v{A}^{\!\varepsilon})$ to (\ref{eq:MBMP_epsilon}). Our proof of these uniform bounds uses that the energy $\E$ (\ref{def:epsilon_tot_energy}) is dissipative together with the fact that the Coulomb energy $| V[ \phi^{\varepsilon} (t) ] |$ along a solution $(\phi^{\varepsilon} , \v{A}^{\!\varepsilon})$ is bounded, with upper bound independent of $\varepsilon$ and $t$. This latter fact is only true when the energy $\E$ is uniformly bounded below, and from the results on the stability of matter in magnetic fields (discussed in Chapter \ref{chap:stability}) we know this requires sufficiently small $\alpha$ and $\alpha^2 \max{\Z}$. We express the fact $V [ ~ \cdot ~ ]$ is a bounded functional on $\fC_N$ when $\alpha$ and $\alpha^2 \max{\Z}$ are sufficiently small and that low regularity $H^1 \oplus H^1 \oplus L^2$ local solutions to (\ref{eq:MBMP_epsilon}) are global as the following Theorems. 

\begin{thm}[Uniform Bound on the Coulomb Energy]\label{lem:bound_on_coulomb}
Let $\{ (\phi^{n} , \v{A}^{\!n}) \}_{n \geq 1} \subset  \fC_N$, where $\fC_N$ is defined by (\ref{def:function_space_C}), and assume that $E [\phi^n , \v{A}^{\!n} , \v{0}] \leq C$ where $C$ is a constant depending on $N$, $K$, $\alpha$, $\Z$, $\ul{\v{R}}$, and $ ( \phi^0 , \v{A}^{\!0} )$, but independent of $n$. Assume $\alpha$ and $\alpha^2 \max{\Z}$ are sufficiently small to ensure $E_{\mr{P}}^{\mr{G}} > - \infty$. Then the sequence of Coulomb energies $\{ V[\phi^{n}] \}_{n=1}^{\infty}$ is uniformly bounded, $\sup_n |V[\phi^n]| < \infty$.
\end{thm} 

\begin{thm}[Dissipation of Energy and Uniform Bounds]\label{thm:MBMP_epsilon_dissipation-laws}
Fix $\varepsilon > 0$ and $m \in [1 , 2]$. Let $(\phi_0 , \v{a}_0 , \dot{\v{a}}_0) \in \X^m_0$ with $\phi_0 \in \BWN{H^m(\R^3;\C^2)}$ and $\| \phi_0 \|_2 = 1$. Let $(\phi , \v{A}) \in C_{\I} H^m \times [ C_{\I} H^m \cap C^1_{\I} H^{m-1} ]$ be the corresponding solution to (\ref{eq:MBMP_epsilon}) provided by Theorem \ref{thm:local_exist_MBMP_epsilon}. Then $\phi (t)$ remains completely antisymmetric and normalized for $t \in \I$, and, if $m = 2$,
\begin{align*}
& \E [\phi , \v{A} , \partial_t \v{A}] (t) - \E [\phi_0 , \v{a}_0 , \dot{\v{a}}_0] \\
& \hspace{1cm} = - 2 \varepsilon \int_0^t \left[ \| \Ham (\v{A} (\tau )) \phi (\tau) \|_{2}^2 - \langle \phi (\tau) , \Ham (\v{A} (\tau ))  \phi (\tau) \rangle^2_{2} \right] \dd \tau , \numberthis \label{eq:Diss_Energy} 
\end{align*}
for all $t \in \I$. Moreover, if $\alpha$ and $\alpha^2 \max{\Z}$ are sufficiently small to ensure $E^{\mr{G}}_{\mr{P}} > - \infty$, then 
\begin{align}
 \| \nabla \phi (t) \|_2 \leq C_1 , \hspace{1cm} F [\v{A} , \partial_t \v{A}] (t) \leq C_2 , \hspace{1cm} \| \v{A} (t) \|_2 \leq C_3 ( 1 +  t ) , \label{eq:uniform-bounds}  
\end{align}
for all $t \in \I$, where $C_1 , C_2 , C_3 > 0$ are constants depending on $N$, $K$, $\Z$, $\alpha$, and the initial data, but not $\varepsilon$ or $t$. As a consequence, for $m = 1$ and for each fixed $\varepsilon > 0$, the solution $(\phi , \v{A})$ exists for all $t \in \R_+$. 
\end{thm}

The are two reasons the bounds (\ref{eq:uniform-bounds}) are important. First, and as already mentioned in the paragraph preceding Theorem \ref{thm:MBMP_epsilon_dissipation-laws}, for each fixed $\varepsilon > 0$, it is necessary to have time-independent bounds on $(\phi (t) , \v{A} (t) , \partial_t \v{A} (t))$ in $H^1 \times H^1 \times L^2$-norm in order to apply the blow-up alternative of Theorem \ref{thm:local_exist_MBMP_epsilon} and assert the $m = 1$ solutions of Theorem \ref{thm:local_exist_MBMP_epsilon} exist for all time. Second, in order to apply a compactness argument to take the $\varepsilon \rightarrow 0$ limit, we need $\varepsilon$-independent bounds on $(\phi (t) , \v{A} (t) , \partial_t \v{A} (t))$ in $H^1 \times H^1 \times L^2$-norm to apply the Banach-Alaoglu Theorem and extract a weak$^*$ converging subsequence. This weak$^*$ limit will be shown to be a finite-energy weak solution to (\ref{eq:MBMP}), thus yielding a proof of Theorem \ref{thm:weak_solns_MBMP}. We emphasize that the complete antisymmetry and normalization of $\phi (t)$ is crucial, as otherwise we cannot make use of the stability result (\ref{eq:general_stability_estimate}) to control the Coulomb energy.

This thesis is organized as follows. In Chapter \ref{chap:prelim} we introduce some notation and review several definitions that will help set the mathematical stage on which the well-posedness of (\ref{eq:MBMP}) is studied. Chapter \ref{chap:stability} is concerned with reviewing the concept of the energetic stability of matter and the results concerning it that are relevant for this thesis. Chapter \ref{chap:othermodels} is presented to further motivate our study of (\ref{eq:MBMP}) by discussing other possible models for an electron interacting with its self-generated electromagnetic field and the mathematical literature devoted to studying these models.
Chapter \ref{chap:epsilon_system} is devoted to studying the $\varepsilon$-modified MBMP equations and proving the main theorems of this thesis. We conclude with Chapter \ref{chap:open} which discusses some interesting open problems regarding the MBMP equations.

\chapter{Mathematical Preliminaries}\label{chap:prelim}

\section{Notation and Definitions}\label{sec:notation}

In this section we review the notation that will be employed throughout the rest of the thesis. In doing so we will also review the definitions of various function spaces that will appear. The majority of this material has been pulled from the standard references \cite{RS72, Adams75, RS75, Barbu10, Evans10, LL01}. Also, in the majority of the thesis we will not be concerned with keeping track of exact constants appearing in various inequalities. Therefore, if $a , b \in \R$, $a \lesssim b$ means that there exists a constant $C > 0$, independent of $a$ and $b$, such that $a \leq C b$. Occasionally, we may choose to emphasize that the implied constant in $a \lesssim b$ depends on a particular quantity, say $\varepsilon$, and for this we will write $a \lesssim_{\varepsilon} b$. We begin with some general definitions and then specialize to specific cases that we will be mostly concerned about. 

Let $( X , \| \cdot \|_X )$ be a Banach space. The dual of $X$ will be denoted by $(X^{\ast} , \| \cdot \|_{X^{\ast}})$ where 
\begin{align*}
\| \cdot \|_{X^{\ast}} := \sup{ \{ | \langle \cdot , u \rangle_{X^{\ast} , X} | ~ : ~ \|u\|_{X} \leq 1 \} } ,
\end{align*}
with $\langle \cdot , \cdot \rangle_{X^{\ast} , X}$ denoting the pairing between $X^{\ast}$ and $X$. Note we've made the "physicist choice" for the pairing between $X^{\ast}$ and $X$: $\langle \cdot , \cdot \rangle_{X^{\ast} , X}$ is anti-linear in the first argument and linear in the second. Frequently, it will be the case that $X$ is a Hilbert space and in this case we will denote the inner product on $X$ as $\langle \cdot , \cdot \rangle_X$, where to norm is given by $\| \cdot \|_X = \sqrt{ \langle \cdot , \cdot \rangle_X }$. We will let $\B (X ; Y)$ denote the set of bounded operators between one normed linear space $X$ to another $Y$, which is itself a Banach space under the usual operator norm when $Y$ is complete. When $X = Y$ we will simply write $\B (X) \equiv \B (X ; X)$. 

Let $\Omega \subset \R^d$ be a Lebesgue measurable set. For $p \in [1, \infty]$, $L^p (\Omega ; X)$ will denote the usual Lebesgue space of strongly measurable functions $f : \Omega \rightarrow X$ with the property that
\begin{align*}
\| f \|_{L^p (\Omega ; X)} = \left\lbrace \begin{array}{lcr}
\displaystyle \left( \int_{\Omega} \| f (\v{x}) \|_X^p \dd \v{x} \right)^{\frac{1}{p}} & \mr{for} & 1 \leq p < \infty \\
\displaystyle \esssup_{\v{x} \in \Omega} \| f (\v{x}) \|_X & \mr{for} & p = \infty
\end{array} \right. 
\end{align*}
is finite. An important special case is when $p = 2$ and $X$ is a Hilbert space with inner product $\langle \cdot , \cdot \rangle_{X}$. In this case $L^2 (\Omega ; X)$ is itself a Hilbert space where for $f , g \in L^{2} (\Omega ; X)$ the inner product is defined as
\begin{align*}
\langle f , g \rangle_{L^{2} (\Omega ; X)} = \int_{\Omega} \langle f (\v{x}) , g (\v{x}) \rangle_{X} \dd \v{x} . 
\end{align*}
The symbol $L^p_{\mr{loc}} ( \Omega ; X )$ will denote the space of locally integrable function $\Omega \rightarrow X$. Specifically,
\begin{align*}
L_{\mr{loc}}^p ( \Omega ; X ) = \left\lbrace f : \Omega \rightarrow X  ~ : ~ \int_K \|f (\v{x}) \|_X^p < \infty , ~ \forall K \subset \Omega ~ \mr{compact} \right\rbrace  .
\end{align*}

In a similar fashion, the usual Sobolev space $W^{m,p} (\Omega ; X)$ for $p \in [1 , \infty]$ and $m \in \N_0$ is the space of functions $f \in L^p (\Omega ; X)$ whose distributional partial derivatives $\partial_1^{k_1} \cdots \partial_d^{k_d} f \in L^p (\Omega ; X)$ for all $k_1 + \cdots + k_d \leq m$, where $k_j \in \N_0$ for $j \in \{1 , \cdots , d\}$. We equip $W^{m,p} (\Omega ; X)$ with the usual norm
\begin{align*}
\| f \|_{W^{m,p} (\Omega ; X)} = \left\lbrace \begin{array}{lcr}
\displaystyle \left( \sum_{k_1 + \cdots + k_d \leq m} \int_{\Omega} \| \partial_1^{k_1} \cdots \partial_d^{k_d} f (\v{x}) \|_X^p \dd \v{x} \right)^{\frac{1}{p}} & \mr{for} & 1 \leq p < \infty \\
\displaystyle \sum_{k_1 + \cdots + k_d \leq m} \esssup_{\v{x} \in \Omega} \| \partial_1^{k_1} \cdots \partial_d^{k_d} f (\v{x}) \|_X & \mr{for} & p = \infty .
\end{array} \right. 
\end{align*}
If $p = 2$ and $X$ is a Hilbert space, then $H^m (\Omega ; X) := W^{m,2} (\Omega; X)$ is itself a Hilbert space with the obvious inner product. We note $L^p \equiv W^{0,p}$ and $L^2 \equiv H^0$. We remind ourselves that, when $X$ is reflexive, for $m \geq 0$ and $p \in [1,\infty )$, the dual of $W^{m,p} (\Omega ; X)$ is isometrically anti-isomorphic to $W^{-m , p'} (\Omega ; X^{\ast})$, where
\begin{align*}
p' = \left\lbrace \begin{array}{ccl}
\infty & \mr{for} & p = 1 \\[1.2ex]
\dfrac{p}{p-1} & \mr{for} & p \in (1 , \infty ) \\[1.5ex]
1 & \mr{for} & p = \infty 
\end{array} \right.
\end{align*}
is the H\"{o}lder conjugate exponent to $p$. When $\Omega = [0,T]$ for some $T > 0$, we will often write $L^p_{\Omega} X = L^p (\Omega ; X)$ and $W^{m,p}_{\Omega} X = W^{m,p} (\Omega ; X)$, or simply $L^p_T X = L^p (\Omega ; X)$ and $W^{m,p}_T X = W^{m,p} (\Omega ; X)$. If $\Omega$ and $X$ are to be understood from a given context, we will simply write $L^p$ and $W^{m,p}$ in place of $L^p (\Omega ; X)$ and $W^{m,p} (\Omega ; X)$, respectively.

For $\Omega \subset \R^d$ open, we denote the space of strongly continuous mappings from $\Omega \rightarrow X$ as $C (\Omega ; X)$ and $C^1 (\Omega ; X)$ as the subspace of maps $f \in C (\Omega ; X)$ whose strong partial derivatives satisfy $\partial_1 f , \cdots , \partial_d f \in C (\Omega ; X)$. Continuing in this way one defines $C^{\infty} (\Omega ; X)$ and $C^{\infty}_c (\Omega ; X)$ as the space of infinitely differentiable and compactly supported, infinitely differentiable maps $\Omega \rightarrow X$, respectively. We let $C^{\mr{w}} (\Omega ; X)$ denote the space of weakly continuous mappings from $\Omega \rightarrow X$. In particular, $f \in C^{\mr{w}} (\Omega ; X)$ if for $\{ \v{x}_n \}_{n \geq 1} \subset \Omega$ with $\v{x}_n \rightarrow \v{x} \in \Omega$, then
\begin{align*}
\lim_{n \rightarrow \infty} | \langle L ,  f (\v{x}_n) - \phi (\v{x}) \rangle_{X^{\ast} , X} | = 0 , \hspace{1cm} \forall L \in X^{\ast}. 
\end{align*}
When $\Omega = [0,T]$ for some $T > 0$, we will often write $C_{\Omega} X = C(\Omega ; X)$, $C_{\Omega}^1 X = C^1 (\Omega ; X)$ and $C^{\mr{w}}_{\Omega} X = C^{\mr{w}} (\Omega ; X)$, or simply $C_T X = C (\Omega ; X)$, $C_T^1 X = C^1 (\Omega ; X)$ and $C^{\mr{w}}_T X = C^{\mr{w}} (\Omega ; X)$.

The space of test functions from $\Omega \rightarrow X$ will be denote by $\fD (\Omega ; X)$. In particular, $\fD (\Omega ; X)$ consists of all the functions in $C^{\infty}_c (\Omega ; X)$ equipped with the following notion of convergence. A sequence $\{ \phi^n \}_{n \geq 1} \subset C^{\infty}_c (\Omega ; X)$ converges in $\fD (\Omega ; X)$ to the function $\phi \in C^{\infty}_c (\Omega ; X)$ if and only if there is some fixed, compact $K \subset \Omega$ such that the support of $\phi^n - \phi$ is in $K$ for all $n \in \N$ and, for all $m \in \N_0$, $\|\partial_t^m \phi^n  - \partial_t^m \phi \|_{L_{\Omega}^{\infty} X} \rightarrow 0$ as $n \rightarrow \infty$. Then, $\fD' (\Omega ; X)$ denotes the space of distributions from $\Omega \rightarrow X$. Specifically, $\fD' (\Omega ; X)$ is the set of strongly continuous linear maps from $C_{c}^{\infty} (\Omega ; \C)  \rightarrow X$. When $g \in L_{\mr{loc}}^1 (\Omega ; X)$ we denote the corresponding distribution in $\fD' (\Omega ; X)$ defined via the formula
\begin{align*}
C_{c}^{\infty} (\Omega ; \C) \ni \phi \longmapsto \int_{\Omega} g (\v{x}) \phi (\v{x}) \dd \v{x} \in X
\end{align*}
by the same symbol. 

Very often $\Omega$ is all of $\R^d$ and $X$ is either $\R^n$, $\C^{m}$, a Sobolev space $W^{m,p}$, or a direct sum of Sobolev spaces $W^{m_1 , p_1} \oplus W^{m_2 , p_2}$. For this reason we introduce some special notations we employ in this setting (some of this notation is highlighted at the beginning of this thesis). First of all, when $X = \R^n$ or $\C^m$ we simply write $\| \cdot \|_X = | \cdot |$ where $| \cdot |$ is the usual Euclidean distance on $\R^n$ or $\C^m$. More importantly, we will abbreviate $\| \cdot \|_{p}$, $\| \cdot \|_{m , p}$, $\| \cdot \|_{q ; m , p}$, and $\| \cdot \|_{q_1 ; m_1 , p_1 \oplus q_2 ; m_2 , p_2}$ for the norms on $L^p$, $W^{m,p}$, $L^q W^{m,p}$, and $L^{q_1} W^{m_1 , p_1} \oplus L^{q_2} W^{m_2 , p_2}$, respectively. This notation comes with the understanding that $\| \cdot \|_{p} \equiv \| \cdot \|_{0 , p} \equiv \| \cdot \|_{0 ; 0 , p}$ and $\| \cdot \|_{q ; p} \equiv \| \cdot \|_{q ; 0 , p}$. 

When discussing many-body wave functions $\psi : \R^{3N} \rightarrow \C^{2^N}$, with $N \geq 1$, we always consider 
\begin{align*}
L^p (\R^{3N} ; \C^{2^N}) \simeq [L^p (\R^{3N})]^{2^N} \equiv \bigotimes^N [L^p (\R^3)]^2
\end{align*}
through the canonical isomorphism, and we recall that $\BWN{L^p(\R^3;\C^2)}$ denotes the closed subspace of $\bigotimes^N [L^p (\R^3)]^2$ consisting of completely antisymmetric many-body wave functions. The $L^p$-norm of a many-body wave function $\psi \in L^p (\R^{3N} ; \C^{2^N})$ will be denoted as
\begin{align*}
\| \psi \|_p^p = \int |\psi (\ul{\v{z}})|^p \dd \ul{\v{z}} \equiv \sum_{\ul{\v{s}} \in \{\uparrow , \downarrow \}^N} \int_{\R^{3N}} | \psi ( \v{x}_1 , s_1 ; \cdots ; \v{x}_N , s_N ) |^p \dd \ul{\v{x}} ,
\end{align*}
where $\ul{\v{s}} = (s_1 , \cdots , s_N)$. When considering vector fields $\v{A} \in L^p (\R^d ; \R^d)$, with components $\v{A} = (A^1 , \cdots , A^d)$, we define $\| \nabla \v{A} \|_{p}$ as
\begin{align*}
\| \nabla \v{A} \|_p^p = \sum_{j = 1}^d \| \nabla A^j \|_{p}^p  = \sum_{i,j = 1}^d \int_{\R^d} | \partial_i A^j (\v{x}) |^p  \dd \v{x} .
\end{align*}
Regarding the magnetic vector potential $\v{A} : \R^3 \rightarrow \R^3$ in particular, we will frequently use the identity $\| \curl{ \v{A} } \|_2 = \| \nabla \v{A} \|_2$, when $\diver{\v{A}} = 0$ and $\v{A} \in \dot{H}^1 (\R^3 ; \R^3)$ (see the discussion below for the definition of $\dot{H}^1 (\R^3 ; \R^3)$). This follows from the well-known vector calculus identity $\curl \curl = \nabla \diver - \Delta$. Concerning the Laplacian, unless otherwise specified, we will always use the notation $\Delta$ and assume the reader will understand what is meant from the given context. For example, when $f : \R^d \rightarrow \R$ and $g : \R^n \rightarrow \R$ we write $\Delta f$ and $\Delta g$ with the understanding that $\Delta f = \sum_{j = 1}^d \partial_j^2 f$ and $\Delta g = \sum_{j = 1}^n \partial_j^n g$. 

We will also make use of fractional order Sobolev spaces. Let $\fS (\R^d ; \C)$ denote the usual space of Schwartz functions. Specifically, $f \in \fS (\R^d ; \C)$ satisfies $f \in C^{\infty} (\R^d ; \C)$ and 
\begin{align*}
 \sup_{\v{x} \in \R^d} |x_1^{\mu_1} \cdots x_d^{\mu_d} \partial_1^{\nu_1} \cdots \partial_d^{\nu_d} f (\v{x})| < \infty 
\end{align*}
for all $\mu_1 , \cdots , \mu_d , \nu_1 , \cdots , \nu_d \in \N_0$. In words, $\fS (\R^d ; \C)$ is the space of smooth functions whose derivatives of all orders decreasing more rapidly than any polynomial. Let $\fS' (\R^d ; \C)$ denote the space of temper distributions. The Fourier transform $\cF : \fS' (\R^d ; \C) \rightarrow \fS' (\R^d ; \C)$ of $L \in \fS' (\R^d ; \C)$ is defined as $(\cF L) (f) = L (\cF f)$ for all $f \in \fS (\R^d ; \C)$, where 
\begin{align*}
(\cF f) (\v{k}) = \int_{\R^d} e^{- 2 \pi i \v{k} \cdot \v{x}} f (\v{x}) \dd \v{x} . 
\end{align*}
For $s \geq 0$ and $p \in [ 1 , \infty ]$, $W^{s,p} (\R^d ; \C)$ is the fractional order Sobolev space defined by
\begin{align*}
W^{s,p} (\R^d ; \C) = \left\lbrace f \in \fS' (\R^d ; \C) ~ : ~ \| \Lambda^s f \|_p < \infty \right\rbrace ,
\end{align*}
where $\Lambda = \sqrt{1 - \Delta}$. For $f \in \fS' (\R^d ; \C)$, the operator $\Lambda^s$ is defined via
\begin{align*}
\Lambda^s f = \cF^{-1} ( 1 + 4 \pi^2 | \cdot |^2 )^{s/2} \cF f .
\end{align*}
We equip $W^{s,p} (\R^d ; \C)$ with the norm $\| f \|_{s,p} = \| \Lambda^s f \|_{p}$, which turns $(W^{s,p} (\R^d ; \C ) , \| \cdot \|_{s,p})$ into a Banach space. For integer $s$ this definition coincides with the Sobolev spaces defined early. An important special case is $p = 2$, in which case $H^s (\R^d ; \C) := W^{s,2} ( \R^d ; \C )$ is a Hilbert space with inner product
\begin{align*}
\langle f , g \rangle_{H^s} = \langle \Lambda^s f , \Lambda^s g \rangle_{L^2} .
\end{align*}
The fractional order homogeneous Sobolev space $\dot{W}^{s , p} (\R^d ; \C)$ is the space of function $f : \R^d \rightarrow \C$ equipped with the seminorm $\| \dot{\Lambda}^s f \|_{p}$ where $\dot{\Lambda} = \sqrt{- \Delta}$. For $p = 2$ we will use the obvious notation $\dot{H}^s (\R^d ; \C) := \dot{W}^{s,2} (\R^d ; \C)$. For $p \in [ 1 , \infty )$ the dual of $W^{s,p} (\R^d ; \C)$ is (isometrically anti-isomorphic to) $W^{-s , p'} (\R^d ; \C)$. This dual space is equip with the norm
\begin{align*}
\| f \|_{-s , p'} = \sup{ \left\lbrace | \langle f , \eta \rangle_{L^2} | ~ : ~ \| \eta \|_{s,p} \leq 1 \right\rbrace } .
\end{align*}
As before, if $p = 2$, then we denote $H^{-s} (\R^d ; \C) := W^{-s , 2} (\R^d ; \C)$. We note that for $f \in H^{-s}$ and $\eta \in H^s$, $\langle f , \eta \rangle_{H^{-s} , H^s} = \langle \Lambda^{-s} f , \Lambda^s \eta \rangle_{L^2}$, and recall that if $f \in L^2 \subset H^{-s}$, then $\langle f ,  \eta \rangle_{H^{-s} , H^s} = \langle f ,  \eta \rangle_{L^2}$. 

With these notational preliminaries sorted out, we can clarify what we mean by a solution to (\ref{eq:MBMP}) and (\ref{eq:MBMP_epsilon}). The \textbf{weak} solution 
\begin{align*}
(\psi , \v{A} , \partial_t \v{A}) \in C^{\mr{w}} ( \R_+ ; \BWN{H^1(\R^3;\C^2)} \times H^1 (\R^3 ; \R^3)  \times L^2 (\R^3 ; \R^3) )
\end{align*}
to (\ref{eq:MBMP}) appearing in Theorem \ref{thm:weak_solns_MBMP} is a distributional solution $(\psi , \v{A})$ in
\begin{align*}
\fD' ( \R_+ ; \BWN{H^{-1}(\R^3;\C^2)} \times H^{-1} (\R^{3} ; \R^3) ) .
\end{align*}
In particular, the solution $(\psi , \v{A})$ satisfies
\begin{align*}
& ~ \int_0^{\infty} \langle \xi , \psi (s) i \partial_t f(s) \rangle_{L^2} \dd s + \int_0^{\infty} \sum_{i = 1}^N \langle \nabla_{\v{x}_i} \xi , f(s) \nabla_{\v{x}_i} \psi (s) \rangle_{L^2} \dd s \\
& \hspace{5cm} = - \int_0^{\infty} \langle \xi , f(s) [ \fL (\v{A} (s)) - V(\ul{\v{R}} , \Z) ] \psi (s) \rangle_{L^2} \dd s , \\
& \int_0^{\infty} \sum_{k = 1}^3 \langle \partial_k \eta , f(s) \partial_k \v{A} (s) \rangle_{L^2} \dd s - \alpha^2 \int_0^{\infty} \langle \eta , \partial_t f(s) \partial_t \v{A} (s) \rangle_{L^2} \dd s \\
& \hspace{5cm} = 4 \pi \alpha \int_0^{\infty} \langle \eta , f (s) \Hproj{ \J_{\!\mr{P}} [\psi (s) , \v{A} (s)] } \rangle_{L^2} \dd s ,
\end{align*}
for all $f \in C_c^{\infty} (\R_+)$, $\xi \in [H^1 (\R^{3N})]^{2^N}$, and $\eta \in H^1 (\R^3 ; \R^3)$, where $\fL (\v{A}) = \sum_{j = 1}^N \cL_j (\v{A})$ and $\cL_j (\v{A}) = 2 \v{A}_j \cdot \v{p}_j + | \v{A}_j |^2 + \gvsig_j \cdot \v{B}_j$.

The solutions $(\phi , \v{A}) \in C_{\I} [ H^m (\R^{3N}) ]^{2^N} \times [ C_{\I} H^m (\R^3 ; \R^3) \cap C^1_{\I} H^{m-1} (\R^3 ; \R^3) ]$, where $\I = [0,T]$, constructed in Theorem \ref{thm:local_exist_MBMP_epsilon} are considered to satisfy the integrated versions of (\ref{eq:MBMP_epsilon}): 
\begin{align*}
\left\lbrace \begin{array}{l}
\phi (t) = e^{(i + \varepsilon) t \Delta} \phi_0 + \int_0^t e^{(i + \varepsilon) (t - \tau) \Delta} f[\phi (\tau) , \tilde{\v{A}} (\tau)] \dd \tau \\
\v{A} (t) = \dot{\s}( t / \alpha ) \v{a}_0 + \alpha ~ \s(t/\alpha) \dot{\v{a}}_0 + 4 \pi \int_0^t \s ( (t-\tau) / \alpha ) \Lambda_{\varepsilon}^{-1} \Hproj{ \J_{\!\mr{P}} [ \phi (\tau) , \tilde{\v{A}} (\tau) ] } \dd \tau ,
\end{array} \right.  
\end{align*}
where
\begin{align*}
f[\phi , \tilde{\v{A}}]  = - (i + \varepsilon) \left( \fL ( \tilde{\v{A}} ) + V(\ul{\v{R}} , \Z) \right) \phi + \varepsilon \left( T_{\mr{P}} [ \phi , \tilde{\v{A}} ] + V [\phi ] \right) \phi , 
\end{align*}
and $e^{(i + \varepsilon) t \Delta}$, $\s (t) = \sin{\left(\sqrt{ - \Delta} t \right)}/\sqrt{ - \Delta}$, and $\dot{\s} (t) = \cos{\left(\sqrt{ - \Delta} t \right)}$ are all defined by their Fourier multipliers (or, equivalently, as convolutions against the respective kernels). In particular, $(\phi , \v{A})$ satisfy (\ref{eq:MBMP_epsilon}) pointwise a.e. when $m = 2$. 

\section{Fundamental Results}\label{sec:fundamental_results}

The proofs of Theorem \ref{thm:local_exist_MBMP_epsilon}, Theorem \ref{thm:MBMP_epsilon_dissipation-laws}, and, ultimately, Theorem \ref{thm:weak_solns_MBMP} will heavily rely on several inequalities involving Sobolev spaces and properties of solutions to two well-studied PDEs: the heat and wave equation. The inequalities needed include the Sobolev inequalities, the Kato-Ponce commutator estimates, the dispersive estimates for the heat kernel, and the Strichartz estimates for the wave equation. In particular, such results will be used repeatedly in \S\ref{sec:estimates} in the proofs of Lemmas \ref{lem:Estimates-Pauli}-\ref{lem:Estimates_Prob_current}, and in \S\ref{sec:metric_space_linearization}. For this reason, we devote an entire section to a detailed discussion of these facts. We claim no originality regarding the material in this section. 

Sobolev inequalities have come to mean an $L^q$-estimate of lower order derivatives of a function in terms of high order derivatives in other $L^p$ spaces. These inequalities are completely standard and used widely in the analysis of PDEs, but they are still worth recalling here. For $d \in \N$, $s \geq 0$, and $p \in [ 1 , \infty )$ satisfying $sp < d$, we define
\begin{align*}
p^* = p^* (d,s,p) = \frac{dp}{d-sp} .
\end{align*}   
In the following Lemma we state the Sobolev inequalities that will be used in this thesis. 
\begin{lem}[Sobolev Inequalities]\label{lem:sobolev}
Let $s \geq 0$, $1 \leq p \leq q$, and $f \in W^{s,p} (\R^d)$.
\begin{enumerate}
\item If $sp < d$, then $\| f \|_{q} \lesssim \| f \|_{s , p}$ when $q \leq p^*$.
\item If $sp = d$, then $\| f \|_{q} \lesssim \| f \|_{s , p}$ when $q < \infty$.
\item If $sp > d$, then $\| f \|_{\infty} \lesssim \| f \|_{s,p}$.
\end{enumerate}
\end{lem}

An important special case of the Sobolev inequalities above is $d = 3$, $s = 1$, and $p = 2$. Then $p^* = 6$ and we arrive at the Sobolev inequality for gradients: $S_3 \| f \|_6^2 \leq \| \nabla f \|_2^2$ where $S_3$ is the sharp constant and is given by $S_3 = 3 \left( \pi / 2 \right)^{4/3}$. Equality is achieved if and only if $f$ is a multiple of the function $1/ \sqrt{\mu^2 + | \v{x} - \v{a} |^2}$ with $\mu > 0$ and $\v{a} \in \R^3$ arbitrary (see \cite[Theorem 8.3]{LL01} for a detailed discussion and proof). This case plays a special role in Chapter \ref{chap:stability} where it is used to demonstrate energetic stability of the Hydrogenic atom. Proofs of the \textit{sharp} Sobolev inequality for $s = \frac{1}{2} , 1$ and $p = 2$ may be found in \cite[Theorem 8.3, 8.4]{LL01}, whereas the sharp constants for the $s = 1$ and general $p,q$ case were derived in \cite{Tal76}. The proof of the more general case for integer $s \geq 1$ and $p , q \in [1 , \infty )$ may be found in, for example, \cite[Chapter IV]{Adams75} or \cite[Chapter 5]{Evans10}. For the case of non-integer $s > 0$ see, for example, \cite[Chapter VII]{Adams75} or \cite[Theorem 6.5]{DnPV12}. A generalization of the Sobolev inequalities which will also be useful to us later on are the Gagliardo–Nirenberg interpolation inequalities. We do not state these inequalities in their most general form. For this we refer the reader to \cite[Chapter 12]{Leoni17}, where one may also find a detailed proof of the special case below \cite[Theorem 12.83]{Leoni17}. 
\begin{lem}[Gagliardo–Nirenberg Inequalities]\label{lem:GN_inequality}
Let $\theta \in [0,1]$. For every $f \in L^q (\R^d) \cap \dot{W}^{1,p} (\R^d)$ with $p,q \in [1 , \infty ]$ we have the inequality 
\begin{align*}
\| f \|_r \lesssim \| f \|_q^{\theta} \| \nabla f \|_p^{1 - \theta} \hspace{5mm} \text{where} \hspace{5mm} \frac{1}{r} = \theta \frac{1}{q} + (1 - \theta) \left( \frac{1}{p} - \frac{1}{d} \right)  ,
\end{align*} 
with the exception that if $p < d$ and $q = \infty$ we assume $f$ vanishes at $\infty$ and if $p = d > 1$ we take $\theta \in (0 , 1]$.
\end{lem}

An important application of the Sobolev inequalities that will be of relevance to us in the proof of Theorem \ref{thm:local_exist_MBMP_epsilon} in \S\ref{sec:proof_main_result} is the Rellich-Kondrachov Compactness Theorem. For a more complete discussion of the Rellich-Kondrachov Theorem we refer the reader to \cite[Theorem 8.9]{LL01} and \cite[Chapter 5]{Evans10}. Before stating the Rellich-Kondrachov Theorem we recall the fundamental notion of compact embeddings of Banach spaces. For two Banach spaces $X,Y$ with $X \subset Y$, we say $X$ is compactly embedded in $Y$, written $X \hookrightarrow Y$, provided $\| u \|_Y \lesssim \| u \|_X$ for all $u \in X$ and each bounded sequence in $X$ is precompact in $Y$. Precompact means that every bounded sequence in $X$ has a subsequence converging in $Y$. 
\begin{lem}[Rellich-Kondrachov Compactness Theorem]\label{lem:rellich_kon}
Let $\Omega \subset \R^d$ be bounded and open. Assume $\partial \Omega$ is $C^1$ and $1 \leq p < d$. Then $W^{1 ,p} (\Omega) \hookrightarrow L^q (\Omega)$ for each $1 \leq q < p^* = dp / (d-p)$.
\end{lem} 

It is important to note that $q$ is assumed to be strictly less than $p^*$ in the Rellich-Kondrachov Theorem. Indeed, even though the Sobolev inequality guarantees $\| f \|_{L^{p^*} (U)} \lesssim \| f \|_{W^{1,p} (U)}$, it is not the case that every bounded sequence in $W^{1,p} (U)$ has a subsequence that converges in $L^{p^*} (U)$. This is a simple consequence of scaling. Indeed, suppose $f \in W^{1,p} (U)$ for some open set $U \subset \R^d$ and suppose (without loss of generality) $\supp{(f)} \subset B_r (0)$ for some $0 < r \leq 1$. Let $f_{\lambda} (\v{x}) = \lambda^{- d / p^*} f (\v{x} / \lambda)$ and note $\| f_{\lambda} \|_{L^{p^*} (U)} = \| f \|_{L^{p^*} (U)}$, $\| f_{\lambda} \|_{L^p (U)} = \lambda \|f \|_{L^p (U)}$, and $\| \nabla f_{\lambda} \|_{L^p (U)} = \| \nabla f \|_{L^p (U)}$. Hence, for $0 < \lambda \leq r$, the family $f_{\lambda}$ is uniformly bounded in $W^{1,p} (U)$, converges pointwise a.e$.$ to $0$ at $\lambda \rightarrow 0$, but has constant $L^{p^*}$-norm. Consequently, there can be no subsequence of $f_{\lambda}$ converging in $L^{p^*} (U)$. 

On a different note, the Rellich-Kondrachov Theorem in the case $p = 2$ and $q = 4$ will allow us to verify the hypothesis of the Aubin-Lions Lemma, which we now state in generality. Our formulation of the Aubin-Lions Lemma is based on \cite[Theorem 1.20]{Barbu10} and we refer the reader to this text and the references therein for a discussion.
\begin{lem}[Aubin-Lions]\label{lem:aubin_lions}
Let $X_0 , X_1 , X_2$ be reflexive Banach spaces such that $X_0 \hookrightarrow X_1 \subset X_2$. Let $1 < p_i < \infty$, $i \in \{0, 1\}$, and $\I = [a,b] \subset \R$. Then $L_{\I^{\circ}}^{p_0} X_0 \cap W^{1 , p_1}_{\I} X_2 \hookrightarrow L^{p_0}_{\I} X_1$.
\end{lem}

Another result which we rely on heavily in the following pages are the so-called Kato-Ponce commutator estimate (also known as the fractional Leibniz rule). In essence, it makes precise the idea that one should be able to have a "product-rule" $L^p$-estimate on $\Lambda^s (f g)$. The following Lemma was first proved by Kato and Ponce \cite{KP88} as a special case in their study of the Euler and Navier-Stokes equations. A generalization involving the operator $\dot{\Lambda}^s$ was proved by \cite{KPV93}. For a textbook proof see \cite{Gra12}. 
\begin{lem}[Generalized Kato-Ponce Commutator Estimates] \label{lem:Kato-Ponce}
Suppose $1 < p < \infty$, $s, \alpha , \beta \geq 0$, and $1/p_i + 1/q_i  = 1/p$ with $i = 1,2$, $1 < q_1 \leq \infty$, $1 < p_2 \leq \infty$. If $f \in W^{s+\alpha , p_1} \cap W^{- \beta , p_2}$ and $g \in W^{s+ \beta , q_2} \cap W^{- \alpha , q_1}$, then
\begin{align*}
\| f g \|_{s,p} \lesssim \| f \|_{s + \alpha , p_1} \| g \|_{-\alpha , q_1}  + \| f \|_{- \beta , p_2} \| g \|_{s+ \beta , q_2} .
\end{align*}
The same conclusion holds for the Sobolev spaces $W^{s,p}$ replaced by their homogeneous counterparts $\dot{W}^{s,p}$.
\end{lem}

Very important to our proof strategy described in Chapter \ref{chap:intro} are the dispersive estimates for the heat contraction semigroup. To describe these estimates, first define the heat kernel $H_t : \R^d \rightarrow \R$ by
\begin{align*}
H_t (\v{x}) = \frac{1}{(4 \pi t)^{d/2}} \exp{ \left\lbrace - \frac{|\v{x}|^2}{4t} \right\rbrace } .
\end{align*}
The action of the heat semigroup $e^{t \Delta}$ on any function $f : \R^d \rightarrow \C$ is then defined as $e^{t \Delta} f := H_t \ast f$, where $*$ denotes convolution. This can be derived easily by taking the inverse Fourier transform of the identity
\begin{align*}
(\cF e^{t \Delta} f )(\v{k}) = e^{- 4 \pi^2 | \v{k} |^2 t } (\cF f)(\v{k}) . 
\end{align*}
Indeed, one can easily reduce to one-dimensional integrals to see that
\begin{align*}
\int_{\R^d} e^{- 4 \pi^2 |\v{k}|^2 t} e^{2 \pi i \v{k} \cdot \v{x}} \dd \v{k} = \prod_{j = 1}^d \int_{\R} e^{- 4 \pi^2 k_j^2 t} e^{2 \pi i k_j x_j} \dd k_j = \prod_{j=1}^d \frac{1}{\sqrt{4 \pi t}} \exp{ \left\lbrace - \frac{x_j^2}{4t} \right\rbrace } = H_t (\v{x}) 
\end{align*}
Moreover, it is a straightforward computation that when $f \in L^p (\R^d)$ with $p \in [1,\infty]$ the function $g(t,\v{x}) = (e^{t \Delta} f)(\v{x})$ is in $C^{\infty} ( (0 , \infty ) \times \R^d)$, and $g$ satisfies the heat equation:
\begin{align*}
\left\lbrace \begin{array}{l}
\partial_t g = \Delta g  \\[1.2ex]
\lim\limits_{t \rightarrow 0^+} g(t,\v{x}) = f(\v{x}) .
\end{array} \right.
\end{align*}
Dispersive estimates for $e^{t \Delta} f$ in $W^{m,p} (\R^d)$ are given by the following Lemma.
\begin{lem}[Dispersive Estimates for the Heat Kernel] \label{lem:Heat-Kernel}
For any $m \geq 0$, $1 \leq r \leq p \leq \infty$, and $f \in L^r (\R^d)$ we have
\begin{align*}
\| e^{t \Delta} f \|_{m , p} \lesssim t^{- \frac{d}{2} \left( \frac{1}{r} - \frac{1}{p} \right)} \left( 1 +  t^{- \frac{m}{2} } \right) \| f \|_{r} 
\end{align*}
\end{lem}
\begin{proof}
The proof relies on the Young's inequality for convolutions: Let $p,r,q \in [1 , \infty]$ and $1/q + 1/r = 1 + 1/p$. If $g \in L^q (\R^d)$ and $h \in L^r (\R^d)$, then $g \ast h \in L^p (\R^d)$ and
\begin{align*}
\| g \ast h \|_p \lesssim \| g \|_q \| h \|_r .
\end{align*}
Using Young's inequality, we first prove this result for $m = 0$. For the case $p = r$, to see that $e^{t \Delta} : L^r (\R^d) \rightarrow L^r (\R^d)$ is a straightforward application of Minkowski's integral inequality and scaling. For the $p = \infty$ case we rely on Young's inequality to find
\begin{align*}
\| e^{t \Delta} f \|_{\infty} \lesssim \| H_t \|_{r'} \| f \|_r \lesssim t^{- \frac{d}{2} \left( \frac{1}{r'} - 1 \right) } \| f \|_r .
\end{align*}
Since $1/r' + 1/r = 1$, we arrive at $\| e^{t \Delta} f \|_{\infty} \lesssim t^{- d / (2r)} \| f \|_r$. Let $r \leq p$. Then, using the $p = r$ and the $p = \infty$ case shown previously, together with H\"{o}lder's inequality, we conclude
\begin{align*}
\| e^{t \Delta} f \|_{p} \leq \| e^{t \Delta} f \|_{\infty}^{1 - r/p} \| e^{t \Delta} f \|_{r}^{r/p} \lesssim t^{- \frac{d}{2} \left( \frac{1}{r} - \frac{1}{p} \right)} \| f \|_{r} .
\end{align*}

Next, note that
\begin{align*}
\nabla H_t (\v{x}) = - \frac{\v{x}}{2 t} H_t (\v{x}) ,
\end{align*}
and, in the same fashion as the $m = 0$-case, we have
\begin{align*}
\| \nabla e^{t \Delta} f \|_p \lesssim t^{- \frac{d}{2} \left( \frac{1}{r} - \frac{1}{p} \right) - \frac{1}{2}} \| f \|_{r}
\end{align*}
More generally, $\dot{\Lambda}^m e^{t \Delta}$ where $\dot{\Lambda}^m = (- \Delta)^{m/2}$ may be defined as convolution against
\begin{align*}
\int_{\R^3} | 2 \pi \v{k} |^m e^{- 4 \pi^2 |\v{k}|^2 t} e^{2 \pi i \v{k} \cdot \v{x}} \dd \v{k} .
\end{align*}
One may proceed as before to show 
\begin{align*}
\| \dot{\Lambda}^m e^{t \Delta} f \|_{p} \lesssim t^{- \frac{d}{2} \left( \frac{1}{r} - \frac{1}{p} \right) - \frac{m}{2}} \| f \|_{r} .
\end{align*}

\end{proof}

The real utility of Lemma \ref{lem:Heat-Kernel} in our proof of Theorem \ref{thm:local_exist_MBMP_epsilon} is that one may exchange derivatives of the heat semigroup into a decay in time. When the choices of the exponents are made carefully this trade-off allows one to control a solution to (\ref{eq:MBMP_epsilon}) in $W^{m,p} (\R^{3N} ; \C^{2^N})$ much better than if one were just working with the Schr\"{o}dinger unitary group $e^{i t \Delta}$. There are the widely used dispersive estimates for the Schr\"{o}dinger unitary group on $\R^d$, namely
\begin{align*}
\| e^{i t \Delta} f \|_{p} \lesssim t^{- d \left( \frac{1}{2} - \frac{1}{p} \right)} \| f \|_{p'} , \hspace{1cm} p \in [2 , \infty ] .
\end{align*}
From these dispersive estimates, together with some duality arguments, we get the well-known Strichartz estimates for the Schr\"{o}dinger unitary group on $\R^d$:
\begin{align*}
\| e^{i t \Delta} f \|_{q ; p} \lesssim \| f \|_2 , \hspace{1cm} q , p \in [2 , \infty] , \hspace{1cm} \frac{2}{q} + \frac{d}{p} = \frac{d}{2} , 
\end{align*}
where $(q,p,d) \neq (2 , \infty , 2)$ \cite[Theorem 2.3]{taodispersive2006}. These estimates are extremely useful in many contexts. However, these are not as useful to us as it is impossible for such Strichartz estimates to gain regularity \cite[Exercise 2.36]{taodispersive2006}.

The last set of estimates which we discuss in this section are the Strichartz estimates for the wave equation. Similar Strichartz estimates for the Klein-Gordon equation are discussed in Chapter \ref{chap:othermodels}. We will state these estimates for the wave equation in complete generality. However, for our purposes we will only need a special case known as the energy estimate. The main point for us here is that we gain a spatial $L^2$-derivative when estimating in $L^{\infty}_T H^m (\R^d ; \R^d)$ a solution $\v{B}$ of the wave equation $\square \v{B} = \v{F}$ in terms of the inhomogeneity $\v{F}$. Before stating the estimates and proving this special case, we need several definitions (similar definitions appear in Chapter \ref{chap:othermodels}). 

We say the indices $m \geq 0$, $q , \tilde{q} \in [2,\infty]$, and $p , \tilde{p} \in [2 , \infty )$ are wave admissible if
\begin{align*}
\frac{1}{q} + \frac{d}{p} = \frac{d}{2} - m = \frac{1}{\tilde{q}'} + \frac{d}{\tilde{p}'} - 2 
\end{align*}
and
\begin{align*}
\max{ \left\lbrace \frac{1}{q} + \frac{d-1}{2p} , \frac{1}{\tilde{q}} + \frac{d-1}{2 \tilde{p}'} \right\rbrace } \leq \frac{d-1}{4} .
\end{align*}
For $T > 0$ and $(\v{a}_0 , \dot{\v{a}}_0) \in \dot{H}^m (\R^d ; \R^d) \times \dot{H}^{m-1} (\R^d ; \R^d)$, consider $\v{F} \in L_T^{\tilde{q}'} L^{\tilde{p}'} (\R^d ; \R^d)$ and define the $C_T \dot{H}^m (\R^d ; \R^d) \cap C^1_T \dot{H}^{m - 1} (\R^d ; \R^d)$-function
\begin{align}\label{eq:wave_soln}
\U_{\v{F}} [ \v{a}_0 , \dot{\v{a}}_0 ] (t) = \dot{\s}( t / \alpha ) \v{a}_0 + \alpha ~ \s(t/\alpha) \dot{\v{a}}_0 + \frac{1}{\alpha} \int_0^t \s \left( \frac{t - \tau}{\alpha} \right) \v{F} (\tau ) \dd \tau , 
\end{align}
where
\begin{align*}
\dot{\s} (t) = \cos{\left(\sqrt{ - \Delta} t \right)} : \dot{H}^m (\R^d ; \R^d) \rightarrow \dot{H}^m (\R^d ; \R^d)
\end{align*}
and
\begin{align*}
\s (t) = \frac{\sin{\left(\sqrt{ - \Delta} t \right)}}{\sqrt{ - \Delta}} : \dot{H}^{m-1} (\R^d ; \R^d) \rightarrow \dot{H}^m (\R^d ; \R^d) .
\end{align*}
It is straightforward to check that $\v{B} (t) = \U_{\v{F}} [ \v{a}_0 , \dot{\v{a}}_0 ] (t)$ satisfies the wave equation $\square \v{B} = \v{F}$ with initial conditions $( \v{B} (0) , \partial_t \v{B} (0) ) = ( \v{a}_0 , \dot{\v{a}}_0 )$. The following lemma is stated in the spirit of \cite[Theorem 2.6]{taodispersive2006}, and for original proofs of the general result see the discussion therein. We will only sketch the proof for the energy estimate we need in this thesis.  
  
\begin{lem}[Strichartz Estimates for the Wave Equation] \label{lem:Strichartz_wave_general}
Let $m \geq 0$, $2 \leq q , \tilde{q} \leq \infty$ and $2 \leq p , \tilde{p} < \infty$ be wave admissible. Further, let $T > 0$, $(\v{a}_0 , \dot{\v{a}}_0) \in \dot{H}^m (\R^d ; \R^d) \times \dot{H}^{m - 1} (\R^d ; \R^d)$, and $\v{F} \in L^{\tilde{q}'}_T L^{\tilde{p}'} (\R^d ; \R^d)$. Then, the function $\v{B} (t) = \U_{\v{F}} [ \v{a}_0 , \dot{\v{a}}_0 ] (t)$ defined by (\ref{eq:wave_soln}) is contained in $C_T \dot{H}^m (\R^d ; \R^d) \cap C^1_T \dot{H}^{m - 1} (\R^d ; \R^d)$ and satisfies the Strichartz estimate
\begin{align*}
\| \v{B} \|_{q ; p} + \max_{k \in \{0,1\}} \| \dot{\Lambda}^{m-k} \partial_t^k \v{B} \|_{\infty ; 2} \lesssim \| (\dot{\Lambda}^m \v{a}_0 ,  \dot{\Lambda}^{m-1} \dot{\v{a}}_0) \|_{2 \oplus 2} + \| \v{F} \|_{\tilde{q}' ; \tilde{p}'} .
\end{align*}
In particular, if $(\v{a}_0 , \dot{\v{a}}_0) \in H^m (\R^d ; \R^d) \times H^{m - 1} (\R^d ; \R^d)$ and $\v{F} \in L_T^1 H^{m - 1} (\R^d ; \R^d)$, then $\v{B} \in C_T H^m (\R^d ; \R^d) \cap C^1_T H^{m - 1} (\R^d ; \R^d)$ and
\begin{align*}
\max_{k \in \{0,1\}} \| \partial_t^k \v{B} \|_{\infty ; m - k , 2} \lesssim \| (\v{a}_0 , \dot{\v{a}}_0) \|_{m,2 \oplus m-1 , 2} + \| \v{F} \|_{1 ; m-1 , 2} ,
\end{align*}
which is referred to as the energy estimate for the wave equation. 
\end{lem}
\begin{proof}
A typical proof of the energy estimate argues via the Fourier transform. Indeed, the Fourier transform of $\v{B} (t)$ is given by
\begin{align*}
( \cF \v{B} (t) )(\v{k}) = \cos{\left(2 \pi |\v{k}| t / \alpha \right)} \v{a}_0 + \frac{\alpha \sin{(2 \pi |\v{k}|t / \alpha)}}{2 \pi | \v{k} |} + \frac{1}{\alpha} \int_0^t \frac{\sin{(2 \pi |\v{k}| (t - \tau) / \alpha)}}{2 \pi | \v{k} |} \v{F} (\tau) \dd \tau . .
\end{align*} 
Multiplying the above expression through by $(1 + 4 \pi^2 |\v{k}|^2)^{m/2}$, taking the $L^2$-norm and using Plancharel's Theorem, we immediately arrive at the energy estimate. 
\end{proof}

\chapter{Stability of Matter}\label{chap:stability}

As mentioned throughout Chapter \ref{chap:intro}, our study of the Cauchy problem associated with (\ref{eq:MBMP}) is primarily motivated by results concerning the energetic stability of matter in magnetic fields. Moreover, such results will play an important role in our proof of \textit{global} existence of solutions to (\ref{eq:MBMP_epsilon}) (see \S\ref{sec:proof_strat}). Therefore, this Chapter is devoted to reviewing the concept of energetic stability and the developments surrounding it. We will restrict ourselves to discussing energetic stability in the non-relativistic case and we will be especially concerned with the work of Fefferman (unpublished), Fr\"{o}hlich, Lieb, Loss, and Yau circa 1986-1995 \cite{FLL86, LL86, LY86, F95, LLS95}. For far more complete treatments see \cite{Lieb05, Loss07, LS10}. 

In its most general form the question of energetic stability asks: Is the absolute ground state $E^{\mr{G}}$ of a given Hamiltonian $H$ for a many-body quantum system is finite? Roughly, the absolute ground state $E^{\mr{G}}$ is the minimum of $\langle \psi , H \psi \rangle_{L^2}$ over all $\psi$ such that $\| \psi \|_2 = 1$. If electromagnetic fields and/or nuclei are included in the model, then we also minimize over all fields and all nuclear positions (c.f.~equation (\ref{def:abs_grnd_state_energy})). Following \cite{Lieb90} one may refine the notion of stability into two separate kinds.
\begin{itemize}
\item A system is stable of the \textit{first kind} if $E^{\mr{G}} > - \infty$ is satisfied. If $E^{\mr{G}} = - \infty$, then we say the system is \textit{unstable}.
\item An $N$-particle system is stable of the \textit{second kind} if $E^{\mr{G}} \geq - C N$ where $C > 0$ is a universal constant.
\end{itemize}
Stability of the second kind is an important notion, for if $E^{\mr{G}}$ is super-linear in the particle number, then the assembly of any two macroscopic objects would release an energy comparable to that of an atomic bomb \cite{Dyson67}. Typically, stability of the second kind is a more delicate question than stability of the first kind. We note that for the results in this thesis all we will be concerned with is whether the absolute ground state energy is finite, not its dependence on the total particle number.

\section{Stability in the Absence of Magnetic Fields}

A typical example of a system where energetic stability becomes a nontrivial question is a non-relativistic Hydrogenic atom in the absence of an electromagnetic field. That is, a single electron bound to an infinitely heavy nucleus of charge $Z > 0$. For this system, we take the Hamiltonian to be $H_{\text{Hy}} = - \Delta - Z |\v{x}|^{-1}$ and the total energy is
\begin{align*}
\langle \psi , H_{\text{Hy}} \psi \rangle_{L^2} = \int_{\R^3} |\nabla \psi (\v{x})|^2 \dd \v{x} - Z \int_{\R^3} \frac{|\psi (\v{x})|^2}{|\v{x}|} \dd \v{x} . 
\end{align*}
It is important to note that the kinetic energy and the potential energy scale like "square of an inverse length" and "inverse length", respectively. Stability will follow from the Sobolev inequality 
\begin{align*}
S_3 \left( \int_{\R^3} |\psi|^6 \right)^{1/3} \leq \int_{\R^3} |\nabla \psi|^2 ,
\end{align*}
where $S_3 = 3 \left( \pi / 2 \right)^{4/3}$ is the sharp constant (see Lemma \ref{lem:sobolev}). Indeed, if we let $\rho = |\psi|^2$ and demand $\int \rho = 1$, then, for any $R > 0$, H\"{o}lder's inequality gives us
\begin{align*}
\int_{\R^3} \frac{\rho (\v{x})}{|\v{x}|} \dd \v{x} \leq \frac{1}{R} + \left( \frac{8 \pi}{3} \right)^{2/3} R \| \rho \|_3 .
\end{align*}
Optimizing the right hand side over $R$ gives
\begin{align*}
\int_{\R^3} \frac{\rho (\v{x})}{|\v{x}|} \dd \v{x} \leq 2 \left( \frac{8 \pi}{3} \right)^{1/3} \| \rho \|_3^{1/2} .
\end{align*}
Therefore,
\begin{align*}
\langle \psi , H_{\mr{Hy}} \psi \rangle_{L^2} \geq S_3 \| \rho \|_3 - 2 Z \left( \frac{8 \pi}{3} \right)^{1/3} \| \rho \|_3^{1/2} \geq - C Z^2 . 
\end{align*}
Note that the constant $C$ in front of $Z^2$ above is approximately $0.753$. This is not entirely far from the exact truth. It is a well-known result from elementary quantum mechanics that the ground state energy $E^{\mr{G}}$ is
\begin{align*}
E^{\mr{G}} \equiv \inf{ \left\lbrace \langle \psi , H_{\text{Hy}} \psi \rangle_{L^2} ~ : ~ \| \psi \|_{2} = 1 \right\rbrace } = - \frac{Z^2}{4} 
\end{align*}
and the unique, normalized ground state wavefunction $\psi_0 \in L^2 (\R^3 ; \C)$ is the radial function
\begin{align*}
\psi_0 (|\v{x}|) = \frac{e^{- Z |\v{x}| / 2}}{2\sqrt{2 \pi}} .
\end{align*}
The point of Sobolev's inequality here is that it provides us a precise bound on the potential energy in terms of the kinetic energy, something that we will see is not possible when magnetic fields and spin-interactions are introduced. See \cite[Chapter 11]{LL01} for a more detailed discussion regarding the calculus of variations for Hamiltonians of the form $-\Delta + U$ with $U : \R^d \rightarrow \R$ being some potential.

Generalizing the above example to a system of $N$ negatively charged particles and $K$ fixed nuclei with charges $\Z = (Z_1 , \cdots , Z_K) \in [0, \infty)^K$ at distinct positions $\ul{\v{R}} = (\v{R}_1 , \cdots , \v{R}_K) \in \R^{3K}$ the Hamiltonian will become
\begin{align}\label{def:manybody_ham_nofield}
H = - \sum_{j = 1}^N \Delta_{\v{x}_j} + V (\ul{\v{R}} , \Z)
\end{align}
where $\Delta_{\v{x}_j}$ is the Laplacian acting on the $j^{\text{th}}$ variable and $V (\ul{\v{R}} , \Z)$ is given by (\ref{def:totelectrostatpot}). In this case the total energy is
\begin{align*}
\langle \psi , H \psi \rangle_{L^2} = \sum_{j = 1}^N \| \nabla_{\v{x}_j} \psi \|_2^2 + \langle \psi , V(\ul{\v{R}} , \Z) \psi \rangle_{L^2} .
\end{align*}
If the negatively charge particles are Fermions (i.e., the Pauli exclusion principle is satisfied), then the appropriate ground state energy to consider is
\begin{align}\label{eq:many_body_ground_state_asy}
E^{\mr{G}}_{\mr{asy}} = \inf{ \left\lbrace \langle \psi , H \psi \rangle_{L^2} ~ : ~ \psi \in \BWN{H^1 (\R^3 ; \C^2)} , ~ \| \psi \|_2 = 1 , ~ \ul{\v{R}}  \right\rbrace } .
\end{align}
If the negatively charged particles are Bosons (i.e., identical particles may occupy the same state), then the appropriate ground state energy to consider is 
\begin{align}\label{eq:many_body_ground_state_sym}
E^{\mr{G}}_{\mr{sym}} = \inf{ \left\lbrace \langle \psi , H \psi \rangle_{L^2} ~ : ~ \psi \in \bigotimes^N_{\mr{sym}} H^1 (\R^3 ; \C^2) , ~ \| \psi \|_2 = 1 , ~ \ul{\v{R}} \right\rbrace } .
\end{align}
It is not difficult to show that the unrestricted minimization of $\langle \psi , H \psi \rangle_{L^2}$ over normalized $\psi$, i.e., no symmetry condition required on the $\psi$, gives the Bosonic ground state $E^{\mr{G}}_{\mr{sym}}$ (see, for example, \cite[Chapter 3]{LS10}). This implies $E^{\mr{G}}_{\mr{asy}} \geq E^{\mr{G}}_{\mr{sym}}$. 

Stability of the first kind, namely $E^{\mr{G}}_{\mr{sym}} > - \infty$, was originally shown by Kato in 1951 \cite{Kato51}, whereas stability of the second kind for $E^{\mr{G}}_{\mr{asy}}$ was first considered by Dyson and Lenard in 1967-1968 \cite{DL67, DL68}. It is important to note that the Pauli exclusion principle is crucial for stability of the second kind. Dyson in 1967  \cite{Dyson67} showed that the ground state for a system of $N$ negative charges and $N$ positive charges in the absence of the exclusion principle is bounded below by a constant times $N^{7/5}$. More generally, Lieb in 1978 \cite{Lieb78} showed $E^{\mr{G}}_{\mr{sym}} \lesssim - N^{5/3}$, complementing early results showing $E^{\mr{G}}_{\mr{sym}} \gtrsim - N^{5/3}$ (see, for example, \cite{Lieb76}).

New proofs for stability of the second kind were discovered independently by Federbush, and Lieb and Thirring around 1975 \cite{Federbush75, LT75}. In particular, Lieb and Thirring discovered a completely new way to prove the stability of (\ref{eq:many_body_ground_state_asy}). Their method uses the now famous Lieb-Thirring bounds, which in their original article read 
\begin{align*}
\sum_{j \geq 1} |\lambda_j| \leq \frac{4}{15 \pi} \int_{\R^3} U_- (\v{x})^{\frac{5}{2}} \dd \v{x}
\end{align*}
where $\lambda_j$ are the negative eigenvalues of the operator $- \Delta + U$ and $U_- (\v{x}) = \max{\{- U(\v{x}) , 0 \}}$. To be precise, the stability estimate in \cite{LT75} is
\begin{align}\label{eq:stability_no_field}
E^{\mr{G}}_{\mr{asy}} \geq - 2.08 N \left[ 1 + \sqrt{ \frac{\sum_{j=1}^K Z_j^{7/3}}{N} } \right]^2  .
\end{align}

\section{Stability in the Presence of Magnetic Fields}

The introduction of an external magnetic field adds significant complexity to the stability question, especially if one takes spin-magnetic field interactions into account. If one ignores such interactions, however, then the kinetic energies $- \Delta_j$ in the many-body Hamiltonian (\ref{def:manybody_ham_nofield}) get replaced by $(\v{p}_j + \v{A}_j)^2$ where $\v{B} = \curl{\v{A}}$ is the applied magnetic field. In this situation, the Lieb-Thirring inequalities mentioned in the previous section continue to be true with the same constant when $-\Delta + U$ is replaced by $(\v{p} + \v{A})^2 + U$. Therefore, the estimate \ref{eq:stability_no_field} continues to be true with the same constant. This fact was proved by Avron, Herbst, and Simon \cite{AHS78} and independently by Combes, Schrader, and Seiler \cite{CSS78} in 1978.

The situation changes dramatically if spin-magnetic field interactions are taken into account. In the case of a Hydrogenic atom (single electron, fixed nucleus of charge $Z > 0$), we take the total energy to be $E_{\mr{P}} [\psi , \v{A}] \equiv E_{\mr{P}} [\psi , \v{A} , \v{0}]$ in the ($N=K=1$)-case, namely
\begin{align}\label{def:onebody_pauli_energy}
E_{\mr{P}} [\psi , \v{A}] = \|  \gvsig \cdot (\v{p} + \v{A}) \psi \|_{2}^2 - Z \langle \psi , |\cdot |^{-1} \psi \rangle_{L^2} + F [ \v{A} , \v{0} ] ,
\end{align}
where $F [ \v{A} , \v{0} ] = (8 \pi \alpha^2)^{-1} \int_{\R^3} |\curl{\v{A}}|^2$. The energy (\ref{def:onebody_pauli_energy}) is well-defined for any pair $(\psi , \v{A}) \in \fC_1$, where $\fC_N$ is defined by (\ref{def:function_space_C}). Whether 
\begin{align*}
E_{\mr{P}}^{\mr{G}} (Z) = \inf{ \left\lbrace E_{\mr{P}} [ \psi , \v{A} ] ~ : ~ (\psi , \v{A}) \in \fC_1 \right\rbrace }
\end{align*}
is finite or not depends crucially on the fact that there exist non-trivial, zero-energy states $(\psi , \v{A}) \in \fC_1$ to the three-dimensional\footnote{The situation in two-dimensions is completely understood and was originally worked out by Aharonov and Casher in 1979 \cite{AC79}} Dirac equation $\gvsig \cdot (\v{p} + \v{A}) \psi = 0$. The first solutions $(\psi , \v{A}) \in \fC_1$ to $\gvsig \cdot (\v{p} + \v{A}) \psi = 0$, so-called \textit{zero modes}, were discovered by Loss and Yau in 1986 \cite{LY86}. An explicit example is
\begin{align}\label{eq:LossYau_zeromode}
( \psi (\v{x}) , \v{A} (\v{x}) ) = \left( \frac{1 + i \gvsig \cdot \v{x}}{\pi (1 + |\v{x}|^2)^{3/2}} \phi_0 , 3 \frac{(|\v{x}|^2 - 1) \v{w} - 2 (\v{w} \cdot \v{x}) \v{x} - 2 \v{w} \wedge \v{x}}{(1 + |\v{x}|^2)^2}  \right) ,
\end{align} 
where $\phi_0$ is any normalized $\C^2$-vector and $\v{w} = \langle \phi_0 , \gvsig \phi_0 \rangle_{\C^2}$. One may notice that $\diver{\v{A}} \neq 0$ for the vector potential in (\ref{eq:LossYau_zeromode}). However, this can be mended by a gauge transform $(\psi , \v{A}) \mapsto (e^{i \zeta} \psi , \v{A} - \nabla \zeta)$ with the gauge function given explicitly by
\begin{align*}
\zeta (\v{x}) = 3 (\v{w} \cdot \v{x}) \frac{ |\v{x}| - \arctan(|\v{x}|)}{|\v{x}|^3} .
\end{align*}
The magnetic field $\v{B} = \curl{\v{A}}$ for the zero mode (\ref{eq:LossYau_zeromode}) is $\v{B} = - 12 \langle \psi , \gvsig  \psi \rangle_{\C^2}$ since $\v{A} = - 3 (1 + |\v{x}|^2) \langle \psi , \gvsig \psi \rangle_{\C^2}$. The magnetic vector potential $\v{A}$ in (\ref{eq:LossYau_zeromode}) may be obtained by finding a solution $\gvsig \cdot \v{p} \psi = \lambda \psi$ for some $\lambda : \R^3 \rightarrow \R$ and then choosing $\v{A} = - \lambda \langle \psi , \gvsig \psi \rangle_{\C^2} / |\psi |^2$. The vector potential can be obtained by stereographic projection from a parallel basis vector field on the three-dimensional sphere and, therefore, the flow lines are circles on the Hopf tori (see \cite{LY86}).

One can see that the existence of zero modes causes collapse in this model if the nuclear charge $Z$ is too large by the following scaling argument. Let $(\psi , \v{A}) \in \fC_1$ satisfy $\gvsig \cdot ( \v{p} + \v{A} ) \psi = 0$ and let $(\psi_{\lambda} (\v{x}) , \v{A}_{\lambda} (\v{x}) ) = ( \lambda^{3/2} \psi (\lambda \v{x}) , \lambda \v{A} (\lambda \v{x}) )$ for $\lambda > 0$. Note $\gvsig \cdot (\v{p} + \v{A}_{\lambda})\psi_{\lambda} = 0$ and $\| \psi_{\lambda} \|_2 = \| \psi \|_2$. The energy $E_{\mr{P}}$ under this scaling becomes
\begin{align*}
E_{\mr{P}}	 [\psi_{\lambda} , \v{A}_{\lambda}] = \lambda \left( - Z \langle \psi , | \cdot |^{-1} \psi \rangle_{L^2} + F [ \v{A} , \v{0} ] \right) . 
\end{align*}
That is, the Coulomb energy and the field energy scale the \textit{same way} (whereas the kinetic energy scales like "square of an inverse length", but we consider zero modes). The previous scaling argument motivates the definition of a \textit{critical charge} $Z_c$ as
\begin{align}\label{def:Zc}
Z_c := \inf{\left\lbrace \frac{F[\v{A} , \v{0}]}{\langle \psi , | \cdot |^{-1} \psi \rangle_{L^2}} ~ : ~ (\psi , \v{A}) \in \fF  \right\rbrace } .
\end{align}
where $\fF$ is the space of zero modes
\begin{align*}
\fF = \left\lbrace (\psi , \v{A}) \in \fC_1 ~ : ~ \gvsig \cdot (\v{p} + \v{A}) \psi = 0 \right\rbrace . 
\end{align*}
We see that if $Z > Z_c$, then $\lim_{\lambda \rightarrow \infty} E_{\mr{P}} [\psi_{\lambda} , \v{A}_{\lambda}] = - \infty$. Moreover, if $Z < Z_c$, then $E_{\mr{P}}$ is uniformly bounded below as the next Theorem demonstrates.
\begin{thm}[\cite{FLL86}]\label{thm:one-electron_stability_estimate}
The absolute ground state energy $E_{\mr{P}}^{\mr{G}}$ satisfies
\begin{align*}
E_{\mr{P}}^{\mr{G}} (Z) = \inf{ \left\lbrace E_{\mr{P}} [ \psi , \v{A} ] ~ : ~ (\psi , \v{A}) \in \fC_1 \right\rbrace } = \left\lbrace \begin{array}{cc}
\mr{finite} & ~ Z < Z_c \\
- \infty & ~ Z > Z_c
\end{array} \right. .
\end{align*}
\end{thm}

It is important to note that if one omits the field energy in the definition of $E_{\mr{P}}$, then instability always follows for any $Z > 0$ (this is obvious from the scaling argument given above). We also mention that if one replaces $[\gvsig \cdot (\v{p} + \v{A})]^2$ with $(\v{p} + \v{A})^2 + \frac{g}{2} \gvsig \cdot \v{B}$, where $g$ is the electron $g$-factor, then stability always holds if $g < 2$ and stability never holds when $g > 2$, \textit{even} when $Z = 0$ (c.f \cite{FLL86}). Since physically $g \simeq 2.0023$ as a result of Quantum Electrodynamics, we run into an inconsistency in the theory with $(\v{p} + \v{A})^2 + \frac{g}{2} \gvsig \cdot \v{B}$. The only way to mend this is to include all Quantum Electrodynamical effects, but this appears to be a very difficult unresolved problem. 
  
The authors in \cite{FLL86} give the lower bound $Z_c > Z_c^L := 3/(\pi \alpha^2) \simeq 17,932$. An upper bound is $Z_c \leq (3 \pi)^2 / (8 \alpha^2) \simeq 208,398$ which follows from plugging in the Loss-Yau zero mode (\ref{eq:LossYau_zeromode}) into the definition of $Z_c$. Therefore, $Z_c$ is far larger than the charges encountered on the periodic table. For $Z < Z_c^L$, a lower bound for $E_{\mr{P}}^{\mr{G}} (Z)$ is given by
\begin{align*}
E_{\mr{P}}^{\mr{G}} (Z) \geq - \frac{1}{4} Z^2 - \frac{Z^3}{32 Z_c^L} \left( 1 - \frac{3}{4} \frac{Z}{Z_c^L} \right)^{3/2} . 
\end{align*}
We would like to point out that the exact computation of $Z_c$ and $E_{\mr{P}}^{\mr{G}} (Z)$, as well as whether $E_{\mr{P}}^{\mr{G}} (Z)$ diverges as $Z \rightarrow Z_c$, remain interesting open problems.

In the same year, Lieb and Loss \cite{LL86} proved a generalization of Theorem \ref{thm:one-electron_stability_estimate} for the many-electron atom and the one-electron molecule. In the case of the many-electron atom\footnote{By many-electron atom we mean $N > 1$ electrons and a single fixed nucleus of charge $Z > 0$} stability occurs if $Z \alpha^{12/7} < \text{constant}$. Whereas, in the case of the one-electron molecule\footnote{By one-electron molecule we mean a single electron interacting with $K > 1$ static nuclei at distinct positions $\ul{\v{R}} = (\v{R}_1 , \cdots , \v{R}_K) \in \R^{3K}$ and charges $\Z = (Z_1 , \cdots , Z_K) \in [0, \infty)^K$} one requires $\alpha^2 \max{\Z} < \text{constant}$ \textit{and} $\alpha < \alpha_c$ where $0.32 < \alpha_c < 6.67$. It is important to note that for stability of the one-electron molecule one requires a bound on $\alpha$, regardless of the size of the charges. To demonstrate why this is the case, assume, for simplicity, that all the $K > 1$ nuclei possess the same charge $Z > 0$ that is not required to be an integer. If $(\psi , \v{A}) \in \fC_1$ is a zero mode pair, then, in this case, the many-body Pauli energy (\ref{eq:mbmp_energy}) is
\begin{align*}
E_{\mr{P}} [\psi , \v{A}] = E_{\mr{P}} [\psi , \v{A} , \v{0}] = - Z \sum_{j=1}^K \langle \psi , |\v{x} - \v{R}_j|^{-1} \psi \rangle_{L^2} + Z^2 \sum_{1 \leq i < j \leq K} |\v{R}_i - \v{R}_j|^{-1} + F [\v{A} , \v{0}] .
\end{align*}
Multiplying $E_{\mr{P}} [\psi , \v{A}]$ on the left by the $\C^{2^K}$-valued function $\Psi (\ul{\v{R}}) = \otimes_{j = 1}^K \psi (\v{R}_j)$, and on the right by its conjugate transpose, and then integrating, one finds
\begin{align}\label{eq:stability_bound_alpha_1}
\int_{\R^{3K}} \langle \Psi (\ul{\v{R}}) , E_{\mr{P}} [\psi , \v{A}] \Psi (\ul{\v{R}}) \rangle_{\C^{2^K}} \dd \ul{\v{R}} = 2 \left( - Z K + \frac{1}{2} Z^2 K (K-1) \right) I [\psi] + F[\v{A} , \v{0}] , 
\end{align}
where 
\begin{align*}
I [\psi] = \frac{1}{2} \iint_{\R^3 \times \R^3} \frac{| \psi  (\v{x}) |^2 |\psi (\v{y})|^2 }{ |\v{x} - \v{y} |} \dd \v{x} \dd \v{y} .
\end{align*}
Choosing $Z = 1/K$ and taking the $K \rightarrow \infty$ limit, the right hand side of (\ref{eq:stability_bound_alpha_1}) becomes $F[\v{A} , \v{0}] - I[\psi]$. Hence, if we choose $\alpha$ so that
\begin{align*}
8 \pi \alpha^2 > \inf{ \left\lbrace \| \v{B} \|^2 / I[\psi] ~ : ~  (\psi , \v{A}) \in \fF \right\rbrace } ,
\end{align*}
then instability follows. Using the explicit Loss-Yau zero mode (\ref{eq:LossYau_zeromode}), we may actually compute the upper bound $\alpha_c < 3 \pi / \sqrt{2} \simeq 6.67$. 

Stability of the full many-body Pauli Hamiltonian (\ref{def:manybody_Pauli_Hamiltonian}) was first resolved independently by Fefferman (unpublished; see \cite{F95} for an announcement), and Lieb, Loss, and Solovej \cite{LLS95} in 1995. The main Theorem of \cite{LLS95} reads as follows.
\begin{thm}[\cite{LLS95}]\label{thm:many-body_stability_estimate}
If $\alpha \leq 0.06$ and $\alpha^2 \max{\Z} \leq 0.041$, then the ground state energy $E_{\mr{P}}^{\mr{G}}$ defined in (\ref{def:abs_grnd_state_energy}) satisfies
\begin{align}\label{eq:general_stability_estimate}
E_{\mr{P}}^{\mr{G}} \geq - C(\alpha , \Z) N^{1/3} K^{2/3}  ,
\end{align}
where $C (\alpha , \Z) > 0$ is a constant depending only on $\alpha$ and $\Z$. 
\end{thm}

That is, for small enough $\max{\Z}$ and $\alpha$, the total energy $E_{\mr{P}}$ associated with the many-body Pauli Hamiltonian (\ref{def:manybody_Pauli_Hamiltonian}) is bounded below with lower bound linear in the total particle number with a constant independent of the magnetic field $\v{B} = \curl{\v{A}}$ \textit{and} the positions of the nuclei $\ul{\v{R}}$. We note the antisymmetry condition in the definition of $\fC_N$ (\ref{def:function_space_C}) is crucial for this result, as minimizing with respect to Bosonic (completely symmetric) wavefunctions results in collapse \cite[Theorem 9.3]{LS10}. Optimal values for $\alpha$ and $\alpha^2 \max{\Z}$ appears to be a difficult open problem. 

To conclude this section and Chapter, we provide a brief discussion addressing the natural and interesting question of zero mode solutions to the Maxwell-Pauli equations (\ref{eq:MP}). For this, we recall some formulas regarding zero modes. Consider a zero mode pair $(\psi , \v{A}) \in \fC_1$. Since $\gvsig \cdot (\v{p} + \v{A}) \psi = 0$, we can use the relation $\sigma^i \sigma^j = \delta_{ij} I + i \epsilon_{ijk} \sigma^k$ among the Pauli matrices $\sigma^j$ to deduce
\begin{align}\label{eq:zeromode_comp_1}
0 = \langle \gvsig \psi , \gvsig \cdot (\v{p} + \v{A}) \psi \rangle_{\C^2} = \langle \psi , (\v{p} + \v{A}) \psi \rangle_{\C^2} + i \langle \psi , (\v{p} + \v{A}) \wedge \gvsig \psi \rangle_{\C^2} .
\end{align}
Taking the real part of both sides of (\ref{eq:zeromode_comp_1}) and solving for $\v{A}$ we find
\begin{align}\label{eq:zeromode_A_potential}
\v{A} = - \frac{ \curl{ \langle \psi  , \gv{\sigma} \psi \rangle_{\C^2} } + 2 \im{ \langle \psi , \nabla \psi \rangle_{\C^2} } }{ 2 \langle \psi , \psi \rangle_{\C^2}  } .
\end{align}
The same formula (\ref{eq:zeromode_A_potential}) appears in \cite{LY86}. Interestingly, the imaginary part gives $\v{A} \wedge \langle \psi , \gvsig \psi \rangle_{\C^2} + \im{ \langle \psi , \nabla \wedge \gvsig \psi \rangle_{\C^2} } = \frac{1}{2} \nabla |\psi|^2$. 

Let $V (t) \in C^1 (\R^3)$ be a (possibly time-dependent) potential. Suppose there exists a sufficiently smooth (to justify the following computations) solution $(\psi (t) , \v{A} (t)) \in \fC_1$, with $\langle \psi , \psi \rangle_{\C^2} > 0$, to the system
\begin{align*}
\left\lbrace \begin{array}{l}
i \partial_t \psi = V \psi   \\
\square \v{A} = 0  \\
\gvsig \cdot (\v{p} + \v{A}) \psi = 0 .
\end{array}  \right.
\end{align*}
Using $i \partial_t \psi = V \psi$ and (\ref{eq:zeromode_A_potential}) we compute $\partial_t \v{A} = - \nabla V$. Therefore, $\partial_t \curl{\v{A}} = 0$ and, hence, the corresponding magnetic field $\v{B} = \curl{\v{A}}$ is time-independent. However, by taking the curl of $\square \v{A} = 0$ we see that the magnetic field also satisfies $\square \v{B} = 0$. Hence, $\Delta \v{B} = 0$ and, since $\v{B} \in L^2 (\R^3 ; \R^3)$, this forces $\v{B} \equiv 0$. Consequently,   $(\psi , \v{A})$ must be the trivial solution and either the solution $(\psi , \v{A})$ is not smooth enough to justify the previous computations or there are \textit{no} non-trivial zero mode solutions to the MP equations. A more careful argument would also cover the case when $V$ has singularities (e.g., $V (|\v{x}|) = - |\v{x}|^{-1}$).

\chapter{Other Models}\label{chap:othermodels}

To our knowledge, there is currently no well-posedness theory for solutions to the MBMP equations (\ref{eq:MBMP}) for any kind of initial data, \textit{even} in the single electron case with no nuclei present, i.e., equations (\ref{eq:MP}). To contrast this observation and to further motivate our study of the MBMP equations, this chapter is devoted to reviewing two similar models for a charged quantum particle interacting with its self-generated electromagnetic field: the \textbf{Maxwell-Schr\"{o}dinger equations} (abbr.~\textbf{MS equations}) and \textbf{Maxwell-Dirac equations}  (abbr.~\textbf{MD equations}). At the end of this chapter we also mention another possible model based on the Brown-Ravenhall operator that to our knowledge has not yet be considered in the literature.

\section{The Maxwell-Schr\"{o}dinger System}\label{sec:MS}

For the case of a single electron and no nuclei in $d$ dimensions, the MS equations considered in the literature read
\begin{align}\label{eq:MS}
\left\lbrace \begin{array}{l}
i \partial_t \psi = \left( (\v{p} + \v{A})^2 - \frac{1}{\alpha} \varphi \right) \psi  \\
\square \v{A} + \nabla ( \diver{\v{A}} + \alpha \partial_t \varphi ) = 4 \pi \alpha \v{J}_{\mr{S}} [\psi , \v{A}]  \\
- \Delta \varphi - \alpha \partial_t \diver{\v{A}} = - 4 \pi \alpha | \psi |^2 ,
\end{array} \right.  
\end{align}
where $\psi (t) : \R^{d} \rightarrow \C$ is the single-particle wave function without spin, $(\varphi (t) , \v{A} (t)) : \R^{d} \times \R^d \rightarrow \R \times \R^d$ are the electromagnetic potentials, and
\begin{align}\label{def:schr_prob_current}
\v{J}_{\mr{S}} [\psi , \v{A}] = - 2 \alpha \re{ \langle \psi , (\v{p} + \v{A}) \psi \rangle_{\C} }
\end{align}
is the Schr\"{o}dinger probability current. We note that, at least formally, there are two conserved quantities associated to (\ref{eq:MS}): $\| \psi \|_{2}$ and the total energy
\begin{align}\label{def:schr_energy}
E_{\mr{S}} [\psi , \v{A}] = \| (\v{p} + \v{A}) \psi \|_{2}^2 + F [ \v{B} , \v{E} ] ,
\end{align}
where $F[\v{B} , \v{E}] = F [ \curl{ \v{A} } , - \nabla \varphi - \alpha \partial_t \v{A} ]$ is given by (\ref{def:field_energy_potentials}). 

Equations (\ref{eq:MS}) may be considered as the classical approximation to the quantum field equations for an electrodynamical non-relativistic many body system, and these equations also appear as a model for laser physics \cite{Healy1982}. Formal justification for the model (\ref{eq:MS}) follows a standard quantization procedure and such a derivation may be found in \cite[Section 2]{Petersen2014}\footnote{The derivation in \cite{Petersen2014} actually gives the MS system \textit{without} the Coulomb self-interaction $\varphi$ included.}. Furthermore, equations (\ref{eq:MS}) are invariant under the gauge transformation $(\psi , \varphi , \v{A} ) \longmapsto (\psi e^{-i \zeta} , \varphi - \alpha \partial_t \zeta , \v{A} + \nabla \zeta)$ where $\zeta : \R^d \rightarrow \R$ is some gauge function. As with the MBMP equations, we will focus on the Coulomb gauge $\diver{\v{A}} = 0$, and in this gauge equations (\ref{eq:MS}) become 
\begin{align}\label{eq:MS_Coulomb}
\left\lbrace \begin{array}{l}
i \partial_t \psi = \left( (\v{p} + \v{A})^2 - \frac{1}{\alpha} \varphi \right) \psi  \\
\square \v{A} = 4 \pi \alpha \Hproj{ \v{J}_{\mr{S}} [\psi , \v{A}] }  \\
- \Delta \varphi = - 4 \pi \alpha |\psi |^2 \\
\diver{\v{A}} = 0 ,
\end{array} \right.  
\end{align}
As we argued in \S\ref{sec:MP} it doesn't make physical sense to include the term corresponding to the Coulomb self-interaction $\varphi = - \alpha |\cdot|^{-1} * |\psi|^2$ in (\ref{eq:MS_Coulomb}). However, this term is present in the mathematical literature studying (\ref{eq:MS}) and so we leave it in (\ref{eq:MS_Coulomb}) for clarity. 

The three dimensional MS equations (\ref{eq:MS_Coulomb}) (without the electrostatic potential $\varphi$) should be compared to the MP equations, which we rerecord here for convenience:
\begin{align*}
\left\lbrace \begin{array}{l}
i \partial_t \psi = [\gvsig \cdot (\v{p} + \v{A})]^2 \psi  \\
\square \v{A} = 4 \pi \alpha \Hproj{ \v{J}_{\mr{P}} [\psi , \v{A}] }  \\
\diver{\v{A}} = 0 .
\end{array} \right. 
\end{align*}
Using the relation $\sigma^i \sigma^j = \delta_{ij} I + i \epsilon_{ijk} \sigma^k$ for Pauli matrices it is straightforward to show that 
\begin{align}\label{eq:expandedPauli}
[ \gvsig \cdot (\v{p} + \v{A}) ]^2 = (\v{p} + \v{A})^2 + \gvsig \cdot \curl{\v{A}} ,
\end{align}
and
\begin{align}\label{eq:expandedCurrent}
\v{J}_{\mr{P}} [\psi , \v{A}] = \v{J}_{\mr{S}} [\psi , \v{A}] - \alpha \curl{\langle \psi , \gvsig \psi \rangle_{\C^2}} .
\end{align}
From (\ref{eq:expandedPauli}) we see that the difference between the Pauli operator $[\gvsig \cdot (\v{p} + \v{A})]^2$ and the magnetic Schr\"{o}dinger operator $(\v{p} + \v{A})^2$ is the coupling between the spin of the electron and the magnetic field $\v{B} = \curl{\v{A}}$, namely $\gvsig \cdot \v{B}$. Likewise, from (\ref{eq:expandedCurrent}), the difference between the Pauli probability current $\v{J}_{\mr{P}} [\psi , \v{A}]$, defined by (\ref{def:Pauli_prob_current}), and the Schr\"{o}dinger probability current $\v{J}_{\mr{S}} [\psi , \v{A}]$, defined by (\ref{def:schr_prob_current}), is the appearance of the \textit{spin current} $\curl{\langle \psi , \gvsig \psi \rangle_{\C^2}}$. Therefore, we may view (\ref{eq:MS_Coulomb}) (without the electrostatic potential $\varphi$) as the physical approximation to (\ref{eq:MP}) where the spin of the electron is neglected.

There is an extensive research literature studying the MS system (\ref{eq:MS}). For literature concerning the Cauchy problem see \cite{Nakamitsu1985, Nakamitsu1986, tsutsumi1993, guo1995, NW05, NW07, Bejenaru2009, WADA2012, Petersen2014, PjP14}, for scattering theory \cite{tsutsumi1993, Ginibre2003, Shimomura2003, Ginibre2006, Ginibre2007, ginibre2008}, for numerics studies \cite{Chupeng2018, Chupeng2019_1, Chupeng2019_2}, and for nonlinear extensions see \cite{ADM17, AMS19}. To our knowledge, Nakamitsu and Tsutsumi in \cite{Nakamitsu1985, Nakamitsu1986} provide the first mathematical treatment of the Cauchy problem associated with (\ref{eq:MS}). In \cite{Nakamitsu1986} the authors treat the system (\ref{eq:MS}) in the Lorenz gauge $\diver{\v{A}} + \alpha \partial_t \varphi = 0$ and show local (in time) well-posedness for initial data  
\begin{align*}
(\psi (0) , \varphi (0) , \dot{\varphi} (0) , \v{A} (0) , \dot{\v{A}} (0)) \in ~ & H^m (\R^d ; \C) \oplus H^{m+1} (\R^d ; \R) \oplus H^m (\R^d ; \R) \\
& \oplus H^m (\R^d ; \R^d) \oplus H^{m-1} (\R^d ; \R^d)  ,
\end{align*}
where $m$ is a positive integer satisfying $m \geq \frac{d}{2} + 2$. In one and two space dimensions, the local solutions obtained are shown to be global using the conservation of energy and the $L^2$-norm. 

Tsutsumi in \cite{tsutsumi1993} studies global existence and asymptotics for the three dimensional problem (\ref{eq:MS}) in the Coulomb gauge. For a certain set of scattered states, i.e., data given at $t = + \infty$, Tsutsumi constructs the modified wave operator and shows (\ref{eq:MS}) admits a global solutions for initial data in the range of this wave operator. Further investigations into the scattering theory of (\ref{eq:MS}) may be found in the paper of \cite{Shimomura2003} and the papers of Ginibre and Velo \cite{Ginibre2003, Ginibre2006, Ginibre2007, ginibre2008}

Using a strategy similar to the one employed in this thesis to study (\ref{eq:MBMP}), Guo, Nakamitsu, and Strauss in \cite{guo1995} prove the three dimensional problem (\ref{eq:MS}) in the Coulomb gauge admits global weak solutions for finite-energy initial data 
\begin{align*}
(\psi (0) , \v{A} (0) , \dot{\v{A}} (0) ) \in H^1 (\R^3 ; \C) \oplus H^1 (\R^3 ; \R^3) \oplus L^2 (\R^3 ; \R^3) 
\end{align*}
satisfying $\diver{\v{A}} (0) = \diver{\dot{\v{A}}} (0) = 0$. There the authors consider an $\varepsilon$-modified version of the MS equations (\ref{eq:MS}) that, in the Coulomb gauge, read 
\begin{align}\label{eq:MS_epsilon}
\left\lbrace \begin{array}{l}
\partial_t \psi = - ( i + \epsilon ) (\v{p} + \v{A})^2 \psi - i \frac{1}{\alpha} \varphi \psi  \\
\square \v{A} = 4 \pi \alpha \Hproj{ \v{J}_{\mr{S}} [\psi , \v{A}] }  \\
- \Delta \varphi = - 4 \pi \alpha | \psi |^2  \\
\diver{\v{A}} = 0 .
\end{array} \right.  
\end{align}
By taking advantage of the regularity-improving, dispersive properties of the heat kernel $e^{\varepsilon t \Delta}$ (see Lemma \ref{lem:Heat-Kernel}) and the dissipative $L^2$-norm and energy associated with (\ref{eq:MS_epsilon}), the authors in are able to prove the existence of low regularity time global solutions to (\ref{eq:MS_epsilon}). Then, by using a compactness argument to consider the $\varepsilon \rightarrow 0$ limit, they prove these low regularity global solutions to (\ref{eq:MS_epsilon}) converge to time global, finite-energy, weak solutions of (\ref{eq:MS}). The existence of global solutions in the Lorenz gauge and the \textit{temporal gauge} $\varphi = 0$ is then shown by choosing appropriate gauge transformations. 
 
Significant progress concerning the Cauchy problem for the three dimensional MS equations (\ref{eq:MS}) was made by Nakamura and Wada in \cite{NW05, NW07}. Specifically, Nakamura and Wada in \cite{NW05} develop the following local well-posedness theory for solutions to the three dimensional MS equations (\ref{eq:MS}) in the Coulomb gauge (a similar result is also obtained in the Lorenz and temporal gauges). Define 
\begin{align*}
\sR = \left\lbrace (s,r) ~ : ~   s \geq 5/3, ~ \max{\left\lbrace 4/3 , s-2 , (2s-1)/4\right\rbrace} \leq \sigma \leq \min{\left\lbrace s+1 ,(5s-2)/3\right\rbrace} \right\rbrace ,
\end{align*}
and
\begin{align}\label{eq:MS_initial_conditions}
\sZ_{0,\mr{S}}^{s,r} := \left\lbrace  ( \psi_0, \v{a}_0 , \dot{\v{a}}_0) \in H^s (\R^3 ; \C) \oplus (H^{r} \oplus H^{r-1}) (\R^3 ; \R^3) : \diver{\v{a}_0} = \diver{\dot{\v{a}}_0} = 0 \right\rbrace .
\end{align}
Consider, again, the equations (\ref{eq:MS}) in the Coulomb gauge
\begin{align}\label{eq:MS_Coulomb_nophi}
\left\lbrace \begin{array}{l}
i \partial_t \psi = (\v{p} + \v{A})^2 \psi  \\
\square \v{A} = 4 \pi \alpha \Hproj{ \v{J}_{\mr{S}} [\psi , \v{A}] }  \\
\diver{\v{A}} = 0 ,
\end{array} \right.  
\end{align}
where for simplicity we are ignoring the term corresponding to the Coulomb self-interaction $\varphi = - \alpha |\cdot|^{-1} * |\psi|^2$. The following Theorem and discussion does not change in any crucial way whether one chooses to include or not include this nonlinear term. 

\begin{thm}[\cite{NW05}]\label{thm:nakamura05}
Fix $(s,r) \in \sR$ with $(s , r) \neq ( 5/2, 7/2 ), (7/2, 3/2)$. Then for any $(\psi_0 , \v{a}_0 , \dot{\v{a}}_0 ) \in \sZ_{0,\mr{S}}^{s,r}$, there exists a $T > 0$  such that (\ref{eq:MS_Coulomb_nophi}) with initial condition 
\begin{align*}
(\psi (0) , \v{A} (0) , \partial_t \v{A} (0) ) = (\psi_0 , \v{a}_0 , \dot{\v{a}}_0 )
\end{align*}
has a unique solution $(\psi , \v{A})$ satisfying $(\psi , \v{A} , \partial_t \v{A}) \in C_T \sZ_{0,\mr{S}}^{s,r}$. Moreover, if
\begin{align*}
r \geq \max{ \left\lbrace (s-1) , (2s+1)/4 \right\rbrace }
\end{align*}
with $(s , r) \neq (5/2 , 3/2)$, then the map $(\psi_0 , \v{a}_0 , \dot{\v{a}}_0 ) \mapsto (\psi , \v{A} , \partial_t \v{A})$ is continuous as a map from $\sZ_{0,\mr{S}}^{s,r}$ to $C_T \sZ_{0,\mr{S}}^{s,r}$. 
\end{thm}

In \cite{NW07} Nakamura and Wada improve the local well-posedness theory in Theorem \ref{thm:nakamura05} down to $s \geq 11/8$ and $r > 1$ using a variation of Strichartz estimates for the Schr\"{o}dinger unitary group $e^{i t \Delta}$. The Strichartz estimates we refer to were first developed by Koch-Tzvetkov \cite{KT03} and Kenig-Koenig \cite{KK03} for the Benjamin-Ono equation, and the adapted to Schr\"{o}dinger equations by J.~Kato \cite{KatoJ05}. With these very low-regularity local solutions, Nakamura and Wada manage to use conservation of energy $E_{\mr{S}} [\psi , \v{A}]$, together with a blow-up alternative, to show that these solutions in fact exist for all time, thereby proving global well-posedness of (\ref{eq:MS_Coulomb_nophi}). Going even further, in \cite{Bejenaru2009} Bejenaru and Tataru manage to show (\ref{eq:MS_Coulomb}) is globally well-posedness in the energy class $(\psi_0 , \v{a}_0 , \dot{\v{a}}_0 ) \in \sZ_{0,\mr{S}}^{1,1}$ by using the analysis of a short time wave packet parametrix for the magnetic Schr\"{o}dinger equation and the related linear, bilinear, and trilinear estimates. 

Petersen in \cite{Petersen2014}, as part of his PhD thesis, was the first to study the Cauchy problem for the many-body Maxwell-Schr\"{o}dinger system, 
\begin{align}\label{eq:MBMS_Coulomb}
\left\lbrace \begin{array}{l}
i \partial_t \psi = \left( \sum\limits_{j = 1}^N (\v{p}_j + \v{A}_j)^2 + \sum\limits_{1 \leq j < k \leq N} \dfrac{Q_j Q_k}{|\v{x}_j - \v{x}_k|}  \right) \psi  \\[1.2ex]
\square \v{A} = 4 \pi \alpha \Hproj{ \J_{\!\mr{S}} [\psi , \v{A}] }  \\
\diver{\v{A}} = 0 ,
\end{array} \right.  
\end{align}
where $Q_j \in \R$ for $j \in \{1 , \cdots , N\}$ and $\v{J}_{\mr{S}} [\psi , \v{A}] (t) : \R^3 \rightarrow \R^3$ is the total Schr\"{o}dinger probability current density and is given by
\begin{align*}
\J_{\!\mr{S}} [\psi , \v{A}] = - 2 \alpha \sum_{j = 1}^N \re{ \int_{\R^{3 (N-1)}} \langle \psi (\ul{\v{x}}) , (\v{p}_j + \v{A}_j) \psi (\ul{\v{x}}) \rangle_{\C} \dd \ul{\v{x}}_j' } ,
\end{align*}
where $\ul{\v{x}}_j' = ( \v{x}_1 , \cdots , \v{x}_{j-1} , \v{x}_{j+1} , \cdots , \v{x}_N )$. Using some of the techniques of Nakamura and Wada \cite{NW05}, Petersen shows local existence and uniqueness of solutions to (\ref{eq:MBMS_Coulomb}) for initial data $(\psi_0 , \v{a}_0 , \dot{\v{a}}_0) \in \sZ_{0,\mr{S}}^{2 , \frac{3}{2}}$. It is an interesting open problem whether such local solutions to (\ref{eq:MBMS_Coulomb}) can be made global, as we expect to be the case. In this regard, it is not clear that the Koch-Tzvetkov Strichartz estimates for the Schr\"{o}dinger kernel used by Nakamura and Wada in \cite{NW07} can be used to overcome regularity issues posed by the Coulomb potential singularity $|\v{x}|^{-1}$. A very careful analysis that is highly adapted to deal with the Coulomb potential seems to be needed. Lastly, we note that in \cite{PjP14} Petersen and Solovej show the existence of traveling wave solutions, i.e., solutions of the form $(\psi (t) (\v{x}) , \v{A} (t) (\v{x}) ) = (e^{-i \omega t} \psi (\v{x} - \v{v}t) , \v{A} (\v{x} - \v{v} t))$, to both (\ref{eq:MS_Coulomb_nophi}) \textit{and} (\ref{eq:MP}) provided that the speed $|\v{v}|$ of the wave is not too large.

At the present moment, and to the best of our knowledge, Petersen is the only other author in the literature who has brought attention to the lack of a well-posedness theory for the MP equations (\ref{eq:MP}). One possible reason for this is that the Cauchy problem for the MP equations presents unexpected difficulties when attempting to directly adapt any single strategy mentioned previously used to study the MS system. To give one example where difficulties arise, we will give a proof sketch of the existence part of Theorem \ref{thm:nakamura05} in the special case $(s , r) = (2 , 3/2)$. The benefit of doing so is two-fold. First, this will allow us to properly explain precisely the difficulties that arise in a direct adaptation of the strategy used to prove Theorem \ref{thm:nakamura05} to produce a local well-posedness result for (\ref{eq:MP}). Second, the work \cite{NW05} provides us with a set of techniques which will nevertheless be useful in our strategy to prove Theorem \ref{thm:weak_solns_MBMP} (similar techniques may also be found in \cite{Petersen2014, ADM17, AMS19}). For the rest of this section we will use to notation introduced in \S\ref{sec:notation}. 

To solve (\ref{eq:MS_Coulomb_nophi}), one considers the Schr\"{o}dinger equation and the wave equation separately; specifically,
\begin{align}\label{eq:Schrodinger_Equation}
\left\lbrace \begin{array}{l}
i \partial_t \xi = (\v{p} + \v{A})^2 \xi  \\
\xi (0) = \psi_0 ,
\end{array} \right.
\end{align}
and 
\begin{align}\label{eq:Klein_Gordon_Equation}
\left\lbrace \begin{array}{l}
(\square + 1) \v{K} =  4 \pi \alpha \Hproj{ \v{J}_{\mr{S}} [ \psi , \v{A} ] } + \v{A}  \\
(\v{K} (0) , \partial_t \v{K} (0) ) = ( \v{a}_0 , \dot{\v{a}}_0 ) .
\end{array} \right.
\end{align}
We will always assume $\diver{\v{A}} = 0$ and $(\psi_0 , \v{a}_0 , \dot{\v{a}}_0 ) \in \sZ_{0,\mr{S}}^{2,\frac{3}{2}}$. In (\ref{eq:Schrodinger_Equation}) and (\ref{eq:Klein_Gordon_Equation}) we regard $\psi$ and $\v{A}$ as known, time-dependent functions. As is typically the case with existence theorems for PDEs, the proof of Theorem \ref{thm:nakamura05} employs a fixed point argument on a carefully chosen metric space. Therefore, we define the space
\begin{align*}
\sZ_{\mr{S}}^{s,r} (T) = & \left\lbrace (\phi , \v{A} ) \in L^{\infty}_T H^s (\R^3 ; \C) \oplus ( L^{\infty}_T H^{r} (\R^3 ; \R^3) \cap W^{1,6}_T L^3 (\R^3 ; \R^3) ) \right. \\
& ~~ \left. \text{s.t.} ~ \| \psi \|_{\infty; s,2} \leq R_1, ~ \max{ \left\lbrace \| \v{A} \|_{\infty; r,2} , \| \partial_t \v{A} \|_{6 ; 3} \right\rbrace } \leq R_2 , ~ \diver{\v{A}} = 0 \right\rbrace \numberthis \label{eq:MS_metric_space} 
\end{align*}
where $T , R_1 , R_2 > 0$ will be chosen later. By standard functional analysis arguments, endowing $\sZ_{\mr{S}}^{s,r} (T)$ with the metric
\begin{align}\label{eq:MS_metric}
d ( (\psi , \v{A} ) , (\psi' , \v{A}^{\!\prime}) ) = \max{ \left\lbrace \| \psi - \psi' \|_{\infty; 2} , \| \v{A} - \v{A}^{\!\prime} \|_{\infty ; \frac{1}{2} , 2} , \| \v{A} - \v{A}^{\!\prime} \|_{4;4} \right\rbrace } ,
\end{align}
gives us a complete metric space for each $T,R_1,R_2 >0$ (see, for example, \cite[Lemma 14]{Petersen2014} or Lemma \ref{lem:completeness} in \S\ref{sec:metric_space_linearization}). 

The key idea behind the proof of Theorem \ref{thm:nakamura05} is to then consider the solution map
\begin{align*}
\Phi : (\psi , \v{A}) \in \sZ_{\mr{S}}^{2,\frac{3}{2}} (T) \longmapsto (\xi , \v{K})
\end{align*}
where $\xi$ and $\v{K}$ satisfy (\ref{eq:Schrodinger_Equation}) and (\ref{eq:Klein_Gordon_Equation}), respectively, and show that $\Phi$ is a contraction map on $(\sZ_{\mr{S}}^{2,\frac{3}{2}} (T) , d)$ for some choice of $T , R_1 , R_2 > 0$. By the Banach fixed point theorem we will then have the existence of a unique local solution to (\ref{eq:MS_Coulomb_nophi}). To show $\Phi$ is a contraction on $\sZ_{\mr{S}}^{2,\frac{3}{2}} (T)$, we need estimates on the solutions $\xi$ and $\v{K}$ to (\ref{eq:Schrodinger_Equation}) and (\ref{eq:Klein_Gordon_Equation}), respectively.

For the Klein-Gordon equation (\ref{eq:Klein_Gordon_Equation}) the desired estimates are the Strichartz estimates proved by Strichartz \cite{Stri77}, Brenner \cite{Bren84}, and Ginibre and Velo \cite{GV85, GV95}. Following the formulation in \cite[Lemma 4.1]{NW05} (see also Lemma \ref{lem:Strichartz_wave_general} in \S\ref{sec:fundamental_results} of this thesis), let $r \in \R$, $(\v{a}_0 , \dot{\v{a}}_0 ) \in H^r \times H^{r-1}$ and for $\v{F} \in L^1_T H^{r - 1}$ define the continuous function
\begin{align}\label{eq:Klein-Gordon_soln}
\V_{\v{F}} [ \v{a}_0 , \dot{\v{a}}_0 ] (t) = \dot{S} \left( \frac{t}{\alpha} \right) \v{a}_0 + \alpha S \left( \frac{t}{\alpha} \right) \dot{\v{a}}_0 + \frac{1}{\alpha} \int_0^t S \left( \frac{t - \tau}{\alpha} \right) \v{F} (\tau ) \dd \tau ,
\end{align}
where
\begin{align*}
\dot{S} (t) = \cos{\left(\sqrt{1 - \Delta} t \right)} : H^r \rightarrow H^r
\end{align*}
and
\begin{align*}
S (t) = \frac{\sin{\left(\sqrt{1 - \Delta} t \right)}}{\sqrt{1 - \Delta}} : H^{r-1} \rightarrow H^r.
\end{align*}
It is straightforward to check that $\v{K} (t) = \V_{\v{F}} [ \v{a}_0 , \dot{\v{a}}_0 ] (t)$ is contained in $C_T H^r (\R^3 ; \R^3) \cap C^1_T H^{r - 1} (\R^3 ; \R^3)$ and satisfies the Klein-Gordon equation $(\square + 1) \v{K} = \v{F}$ with initial conditions $\v{K} (0) = \v{a}_0$ and $\partial_t \v{K} (0) = \dot{\v{a}}_0$. Moreover, the following Strichartz estimates hold:
\begin{align}\label{eq:Strichartz_KG_general}
\max_{k \in \left\lbrace0,1\right\rbrace} \| \partial_t^k \v{K} \|_{q_0 ; r - k - \frac{2}{q_0} , r_0} \lesssim \| (\v{a}_0 , \dot{\v{a}}_0 ) \|_{r , 2 \oplus r-1 , 2} + \| \v{F} \|_{q_1' ; r - 1 + \frac{2}{q_1}  , r_1'} 
\end{align}
where $0 \leq \frac{2}{q_k} = 1 - \frac{2}{r_k} < 1$ for $k \in \left\lbrace0,1\right\rbrace$, and $(q_1' , r_1')$ are the H\"{o}lder dual indices to $(q_1 , r_1)$. 

To solve the Schr\"{o}dinger equation (\ref{eq:Schrodinger_Equation}) and provide a useful estimate on its solution, one approach is to construct the two-parameter unitary group $U_{\v{A}}^{\mr{S}}$ associated to (\ref{eq:Schrodinger_Equation}). For this, one may use the abstract results of Kato \cite{Kato70, Kato73} concerning general linear evolution equations of the type
\begin{align*}
\left\lbrace \begin{array}{l}
\partial_t \xi + \mb{A} (t) \xi = \mb{F} (t)  \\
\xi (s) = \psi_0
\end{array} \right.
\end{align*}
in a Banach space $X$. Following the formulation of \cite[Theorem 1]{Kato73} as given in \cite[Theorem 6 \& Remark 7]{Petersen2014}, we have the following result concerning the initial value problem
\begin{align}\label{eq:homo_magnetic_schr}
\left\lbrace \begin{array}{l}
i \partial_t \xi = ( \v{p} + \v{A} )^2 \xi \\ 
\xi (s) = \psi_0 \in H^2 (\R^3 ; \C) .
\end{array} \right.
\end{align}
Let $T > 0$, $\I = [0,T]$, and fix a divergence-free vector potential $\v{A} \in L_T^{\infty} H^1 \cap W^{1,1} L^3$. It follows from the Kato-Rellich Theorem on perturbations of self-adjoint operators that $(\v{p} + \v{A})^2$ is an essentially self-adjoint operator with domain $H^2 (\R^3 ; \C)$. Moreover, there exists a unique unitary two-parameter group $U_{\v{A}}^{\mr{S}}$ defined on $\I^2$ with the following properties:
\begin{itemize}
\item $U_{\v{A}}^{\mr{S}}: \I^2 \rightarrow \B (L^2)$ is strongly continuous with $U_{\v{A}}^{\mr{S}}(t,t) = 1$ for $t \in \I$.
\item $U_{\v{A}}^{\mr{S}}(t, \tau) U_{\v{A}}^{\mr{S}}(\tau , s) = U (t,s)$ for all $(t, \tau , s) \in \I^3$. 
\item $U_{\v{A}}^{\mr{S}}(t, s) H^2 \subset H^2$ for $(t, s) \in \I^2$ and $U_{\v{A}}^{\mr{S}}: \I^2 \rightarrow \B (H^2)$ is strongly continuous.
\item The strong partial derivatives $\partial_t U_{\v{A}}^{\mr{S}}(t, s) \psi_0$ and $\partial_{\tau} U_{\v{A}}^{\mr{S}}(t, s) \psi_0$ exist in $L^2$ for all $(t, s) \in \I^2$ and $\psi_0 \in H^2$ and are given by $i \partial_t U_{\v{A}}^{\mr{S}}(t, s) \psi_0 = (\v{p} + \v{A} (t) )^2 U_{\v{A}}^{\mr{S}}(t, s) \psi_0$ and $\partial_s U_{\v{A}}^{\mr{S}}(t, s) \psi_0 = i U_{\v{A}}^{\mr{S}}(t, s) (\v{p} + \v{A} (s))^2 \psi_0$. Moreover, $\partial_t U_{\v{A}}^{\mr{S}}(t,s) , \partial_s U_{\v{A}}^{\mr{S}}(t,s) : \I^2 \rightarrow \B (H^2 , L^2)$ are strongly continuous. 
\end{itemize}

Using the unitarity of $U_{\v{A}}^{\mr{S}}(t,s)$ in $L^2$, namely, $\| U_{\v{A}}^{\mr{S}}(t,s) \psi_0 \|_{2} = \| \psi_0 \|_{2}$, and by interpolating, one can infer $U_{\v{A}}^{\mr{S}}(t,s) H^s \subseteq H^s$ for all $s \in [0,2]$. Further, by a duality argument, one can extend $U_{\v{A}}^{\mr{S}}(t,s) H^s \subseteq H^s$ for all $s \in [-2,2]$. Moreover, $\| U_{\v{A}}^{\mr{S}}(t,s) \|_{\B (H^2)}$ is bounded for each $(t,s) \in \I^2$ and, therefore, by the uniform boundedness principle,
\begin{align*}
\| U_{\v{A}}^{\mr{S}} (t,s) \|_{L^{\infty} (\I^2 ; \B (H^{-2}))} \leq \| U_{\v{A}}^{\mr{S}}(t,s) \|_{L^{\infty}  (\I^2 ; \B (H^2))} < \infty .
\end{align*}
These observations allow us to obtain a solution to
\begin{align}\label{eq:inhomo_magnetic_schr}
\left\lbrace \begin{array}{l}
i \partial_t \xi = ( \v{p} + \v{A} )^2  \xi +  f \\ 
\xi (s) = \psi_0  
\end{array} \right.
\end{align}
using an application of Duhamel's formula. Indeed, if $f \in L^1_T H^{-2}$, then
\begin{align}\label{eq:inhomo_mag_schr_duhamel}
\xi (t) = U_{\v{A}}^{\mr{S}}(t,s) \psi_0 - i \int_s^t U_{\v{A}}^{\mr{S}}(t, \tau) f( \tau ) \dd \tau
\end{align}
is a $C_T L^2 \cap W^{1,1}_T H^{-2}$-solution to (\ref{eq:inhomo_magnetic_schr}).

Using the formula (\ref{eq:inhomo_mag_schr_duhamel}) we can give an explicit estimate on $\| U_{\v{A}}^{\mr{S}}(t,s) \|_{L_{\I^2}^{\infty} \B (H^2)}$ as follows. Observe that, in the sense of distributions, $\xi (t) = U_{\v{A}}^{\mr{S}}(t,s) \psi_0$ satisfies 
\begin{align*}
i \partial_t^2 \xi = (\v{p} + \v{A})^2 \partial_t \xi + 2 \partial_t \v{A} \cdot (\v{p} + \v{A}) \xi 
\end{align*}
since $\v{A}$ is assumed to be divergence-free. Hence, by (\ref{eq:inhomo_mag_schr_duhamel}), we conclude
\begin{align}\label{eq:proof_two_parameter_esti_H2_1} 
\partial_t \xi (t) = U_{\v{A}}^{\mr{S}}(t,s) \partial_t \xi (s) - 2 i \int_s^t U_{\v{A}}^{\mr{S}}(t,\tau) \partial_t \v{A} (\tau) \cdot (\v{p} + \v{A} (\tau)) \xi (\tau) \dd \tau .
\end{align}
Consider taking the spatial $L^2$-norm on both sides of (\ref{eq:proof_two_parameter_esti_H2_1}). Since $U_{\v{A}}^{\mr{S}}$ is unitary, we need to estimate
\begin{align}\label{eq:proof_two_parameter_esti_H2_2}
\| \partial_t \xi (t) \|_{2} \leq \| \partial_t \xi (s) \|_{2} + 2 \int_s^t \| \partial_t \v{A} (\tau ) \cdot (\v{p} + \v{A} (\tau ) \xi (\tau ) \|_{2} \dd \tau .
\end{align}
We recall that $\v{A} \in L_T^{\infty} H^1 \cap W^{1,1} L^3$. From (\ref{eq:proof_two_parameter_esti_H2_2}) we see that we need to estimate $\| \partial_t \v{A} \cdot (\v{p} + \v{A}) \xi \|_{2}$. Using H\"{o}lder's inequality, the Sobolev inequality $\dot{H}^1 (\R^3) \subset L^6 (\R^3)$, and the diamagnetic inequality $| \nabla | f | (\v{x}) | \leq | (\v{p} + \v{A}) f (\v{x}) |$ for a.e.~$\v{x} \in \R^3$, we have the estimate
\begin{align}\label{eq:proof_two_parameter_esti_H2_3}
\| \partial_t \v{A} \cdot (\v{p} + \v{A}) \xi \|_{2} \leq \| \partial_t \v{A} \|_{3} \| (\v{p} + \v{A}) \xi \|_{6} \lesssim  \| \partial_t \v{A} \|_{3} \| (\v{p} + \v{A})^2 \xi \|_{2} .
\end{align}
Estimate (\ref{eq:proof_two_parameter_esti_H2_3}) together with $\partial_t \xi = - i (\v{p} + \v{A})^2 \xi$ allows us to go back to (\ref{eq:proof_two_parameter_esti_H2_2}) and conclude
\begin{align*}
\| (\v{p} + \v{A}(t))^2 \xi (t)\|_{2} \leq \| (\v{p} + \v{A}(s))^2 \xi (s) \|_{2} + C \int_s^t  \| \partial_t \v{A} (\tau ) \|_{3} \| (\v{p} + \v{A} (\tau ))^2 \xi (\tau) \|_{2} \dd \tau .
\end{align*}
Then, by Gronwall's inequality,
\begin{align}\label{eq:proof_two_parameter_esti_H2_4}
\| (\v{p} + \v{A} (t))^2 \xi (t) \|_{2} \leq \| (\v{p} + \v{A} (s))^2 \xi (s) \|_{2} \exp{ \left\lbrace
C \int_s^t \| \partial_t \v{A} (\tau) \|_{3} \dd \tau \right\rbrace } ,
\end{align}
where $C > 0$ is some universal constant independent of $\v{A}$ and $\xi$. 

From the assumption that $\v{A}$ is divergence-free we have the identity 
\begin{align*}
(\v{p} + \v{A})^2 = - \Delta + 2 \v{A} \cdot \v{p} + |\v{A}|^2 ,
\end{align*}
and thus
\begin{align}\label{eq:proof_two_parameter_esti_H2_5}
\| \Delta \xi \|_{2} \leq \| (\v{p} + \v{A})^2 \xi \|_{2} + 2 \| \v{A} \cdot \v{p} \xi \|_{2} + \| |\v{A}|^2 \xi \|_{2} .
\end{align}
The last term on the right hand side of (\ref{eq:proof_two_parameter_esti_H2_5}) may be dealt with by an application of H\"{o}lder's inequality, the Sobolev inequality $\dot{H}^1 (\R^3) \subset L^6 (\R^3)$, and the diamagnetic inequality:
\begin{align}\label{eq:proof_two_parameter_esti_H2_6}
\| |\v{A}|^2 \xi \|_{2} \leq \| \v{A} \|_{6}^2 \| \xi \|_{6} \lesssim \| \v{A} \|_{1,2}^2 \| (\v{p} + \v{A})^2 \xi \|_{2} .  
\end{align}
The middle term on the right hand side of (\ref{eq:proof_two_parameter_esti_H2_5}) may be dealt in a similar fashion, except with an additional use of the Sobolev inequality $H^1 (\R^3) \subset H^{\frac{1}{2}} (\R^3) \subset L^3 (\R^3)$. We find
\begin{align*}
\| \v{A} \cdot \v{p} \xi \|_{2} & \leq \| \v{A} \cdot (\v{p} + \v{A}) \xi \|_{2} + \| |\v{A}|^2 \xi \|_{2} \\
& \leq \| \v{A} \|_{3} \| (\v{p} + \v{A}) \xi \|_{6} + \| |\v{A}|^2 \xi \|_{2} \\
& \lesssim \left( 1 + \| \v{A} \|_{1,2} \right) \| \v{A} \|_{1,2} \| (\v{p} + \v{A})^2 \xi \|_{2} . \numberthis \label{eq:proof_two_parameter_esti_H2_7}
\end{align*}
Combining (\ref{eq:proof_two_parameter_esti_H2_4}) with (\ref{eq:proof_two_parameter_esti_H2_5}), (\ref{eq:proof_two_parameter_esti_H2_6}), and (\ref{eq:proof_two_parameter_esti_H2_7}) we arrive at
\begin{align}\label{eq:two_parameter_esti_H2}
\| U_{\v{A}}^{\mr{S}}(t,s) \|_{L^{\infty} (\I ; \B (H^2))} \lesssim \left( 1 + \| \v{A} \|_{\infty; 1, 2} \right)^4  e^{C \| \partial_t \v{A} \|_{1;3}} .
\end{align}
With estimate (\ref{eq:two_parameter_esti_H2}) at our disposal we are in a position to sketch the proof of Theorem \ref{thm:nakamura05}. 

\begin{proof}[Sketch Proof of Theorem \ref{thm:nakamura05}]\label{proof:NW05}
Let $(\psi_0 , \v{a}_0 , \dot{\v{a}}_0 ) \in \sZ_{0,\mr{S}}^{2,\frac{3}{2}}$ and $(\psi , \v{A}) \in \sZ_{\mr{S}}^{2,\frac{3}{2}} (T)$. We first show that, for some $T,R_1 , R_2 > 0$, the map 
\begin{align*}
\Phi : (\psi , \v{A}) \in \sZ_{\mr{S}}^{2,\frac{3}{2}} (T) \mapsto (\xi , \v{K}) \in \sZ_{\mr{S}}^{2,\frac{3}{2}} (T) ,
\end{align*}
where
\begin{align*}
( \xi (t) , \v{K} (t) ) = \left( U_{\v{A}}^{\mr{S}} (t,0) \psi_0 , \V_{4 \pi \alpha \Hproj{ \v{J}_{\mr{S}} [\psi, \v{A}] } + \v{A}} [ \v{a}_0 , \dot{\v{a}}_0 ] (t) \right) ,
\end{align*}
with $\V$ being defined by (\ref{eq:Klein-Gordon_soln}). That is, for some choice of $T , R_1 , R_2 > 0$, $\Phi$ maps $\sZ_{\mr{S}}^{2,\frac{3}{2}} (T)$ into itself.

Note that $\v{K}$ is divergence-free by virtue of the formula (\ref{eq:Klein-Gordon_soln}). Estimate (\ref{eq:two_parameter_esti_H2}) together with $\| \partial_t \v{A} \|_{1;3} \lesssim T \| \partial_t \v{A} \|_{\infty ; \frac{1}{2} , 2}$ implies that
\begin{align}\label{eq:proof_thm2_1}
\| \xi (t) \|_{2,2} \leq C_1 (1 + R_2 + R_2^2)^2 e^{C T R_2} \| \psi_0 \|_{2,2} ,
\end{align}
where $C_1 > 0$ is a universal constant. The Sobolev inequality $L_T^{\infty} \dot{H}^{1/2} (\R^3) \subset L_T^6 L^3 (\R^3)$ and the estimate (\ref{eq:Strichartz_KG_general}) yields
\begin{align}\label{eq:proof_thm2_2}
 \| ( \v{K} , \partial_t \v{K} ) \|_{\infty ; \frac{3}{2} , 2 \oplus 6 ; 3} \lesssim  \| (\v{a}_0 , \dot{\v{a}}_0 ) \|_{ \frac{3}{2} , 2 \oplus \frac{1}{2} , 2} + \| \Hproj{ \v{J}_{\mr{S}} [\psi , \v{A}] } + \v{A} \|_{1 ; \frac{1}{2} , 2} .
\end{align}
We need an $H^{1/2} (\R^3)$-estimate on $\v{J}_{\mr{S}} [\psi , \v{A}]$. This may be achieved using, for example, Lemma \ref{lem:Kato-Ponce} together with the Sobolev inequalities $H^{\frac{1}{2}} (\R^3 ; \C) \subset L^3 (\R^3 ; \C)$ and $\dot{H}^1 (\R^3 ; \C) \subset L^6 (\R^3 ; \C)$. Indeed, we have the estimate
\begin{align}\label{eq:proof_thm2_4}
\| (1 - \Delta)^{\frac{1}{4}} (\psi \nabla \psi) \|_{2} \lesssim \| (1 - \Delta)^{\frac{1}{4}} \psi \|_{6} \| \nabla \psi \|_{3} + \| \psi \|_{6} \| (1 - \Delta)^{\frac{1}{4}} \nabla \psi \|_{3} \lesssim  \| \psi \|_{2,2}^2 . 
\end{align}
Recalling that $H^s (\R^d)$ is a Banach algebra for $s > d/2$, we have $|\psi|^2 \in H^2 (\R^3)$ and, hence,
\begin{align}\label{eq:proof_thm2_5}
\| (1 - \Delta)^{\frac{1}{4}} (\v{A} |\psi|^2) \|_{2} \lesssim \| (1 - \Delta)^{\frac{1}{4}} \v{A} \|_{6} \| \psi \|_{6}^2 + \| \v{A} \|_{6} \| (1 - \Delta)^{\frac{1}{4}} |\psi|^2 \|_{3}  \lesssim \| \v{A} \|_{\frac{3}{2} , 2}  \| \psi \|_{2,2}^2 .
\end{align}
With (\ref{eq:proof_thm2_4}) and (\ref{eq:proof_thm2_5}) at our disposal, (\ref{eq:proof_thm2_2}) becomes
\begin{align}\label{eq:proof_thm2_6}
\max_{k \in \left\lbrace 0, 1 \right\rbrace} \| \partial_t^k \v{K} \|_{\infty ; 3/2 - k , 2} \leq C_2 \left( \| (\v{a}_0 , \dot{\v{a}}_0 ) \|_{ \frac{3}{2} , 2 \oplus \frac{1}{2} , 2} + (R_1^2 + R_1^2 R_2 + R_2) T \right) ,
\end{align}
where $C_2 > 0$ is a universal constant. Without loss of generality we can assume that $T < 1$. Combining (\ref{eq:proof_thm2_1}) and (\ref{eq:proof_thm2_6}), we choose $R_1 , R_2 > 0$ so that
\begin{align*}
& 2 C_2 \| (\v{a}_0 , \dot{\v{a}}_0 ) \|_{ \frac{3}{2} , 2 \oplus \frac{1}{2} , 2} \leq R_2 , \\
& 2 C_1 \| \psi_0 \|_{2,2} (1 + R_2 + R_2^2 )^2 e^{C R_2} \leq R_1 ,
\end{align*}
and $T \in (0,1)$ so that $T \leq \left( 2 C_2 (R_1^2 + R_2 R_1^2 + R_2) \right)^{-1}$. With these choices of $T , R_1 , R_2 > 0$ we ensure that $\Phi$ maps $\sZ_{\mr{S}}^{2,\frac{3}{2}} (T)$ into itself. 

The final step in our sketch is to show that we may further choose $T \in (0,1)$ so that $\Phi$ is a contraction with respect to the metric (\ref{eq:MS_metric}). Let $(\psi' , \v{A}^{\!\prime}) \in \sZ_{\mr{S}}^{2,\frac{3}{2}} (T)$, and let $\xi' (t) = U_{\v{A}^{\!\prime}} (t , 0) \psi_0$ and $\v{K}' (t) =  \V_{4 \pi \alpha \Hproj{ \v{J}_{\mr{S}} [\psi' , \v{A}^{\!\prime}] } + \v{A}^{\!\prime}} [ \v{a}_0 , \dot{\v{a}}_0 ] (t)$. Note that $\xi - \xi'$ satisfies the equation
\begin{align}\label{eq:xi_minus_xiprime}
\left\lbrace \begin{array}{l}
i \partial_t (\xi - \xi') = (\v{p} + \v{A})^2 (\xi - \xi') + 2 (\v{A} - \v{A}^{\!\prime}) \cdot \v{p} \xi' + (|\v{A}|^2 - |\v{A}^{\!\prime}|^2) \xi' \\
(\xi - \xi') (0) = 0 .
\end{array} \right. 
\end{align}
Using the Duhamel formula (\ref{eq:inhomo_mag_schr_duhamel}) on (\ref{eq:xi_minus_xiprime}) we have
\begin{align}
(\xi - \xi') (t) = - i \int_0^t U_{\v{A}}^{\mr{S}}(t , \tau ) \left( 2 (\v{A} - \v{A}^{\!\prime}) \cdot \v{p} \xi' + (|\v{A}|^2 - |\v{A}^{\!\prime}|^2) \xi' \right) (\tau) \dd \tau .
\end{align}
Using that $U_{\v{A}}^{\mr{S}}(t , \tau)$ is unitary, $|\v{A}|^2 - | \v{A}^{\!\prime} |^2 = (\v{A} + \v{A}^{\!\prime}) \cdot (\v{A} - \v{A}^{\!\prime})$, Cauchy-Schwartz, and Sobolev $H^{3/4} (\R^3) \subset L^4 (\R^3)$ we have
\begin{align*}
\| \xi - \xi' \|_{\infty; 2} & \leq \int_0^T \left( 2 \| (\v{A} - \v{A}^{\!\prime}) \cdot \v{p} \xi' \|_{2} (\tau) + \| (|\v{A}|^2 - |\v{A}^{\!\prime}|^2) \xi' \|_{2} (\tau) \right)  \dd \tau \\
& \lesssim \int_0^T (1 + \| \v{A} + \v{A}^{\!\prime} \|_{L^4} (\tau) ) \| \v{A} - \v{A}^{\!\prime} \|_{L^4} (\tau) \| \xi' \|_{2,2} (\tau) \dd \tau \\
& \lesssim R_1 ( T^{\frac{3}{4}} + 2 R_2 T^{\frac{1}{2}} ) \| \v{A} - \v{A}^{\!\prime} \|_{4;4} . \numberthis \label{eq:proof_thm2_7} 
\end{align*}

Using
\begin{align*}
(\v{K} - \v{K}' )(t) = \V_{4 \pi \alpha \Hproj{ ( \v{J}_{\mr{S}} [\psi , \v{A}] - \v{J}_{\mr{S}} [\psi' , \v{A}^{\!\prime}] )} + \v{A} - \v{A}^{\!\prime}} [ \v{0} , \v{0} ] (t) 
\end{align*}
together with the Strichartz estimates (\ref{eq:Strichartz_KG_general}), we have
\begin{align*}
& \max{ \left\lbrace \| \v{K} - \v{K}' \|_{\infty ; \frac{1}{2} , 2} , \| \v{K} - \v{K}' \|_{4;4} \right\rbrace } \\
& \hspace{2cm} \lesssim \| \Hproj{ \left( \v{J}_{\mr{S}} [\psi , \v{A}] - \v{J}_{\mr{S}} [\psi' , \v{A}^{\!\prime}] \right) } \|_{\frac{4}{3} ; \frac{4}{3}} + \| \v{A} - \v{A}^{\!\prime} \|_{1 ; - \frac{1}{2} , 2} . \numberthis \label{eq:proof_thm2_8} 
\end{align*}
Write $\v{J}_{\mr{S}} [\psi , \v{A}] - \v{J}_{\mr{S}} [\psi' , \v{A}^{\!\prime}] = - 2 \alpha (g_1 + g_2 + g_3) $ where
\begin{align*}
g_1 = \re{ \langle \psi - \psi' , (\v{p} + \v{A}) \psi \rangle_{\C} } , ~~ g_2 = \re{ \langle \psi' , (\v{p} + \v{A}^{\!\prime}) (\psi - \psi') \rangle_{\C} } , ~~ g_3 = \re{ \langle \psi' , (\v{A} - \v{A}^{\!\prime}) \psi \rangle_{\C} }  .
\end{align*}
Therefore, we want to estimate $\| \Hproj{ g_k } \|_{\frac{4}{3} ; \frac{4}{3}}$ for $k \in \{1,2,3\}$ by $d ( (\psi , \v{A} ) , (\psi' , \v{A}^{\!\prime}) )$. The only term that displays a challenge is $g_2$ because we have a gradient applied to $\psi - \psi'$, but in the metric $d ( (\psi , \v{A} ) , (\psi' , \v{A}^{\!\prime}) )$ we have $\| \psi - \psi' \|_{\infty; 2}$. This difficulty is circumvented by noticing that the projection $\Hproj$ satisfies $\Hproj{ \nabla f } = 0$ and, hence, 
\begin{align*}
\Hproj{ \langle \psi' , \v{p} (\psi - \psi') \rangle_{\C} } = - \Hproj{ \langle ( \psi - \psi') , \v{p} \psi' \rangle_{\C} } .
\end{align*} 
Using this, together with H\"{o}lder, Sobolev, and the boundedness of $\Hproj$ on $L^p$ for $p \in (1, \infty)$, it is possible to show
\begin{align*}
& \max{ \left\lbrace \| \v{K} - \v{K}' \|_{\infty ; \frac{1}{2} , 2} , \| \v{K} - \v{K}' \|_{4;4} \right\rbrace } \\
& \hspace{5mm} \lesssim R_1 ( T^{\frac{3}{4}} + R_2 T^{\frac{1}{2}} ) \| \psi - \psi' \|_{\infty; 2} + R_1^2 T^{\frac{1}{2}} \| \v{A} - \v{A}^{\!\prime} \|_{4;4} + T \| \v{A} - \v{A}^{\!\prime} \|_{\infty ; \frac{1}{2} , 2}  . \numberthis \label{eq:proof_thm2_9} 
\end{align*}
Combining (\ref{eq:proof_thm2_7}) and (\ref{eq:proof_thm2_9}) we can bound $d ( ( \xi , \v{K} ) , (\xi' , \v{K}') )$ by $d ( (\psi , \v{A} ) , (\psi' , \v{A}^{\!\prime}))$ with a constant depending on $R_1 , R_2 , T > 0$. We may then choose $T \in (0,1)$ so that this constant is smaller than $1$, proving $\Phi$ is a contraction with respect to $d$. By the Banach fixed point theorem, this gives us a solution to (\ref{eq:MS_Coulomb_nophi}). 
\end{proof}

With this proof sketch of Theorem \ref{thm:nakamura05} in the $(s,r) = (2 , 3/2)$ case completed, let us try to adapt it to the MP equations (\ref{eq:MP}) and see precisely where the difficulties are encountered. The space of initial conditions will be 
\begin{align}\label{eq:MP_initial_conditions}
\sZ_{0,\mr{P}}^{s,r} := \left\lbrace  ( \psi_0, \v{a}_0 , \dot{\v{a}}_0) \in H^s (\R^3 ; \C^2) \oplus (H^{r} \oplus H^{r-1}) (\R^3 ; \R^3) : \diver{\v{a}_0} = \diver{\dot{\v{a}}_0} = 0 \right\rbrace ,
\end{align}
and the metric space will be
\begin{align*}
\sZ_{\mr{P}}^{s,r} (T) = & \left\lbrace (\phi , \v{A} ) \in L^{\infty}_T H^s (\R^3 ; \C^2) \oplus ( L^{\infty}_T H^{r} (\R^3 ; \R^3) \cap W^{1, \infty }_T H^{r-1} (\R^3 ; \R^3) ) \right. \\
& ~~ \left. \text{s.t.} ~ \| \psi \|_{\infty; s,2} \leq R_1, ~ \max{ \left\lbrace \| \v{A} \|_{\infty; r,2} , \| \partial_t \v{A} \|_{\infty ; 1,2} \right\rbrace } \leq R_2 , ~ \diver{\v{A}} = 0 \right\rbrace , \numberthis \label{eq:MP_metric_space} 
\end{align*}
with the metric $d$ having the same definition as before, namely via (\ref{eq:MS_metric}). We begin by considering the Pauli equation,
\begin{align}\label{eq:Pauli_equation_local_exist}
\left\lbrace \begin{array}{l}
i \partial_t \xi = [ \gvsig \cdot (\v{p} + \v{A})]^2 \xi  \\
\xi (0) = \psi_0 ,
\end{array} \right.
\end{align}
and the Klein-Gordon equation,
\begin{align}\label{eq:Klein_Gordon_Equation_Pauli}
\left\lbrace \begin{array}{l}
(\square + 1) \v{K} =  4 \pi \alpha \Hproj{ \v{J}_{\mr{P}} [ \psi , \v{A} ] } + \v{A}  \\
(\v{K} (0) , \partial_t \v{K} (0) ) = ( \v{a}_0 , \dot{\v{a}}_0 ) .
\end{array} \right.
\end{align}
As before, we will always assume $\diver{\v{A}} = 0$. We mention that the Klein-Gordon equation is treated exactly as it was before, with $\v{K} (t) = \V_{4 \pi \alpha \Hproj{ \v{J}_{\mr{P}} [\psi , \v{A}] } + \v{A}} [ \v{a}_0 , \dot{\v{a}}_0 ]$ where $\V$ is defined by (\ref{eq:Klein-Gordon_soln}). 

In our study of (\ref{eq:Pauli_equation_local_exist}) and (\ref{eq:Klein_Gordon_Equation_Pauli}) we will specialize to $(\psi_0 , \v{a}_0 , \dot{\v{a}}_0 ) \in \sZ_{0,\mr{P}}^{2,2}$. We consider the solution map
\begin{align*}
\Phi_P : (\psi , \v{A}) \in \sZ_{\mr{P}}^{2,2} (T) \longmapsto (\xi , \v{K})
\end{align*}
where $\xi$ and $\v{K}$ satisfy (\ref{eq:Pauli_equation_local_exist}) and (\ref{eq:Klein_Gordon_Equation_Pauli}), respectively. The higher regularity on the vector potential $\v{A}$ in $\sZ_{\mr{P}}^{2,2} (T)$ seems to be needed to construct the two parameter unitary group for the Pauli equation (\ref{eq:Pauli_equation_local_exist}) as we shall see below. 

To solve the Pauli equation (\ref{eq:Pauli_equation_local_exist}) it is natural to construct the two parameter unitary group $U_{\v{A}}^{\mr{P}}$ using Kato's abstract method as we did for (\ref{eq:homo_magnetic_schr}). In this regard, we can prove the following. Let $T > 0$, $\I = [0,T]$, and consider a divergence-free vector potential $\v{A} \in L_T^{\infty} H^1 \cap W^{1,1} H^1$. By Kato-Rellich for perturbations of self-adjoint operators $[\gvsig \cdot (\v{p} + \v{A})]^2$ is a essentially self-adjoint operator with domain $H^2 (\R^3 ; \C^2)$. Moreover, there exists a unique unitary two-parameter group $U_{\v{A}}^{\mr{P}}$ defined on $\I^2$ with the following properties:
\begin{itemize}
\item $U_{\v{A}}^{\mr{P}} : \I^2 \rightarrow \B (L^2)$ is strongly continuous with $U_{\v{A}}^{\mr{P}} (t,t) = 1$ for $t \in \I$.
\item $U_{\v{A}}^{\mr{P}} (t, \tau) U_{\v{A}}^{\mr{P}} (\tau , s) = U_{\v{A}}^{\mr{P}} (t,s)$ for all $(t, \tau , s) \in \I^3$. 
\item $U_{\v{A}}^{\mr{P}} (t, s) H^2 \subset H^2$ for $(t, s) \in \I^2$ and $U_{\v{A}}^{\mr{P}} : \I^2 \rightarrow \B (H^2)$ is strongly continuous.
\item The strong partial derivatives $\partial_t U_{\v{A}}^{\mr{P}} (t, s) \psi_0$ and $\partial_{\tau} U_{\v{A}}^{\mr{P}} (t, s) \psi_0$ exist in $L^2$ for all $(t, s) \in \I^2$ and $\psi_0 \in H^2$ and are given by $i \partial_t U_{\v{A}}^{\mr{P}} (t, s) \psi_0 = [\gvsig \cdot (\v{p} + \v{A} (t) )]^2 U_{\v{A}}^{\mr{P}} (t, s) \psi_0$ and $\partial_s U_{\v{A}}^{\mr{P}} (t, s) \psi_0 = i U_{\v{A}}^{\mr{P}} (t, s) [\gvsig \cdot (\v{p} + \v{A}(s))]^2 \psi_0$. Moreover, $\partial_t U_{\v{A}}^{\mr{P}} (t,s) , \partial_s U_{\v{A}}^{\mr{P}} (t,s) : \I^2 \rightarrow \B (H^2 ; L^2)$ are strongly continuous.
\end{itemize}

Again, we may extend $U_{\v{A}}^P$ to $H^{s} (\R^3 ; \C^2)$ for $s \in [-2,2]$ and if $f \in L^1_T H^{-2}$, then
\begin{align}\label{eq:inhomo_pauli_duhamel}
\xi (t) = U_{\v{A}}^{\mr{P}} (t,s) \psi_0 - i \int_s^t U_{\v{A}}^{\mr{P}} (t, \tau) f( \tau ) \dd \tau
\end{align}
is a $C_T L^2 \cap W^{1,1}_T H^{-2}$-solution to
\begin{align*}
\left\lbrace \begin{array}{l}
i \partial_t \xi = [\gvsig \cdot (\v{p} + \v{A})]^2 \xi + f \\
\xi (s) = \psi_0  \in H^2 ( \R^3 ; \C^2) .
\end{array} \right.
\end{align*}
It is possible to give an explicit bound on $\| U_{\v{A}}^{\mr{P}} \|_{L^{\infty} (\I^2 ; \B (H^2))}$ using a similar strategy as the proof of (\ref{eq:two_parameter_esti_H2}). Indeed, observe that, in the sense of distributions, $\xi (t) = U_{\v{A}}^{\mr{P}} (t,s) \psi_0$ satisfies 
\begin{align*}
i \partial_t^2 \xi = [\gvsig \cdot (\v{p} + \v{A})]^2 \partial_t \xi + \left( 2 \partial_t \v{A} \cdot (\v{p} + \v{A}) + \gvsig \cdot \partial_t \v{B} \right) \xi ,
\end{align*}
where we've used the identity $[\gvsig \cdot (\v{p} + \v{A}) ]^2 = (\v{p} + \v{A})^2 + \gvsig \cdot \v{B}$, $\v{B} = \curl{\v{A}}$. Hence, by (\ref{eq:inhomo_pauli_duhamel}), we conclude
\begin{align}\label{eq:proof_two_parameter_esti_pauli_1} 
\partial_t \xi (t) = U_{\v{A}}^{\mr{P}} (t,s) \partial_t \xi (s) - i \int_s^t U_{\v{A}}^{\mr{P}} (t,\tau) \left( 2 \partial_t \v{A} (\tau) \cdot (\v{p} + \v{A} (\tau)) + \gvsig \cdot \partial_t \v{B} (\tau) \right) \xi (\tau) \dd \tau .
\end{align}
Consider taking the spatial $L^2$-norm on both sides of (\ref{eq:proof_two_parameter_esti_pauli_1}). Since $U_{\v{A}}$ is unitary, we arrive at the estimate
\begin{align}\label{eq:proof_two_parameter_esti_pauli_2}
\| \partial_t \xi (t) \|_{2} \leq \| \partial_t \xi (s) \|_{2} + \int_s^t \| \left( 2 \partial_t \v{A} (\tau ) \cdot (\v{p} + \v{A} (\tau)) + \gvsig \cdot \partial_t \v{B} (\tau) \right) \xi (\tau ) \|_{2} \dd \tau .
\end{align}
We recall that $\v{A}$ is assumed to be divergence-free and $\v{A} \in L_T^{\infty} H^1 \cap W^{1,1} H^1$. The term $\| \partial_t \v{A} \cdot (\v{p} + \v{A}) \xi \|_{2}$ is estimated using H\"{o}lder's inequality and the Sobolev inequality $H^s (\R^3) \subset L^{\infty} (\R^3)$, $s > 3/2$. One finds
\begin{align*}
\| \partial_t \v{A} \cdot (\v{p} + \v{A}) \xi \|_{2} \lesssim \left( 1 + \| \v{A} \|_{1,2} \right) \| \partial_t \v{A} \|_{1,2} \| \xi \|_{2,2} .
\end{align*}
To estimate $\| \gvsig \cdot \partial_t \v{B} \xi \|_{2}$ we again use the Sobolev inequality $H^s (\R^3) \subset L^{\infty} (\R^3)$, $s > 3/2$, to find
\begin{align*}
\| \gvsig \cdot \partial_t \v{B} \xi \|_{2} \lesssim \| \partial_t \v{B} \|_{2} \| \xi \|_{2,2} \lesssim \| \partial_t \v{A} \|_{1,2} \| \xi \|_{2,2} . 
\end{align*}
Feeding these estimates back into (\ref{eq:proof_two_parameter_esti_pauli_2}) and using that $\partial_t \xi  = - i [ \gvsig \cdot (\v{p} + \v{A}) ]^2 \xi$  we find
\begin{align*}
& \| [ \gvsig \cdot (\v{p} + \v{A} (t)) ]^2 \xi (t) \|_{2} \\
& \lesssim \| [ \gvsig \cdot (\v{p} + \v{A} (s)) ]^2 \xi (s) \|_{2} + \int_s^t \left[ \| \partial_t \v{A} (\tau ) \|_{1,2} (2 + \| \v{A} (\tau ) \|_{1,2}) \right] \| \xi (\tau ) \|_{2,2} \dd \tau . \numberthis \label{eq:proof_two_parameter_esti_pauli_3}
\end{align*}

At this point we do not immediately proceed with Gronwall's inequality as we did in the Schr\"{o}dinger case because an estimate of the form $\| \xi \|_{2,2} \lesssim \| [\gvsig \cdot (\v{p} +\v{A})]^2 \xi \|_{2}$ is impossible due to the existence of zero modes (see Chapter \ref{chap:stability}). Instead, using that $\partial_t \xi = - i [ \gvsig \cdot (\v{p} + \v{A}) ]^2 \xi$ together with the reverse triangle inequality, H\"{o}lder's inequality, Sobolev's inequalities $\dot{H}^{\frac{1}{2}} (\R^3) \subset L^3 (\R^3)$, $\dot{H}^1 (\R^3) \subset L^6 (\R^3)$ and $H^{\frac{3}{2} + \delta} (\R^3) \subset L^{\infty} (\R^3)$, $\delta > 0$, Plancherel's theorem, and the Young's inequalities $|\v{k}|^{3/2} \lesssim \epsilon^{-3} + \epsilon |\v{k}|^2$, $|\v{k}| \lesssim \epsilon + \epsilon |\v{k}|^2$, and $|\v{k}|^{3/2 + \delta} \lesssim \epsilon^{- (3-2\delta)/ (1 - 2 \delta)} + \epsilon |\v{k}|^2$, with $\epsilon > 0$ and $\delta \in (0, 1/2)$, we find
\begin{align*} 
& \|  [ \gvsig \cdot (\v{p} + \v{A}) ]^2 \xi \|_{2} \\
& \geq \| \Delta \xi \|_{2} - C \left[ \| \v{A} \|_{\infty; 1, 2} \left( \frac{1}{\epsilon_1^3} \| \psi_0 \|_{2} + \epsilon_1 \| \Delta \xi \|_{2} \right) \right. \\
& \hspace{1cm} \left. + \| \v{A} \|_{\infty; 1, 2}^2 \left( \frac{1}{\epsilon_2} \| \psi_0 \|_{2} + \epsilon_2 \| \Delta \xi \|_{2} \right)  + \| \v{B} \|_{\infty; 2} \left( \epsilon_3^{- \frac{3-2\delta}{1-2\delta}} \| \psi_0 \|_{2} + \epsilon_3 \| \Delta \xi \|_{2} \right) \right] , \numberthis \label{eq:proof_two_parameter_esti_pauli_4}
\end{align*}
for some $\epsilon_1 , \epsilon_2 , \epsilon_3 > 0$, $\delta \in (0 , 1/2)$, and a universal constant $C > 0$. Choosing $\epsilon_1 = (4 C \| \v{A} \|_{\infty; 1, 2})^{-1}$, $\epsilon_2 = (4 C \| \v{A} \|_{\infty; 1, 2}^2 )^{-1}$, and $\epsilon_3 = (4 C \| \v{B} \|_{\infty; 2})^{-1}$, (\ref{eq:proof_two_parameter_esti_pauli_4}) becomes
\begin{align}\label{eq:proof_two_parameter_esti_pauli_5} 
\|  [ \gvsig \cdot (\v{p} + \v{A}) ]^2 \xi \|_{2} \geq  \frac{1}{4} \| \xi \|_{2,2} - C \left[ 1 +  \| \v{A} \|_{\infty; 1, 2}^4 + \| \v{A} \|_{\infty; 1, 2}^{\frac{4 - 4 \delta}{1 - 2 \delta}} \right] \| \psi_0 \|_{2} ,
\end{align}
where $C > 1$ is a different universal constant. In an identical fashion, we find
\begin{align}\label{eq:proof_two_parameter_esti_pauli_6} 
\|  [ \gvsig \cdot (\v{p} + \v{A}) ]^2 \xi \|_{2} \leq \frac{7}{4} \| \xi \|_{2,2} + C \left[  \| \v{A} \|_{\infty; 1, 2}^4 + \| \v{A} \|_{\infty; 1, 2}^{\frac{4 - 4 \delta}{1 - 2 \delta}} \right] \| \psi_0 \|_{2} .
\end{align}
Plugging (\ref{eq:proof_two_parameter_esti_pauli_5}) and (\ref{eq:proof_two_parameter_esti_pauli_6}) into (\ref{eq:proof_two_parameter_esti_pauli_3}), take $\delta = 1/3$, and now applying Gronwall's inequality, we conclude
\begin{align}\label{eq:two_parameter_esti_pauli}
\| U_{\v{A}}^{\mr{P}} \|_{L^{\infty} (\I^2 ; \B (H^2))} \leq C \left[ 1 + \| \v{A} \|_{\infty; 1 , 2}^4 \right]^{2} \exp{ \left( C T \| \partial_t \v{A} \|_{\infty ; 1,2} (2 + \| \v{A} \|_{\infty ; 1,2}) \right] }
\end{align}
where $C > 0$ is some universal constant independent of $\v{A}$ and $\xi$. 

From the proof of (\ref{eq:two_parameter_esti_pauli}) we see that it is the presence of the full magnetic field $\v{B} = \curl{\v{A}}$ implicit in the Pauli operator $[\gvsig \cdot (\v{p} + \v{A})]^2$ that forces us to consider $\partial_t \v{A} \in L_T^1 H^1$. It is for this reason we choose to consider initial data in $\sZ_{0,\mr{P}}^{2,2}$ and work with the metric space $\sZ_{\mr{P}}^{2,2} (T)$ as $\v{A}$ and $\partial_t \v{A}$ are generally assumed to be one spatial $L^2$-derivative apart. In the same fashion as in the sketch proof of Theorem \ref{thm:nakamura05}, we may use the estimate (\ref{eq:two_parameter_esti_pauli}) and the Strichartz estimates for the Klein-Gordon equation (\ref{eq:Strichartz_KG_general}) to prove that $\Phi_P$ maps $\sZ^{2,2}_{T,P}$ into itself for appropriately chosen $T , R_1 , R_2 > 0$. 

The difficulties begin when trying to prove $d ( (\xi , \v{K} ) , (\xi' , \v{K}) ) \lesssim_{T,R_1,R_2} d ( (\psi , \v{A}) , (\psi' , \v{A}^{\!\prime}) )$ where $\Phi_P (\psi , \v{A}) = (\xi , \v{K})$ and likewise for the primed variables. Similar to (\ref{eq:xi_minus_xiprime}) in the Schr\"{o}dinger case, we use the Duhamel formula (\ref{eq:inhomo_pauli_duhamel}) and find
\begin{align*}
& (\xi - \xi') (t) \\
& \hspace{5mm} = - i \int_0^t U_{\v{A}}^{\mr{P}} (t , \tau ) \left( 2 (\v{A} - \v{A}^{\!\prime}) \cdot \v{p} \xi' + (|\v{A}|^2 - |\v{A}^{\!\prime}|^2) \xi' + \gvsig \cdot (\v{B} - \v{B}') \xi' \right) (\tau) \dd \tau . \numberthis \label{eq:xi_minus_xiprime_pauli}
\end{align*}
Taking the spatial $L^2$-norm on both sides of (\ref{eq:xi_minus_xiprime_pauli}) we immediately see that we are stuck with $\| \v{B} - \v{B}' \|_2 = \| \v{A} - \v{A}^{\!\prime} \|_{\dot{H}^1}$ in making an upper bound. We can try to cope with this by changing the metric to
\begin{align*}
d ( (\psi , \v{A} ) , (\psi' , \v{A}^{\!\prime}) ) = \max{ \left\lbrace \| \psi - \psi' \|_{\infty ; 2} , \| \v{A} - \v{A}^{\!\prime} \|_{\infty ; 1 , 2} , \| \v{A} - \v{A}^{\!\prime} \|_{4 ; 4} \right\rbrace } . 
\end{align*}
However, this will cause a problem in our estimates of $\v{K} - \v{K}'$ as the best we could do using the Strichartz estimates (\ref{eq:Strichartz_KG_general}) would then be
\begin{align*}
\| \v{K} - \v{K}' \|_{\infty ; 1 , 2} \lesssim  \| \Hproj{ \left( \v{J}_{\mr{P}} [\psi , \v{A}] - \v{J}_{\mr{P}} [\psi' , \v{A}^{\!\prime}] \right) } \|_{q_1' ; \frac{2}{q_1} , r_1'} + \| \v{A} - \v{A}^{\!\prime} \|_{q_1' ; \frac{2}{q_1} , r_1'} .
\end{align*}
We can write $\v{J}_{\mr{P}} [\psi , \v{A}] - \v{J}_{\mr{P}} [\psi' , \v{A}^{\!\prime}] = - 2 \alpha \sum_{k = 1}^5 g_k$ where
\begin{align*}
& g_1 = \re{ \langle \psi - \psi' , (\v{p} + \v{A}) \psi \rangle_{\C^2} } , \hspace{1cm} g_2 = \re{ \langle \psi' , (\v{p} + \v{A}^{\!\prime}) (\psi - \psi') \rangle_{\C^2} } , \\
& g_3 = \re{ \langle \psi' , (\v{A} - \v{A}^{\!\prime}) \psi \rangle_{\C^2} } , ~~ g_4 =  \frac{1}{2} \curl{\langle \psi - \psi' , \gvsig \psi \rangle_{\C^2} } , ~~ g_5 = \frac{1}{2} \curl{\langle \psi' , \gvsig (\psi - \psi') \rangle_{\C^2}} .
\end{align*}
However, estimating $\| \Hproj{ g_k } \|_{q_1' ; \frac{2}{q_1} , r_1'}$ for $k \in \{ 4 , 5 \}$ and $2/q_1 = 1 - 2/r_1$ by $d ( (\psi , \v{A} ) , (\psi' , \v{A}^{\!\prime}) )$ is not possible. Indeed, the projection $\Hproj$ acts as the identity on a purl curl: $\Hproj{ g_k } = g_k$ for $k \in \{ 4 , 5 \}$. Therefore, we are stuck with trying to estimate, for example, 
\begin{align*}
\| \curl{ \langle \psi - \psi' , \gvsig \psi \rangle_{\C^2} } \|_{q_1' ; \frac{2}{q_1} , r_1'}
\end{align*}
by $\| \psi - \psi' \|_{L_T^{\infty} L^2} \| \psi \|_{H^2}$, and such an estimate seems very unlikely.

One could try to mend this approach by considering a more sophisticated metric space than $(\sZ_{\mr{P}}^{2,2} (T) , d)$. For example, one could increase the number of spatial derivatives on $\psi - \psi'$ in the definition of the metric $d$. However, one would then have to estimate $\| \xi - \xi' \|_{\infty; s , 2}$ for some $s > 0$, and this would serve to increase the number of derivatives we need on $\v{A} - \v{A}^{\!\prime}$ in the metric. Hence, we will fall into a negative feedback loop in trying to "close the metric". We expect that to overcome this problem and prove local well-posedness of the MP system (\ref{eq:MP}) one needs a very careful analysis of the spin current $\curl{\langle \psi , \gvsig \psi \rangle_{\C^2}}$ and perhaps sharper results concerning the two-parameter unitary group $U_{\v{A}}^{\mr{P}}$ such a Strichartz estimates. We admit the possibility of a clever choice of metric space involving of classical Sobolev spaces $L_T^q W^{m , p}$ that allows one to "close the metric" in the way just mentioned. However, our attempts at constructing such a metric space have been unsuccessful. We remark that if one ignores the spin current term and considers the initial value problem
\begin{align*}
\left\lbrace \begin{array}{l}
i \partial_t \psi = [ \gvsig \cdot (\v{p} + \v{A}) ]^2 \psi \\
\square \v{A} = - 8 \pi \alpha^2 \Hproj{ \re{\langle \psi , (\v{p} + \v{A}) \psi \rangle_{\C^2}} } \\
\diver{\v{A}} = 0 \\
( \psi (0) , \v{A} (0) , \partial_t \v{A} (0) ) \in H^s \oplus H^s \oplus H^{s-1} ,
\end{array} \right.
\end{align*}
then it does seem possible to prove local well-posedness of this system using the methods described in this section for regularities $s > 5/2$ (this range of regularities comes from the need to estimate the term $\| \v{p} \psi \|_{\infty}$).

\section{The Maxwell-Dirac System}\label{sec:MD}

Another direction one may pursue is to consider relativistic effects. In this regard, the MD equations may serve as another possible model for an electron interacting with its classical self-generated electromagnetic field. For an electron moving in three dimensions with no nuclei present, the MD equations considered in the mathematical literature read
\begin{align}\label{eq:MD}
\left\lbrace \begin{array}{l}
i \partial_t \psi = \left(  \dfrac{1}{\alpha} \gvalp \cdot \left( \v{p} + \v{A} \right) + \dfrac{1}{2 \alpha^2} \beta - \dfrac{1}{\alpha} \varphi \right) \psi  \\[1.2ex]
\square \v{A} + \nabla ( \diver{\v{A}} + \alpha \partial_t \varphi ) = 4 \pi \alpha \v{J}_{\mr{D}} [\psi] \\
- \Delta \varphi - \alpha \partial_t \diver{\v{A}} = - 4 \pi \alpha |\psi |^2 , \\
\end{array} \right. 
\end{align}
where $\psi (t) : \R^3 \rightarrow \C^4$ is the single-particle $4$-component spinor, and
\begin{align}\label{def:dirac_prob_current}
\v{J}_{\mr{D}} [\psi]  = - \langle \psi , \gvalp \psi \rangle_{\C^4}
\end{align}
is the Dirac probability current. In (\ref{eq:MD}), $( \beta , \gvalp ) \in \R^4 \otimes M_{4 \times 4} (\C)$ is the vector of Dirac matrices $(\beta , \gvalp) = ( \beta , \alpha^1 , \alpha^2 , \alpha^3)$, which are $4\times 4$ Hermitian matrices assumed to anti-commute and have square equal to $I$. A commonly chosen representation of these matrices is
\begin{align*}
\beta = \left( \begin{array}{cc}
I & 0 \\
0 & -I 
\end{array} \right) , \hspace{1cm} \gvalp = \left( \begin{array}{cc}
0 & \gvsig \\
\gvsig & 0
\end{array} \right) ,
\end{align*}
where $\gvsig \in \R^3 \otimes M_{2 \times 2} (\C)$ is the vector of Pauli matrices. The matrices $\alpha^j$ should not be confused with the fine structure constant $\alpha$. Associated with the MD equations (\ref{eq:MD}) are several conserved quantities \cite{Glassey1979} and, in particular, $\| \psi \|_2$ and the total energy
\begin{align}\label{eq:Dirac_Energy}
E_{\mr{D}} [\psi , \v{A}]= \frac{1}{\alpha} \langle \psi , \gvalp \cdot \left( \v{p} + \v{A} \right) \psi \rangle_{L^2} + F [ \curl{ \v{A} } , - \nabla \varphi - \alpha \partial_t \v{A} ] ,
\end{align}
where $F$ is given by (\ref{def:field_energy_potentials}), are both conserved. The total energy (\ref{eq:Dirac_Energy}) does not have a definite sign (see, for example, \cite[Chapter 10]{Schwabl08}) and, therefore, is not as useful a quantity as it is in the case of the MS system (\ref{eq:MS}). As was the case with the MS system and the MP system, the presence of the electrostatic potential in (\ref{eq:MD}) is not physical. However, it is included in most of the mathematical literature concerning the MD system and so we include it for consistency. 

The relationship between the MP system (\ref{eq:MP}) and the MD system (\ref{eq:MD}) comes from the well-known fact that the non-relativistic limit of the Dirac equation,
\begin{align}\label{eq:Dirac_eqn}
i \partial_t \psi = \left(  \dfrac{1}{\alpha} \gvalp \cdot \left( \v{p} + \v{A} \right) + \dfrac{1}{2 \alpha^2} \beta - \dfrac{1}{\alpha} \varphi \right) \psi ,
\end{align}
produces the Pauli equation. In what follows we give a rough, non-rigorous argument as to why this is the case following the presentation in \cite[\S 5.3]{Schwabl08}. For the mathematically rigorous result see \cite{BMP98}, and for a derivation using operator resolvents that shows the non-relativistic limit of the Dirac operator is the Pauli operator may be found in \cite[Chapter 6]{Thaller92}. Let $\psi (t) : \R^3 \rightarrow \C^4$ satisfy (\ref{eq:Dirac_eqn}) and write $\psi (t) = ( \chi_1 (t) , \chi_2 (t))$ with $\chi_j (t) : \R^3 \rightarrow \C^2$, $j \in \{1,2\}$. Introduce $\tilde{\psi} (t) : \R^3 \rightarrow \C^4$ as
\begin{align*}
\tilde{\psi} (t) = \left( \begin{array}{c}
\tilde{\chi}_1 (t) \\
\tilde{\chi}_2 (t) 
\end{array} \right) = \exp{\left\lbrace \frac{i t}{2 \alpha^2} \right\rbrace} \left( \begin{array}{c}
\chi_1 (t) \\
\chi_2 (t) 
\end{array} \right) .
\end{align*}
In terms of $\tilde{\chi}_j$, $j \in \{1,2\}$, the Dirac equation (\ref{eq:Dirac_eqn}) becomes
\begin{align}\label{eq:symmetry_broken_dirac}
\left\lbrace \begin{array}{l}
i  \partial_t \tilde{\chi}_1 =  \frac{1}{\alpha} \gvsig \cdot (\v{p} + \v{A}) \tilde{\chi}_2 - \frac{1}{\alpha} \varphi \tilde{\chi}_1  \\
i  \partial_t \tilde{\chi}_2 = \frac{1}{\alpha} \gvsig \cdot (\v{p} + \v{A}) \tilde{\chi}_1 - \frac{1}{\alpha} \varphi \tilde{\chi}_2 - \frac{1}{\alpha^2} \tilde{\chi}_2 .
\end{array} \right.  
\end{align}
Taking the second equation in (\ref{eq:symmetry_broken_dirac}) and solving for $\tilde{\chi}_2$ we find
\begin{align*}
\tilde{\chi}_2 = \alpha \gvsig \cdot (\v{p} + \v{A}) \tilde{\chi}_1 - \alpha \varphi \tilde{\chi}_2 - i \alpha^2 \partial_t \tilde{\chi}_2 .
\end{align*}
Subsituting the previous expression into the first equation in (\ref{eq:symmetry_broken_dirac}) we arrive at
\begin{align}\label{eq:symmetry_boken_dirac_2}
i  \partial_t \tilde{\chi}_1 = \left( [ \gvsig \cdot (\v{p} + \v{A}) ]^2- \frac{1}{\alpha} \varphi \right) \tilde{\chi}_1 - \gvsig \cdot (\v{p} + \v{A}) \left( \varphi \tilde{\chi}_2 + i \alpha \partial_t \tilde{\chi}_2 \right) 
\end{align}
To consider the non-relativistic limit we assume all energies are very small compared to the rest energy $m c^2 = 1 / (2 \alpha^2)$ or, equivalently, $\alpha \gg \alpha^2$. Therefore, in (\ref{eq:symmetry_boken_dirac_2}) we consider the term $\gvsig \cdot (\v{p} + \v{A}) \left( \varphi \tilde{\chi}_2 + i \alpha \partial_t \tilde{\chi}_2 \right)$ to be very small relative to the other terms in (\ref{eq:symmetry_boken_dirac_2}). Hence, we neglect these terms and arrive at the Pauli equation for $\tilde{\chi}_1$:
\begin{align*}
i \partial_t \tilde{\chi}_1 = \left( [\gvsig \cdot (\v{p} +\v{A}) ]^2 - \frac{1}{\alpha} \varphi \right) \tilde{\chi}_1 . 
\end{align*}
Under these approximations, we also see that the Dirac probability current density becomes the Pauli probability current density. Indeed,
\begin{align*}
- \langle \psi , \gvalp \psi \rangle_{\C^4} = - 2 \re{ \langle \tilde{\chi}_1 , \gvsig \tilde{\chi}_2 \rangle_{\C^2} } = - 2 \alpha \re{ \langle \gvsig \tilde{\chi}_1 , \gvsig \cdot (\v{p} + \v{A}) \tilde{\chi}_1 \rangle_{\C^2} } .  
\end{align*}

Similar to the MS equations (\ref{eq:MS}), there is an extensive literature studying the MD equations (\ref{eq:MD}). To our knowledge, Gross in \cite{Gross1966} was the first to study the Cauchy problem associated with (\ref{eq:MD}). Gross shows the local existence and uniqueness of solutions to (\ref{eq:MD}) in the Lorenz gauge $\diver{\v{A}} + \alpha \partial_t \varphi = 0$ for initial data 
\begin{align*}
(\psi_0 , \v{a}_0 , \dot{\v{a}}_0 ) \in H^1 (\R^3 ; \C^4) \times ( \dot{H}^{\frac{3}{2}} (\R^3 ; \R^3) \cap \dot{H}^1 (\R^3 ; \R^3) ) \times H^{\frac{1}{2}} (\R^3 ; \R^3) . 
\end{align*}
Gross points out however that obtaining global solutions is a far more delicate matter since, in principle, one needs a local existence theory in the energy class $H^{\frac{1}{2}} (\R^3 ; \C^4) \times \dot{H}^1 (\R^3 ; \R^3) \times L^2 (\R^3 ; \R^3)$ where conservation of energy and momentum holds. However, the total energy for the MD system (\ref{eq:Dirac_Energy}) does not have a definite sign, and, therefore, conservation of energy together with a local result in the energy class cannot be used to construct a global solution.   

Significant work has gone into understanding the well-posedness of the MD system since the work of Gross in 1966. Chadam in 1972 \cite{Chadam1972} extended the work of Gross by showing that one can solve the Cauchy problem for (\ref{eq:MD}) in a arbitrary bounded region of space-time (provided that either the coupling between the electron and the field, or the initial data, was sufficiently small) and that for a cut-off version of (\ref{eq:MD}) admits a global solution. In the following years, Chadam \cite{Chadam1973} shows that a one dimensional version of (\ref{eq:MD}) admits a global solution, and Chadam-Glassey in \cite{Glassey1974} analyze the asymptotic behaviour of these one-dimensional, global solutions and show that they blow up at each spatial point as $|t| \rightarrow + \infty$. 
 
Much later, Flato, Simon and Taflin \cite{flato1987} argue global existence and uniqueness to the three dimension MD equations in the Lorenz gauge for initial data in the image of certain modified wave operators. Georgiev \cite{Geo1991} then proved global existence and uniqueness to the MD system in the Lorenz gauge for sufficiently small initial data with respect to weighted Sobolev norms. Using the \textit{null form} estimates of Klainerman and Machedon for nonlinear wave equations \cite{KM93}, Bournaveas \cite{Bournaveas1996} as part of his PhD thesis shows that the MD equations in the Coulomb gauge is locally well-posed in the space of initial data 
\begin{align*}
(\psi_0 , \v{a}_{0} , \dot{\v{a}}_{0} ) \in H^{1} (\R^3 ; \C^4) \times H^{2} (\R^3 ; \R^3)  \times H^1 (\R^3 ; \R^3) .
\end{align*}
Using space-time estimates for the wave equation, Bournaveas in the same paper shows that MD system in the Lorenz gauge is locally well-posed for initial data in Sobolev spaces which are an $\epsilon$ above the energy class.

To our knowledge the state-of-the-art local well-posedness result for the Cauchy problem of the MD system is from the work of D'Ancona, Foschi, and Selberg \cite{Ancona2010}. There the authors uncovered the full \textit{null structure} of the MD system in the Lorenz gauge and use it to prove local well-posedness for initial data
\begin{align*}
(\psi_0 , \v{E}_{0} , \v{B}_{0} ) \in H^s (\R^3 ; \C^4) \times H^{s-\frac{1}{2}} (\R^3 ; \R^3)  \times H^{s-\frac{1}{2}} (\R^3 ; \R^3) 
\end{align*}
where $s > 0$, and $\v{E}_0$ and $\v{B}_0$ are the initial values of the electric and magnetic fields, respectively. Note that global well-posedness could be deduced from the conservation of $\| \psi (t) \|_{2} = \| \psi_0 \|_{2}$ if we could take $s = 0$. This is currently an open problem. Furthermore, D'Ancona and Selberg \cite{Ancona2011} solve the \textit{two-dimensional} Cauchy problem for the MD system completely by proving global well-posedness in the $L^2$-class, as well as persistence of higher regularity. 

It is an interesting question whether the null structures in the MD system that have been exploited with success to study well-posedness can be generalized to the MP system. It is clear, however, that some strategies involving null form estimates cannot be adapted to the MP system. For example, those strategies introduced in \cite{Bournaveas1996} rely heavily on the fact that the square of the massive free Dirac operator, $\frac{1}{\alpha} \gvalp \cdot \v{p} + \frac{1}{2 \alpha^2} \beta$, is the operator $- \frac{1}{\alpha^2} \Delta + \frac{1}{4 \alpha^4}$. A potentially promising avenue for future research is to investigate the possibility that a null structure is present in the MP system which may be used to control the spin current $\curl{\langle \psi , \gvsig \psi \rangle_{\C^2}}$.

Another interesting question is whether it is possible to consider the non-relativistic limit of solutions to the MD system and produce solutions to the MP system. This question of a non-relativistic limit has been considered in the work of Masmoudi and Mauser \cite{MM01}, Masmoudi and Nakanishi \cite{MN03}, and Bechouche, Mauser, and Selberg \cite{BMS05}. To discuss the aforemention mentioned articles, we express the Coulomb gauged MD system in Gaussian units:
\begin{align}\label{eq:MD_Gaussian}
\left\lbrace \begin{array}{l}
i \hbar \partial_t \psi = \left(  c \gvalp \cdot \left( \v{p} - \dfrac{e_0}{c} \v{A} \right) + \beta m_e c^2 + e_0 \varphi \right) \psi \\[1.2ex]
\square \v{A} = - 4 \pi \alpha \Hproj{ \langle \psi , \gvalp \psi \rangle_{\C^4} } \\ 
- \Delta \varphi = 4 \pi e | \psi |^2 \\
\diver{\v{A}} = 0 ,
\end{array} \right. 
\end{align}
Considering the $c \rightarrow +\infty$-limit of (\ref{eq:MD_Gaussian}), the authors in \cite{MM01, MN03, BMS05} produce solutions to a Schr\"{o}dinger-Poisson system involving the upper and lower 2-spinors of $\psi$. However, in doing so, one has effectively "turned off the lights" and removed the electromagnetic field in the $c \rightarrow + \infty$-limit. To avoid this one should absorb the factor of $1/c$ into the definition of the vector potential in the operator $\v{p} - (e/c) \v{A}$. This factor of $c$ has nothing to do with relativistic effects and would also appear in the Schr\"{o}dinger equation in an external magnetic field. Bechouche, Mauser, and Selberg in \cite[Theorem 1.7]{BMS05} do give a result concerning the MP system as a semi-nonrelativistic limit of (\ref{eq:MD_Gaussian}). There, however, the authors assume one has apriori a solution to the MP system. 
 
For some concluding remarks to this section, we mention that the MD system has been looked at numerically by Shebalin \cite{SHEBALIN1997}, Bao and Li \cite{BL04}, and Li, Chan, and Hou \cite{LCH10}. Moreover, we point out that there are no results for the MD system with a Coulomb interaction added. Such a system would read 
\begin{align}\label{eq:MD_coulomb}
\left\lbrace \begin{array}{l}
i \partial_t \psi = \left(  \dfrac{1}{\alpha} \gvalp \cdot \left( \v{p} + \v{A} \right) + \dfrac{1}{2 \alpha^2} \beta - \dfrac{Z}{|\v{x}|} \right) \psi , \\[1.2ex]
\square \v{A} = 4 \pi \alpha \Hproj{ \v{J}_{\mr{D}} [\psi] } \\
\diver{\v{A}} = 0 ,
\end{array} \right. 
\end{align}
for some interaction strength $Z > 0$. Using that the massive Dirac-Coulomb operator $\gvalp \cdot \v{p} + \beta - \nu |\v{x}|^{-1}$ is essentially self-adjoint on $C_0^{\infty} (\R^3 ; \C^4)$ for $\nu \in [0 , \sqrt{3}/2)$, it is possible to directly adapt Gross's proof strategy \cite{Gross1966} to prove the local well-posedness of (\ref{eq:MD_coulomb}) for initial data in the same Sobolev spaces considered there. One could in principle extend this result to $\nu \geq \sqrt{3}/2$ by choosing a distinguished self-adjoint extension (c.f.~\cite{EL07}) An interesting open problem for (\ref{eq:MD_coulomb}) would be to improve this local well-posedness down to regularity considered in \cite{Ancona2010}. Another problem would be to study a two-dimensional analogue of (\ref{eq:MD_coulomb}) and proving a result analogous to that in \cite{Ancona2011}.

\section{The Brown-Ravenhall-Maxwell System}\label{sec:BR}

We conclude this chapter with a section entertaining another possible model for an electron interacting with its self-generated electromagnetic field. This model is based on the Brown-Ravenhall Hamiltonian as introduced by Brown and Ravenhall \cite{BR51} as an approximation to Quantum Electrodynamics (see also \cite{Sucher80}). Let $\D (\v{A})$ be the massive free Dirac operator, i.e.
\begin{align*}
\D (\v{A}) = \frac{1}{\alpha} \gvalp \cdot (\v{p} + \v{A}) + \frac{1}{2 \alpha^2} \beta ,
\end{align*}
and let $\Lambda_{+, \v{A}} = \chi_{(0,\infty)} (\D (\v{A}))$ be the spectral projection onto the positive spectral subspace of $\D (\v{A})$ (see \S\ref{sec:MD} for the definitions of $\beta$ and $\gvalp$). Here $\chi_{\Omega}$ is the indicator function of the set $\Omega$. We define the \textit{Brown-Ravenhall operator} as
\begin{align*}
\D_+ (\v{A}) = \Lambda_{+ , \v{A}} \D (\v{A}) = | \D (\v{A}) |
\end{align*}
and the \textbf{Brown-Ravenhall-Maxwell equations} as 
\begin{align}\label{eq:BRM}
\left\lbrace \begin{array}{l}
i \partial_t \psi = \D_+ (\v{A}) \psi \\
\square \v{A} = 4 \pi \alpha \Hproj{ \v{J}_{\mr{BR}} [ \psi , \v{A} ] }  \\
\diver{\v{A}} = 0 .
\end{array} \right.
\end{align}
The probability current associated with the operator $\D_+(\v{A})$ is denoted by $\v{J}_{\mr{BR}} [ \psi , \v{A} ]$ and takes the form
\begin{align}\label{eq:BR_Current}
\v{J}_{\mr{BR}} [\psi , \v{A}] = - \langle \psi , \Lambda_{+ , \v{A}} \gvalp \psi \rangle_{\C^4} - \frac{1}{2 \pi i} \oint_{\Gamma} \langle R_{z} (\v{A}) \psi , \gvalp R_{\ol{z}} (\v{A}) \D (\v{A}) \psi \rangle_{\C^4} \dd z ,
\end{align}
where $R_{z} (\v{A}) = (\D (\v{A}) - z)^{-1}$ and $\Gamma$ is any contour in the right half plane $\{ z \in \C : \re{z} > 0 \}$ that contains the spectrum of $\D_+ (\v{A})$ (note that the spectrum of $\D_+ (\v{A})$ is real, and stricly positive).

To our knowledge, the system (\ref{eq:BRM}) has not been considered in the literature. It would be interesting to study the local and global well-posedness of the system (\ref{eq:BRM}). Perhaps even more interesting is to consider the system
\begin{align}\label{eq:BRMC}
\left\lbrace \begin{array}{l}
i \partial_t \psi = \left( \D_+ (\v{A}) - \dfrac{Z}{|\v{x}|} \right) \psi \\[1.2ex]
\square \v{A} = 4 \pi \alpha \Hproj{ \v{J}_{\mr{BR}} [ \psi , \v{A} ] }  \\
\diver{\v{A}} = 0 .
\end{array} \right.
\end{align}
We speculate that global existence of (\ref{eq:BRMC}) will require assumptions on the size of $\alpha$ and $Z \alpha$. Indeed, the energetic stability of the Brown-Ravenhall, or "no-pair", model has been considered by several authors in the literature \cite{LSS97, HS99, LL02}. There one finds energetic stability if $\alpha$ and $Z \alpha$ are small enough, and instability if $\alpha$ or $Z \alpha$ are too large. We find it a very interesting problem to prove global well-posedness to (\ref{eq:BRMC}) under the assumption that $\alpha$ and $Z \alpha$ are small enough. 

We conclude this last section with a semi-formal derivation of the probability current (\ref{eq:BR_Current}). 
As usual, the probability current (\ref{eq:BR_Current}) is obtained by taking the variational derivative of the kinetic energy $\langle \psi , \D_+ (\v{A}) \psi \rangle_{L^2 (\R^3 ; \C^4)}$ with respect to $\v{A}$. Deriving $\v{J}_{\mr{BR}} [\psi , \v{A}]$ is nontrivial since we need to understand the variational derivative of the spectral projections $\Lambda_{+ , \v{A}}$ with respect to $\v{A}$. In the following we note the important fact that the spectrum of $\D (\v{A})$ does not contain the interval $(- (2 \alpha^2)^{-1} , (2 \alpha^2)^{-1} )$ (in particular, it does not contain $0$) for a large class of vector potentials $\v{A}$ (see, for example, \cite[Theorem 7.1]{Thaller92}). Also, we will always assume $\v{A}$ is some fixed, divergence-free vector potential that is smooth enough to justify the computations and $\psi (t) \in H^1 (\R^3 ; \C^4)$ is such that $\| \D (\v{A}) \psi \|_{\infty ; 2} < \infty$. 

The first term on the right hand side (\ref{eq:BR_Current}) comes from the variational derivative of $\D (\v{A})$. The second term on the right hand side of (\ref{eq:BR_Current}) is due to the $\v{A}$-dependence of the spectral projection, and its derivation is based on Stone's formula 
\begin{align}\label{eq:stone_formula_function}
\frac{1}{2 \pi i} \int_0^{\infty} \left( \frac{1}{x - \lambda - i \epsilon} - \frac{1}{x - \lambda + i \epsilon} \right) d \lambda \xrightarrow{~ \epsilon \rightarrow 0^+ ~} \left\lbrace \begin{array}{ll}
0, & x < 0 \\
1/2, & x = 0 \\
1, & x > 0 .
\end{array} \right. 
\end{align}
Therefore,
\begin{align}\label{eq:Stone_formula}
\Lambda_{+ , \v{A}} := \chi_{(0 , \infty)} ( \D (\v{A})) = \frac{1}{\pi} \slim_{~ \epsilon \rightarrow 0^+ ~} \int_0^{\infty} \frac{\epsilon}{(\D (\v{A}) - \lambda)^2 + \epsilon^2} \dd \lambda .
\end{align}
Observe that $\langle ( \Lambda_{+, \v{A} + \delta \v{A}^{\!\prime}} - \Lambda_{+, \v{A}} ) \psi , \D (\v{A}) \psi \rangle_{L^2}$ contains the term
\begin{align*}
& \frac{1}{\D (\v{A}) + \delta \gvalp \cdot \v{A}^{\!\prime} / \alpha - z} - \frac{1}{\D (\v{A}) - z} \\
& \hspace{2cm} = \frac{\delta}{\alpha} \left( \D (\v{A}) + \delta \gvalp \cdot \v{A}^{\!\prime} / \alpha - z \right)^{-1}  \gvalp \cdot \v{A}^{\!\prime}  \left( \D (\v{A}) - z \right)^{-1} .
\end{align*}
Hence, in the limit $\delta \rightarrow 0$ we find
\begin{align*}
& \lim_{\delta \rightarrow 0} \frac{ \langle \left( \Lambda_{+, \v{A} + \delta \v{A}^{\!\prime}} - \Lambda_{+, \v{A}} \right) \psi , \D (\v{A}) \psi \rangle_{L^2} }{\delta} \\
& = \frac{1}{2 \pi \alpha i} \lim_{\im{z} \rightarrow 0^+} \int_0^{\infty} \left[ \langle R_{z} (\v{A}) \gvalp \cdot \v{A}^{\!\prime} R_{z} (\v{A}) , \D (\v{A}) \psi \rangle_{L^2} \right. \\
& \hspace{5cm} \left. - \langle R_{\ol{z}} (\v{A}) \gvalp \cdot \v{A}^{\!\prime} R_{\ol{z}} (\v{A}) \psi , \D (\v{A}) \psi \rangle_{L^2} \right] \dd \re{z} . \numberthis \label{eq:BR_1} 
\end{align*}
To proceed, recall that for $f : \C \rightarrow \C$ we have formula
\begin{align}\label{eq:BR_2}
\int_{\nu}^{\infty} ( f(\lambda + i \epsilon) - f (\lambda - i \epsilon) ) \dd \lambda = \oint_{\Gamma} f(z) \dd z - i \int_{- \epsilon}^{\epsilon} f( \nu + i x ) \dd x ,
\end{align}
where $\Gamma$ is the rectangular contour of width $2 \epsilon$ with left side at $\nu \in (0 , (2 \alpha^2)^{-1} )$ which encloses the positive real axis and has counterclockwise orientation. We apply this formula to the function $f : \C \rightarrow \C$ given by
\begin{align}\label{eq:BR_3}
f(z) = \langle R_{z} (\v{A}) \gvalp \cdot \v{A}^{\!\prime} R_{z} (\v{A}) \psi , D (\v{A}) \psi \rangle_{L^2} .
\end{align}
Note that $f(z)$ analytic since $R_{z} (\v{A})$ is a strongly analytic function from the resolvent set $\rho (\D (\v{A})) \subset \C$ to $\B (L^2)$. With this, (\ref{eq:BR_1}) will to equal
\begin{align*}
\frac{1}{2 \pi \alpha i} \oint_{\Gamma} f (z) \dd z - \frac{1}{2 \pi \alpha} \slim_{\epsilon \rightarrow 0^+} \int_{- \epsilon}^{\epsilon} f(\nu + i x) \dd x  ,
\end{align*}
where $f(z)$ is given by (\ref{eq:BR_3}). We've used that the first term $\oint_{\Gamma} f(z) \dd z$ above does not depend on $\epsilon$. Moreover, the first term
\begin{align*}
\frac{1}{2 \pi \alpha i} \oint_{\Gamma} f (z) \dd z
\end{align*}
is real by Cauchy's integral formula, and only knows about the positive part of spectrum of $\D (\v{A})$. This gets us most of the way to the current (\ref{eq:BR_Current}). To finish the derivation, we need to show that
\begin{align}\label{eq:BR_4}
\lim_{\epsilon \rightarrow 0^+} \int_{- \epsilon}^{\epsilon} f(\nu + i x) \dd x  = 0 .
\end{align}
For this, recall the identity $\| (T - z)^{-1} \| = [ \dist{ ( z , \sigma (T) ) } ]^{-1}$ for any self-adjoint operator $T$ ($\sigma (T)$ denotes the spectrum of $T$). Then note that
\begin{align*}
\dist{ ( \nu + i x , \sigma ( \D (\v{A}) ) ) } \geq \left| \nu + i x - \frac{1}{2 \alpha^2} \right| = \sqrt{ \left( \frac{1}{2 \alpha^2} - \nu \right)^2 + x^2 } ,
\end{align*}
and, hence,
\begin{align*}
| f(\nu + i x) | \leq \| \v{A}^{\!\prime} \|_{\infty} \| \psi \|_2 \| \D (\v{A}) \psi \|_2 \left[ \left( \frac{1}{2 \alpha^2} - \nu \right)^2  + x^2 \right]^{-1} .
\end{align*} 
The previous to inequality implies
\begin{align*}
\int_{- \epsilon}^{\epsilon} |f(\nu + i x)| \dd x \leq 4 \alpha^2 \| \v{A}^{\!\prime} \|_{\infty} \| \psi \|_2 \| \D (\v{A}) \psi \|_2 \frac{\arctan{\left( \frac{2 \alpha^2 \epsilon}{ 1 - 2\alpha^2 \nu}\right)}}{ 1 - 2 \alpha^2 \nu} ,
\end{align*}
from which we conclude (\ref{eq:BR_4}). This completes the derivation of (\ref{eq:BR_Current}).

\chapter{The $\varepsilon$-Modified System}\label{chap:epsilon_system}

This Chapter is devoted to proving Theorem \ref{thm:local_exist_MBMP_epsilon}, Theorem \ref{thm:MBMP_epsilon_dissipation-laws}, and eventually Theorem \ref{thm:weak_solns_MBMP}. The first section \S\ref{sec:estimates} is a collection of technical estimates, stated as Lemmas \ref{lem:Estimates-Pauli}-\ref{lem:Estimates_Prob_current}, for the right hand side of (\ref{eq:MBMP_epsilon}) in various Sobolev spaces. The second section \S\ref{sec:metric_space_linearization} consists of setting up the contraction mapping scheme that we will use to prove Theorem \ref{thm:local_exist_MBMP_epsilon}. The following section \S\ref{sec:local} is devoted to actually proving Theorem \ref{thm:local_exist_MBMP_epsilon}. The fourth and fifth sections, \S\ref{sec:conserve} and \S\ref{sec:proof_main_result}, are devoted to proofs of Theorem \ref{lem:bound_on_coulomb} and \ref{thm:MBMP_epsilon_dissipation-laws}, and the main result, Theorem \ref{thm:weak_solns_MBMP}, respectively. 

\section{Technical Estimates}\label{sec:estimates}

This section is devoted to the derivation of several estimates, stated as Lemmas \ref{lem:Estimates-Pauli}-\ref{lem:Estimates_Prob_current}, for the right hand side of (\ref{eq:MBMP_epsilon}) in various Sobolev spaces. To obtain such estimates we will repeatedly make use of Lemmas \ref{lem:Kato-Ponce} and \ref{lem:Heat-Kernel}. The estimates will be crucial for our proof of Theorem \ref{thm:local_exist_MBMP_epsilon}. Some remarks on a particular notation used in this section are in order. Recall that $\Lambda = \sqrt{1 - \Delta}$ and $\dot{\Lambda} = \sqrt{- \Delta}$ (see \S\ref{sec:notation}). For $k \in \{ 1 , \cdots , N \}$, we will use the notation $\Lambda_k^s = (1 - \Delta_{\v{x}_k})^{s/2}$ (likewise for $\dot{\Lambda}_k^s$), where $\Delta_{\v{x}_k}$ is the Laplacian acting on the $k^{\mr{th}}$ electron coordinates, $\Delta_{\v{x}_k} = \sum_{j=1}^3 \partial_{x^j_k}^2$. We emphasize that if no subscript $k$ is present on $\Lambda$, then the Laplacian in the definition of $\Lambda$ is taken to be the full Laplacian acting on all the coordinates in the given context. 

\begin{lem}[Estimates for the Pauli Term] \label{lem:Estimates-Pauli}
Let $m \in [1 , \infty)$ and $N \geq 1$. For all $(\phi , \v{A}) \in [H^m (\R^{3N}) ]^{2^N} \times H^m (\R^3 ; \R^3)$, with $\diver{\v{A}} = 0$ and $\v{B} = \curl{\v{A}}$, and for each $j \in \{ 1, \cdots , N \}$, the operator $\cL_j (\v{A})$ given by 
\begin{align}\label{def:D_j}
\cL_j (\v{A}) = 2 \v{A}_j \cdot \v{p}_j + | \v{A}_j |^2 + \gvsig_j \cdot \v{B}_j 
\end{align}
satisfies the estimates
\begin{align} \label{eq:Estimate-Pauli-1}
\| \cL_j (\v{A}) \phi \|_{m-1,\frac{3}{2}} \lesssim (1 + \|\v{A}\|_{m,2}) \| \v{A} \|_{m,2} \| \phi \|_{m,2} ,
\end{align}
and 
\begin{align} \label{eq:Estimate-Pauli-2}
\| e^{t \Delta} \cL_j ( \v{A} ) \phi \|_{m,2} \lesssim t^{- \frac{1}{4}} \left[ 1 + t^{-\frac{1}{2}} \right] \left( 1 + \| \v{A} \|_{m,2} \right) \| \v{A} \|_{m,2} \| \phi \|_{m,2} ,
\end{align}
for all $t > 0$. Furthermore, for $(\phi , \v{A}) , (\phi' , \v{A}') \in [H^1 (\R^{3N})]^{2^N} \times H^1 (\R^3 ; \R^3)$, with $\diver{\v{A}} = \diver{\v{A}'} = 0$, and each $j \in \{1 , \cdots , N\}$, we have, for all $t > 0$,
\begin{align*}
\| e^{t \Delta} \left[ \cL_j ( \v{A} ) \phi - \cL_j ( \v{A}' ) \phi' \right] \|_{1,2} & \lesssim t^{- \frac{1}{4}} \left( 1 + t^{-\frac{1}{2}} \right) [ ( 1 + \| \v{A} \|_{1,2} + \| \v{A}' \|_{1,2} ) \| \phi' \|_{1,2} \\
& + (1  + \| \v{A} \|_{1,2} ) \| \v{A} \|_{1,2} ] \max{\{ \| \phi - \phi' \|_{1,2} , \| \v{A} - \v{A}' \|_{1,2} \}} . \numberthis \label{eq:Estimate-Pauli-3}
\end{align*}
\end{lem}
\begin{proof}
To show (\ref{eq:Estimate-Pauli-1}) it suffices to consider the case $N = 1$, as the general case follows in a similar fashion. We use Lemma \ref{lem:Kato-Ponce} and the Sobolev inequality $H^1 (\R^3) \subset L^r (\R^3)$, $2 \leq r \leq 6$, to prove (\ref{eq:Estimate-Pauli-1})
\begin{align*}
\| \cL (\v{A}) f \|_{m-1 ,\frac{3}{2}} & \lesssim \| \Lambda^{m-1} \v{A} \|_{6} \| \v{p} f \|_{2} + \| \v{A} \|_{6} \| \Lambda^{m-1} \v{p} f \|_{2} + \| \Lambda^{m-1} \v{A} \|_{6} \| \v{A} f \|_{2} \\
& \hspace{0.2cm}   + \| \v{A} \|_{6} \| \Lambda^{m-1} (\v{A} f) \|_{2} + \| \Lambda^{m-1} \v{B} \|_{2} \| f \|_{6} + \| \v{B} \|_{2} \| \Lambda^{m-1} f \|_{6} \\
& \lesssim \| \v{A} \|_{m,2} \| f \|_{1,2} + \| \v{A} \|_{1,2} \| f \|_{m,2} + \| \v{A} \|_{m,2} \| \v{A} \|_{6} \| f \|_{3}  + \| \v{A} \|_{m,2} \| f \|_{1,2} \\
& \hspace{0.2cm} + \| \v{A} \|_{1,2} \| f \|_{m,2}  + \| \v{A} \|_{1,2} \left( \| \Lambda^{m-1} \v{A} \|_{6} \| f \|_{3}  + \| \v{A} \|_{6} \| \Lambda^{m-1} f \|_{3} \right) \\
& \lesssim (1 + \|\v{A}\|_{m,2}) \| \v{A} \|_{m,2} \| f \|_{m,2} . 
\end{align*} 

To prove (\ref{eq:Estimate-Pauli-2}), fix $j \in \{1 , \cdots , N\}$, and note that (with hopefully obvious notation)
\begin{align}\label{eq:Estimate-Pauli-4}
\| e^{t \Delta} \cL_j ( \v{A} ) \phi \|_{m,2} \lesssim \sum_{k = 1}^N \| \Lambda^{m}_k e^{t \Delta} \cL_j (\v{A}) \phi \|_2 .
\end{align}
We separate into two cases: (a) $k \neq j$ and (b) $k = j$. For case (a) we use Lemma \ref{lem:Heat-Kernel} and (\ref{eq:Estimate-Pauli-1}) to find
\begin{align*}
\| \Lambda^{m}_k e^{t \Delta} \cL_j (\v{A}) \phi \|_2 & \leq \| \Lambda_k e^{t \Delta_{\v{x}_k}} \cL_j (\v{A}) \Lambda_k^{m-1} \phi \|_2 \\
& \lesssim t^{- \frac{1}{4}} \left[ 1 + t^{- \frac{1}{2}} \right] \|  \cL_j (\v{A}) \Lambda_k^{m-1} \phi \|_{\frac{3}{2}} \\
& \lesssim t^{- \frac{1}{4}} \left[ 1 + t^{- \frac{1}{2}} \right]  (1 + \|\v{A}\|_{1,2}) \| \v{A} \|_{1,2} \| \Lambda_k^{m-1} \phi \|_{1,2} \\
& \lesssim t^{- \frac{1}{4}} \left[ 1 + t^{- \frac{1}{2}} \right]  (1 + \|\v{A}\|_{1,2}) \| \v{A} \|_{1,2} \| \phi \|_{m,2} . \numberthis \label{eq:Estimate-Pauli-5}
\end{align*}  
For case (b) we use Lemma \ref{lem:Kato-Ponce} and Lemma \ref{lem:Heat-Kernel}, and the estimate (\ref{eq:Estimate-Pauli-1}), to find
\begin{align*}
\| \Lambda_j^{m} e^{t \Delta} \cL_j (\v{A}) \phi \|_2 & = \| \Lambda_j e^{t \Delta_{\v{x}_j}} \Lambda_j^{m-1} ( \cL_j (\v{A}) \phi ) \|_2 \\
& \lesssim t^{- \frac{1}{4}} \left[ 1 + t^{- \frac{1}{2}} \right] \| \Lambda_j^{m-1} ( \cL_j (\v{A}) \phi ) \|_{\frac{3}{2}} \\
& \lesssim t^{- \frac{1}{4}} \left[ 1 + t^{- \frac{1}{2}} \right] (1 + \|\v{A}\|_{m,2}) \| \v{A} \|_{m,2} \| \phi \|_{m,2} . \numberthis \label{eq:Estimate-Pauli-6}
\end{align*}  
Combining (\ref{eq:Estimate-Pauli-4}) through (\ref{eq:Estimate-Pauli-6}) we arrive at (\ref{eq:Estimate-Pauli-3}). 

To prove (\ref{eq:Estimate-Pauli-4}) we write 
\begin{align*}
\cL_j ( \v{A} ) \phi - \cL_j ( \v{A}' ) \phi' = L_{1,j} [\phi - \phi' , \v{A} ] + L_{2,j} [ \phi , \v{A} - \v{A}']
\end{align*}
where
\begin{align*}
& L_{1,j} [\phi - \phi' , \v{A} ] =  2 \v{A}_j \cdot \v{p}_j (\phi - \phi') + |\v{A}_j|^2 (\phi - \phi' ) + \gvsig_j \cdot \v{B}_j (\phi - \phi' )  , \\
& L_{2,j} [ \phi' , \v{A} - \v{A}'] = 2 ( \v{A}_j - \v{A}'_j ) \cdot \v{p}_j \phi' + (|\v{A}|^2_j - |\v{A}'|^2_j) \phi' + \gvsig_j \cdot (\v{B}_j - \v{B}'_j) \phi' . 
\end{align*}
Using H\"{o}lder's inequality and the Sobolev inequality $H^1 (\R^3) \subset L^r (\R^3)$, $2 \leq r \leq 6$, to find 
\begin{align*}
\| L_{1,j} [\phi - \phi' , \v{A} ] \|_{\frac{3}{2}} & \lesssim \| \v{A} \|_6 \| \phi - \phi' \|_{1, 2} + \| \v{A} \|_4^2 \| \phi - \phi' \|_{6}   + \| \v{B} \|_2 \| (\phi - \phi' ) \|_{6} \\
& \lesssim \left( 2  + \| \v{A} \|_{1,2} \right) \| \v{A} \|_{1,2} \| (\phi - \phi' ) \|_{1,2} . \numberthis \label{eq:Estimate-Pauli-7}
\end{align*}
and
\begin{align*}
& \| L_{2,j} [ \phi' , \v{A} - \v{A}'] \|_{\frac{3}{2}} \\
& \hspace{2mm} \lesssim \| \v{A} - \v{A}' \|_6 \| \phi' \|_{1,2} + \| \v{A} - \v{A}' \|_3 ( \| \v{A} \|_6 + \| \v{A}' \|_6 ) \| \phi' \|_6 + \| \v{B} - \v{B}' \|_2 \| \phi' \|_{6}  \\
& \hspace{2mm} \lesssim \left( 2 + ( \| \v{A} \|_{1,2} + \| \v{A}' \|_{1,2} ) \right) \| \phi' \|_{1,2} \| \v{A} - \v{A}' \|_{1,2}  . \numberthis \label{eq:Estimate-Pauli-8}
\end{align*}
Lemma \ref{lem:Heat-Kernel} gives
\begin{align*}
\| e^{t \Delta} \left[ \cL_j ( \v{A} ) \phi - \cL_j ( \v{A}' ) \phi' \right] \|_{1,2} \lesssim  t^{- \frac{1}{4}} \left[ 1 + t^{-\frac{1}{2}} \right] \|  \cL_j ( \v{A} ) \phi - \cL_j ( \v{A}' ) \phi' \|_{\frac{3}{2}}  ,
\end{align*}
which, together with (\ref{eq:Estimate-Pauli-7}) and (\ref{eq:Estimate-Pauli-8}) allows us to conclude (\ref{eq:Estimate-Pauli-3}).
\end{proof}

\begin{lem}[Estimates for the Coulomb Term] \label{lem:Estimates-Coulomb}
Fix $m \in [1,2]$ and let $N, K \geq 1$, $\Z \in [0 , \infty )^K$, and $\ul{\v{R}} = (\v{R}_1 , \cdots , \v{R}_K) \in \R^{3K}$, with $\v{R}_i \neq \v{R}_j$ for all $i \neq j$. Then, for all $\phi \in H^m (\R^{3N} ; \C)$, the operator $V (\ul{\v{R}} , \Z)$ given by (\ref{def:totelectrostatpot}), satisfies the estimate
\begin{align}\label{eq:Estimate-Coulomb}
\| e^{t \Delta} V (\ul{\v{R}} , \Z) \phi \|_{m , 2} \lesssim \left[ 1 +  \left(1 + t^{- \frac{1}{2} }\right) \left(t^{- \frac{9}{20}} + t^{-\frac{1}{4}} \right) \right] \| \phi \|_{m,2} , 
\end{align}
for all $t > 0$.
\end{lem}

\begin{proof}
To prove (\ref{eq:Estimate-Coulomb}) we need to first prove the following inequalities. Let $v : \R^3 \rightarrow \R^3$ be the function $v(\v{x}) = |\v{x}|^{-1}$. Then, for all $\psi \in H^m (\R^3 ; \C)$, we have 
\begin{align}\label{eq:Estimate-Coulomb-appen_1}
\| v \psi \|_{\frac{3}{2}} \lesssim \| \psi \|_{1,2}  
\end{align} 
and
\begin{align}\label{eq:Estimate-Coulomb-appen_2}
\| v \psi \|_{m-1 , \frac{5}{4}} \lesssim \| \psi \|_{m,2} .
\end{align} 
Moreover, for all $\psi \in H^m (\R^6 ; \C)$, we have
\begin{align}\label{eq:Estimate-Coulomb-appen_3}
\int_{\R^3} \left( \int_{\R^3} \left| \frac{ \psi (\v{x}_1 , \v{x}_2) }{| \v{x}_1 - \v{x}_2 |} \right|^{\frac{3}{2}} \dd \v{x}_1 \right)^{\frac{4}{3}} \dd \v{x}_2 \lesssim \| \Lambda_1\psi \|_{2}^2 
\end{align}
and
\begin{align}\label{eq:Estimate-Coulomb-appen_4}
\int_{\R^3} \left( \int_{\R^3} \left|\Lambda_1^{m-1}\frac{ \psi (\v{x}_1 , \v{x}_2) }{| \v{x}_1 - \v{x}_2 |} \right|^{\frac{5}{4}} \dd \v{x}_1 \right)^{\frac{8}{5}} \dd \v{x}_2 \lesssim \| \Lambda_1^{m} \psi \|_{2}^2  .
\end{align}

Let $B_1$ denote the unit ball in $\R^3$, and $B_1^c = \R^3 \backslash B_1$. Using H\"{o}lder's inequality we find
\begin{align*}
\| v \psi \|_{\frac{3}{2}}^{\frac{3}{2}} & = \int_{B_1} \frac{|\psi (\v{x})|^{\frac{3}{2}}}{|\v{x}|^{\frac{3}{2}}} \dd \v{x} +  \int_{B_1^c} \frac{|\psi (\v{x})|^{\frac{3}{2}}}{|\v{x}|^{\frac{3}{2}}} \dd \v{x} \\
& \leq \left( \int_{B_1} |\v{x}|^{ - 2 } \dd \v{x} \right)^{\frac{3}{4}} \| \psi \|_{L^6 (B_1)}^{\frac{3}{2}} + \left( \int_{B_1^c} |\v{x}|^{- 6} \dd \v{x} \right)^{\frac{1}{4}} \| \psi \|_{L^2 (B_1^c)}^{\frac{3}{2}} \\
& \lesssim \|\psi \|_{6}^{ \frac{3}{2} } + \| \psi \|_{2}^{\frac{3}{2}} . \numberthis  \label{eq:Estimate-Coulomb-appen_1_proof}
\end{align*}
The estimate (\ref{eq:Estimate-Coulomb-appen_1_proof}) and the Sobolev inequality $\| \psi \|_{6} \lesssim \| \nabla \psi \|_{2}$ imply (\ref{eq:Estimate-Coulomb-appen_1}). 

For estimate (\ref{eq:Estimate-Coulomb-appen_2}) we focus on the case $m = 2$, as the $m = 1$ case is proved in the same way as (\ref{eq:Estimate-Coulomb-appen_1}) and then general case $m \in (1,2)$ will follow similarly. Below we will make use of the homogeneous Sobolev space $\dot{W}^{1, \frac{5}{4}} (\R^3)$ defined through the seminorm $\| f \|_{\dot{W}^{1, 5/4} (\R^3)} = \| \dot{\Lambda} f \|_{\frac{5}{4}}$. As before, we write
\begin{align}\label{eq:Estimate-Coulomb-appen_5}
\| v \psi \|_{\dot{W}^{1 , \frac{5}{4}} (\R^3) }^{\frac{5}{4}} = \| v\psi \|_{\dot{W}^{1 , \frac{5}{4}}(B_1) }^{\frac{5}{4}} + \| v\psi \|_{\dot{W}^{1 , \frac{5}{4}}( B_1^c ) }^{\frac{5}{4}} .
\end{align}
We argue, separately, that both terms on the right hand side of (\ref{eq:Estimate-Coulomb-appen_3}) are bounded by $\| \psi \|_{2,2}$. For this it will be useful to remind ourselves of the identity $\dot{\Lambda} | \v{x} |^{-1} = C |\v{x}|^{-2}$ where $C$ is a nonessential constant. To show 
\begin{align}\label{eq:Estimate-Coulomb-appen_6}
\| v\psi \|_{\dot{W}^{1 , \frac{5}{4}}(B_1^c) } \lesssim \| \psi \|_{2,2} 
\end{align}
we use Lemma \ref{lem:Kato-Ponce} to find
\begin{align*}
\|  v\psi \|_{\dot{W}^{1 , \frac{5}{4}}(B_1^c) } & \lesssim \| \dot{\Lambda} v\|_{L^{\frac{10}{3}} (B_1^c)} \| \psi \|_{L^2 (B_1^c)} +  \| v\|_{L^{\frac{10}{3}} (B_1^c)} \| \dot{\Lambda} \psi \|_{L^{2} (B_1^c)}  \\
& \lesssim \| v^2 \|_{L^{\frac{10}{3}} (B_1^c)} \| \psi \|_2 +  \| v\|_{L^{\frac{10}{3}} (B_1^c)} \| \psi \|_{1 , 2} . \numberthis  \label{eq:Estimate-Coulomb-appen_6.1}
\end{align*}
Since $\| |\cdot|^{-k} \|_{L^{\frac{10}{3}} (B_1^c)} < \infty$ for $k \in \{ 1,2 \}$, (\ref{eq:Estimate-Coulomb-appen_6.1}) implies (\ref{eq:Estimate-Coulomb-appen_6}). 

Showing the inequality
\begin{align}\label{eq:Estimate-Coulomb-appen_7}
\| v\psi \|_{\dot{W}^{1 , \frac{5}{4}}(B_1) } \lesssim \| \psi \|_{2,2} .
\end{align}
follows in a similar fashion. Indeed, using Lemma \ref{lem:Kato-Ponce} we find
\begin{align*}
\| v\psi \|_{\dot{W}^{1 , \frac{5}{4}}(B_1) } & \lesssim \|\dot{\Lambda} v\|_{L^{\frac{5}{4}} (B_1)} \| \psi \|_{L^{\infty} (B_1)} + \| v\|_{L^{\frac{30}{19}} (B_1)} \|\dot{\Lambda} \psi \|_{L^{6} (B_1)} \\
& \lesssim \| v^2 \|_{L^{\frac{5}{4}} (B_1)} \| \psi \|_{L^{\infty} (B_1)} + \| v\|_{L^{\frac{30}{19}} (B_1)} \| \Lambda^2 \psi \|_{2} . \numberthis \label{eq:Estimate-Coulomb-appen_8}
\end{align*}
Estimate (\ref{eq:Estimate-Coulomb-appen_8}), together with the Sobolev inequality $\| \psi \|_{\infty} \lesssim \| \psi \|_{2,2}$ and  the observation that $\max{ \{ \| v^2 \|_{L^{\frac{5}{4}} (B_1)} , \| v\|_{L^{\frac{30}{19}} (B_1)} \} } < \infty$, implies (\ref{eq:Estimate-Coulomb-appen_7}). With (\ref{eq:Estimate-Coulomb-appen_5}), (\ref{eq:Estimate-Coulomb-appen_6}), and (\ref{eq:Estimate-Coulomb-appen_7}) we are able to conclude $\| v\psi \|_{\dot{W}^{1,\frac{5}{4}}} \lesssim \| \psi \|_{2,2}$.

Proving (\ref{eq:Estimate-Coulomb-appen_3}) is similar to showing (\ref{eq:Estimate-Coulomb-appen_1}). Indeed, using H\"{o}lder's inequality and the Sobolev inequality $\| f \|_6 \lesssim \| \nabla f \|_2$ we find
\begin{align*}
& \int_{\R^3} \left( \int_{\R^3} \left| \frac{ \psi (\v{x}_1 , \v{x}_2) }{| \v{x}_1 - \v{x}_2 |} \right|^{\frac{3}{2}} \dd \v{x}_1 \right)^{\frac{4}{3}} \dd \v{x}_2 \\
& =  \int_{\R^3} \left( \int_{B_1} \left| \frac{ \psi (\v{y} + \v{x}_2 , \v{x}_2) }{| \v{y} |} \right|^{\frac{3}{2}} \dd \v{y} + \int_{B_1^c} \left| \frac{ \psi (\v{y} + \v{x}_2 , \v{x}_2) }{| \v{y} |} \right|^{\frac{3}{2}} \dd \v{y}  \right)^{\frac{4}{3}} \dd \v{x}_2 \\
& \lesssim \int_{\R^3} \left( \left( \int_{\R^3} | \v{p}_1 \psi (\v{x}_1 , \v{x}_2) |^2 \dd \v{x}_1 \right)^{\frac{3}{4}} + \left( \int_{\R^3} | \psi (\v{x}_1 , \v{x}_2) |^2 \dd \v{x}_1 \right)^{\frac{3}{4}}  \right)^{\frac{4}{3}} \dd \v{x}_2 \\
& \lesssim \| \Lambda_1\psi \|_{2}^2 .
\end{align*}
To show estimate (\ref{eq:Estimate-Coulomb-appen_4}) one combines the strategy used to show (\ref{eq:Estimate-Coulomb-appen_2}) and (\ref{eq:Estimate-Coulomb-appen_3}).

With estimates (\ref{eq:Estimate-Coulomb-appen_1}) through (\ref{eq:Estimate-Coulomb-appen_4}) at our disposal we may prove (\ref{eq:Estimate-Coulomb}). We split $V(\ul{\v{R}} , \Z)$ into three pieces: $V (\ul{\v{R}} , \Z) = \sum_{n = 1}^3 V_n  (\ul{\v{R}} , \Z) $ where 
\begin{align*}
& V_1  (\ul{\v{R}} , \Z) = \sum_{1 \leq i < j \leq N} \frac{1}{|\v{x}_i - \v{x}_j|} , \\
& V_2  (\ul{\v{R}} , \Z) =  - \sum_{i = 1}^N \sum_{j = 1}^K \frac{Z_j}{| \v{x}_i - \v{R}_j |} , \\
& V_3  (\ul{\v{R}} , \Z) = \sum_{1 \leq i < j \leq K} \frac{Z_i Z_j}{|\v{R}_i - \v{R}_j|} .
\end{align*} 
We show (\ref{eq:Estimate-Coulomb}) with $V(\ul{\v{R}} , \Z)$ replaced by $V_n (\ul{\v{R}} , \Z)$, $n = 1 , 2 , 3$. The estimate is trivial for $V_3 (\ul{\v{R}} , \Z)$ since $\ul{\v{R}}$ is fixed. Indeed, we find
\begin{align}\label{eq:Estimate-Coulomb-1}
\| e^{t \Delta} V_3 (\ul{\v{R}} , \Z) \phi \|_{m , 2} \leq \left( \sum_{i,j = 1}^K \frac{Z_i Z_j}{|\v{R}_i - \v{R}_j|} \right) \| e^{t \Delta} \phi \|_{m,2} \lesssim \| \phi \|_{m,2} .
\end{align}
For $V_2 (\ul{\v{R}} , \Z)$, the desired estimate is equivalent to controlling $\| e^{t \Delta} | \v{x}_i |^{-1} \phi \|_{m,2}$ by $\| \phi \|_{m,2}$ for each $i = 1 , \cdots , N$. For this, fix $i \in \{ 1 , \cdots , N \}$ and note that 
\begin{align}\label{eq:Estimate-Coulomb-2}
\| e^{t \Delta} | \v{x}_i |^{-1} \phi \|_{m,2} & \lesssim \sum_{k = 1}^N \| \Lambda^{m}_k e^{t \Delta} | \v{x}_i |^{-1} \phi \|_{2} .
\end{align}
To estimate the right hand side of (\ref{eq:Estimate-Coulomb-2}) we consider two cases: (a) $k \neq i$ and (b) $k = i$. For case (a), we use Lemma \ref{lem:Heat-Kernel} and the estimate (\ref{eq:Estimate-Coulomb-appen_1}) to find 
\begin{align*}
\| \Lambda^{m}_k e^{t \Delta} | \v{x}_i |^{-1} \phi \|_{2} & \leq \| e^{t \Delta_{\v{x}_i}} | \v{x}_i |^{-1} \Lambda^{m}_k e^{t \Delta_{\v{x}_k}} \phi \|_2 \\
& \lesssim t^{- \frac{1}{4}} \left( \int_{\R^{3(N-1)}} \left( \int_{\R^3} \left| | \v{x}_i |^{-1} \Lambda^{m}_k e^{t \Delta_{\v{x}_k}} \phi ( \ul{\v{x}} )  \right|^{\frac{3}{2}} \dd \v{x}_i \right)^{\frac{4}{3}} \dd \ul{\v{x}}_i' \right)^{\frac{1}{2}} \\
& \lesssim t^{- \frac{1}{4}} \| \Lambda_k  e^{t \Delta_{\v{x}_k}} \Lambda_i \Lambda_k^{m-1}  \phi \|_2  \\
& \lesssim t^{- \frac{1}{4}} \left[1 + t^{-\frac{1}{2}} \right] \| \Lambda_i \Lambda_k^{m-1} \phi \|_2 \\
& \lesssim t^{- \frac{1}{4}} \left[1 + t^{-\frac{1}{2}} \right] \| \phi \|_{m,2} . \numberthis \label{eq:Estimate-Coulomb-3}
\end{align*}
For case (b) the estimating is similar to that of (\ref{eq:Estimate-Coulomb-3}). Using (\ref{eq:Estimate-Coulomb-appen_2}) we find
\begin{align*}
\| \Lambda_i^{m}e^{t \Delta} | \v{x}_i |^{-1} \phi \|_2 & \lesssim \| \Lambda_i e^{t \Delta_{\v{x}_i}}\Lambda_i^{m-1}| \v{x}_i |^{-1} \phi \|_2 \\
& \lesssim t^{- \frac{9}{20}} \left[1 + t^{- \frac{1}{2} } \right] \left( \int_{\R^{3(N-1)}} \left( \int_{\R^3} \left|\Lambda_i^{m-1}\frac{\phi ( \ul{\v{x}} )}{|\v{x}_i|}  \right|^{\frac{5}{4}} \dd \v{x}_i \right)^{\frac{8}{5}} \dd \ul{\v{x}}_i' \right)^{\frac{1}{2}} \\
& \lesssim  t^{- \frac{9}{20}} \left[1 + t^{- \frac{1}{2} } \right] \| \Lambda_i^m \phi \|_2 \\
& \lesssim t^{- \frac{9}{20}} \left[1 + t^{- \frac{1}{2} } \right] \| \phi \|_{m,2} . \numberthis \label{eq:Estimate-Coulomb-4}
\end{align*}
Combining estimates (\ref{eq:Estimate-Coulomb-3}) and (\ref{eq:Estimate-Coulomb-4}) we arrive at
\begin{align*}
\| e^{t \Delta} V_2 (\ul{\v{R}} , \Z) \phi \|_{m , 2} & \leq  \sum_{i = 1}^N \sum_{j = 1}^K Z_j \| e^{t \Delta} | \v{x}_i - \v{R}_j |^{-1} \phi \|_{m,2} \\
& \lesssim \left( 1 + t^{- \frac{1}{2} } \right) \left( t^{- \frac{9}{20}} + t^{-\frac{1}{4}} \right) \| \phi \|_{m,2} . \numberthis \label{eq:Estimate-Coulomb-5}
\end{align*}

Finally we need to control $\| e^{t \Delta} | \v{x}_i - \v{x}_j |^{-1} \phi \|_{m,2}$ by $\| \phi \|_{m,2}$ for each $i,j = 1 , \cdots , N$ with $i \neq j$. The estimates involved are similar to those involved with controlling $\| e^{t \Delta} V_2 (\ul{\v{R}} , \Z) \phi \|_{m , 2}$, and thus we choose to be brief with the computations. Fix $( i , j ) \in \{ 1 , \cdots , N \}^2$ with $i \neq j$. Note that
\begin{align}\label{eq:Estimate-Coulomb-6}
\| e^{t \Delta} | \v{x}_i - \v{x}_j |^{-1} \phi \|_{m,2} & \lesssim \sum_{k = 1}^N \| \Lambda^{m}_k e^{t \Delta} | \v{x}_i - \v{x}_j |^{-1} \phi \|_{2} .
\end{align}
Estimating the right hand side of (\ref{eq:Estimate-Coulomb-6}) is similar to estimating the right hand side of (\ref{eq:Estimate-Coulomb-2}). We again consider two cases: (a) $k \neq j, i$ and (b) $k = j, i$. For case (a) we use Lemma \ref{lem:Heat-Kernel} and (\ref{eq:Estimate-Coulomb-appen_3}) to find
\begin{align*}
\| \Lambda^{m}_k e^{t \Delta} \frac{\phi}{| \v{x}_i - \v{x}_j |} \|_{2} & \leq \| e^{t \Delta_{\v{x}_i}} \frac{\Lambda^{m}_k e^{t \Delta_{\v{x}_k}} \phi}{|\v{x}_i - \v{x}_j|} \|_{2} \\
& \lesssim t^{- \frac{1}{4}} \left( \int_{\R^{3(N-1)}} \left( \int_{\R^3} \left| \frac{\Lambda^{m}_k e^{t \Delta_{\v{x}_k}} \phi ( \ul{\v{x}} ) }{ | \v{x}_i - \v{x}_j | } \right|^{\frac{3}{2}} \dd \v{x}_i \right)^{\frac{4}{3}}  \dd \ul{\v{x}}'_{i} \right)^{\frac{1}{2}} \\
& \lesssim t^{- \frac{1}{4}} \| \Lambda_i \Lambda^{m}_k e^{t \Delta_{\v{x}_k}} \phi \|_2 \\
& \lesssim t^{- \frac{1}{4}} [1 + t^{-\frac{1}{2}}] \| \Lambda_i \Lambda_k^{m-1} \phi \|_2 \\
& \lesssim t^{- \frac{1}{4}} [1 + t^{-\frac{1}{2}}] \| \phi \|_{m,2} . \numberthis \label{eq:Estimate-Coulomb-7}
\end{align*}
For case (b) the estimating is similar. We choose $k = i$, and note that the case $k = j$ is identical by symmetry. Using Lemma \ref{lem:Heat-Kernel} and (\ref{eq:Estimate-Coulomb-appen_4}) we find
\begin{align*}
\|\Lambda_i^{m}e^{t \Delta} \frac{\phi}{| \v{x}_i - \v{x}_j |}  \|_{2} & \leq \| \Lambda_i  e^{t \Delta_{\v{x}_i}}\Lambda_i^{m-1} | \v{x}_i - \v{x}_j |^{-1} \phi \|_{2} \\
& \lesssim t^{-\frac{9}{20}} [ 1 + t^{-\frac{1}{2}} ] \left(  \int_{\R^{3(N-1)}} \left( \int_{\R^3} \left|\Lambda_i^{m-1}\frac{\phi ( \ul{\v{x}} ) }{ | \v{x}_i - \v{x}_j | } \right|^{\frac{5}{4}} \dd \v{x}_j \right)^{\frac{8}{5}} \dd \ul{\v{x}}_i' \right)^{\frac{1}{2}} \\
& \lesssim t^{- \frac{9}{20}} [1 + t^{- \frac{1}{2} }] \|\Lambda_i^{m}\phi \|_2 \\
& \lesssim t^{- \frac{9}{20}} [1 + t^{- \frac{1}{2} }] \| \phi \|_{m,2}  . \numberthis \label{eq:Estimate-Coulomb-8}
\end{align*}
Combining estimates (\ref{eq:Estimate-Coulomb-7}) and (\ref{eq:Estimate-Coulomb-8}) we arrive at
\begin{align*}
\| e^{t \Delta} V_3 (\ul{\v{R}} , \Z) \phi \|_{m , 2} & \leq  \sum_{1 \leq i < j \leq N} \| e^{t \Delta} | \v{x}_i - \v{x}_j |^{-1} \phi \|_{m,2} \\
& \lesssim \left( 1 + t^{- \frac{1}{2} } \right) \left( t^{- \frac{9}{20}} + t^{-\frac{1}{4}} \right) \| \phi \|_{m,2} . \numberthis \label{eq:Estimate-Coulomb-9}
\end{align*}
Collecting estimates (\ref{eq:Estimate-Coulomb-1}), (\ref{eq:Estimate-Coulomb-5}), and (\ref{eq:Estimate-Coulomb-9}) we arrive at (\ref{eq:Estimate-Coulomb}). 

\end{proof}

\begin{lem}[Estimates for the Energies] \label{lem:Estimates-Energy}
Fix $\varepsilon > 0$, $N , K \geq 1$, and let and $\Z \in [0 , \infty )^K$, $\ul{\v{R}} = (\v{R}_1 , \cdots , \v{R}_K) \in \R^{3K}$, with $\v{R}_i \neq \v{R}_j$ for all $i \neq j$. For all $(\phi , \v{A} ) \in [H^1 (\R^{3N}) ]^{2^N} \times \dot{H}^1 (\R^3 ; \R^3)$, with $\diver{\v{A}} = 0$, the kinetic energy $T_{\mr{P}} = T_{\mr{P}} [\phi , \v{A}]$, as defined in (\ref{def:total_kinetic}), and the potential energy $V = V[\phi]$, as defined in (\ref{def:tot_potential_energy}), satisfy the estimates
\begin{align}\label{eq:Estimate-Kinetic-Coulomb}
T_{\mr{P}} \lesssim ( 1 + \| \nabla \v{A} \|_2 )^2 \| \phi \|_{1,2}^2  \hspace{0.5cm} \mr{and} \hspace{0.5cm} V \lesssim \| \phi \|_{1,2}^2 ,
\end{align}
respectively. Moreover, for all $(\phi , \v{A} ) , (\phi' , \v{A}'  ) \in [H^1 (\R^{3N}) ]^{2^N} \times \dot{H}^1 (\R^3 ; \R^3)$, the difference of the total kinetic energies and potential energies $T_{\mr{P}} - T'_{\mr{P}} + V - V' \equiv T_{\mr{P}} [\phi , \v{A} ] - T_{\mr{P}} [\phi' , \v{A}' ] + V[\phi] - V[\phi']$ satisfies the estimate 
\begin{align*}
& | T_{\mr{P}} - T'_{\mr{P}} + V - V' | \\
& \hspace{2mm} \lesssim \omega ( \| \phi \|_{1,2} , \| \phi' \|_{1,2} , \| \nabla \v{A} \|_{2} , \| \nabla \v{A}' \|_{2} ) \max{ \{ \| \phi - \phi' \|_{1,2} , \|\nabla ( \v{A} - \v{A}' ) \|_{2}  \} } , \numberthis \label{eq:Estimate-Energy-diff} 
\end{align*}
where 
\begin{align}\label{eq:def_omega}
\omega (x_1 , x_2 , x_3 , x_4 ) & =  \left( 1 + x_2  + x_3 \right) \left[ ( 1 + x_3 ) x_1 + ( 1 + x_4 ) x_2 \right] + ( x_1 + x_2 ) .
\end{align}
\end{lem}
\begin{proof}
To show the first estimate in (\ref{eq:Estimate-Kinetic-Coulomb}) it suffices to prove the $N = 1$ case, as for general $N \geq 1$ the estimating goes in a similar fashion. Using H\"{o}lder's inequality and Sobolev's inequality $H^1 (\R^3) \subset L^r (\R^3)$, $1 \leq r \leq 6$, we find
\begin{align*}
\| \gvsig \cdot (\v{p} + \v{A}) \phi \|_2 & \leq \| \v{p} \phi \|_2 + \| \v{A} \phi \|_2 \\
& \lesssim \| \phi \|_{1,2} + \| \v{A} \|_6 \| \phi \|_3 \\
& \lesssim ( 1 + \| \nabla \v{A} \|_2 ) \| \phi \|_{1,2} .
\end{align*}

To show the second estimate in (\ref{eq:Estimate-Kinetic-Coulomb}), first note that
\begin{align}\label{eq:Estimate-Energy-4} 
V[\phi] \leq \sum_{1 \leq i < j \leq N} \langle \phi , |\v{x}_i - \v{x}_j|^{-1} \phi \rangle_{L^2} + \left( \sum_{1 \leq i < j \leq K} \frac{Z_i Z_j}{|\v{R}_i - \v{R}_j|} \right) \| \phi \|_{2}^2 .
\end{align}
Considering (\ref{eq:Estimate-Energy-4}) we focus on controlling the electron-electron repulsion energy since the nuclei-nuclei repulsion energy is trivially bounded by $\| \phi \|_{1,2}$. The desired estimate on the electron-electron repulsion energy follows from the uncertainty principle for Hydrogen, namely $\langle \psi , |\v{x}|^{-1} \psi \rangle \leq \| \psi \|_2 \| \nabla \psi \|_2$. It suffices to consider the case $N = 2$. Using H\"{o}lder's inequality and Sobolev's inequality we find
\begin{align*}
& \langle \phi , |\v{x}_1 - \v{x}_2|^{-1} \phi \rangle_{L^2} \\
& \hspace{0.1cm} = \int_{\R^3} \int_{\R^3} \frac{|\phi (\v{y} + \v{x}_2 , \v{x}_2)|^2}{|\v{y}|} \dd \v{y} \dd \v{x}_2 \\
& \hspace{0.1cm} \leq \int_{\R^3} \left( \int_{\R^3} | \phi (\v{y} + \v{x}_2 , \v{x}_2) |^2 \dd \v{y} \right)^{\frac{1}{2}} \left( \int_{\R^3} | \v{p}_1 \phi (\v{y} + \v{x}_2 , \v{x}_2) |^2 \dd \v{y} \right)^{\frac{1}{2}} \dd \v{x}_2  \\
& \hspace{0.1cm} \leq \frac{1}{2} \left( \| \phi \|_2^2 + \| \v{p}_1 \phi \|_2^2 \right) . \numberthis \label{eq:Estimate-Energy-4.1}
\end{align*}
Estimates (\ref{eq:Estimate-Energy-4}) and (\ref{eq:Estimate-Energy-4.1}) imply the second estimate in (\ref{eq:Estimate-Kinetic-Coulomb}).

To estimate $T_{\mr{P}} - T'_{\mr{P}}$ it suffices to consider the $N = 1$ case. Write $T_{\mr{P}} - T'_{\mr{P}} = \sum_{k = 1}^6 T_k$ where
\begin{align*}
& T_1 [ \phi , \phi' , \v{A} , \v{A}' ] = \langle \gvsig \cdot \v{p} ( \phi - \phi' ) ,  \gvsig \cdot (\v{p} + \v{A}) \phi \rangle , \\
& T_2 [ \phi , \phi' , \v{A} , \v{A}' ] = \langle \gvsig \cdot (\v{A} - \v{A}') \phi' ,  \gvsig \cdot (\v{p} + \v{A}) \phi \rangle , \\
& T_3 [ \phi , \phi' , \v{A} , \v{A}' ] =  \langle \gvsig \cdot \v{A} (\phi -  \phi')  ,  \gvsig \cdot (\v{p} + \v{A}) \phi \rangle , \\
& T_4 [ \phi , \phi' , \v{A} , \v{A}' ] = \langle \gvsig \cdot (\v{p} + \v{A}') \phi' , \gvsig \cdot \v{p} ( \phi - \phi' )  \rangle  ,\\
& T_5 [ \phi , \phi' , \v{A} , \v{A}' ] = \langle  \gvsig \cdot (\v{p} + \v{A}') \phi' ,  \gvsig \cdot (\v{A} - \v{A}') \phi'  \rangle ,  \\
& T_6 [ \phi , \phi' ,\v{A} , \v{A}' ] =  \langle \gvsig \cdot (\v{p} + \v{A}') \phi'  ,  \gvsig \cdot \v{A} (\phi -  \phi') \rangle .
\end{align*}
Using Cauchy-Schwartz together with first estimate in (\ref{eq:Estimate-Kinetic-Coulomb}) we find
\begin{align}
& T_1 [ \phi , \phi' , \v{A} , \v{A}' ]  \lesssim ( 1 + \| \nabla \v{A} \|_2 ) \| \phi \|_{1,2} \| \phi - \phi' \|_{1,2} , \label{eq:Estimate-Energy-5}  \\
& T_2 [ \phi , \phi' , \v{A} , \v{A}' ]  \lesssim  ( 1 + \| \nabla \v{A} \|_2 ) \| \phi \|_{1,2} \| \phi' \|_{1,2} \| \nabla ( \v{A} - \v{A}' ) \|_2 , \label{eq:Estimate-Energy-6}  \\
&  T_3 [ \phi , \phi' , \v{A} , \v{A}' ]  \lesssim ( 1 + \| \nabla \v{A} \|_2 ) \| \phi \|_{1,2} \| \nabla \v{A} \|_2 \| \phi - \phi' \|_{1,2} , \label{eq:Estimate-Energy-7}  \\
& T_4 [ \phi , \phi' , \v{A} , \v{A}' ]   \lesssim ( 1 + \| \nabla \v{A}' \|_2 ) \| \phi' \|_{1,2} \| \phi - \phi' \|_{1,2} , \label{eq:Estimate-Energy-8}  \\
& T_5 [ \phi , \phi' , \v{A} , \v{A}' ]  \lesssim ( 1 + \| \nabla \v{A}' \|_2 ) \| \phi' \|_{1,2}^2 \| \nabla ( \v{A} - \v{A}' ) \|_2 , \label{eq:Estimate-Energy-9}  \\
& T_6 [ \phi , \phi' , \v{A} , \v{A}' ]  \lesssim ( 1 + \| \nabla \v{A}' \|_2 ) \| \phi' \|_{1,2} \| \nabla \v{A} \|_2 \| \phi - \phi' \|_{1,2} . \label{eq:Estimate-tot-Energy0}  
\end{align}
Collecting estimates (\ref{eq:Estimate-Energy-5}) through (\ref{eq:Estimate-tot-Energy0}) we conclude
\begin{align*}
| T_{\mr{P}} - T'_{\mr{P}} |  \lesssim \omega_1 ( \|\phi\|_{1,2} , \| \phi' \|_{1,2} , \| \nabla \v{A} \|_2 , \| \nabla \v{A}' \|_2 ) \max{\{ \| \phi - \phi' \|_{1,2} , \| \nabla ( \v{A} - \v{A}' ) \|_2 \}} \numberthis \label{eq:Estimate-tot-Energy1} 
\end{align*}
where $\omega_1$ function
\begin{align*}
\omega_1 ( x , y , z , w ) = \left( 1 + y  + z \right) \left[ ( 1 + z ) x + ( 1 + w ) y \right] .
\end{align*}

To estimate $V - V'$, write $V - V' = V_1 + V_2$ where
\begin{align*}
V_1 [ \phi , \phi' ] = \langle \phi - \phi' , V (\ul{\v{R}} , \Z) \phi \rangle_{L^2} , \hspace{1cm} V_2 [\phi , \phi'] = \langle \phi' , V (\ul{\v{R}} , \Z) (\phi - \phi') \rangle_{L^2} .
\end{align*}
We want to control $\max{ \{ V_1 , V_2 \} }$ by $\| \phi \|_{1,2}$, $\| \phi' \|_{1,2}$, and $\| \phi - \phi' \|_{1,2}$. Therefore, we show the inequality
\begin{align}\label{eq:Estimate-tot-Energy1.5}
| \langle h , V(\ul{\v{R}} , \Z) g \rangle | \lesssim \| h \|_{1,2} \| g \|_{1,2} , \hspace{1cm} \forall h,g \in H^1 (\R^{3N} , \C) .
\end{align}
Note that 
\begin{align*}
\langle h , V( \ul{\v{R}} , \Z ) g \rangle_{L^2} & = \sum_{i < j}^N \langle h , | \v{x}_i - \v{x}_j |^{-1} g \rangle_{L^2} - \sum_{i = 1}^N \sum_{j = 1}^K Z_j \langle h , | \v{x}_i - \v{R}_j |^{-1} g \rangle_{L^2} \\
& \hspace{0.25cm} + \sum_{i < j}^K \frac{Z_i Z_j}{|\v{R}_i - \v{R}_j|} \langle h , g \rangle_{L^2} . \numberthis  \label{eq:Estimate-tot-Energy2}
\end{align*}
The third term on the right hand side of (\ref{eq:Estimate-tot-Energy2}) is bounded by $\| g \|_2 \| h \|_2$ via Cauchy-Schwartz. To estimate the second term on the right hand side of (\ref{eq:Estimate-tot-Energy2}) it suffices to consider the case $N, K = 1$ and $\v{R}_1 = 0$. Indeed, in this situation $\langle h , |\v{x}|^{-1} g \rangle \lesssim \sqrt{ \| h \|_6 \| g \|_6 \| h \|_2 \| g \|_2 }$. This follows by writing $\langle h , |\v{x}|^{-1} g \rangle$ as the sum of an integral over the ball of radius $R$ and its complement, using H\"{o}lder's inequality, and then optimizing over $R$. The desired estimate (\ref{eq:Estimate-tot-Energy1.5}) then follows from the Sobolev inequality. Estimating the first term on the right hand side of (\ref{eq:Estimate-tot-Energy2}) by $\| h \|_{1,2} \| g \|_{1,2}$ follows the same proof as that of (\ref{eq:Estimate-Energy-4.1}). Hence (\ref{eq:Estimate-tot-Energy1.5}) holds, and therefore
\begin{align}\label{eq:Estimate-tot-Energy3}
|V - V'| \lesssim |V_1| + |V_2| \lesssim ( \| \phi \|_{1,2} + \| \phi' \|_{1,2} ) \| \phi - \phi' \|_{1,2} .
\end{align}
Collecting estimates (\ref{eq:Estimate-tot-Energy1}) and (\ref{eq:Estimate-tot-Energy3}), we arrive at (\ref{eq:Estimate-Energy-diff}).
\end{proof}

\begin{lem}[Estimates for the Probability Current Density] \label{lem:Estimates_Prob_current}
Fix $m \in [1 , \infty)$ and $N \geq 1$. For all $(\phi , \v{A}) \in [H^{m} (\R^{3N})]^{2^N} \times H^{m} (\R^3 ; \R^3)$, with $\diver{\v{A}} = 0$, and each $j \in \{1 , \cdots , N\}$, the probability current density $\J_{\!\mr{P}} [\phi , \v{A}]$ as given by (\ref{def:prob_current_compact}) is in the Sobolev space $H^{m - 2} (\R^3 ; \R^3)$ and satisfies the estimate
\begin{align} \label{eq:Estimate-KG-1}
\| \J_{\! \mr{P}} [\phi , \v{A} ] \|_{m-2,2} \lesssim (1 + \| \v{A} \|_{m,2} ) \| \phi \|_{m,2}^2 .
\end{align}
Moreover, for $(\phi , \v{A}) , (\phi' , \v{A}') \in [H^1 (\R^{3N})]^{2^N} \times H^1 (\R^3 ; \R^3)$, with $\diver{\v{A}} = \diver{\v{A}'} = 0$, and each $j \in \{1 , \cdots , N\}$, we have
\begin{align*} 
\| \J_{\!\mr{P}} [\phi , \v{A}] - \J_{\!\mr{P}} [\phi' , \v{A}'] \|_{-1 , 2} & \lesssim \left\lbrace \left[ (1 + \| \v{A} \|_{1,2}) \| \phi \|_{1,2} + (1 + \| \v{A}' \|_{1,2}) \|\phi'\|_{1,2} \right] \right. \\
& \hspace{1cm} \left. + \| \phi \|_{1,2} \| \phi' \|_{1,2} \right\rbrace \max{\{ \| \phi - \phi' \|_{1,2} \| \v{A} - \v{A}' \|_{1,2} \}} . \numberthis \label{eq:Estimate-KG-2}
\end{align*}
\end{lem}
\begin{proof}
To prove (\ref{eq:Estimate-KG-1}) we split into two cases: (a) $1 \leq m \leq 2$ and (b) $m > 2$. For (a), we specialize to $m = 1$ and note that the general case $1 \leq m \leq 2$ follows in a similar fashion. Since
\begin{align*}
\| \J_{\!\mr{P}} [\phi , \v{A}] \|_{-1 , 2} \lesssim \| \J_{\!\mr{P}} [\phi , \v{A}] \|_{\frac{6}{5}} 
\end{align*}
we need to estimate $\| \J_{\!\mr{P}} [\phi , \v{A}] \|_{\frac{6}{5}}$ by $(1 + \| \v{A} \|_{1,2}) \| \phi \|_{1,2}^2$. Write $\J_{\!\mr{P}} = \sum_{j = 1}^N \v{J}_j$. Using Minkowski's integral inequality, H\"{o}lder's inequality, and the Sobolev inequality $H^1 (\R^3) \subset L^r (\R^3)$, $2 \leq r \leq 6$, we have
\begin{align*}
\| \v{J}_j [\phi , \v{A}] \|_{\frac{6}{5}} & = \alpha \left( \int_{\R^3} \left| \int \langle \gvsig \phi_{\ul{\v{z}}_j'} , \gvsig \cdot (\v{p} + \v{A}_j) \phi_{\ul{\v{z}}_j'} \rangle_{\C^2} (\v{x}_j) \dd \ul{\v{z}}_j' \right|^{\frac{6}{5}} \dd \v{x}_j \right)^{\frac{5}{6}} \\
& \leq \alpha \int \left( \int_{\R^{3}} \left| \langle \gvsig \phi_{\ul{\v{z}}_j'} , \gvsig \cdot (\v{p} + \v{A}_j) \phi_{\ul{\v{z}}_j'} \rangle_{\C^2} (\v{x}_j) \right|^{\frac{6}{5}}  \dd \v{x}_j \right)^{\frac{5}{6}} \dd \ul{\v{z}}_j' \\
&  \lesssim \int \left[ \| \phi_{\ul{\v{z}}_j'}  \|_{3} \| (\v{p} + \v{A}) \phi_{\ul{\v{z}}_j'}  \|_{2} \right] \dd \ul{\v{z}}_j' \\
& \lesssim ( 1 + \| \v{A} \|_{1,2} ) \| \phi \|_{1,2}^2 . \numberthis \label{eq:Estimate-KG-3}
\end{align*}
The estimate (\ref{eq:Estimate-KG-3}) thus yields $\| \v{J}_j [\phi , \v{A}] \|_{-1 , 2} \lesssim (1 + \| \v{A} \|_{1,2}) \| \phi \|_{1,2}^2$ For case (b), we use Minkowski's integral inequality, Lemma \ref{lem:Kato-Ponce}, and the Sobolev inequality to find 
\begin{align*}
 \| \v{J}_j [\phi , \v{A}] \|_{m-2,2} & = \alpha \left( \int_{\R^3} \left| \int \Lambda_j^{m-2} \langle \gvsig \phi_{\ul{\v{z}}_j'} , \gvsig \cdot (\v{p} + \v{A}_j) \phi_{\ul{\v{z}}_j'} \rangle_{\C^2} (\v{x}_j) \dd \ul{\v{z}}_j' \right|^2 \dd \v{x}_j \right)^{\frac{1}{2}} \\
& \leq \alpha \int \left( \int_{\R^{3}} \left| \Lambda_j^{m-2} \langle \gvsig \phi_{\ul{\v{z}}_j'} , \gvsig \cdot (\v{p} + \v{A}_j) \phi_{\ul{\v{z}}_j'} \rangle_{\C^2} (\v{x}_j) \right|^2  \dd \v{x}_j \right)^{1/2} \dd \ul{\v{z}}_j' \\
& \lesssim \int \left[ \| \phi_{\ul{\v{z}}_j'} \|_{m-2,6} \| \phi_{\ul{\v{z}}_j'} \|_{1 , 3} + \| \phi_{\ul{\v{z}}_j'} \|_{3} \| \phi_{\ul{\v{z}}_j'} \|_{m-1,6} \right. \\
& \hspace{3cm} \left. + \| \v{A} \|_{m-2 , 6} \| \phi_{\ul{\v{z}}_j'} \|_{6}^2  + \| \v{A} \|_{6} \| \phi_{\ul{\v{z}}_j'} \|_{m-2 , 6} \| \phi_{\ul{\v{z}}_j'} \|_{3} \right] \dd \ul{\v{z}}_j' \\
& \lesssim (1 + \| \v{A} \|_{m,2} ) \int \| \phi_{\ul{\v{z}}_j'} \|_{m,2}^2 \dd \ul{\v{z}}_j' \lesssim (1 + \| \v{A} \|_{m,2}) \| \phi \|_{m,2}^2 . \numberthis \label{eq:Estimate-KG-4}
\end{align*}
Combining (\ref{eq:Estimate-KG-3}) and (\ref{eq:Estimate-KG-4}) we arrive at (\ref{eq:Estimate-KG-1}).

Arguing (\ref{eq:Estimate-KG-2}) in similar to the case $m = 1$ in proving (\ref{eq:Estimate-KG-1}). Specifically, we need to estimate $ \v{J}_j [\phi , \v{A}] - \v{J}_j [\phi' , \v{A}']$ in $L^{\frac{6}{5}}$-norm. We write
\begin{align}\label{eq:Estimate-KG-9}
\v{J}_j [\phi , \v{A}] - \v{J}_j [\phi' , \v{A}'] = - \alpha \re{ \sum_{\alpha = 1}^4 \v{F}_j^{\alpha} [\phi , \phi' , \v{A} , \v{A}']  } 
\end{align}
where
\begin{align*}
\v{F}_j^1 [\phi , \phi' , \v{A} , \v{A}'] ( \v{x}_j ) & = \int \langle \gvsig \left( \phi_{\ul{\v{z}}_j'} - \phi'_{\ul{\v{z}}_j'} \right) , \gvsig \cdot (\v{p} + \v{A}_j) \phi_{\ul{\v{z}}_j'} \rangle_{\C^2} (\v{x}_j) \dd \ul{\v{z}}_j' , \\
\v{F}^2_j [\phi , \phi' , \v{A} , \v{A}'] ( \v{x}_j ) & = \int \langle \gvsig \phi_{\ul{\v{z}}_j'}' , \gvsig \cdot \v{p} \left( \phi_{\ul{\v{z}}_j'} - \phi'_{\ul{\v{z}}_j'} \right) \rangle_{\C^2} (\v{x}_j) \dd \ul{\v{z}}_j' , \\
\v{F}^3_j [\phi , \phi' , \v{A} , \v{A}'] ( \v{x}_j ) & = \int \langle \gvsig \phi'_{\ul{\v{z}}_j'} , \gvsig \cdot \left( \v{A}_j - \v{A}'_j \right) \phi_{\ul{\v{z}}_j'} \rangle_{\C^2} (\v{x}_j) \dd \ul{\v{z}}_j' , \\
\v{F}^4_j [\phi , \phi' , \v{A} , \v{A}'] ( \v{x}_j ) & = \int \langle \gvsig \phi'_{\ul{\v{z}}_j'} , \gvsig \cdot \v{A}'_j \left( \phi_{\ul{\v{z}}_j'} - \phi'_{\ul{\v{z}}_j'} \right) \rangle_{\C^2} (\v{x}_j) \dd \ul{\v{z}}_j' .
\end{align*}
Estimating $\v{F}^{\alpha}_j$, for $\alpha = 1, \cdots ,4$, in $L^{\frac{6}{5}}$-norm is straightforward and involves the same strategy used to show (\ref{eq:Estimate-KG-3}). We find
\begin{align}
& \| \v{F}^1_j [\phi , \phi' , \v{A} , \v{A}'] \|_{\frac{6}{5}} \lesssim ( 1 + \| \v{A} \|_{1,2} ) \| \phi \|_{1,2} \| \phi - \phi' \|_{1,2} . \numberthis \label{eq:Estimate-KG-5} \\
& \| \v{F}_j^2 [\phi , \phi' , \v{A} , \v{A}'] \|_{\frac{6}{5}} \lesssim \| \phi' \|_{1,2} \| \phi - \phi' \|_{1,2} , \label{eq:Estimate-KG-6} \\
& \| \v{F}_j^3 [\phi , \phi' , \v{A} , \v{A}'] \|_{\frac{6}{5}} \lesssim \| \phi \|_{1,2} \|\phi' \|_{1,2} \| \v{A} - \v{A}' \|_{1,2} , \label{eq:Estimate-KG-7} \\
& \| \v{F}_j^4 [\phi , \phi' , \v{A} , \v{A}'] \|_{\frac{6}{5}} \lesssim \| \v{A}' \|_{1,2} \| \phi' \|_{1,2} \| \phi - \phi' \|_{1,2} . \label{eq:Estimate-KG-8} 
\end{align}
Estimates (\ref{eq:Estimate-KG-5}) through (\ref{eq:Estimate-KG-8}) imply (\ref{eq:Estimate-KG-2}).
\end{proof}

\section{The Contraction Mapping Argument}\label{sec:metric_space_linearization}

This section is devoted to setting the stage for a contraction mapping argument that will be used to prove the local well-posedness of the $\varepsilon$-modified MBMP equations (\ref{eq:MBMP_epsilon}) as stated in Theorem \ref{thm:local_exist_MBMP_epsilon}. As discussed in detail in \S\ref{sec:proof_strat}, we will employ a proof strategy inspired by the one used to study the MS equations (\ref{eq:MS}) in \cite{NW05}. This proof strategy is based on the well-known Banach fixed point theorem, namely, every contraction mapping on a complete metric space has a fixed point. Therefore, we need to carefully choose an appropriate solution map and a metric space on which this solution map is a contraction. Throughout this section we fix $N , K \geq 1$, $m \in [1 , \infty)$, $\varepsilon > 0$, $\Z \in [0 , \infty )^K$, and $\ul{\v{R}} = (\v{R}_1 , \cdots , \v{R}_K) \in \R^{3K}$, with $\v{R}_i \neq \v{R}_j$ for all $i \neq j$. 

We begin by specifying the metric space on which an appropriate solution map will be defined. Given $T , R \in (0, \infty)$, consider the $(T, R)$-dependent space
\begin{align*}
\X^m_{T} (R) = & \{ (\phi , \v{A} ) \in L_T^{\infty} [H^m (\R^{3N})]^{2^N} \times [L_T^{\infty} H^m (\R^3 ; \R^3) \cap W_T^{1 , \infty} H^{m-1} (\R^3 ; \R^3)]  \\
& ~~ \text{s.t.} ~ \max{ \{ \| \phi \|_{\infty ; m, 2 } , \| \v{A} \|_{\infty ; m, 2 } , \| \partial_t \v{A} \|_{\infty ; m-1 , 2 } \} } \leq R , ~ \diver{\v{A}} = 0 \} .
\end{align*}
We equip $\X^m_{T} (R)$ with the metric
\begin{align}
d( (\phi , \v{A}) , (\phi', \v{A}') ) = \max{ \{ \| \phi - \phi' \|_{\infty ; 1,2} , \| \v{A} - \v{A}' \|_{\infty ; 1,2} , \| \partial_t \v{A} - \partial_t \v{A}' \|_{\infty; 2} \} } .
\end{align}
We emphasize that the $L^{\infty}$-norm in the definition of $( \X^m_{T} (R) , d )$ is taken over the time interval $[0,T]$. When the radius $R > 0$ is understood we will simply write $\X^m_T$ for $\X_T^m (R)$. That $( \X^m_{T } , d )$ is complete is expressed by the following Lemma.
\begin{lem}\label{lem:completeness}
For all $T , R > 0$ and $m \in [ 1 , \infty )$, $(\X^m_{T}, d)$ is a complete metric space. 
\end{lem}
\begin{proof}
Fix $T , R > 0$. Let $(\phi^n , \v{A}^{\!n})_{n \geq 1}$ be a Cauchy sequence in $(\X^m_{T}, d)$. Then $(\phi^n)_{n \geq 1}$ is a Cauchy sequence in the Banach space $L^{\infty}_T H^1$. Therefore, there exists $\phi \in L_T^{\infty} H^1$ so that
\begin{align*}
\lim_{n \rightarrow \infty} \| \phi^n - \phi \|_{\infty ; 1 , 2} = 0 .
\end{align*}
Moreover, the sequence $(\phi^n)_{n \geq 1}$ is bounded by the constant $R > 0$ in $L_T^{\infty} H^m$, a space that can be considered the dual of the separable space $L_T^1 H^{-m}$ via the isometric anti-isomorphism
\begin{align*}
L^{\infty}_T H^m \ni F \longmapsto \int_0^T \langle F (s) , ( \cdot ) (s) \rangle_{H^m , H^{-m}} \dd s \in \left( L_T^1 H^{-m} \right)^{\ast} .
\end{align*}
By the Banach-Alaoglu Theorem, namely, the unit ball in $(L_T^1 H^{-m})^{\ast}$ is compact in the weak-$\ast$ topology, there exists $\psi \in L^{\infty}_T H^m$ and a subsequence $(\phi^{n_k})_{k \geq 1} \subset (\phi^n)_{n \geq 1}$ so that 
\begin{align*}
\phi^{n_k} \xrightarrow{~~ \mr{w}^{\ast} ~~} \psi \hspace{5mm} \mr{in} ~~ L_T^{\infty} H^m 
\end{align*}
as $k \rightarrow \infty$. However, for any $\eta \in L^2_T L^2$, the sequence $( \langle \eta , \phi^{n_k} \rangle_{L_T^2 L^2}  )_{k \geq 1}$ converges to both $\langle \eta , \phi \rangle_{L_T^2 L^2}$ and $\langle \eta , \psi \rangle_{L_T^2 L^2}$. Since this is true for all $\eta \in L_T^2 L^2$, we conclude $\phi = \psi$. Note that weak-$\ast$ convergence preserves the upper bound $\| \phi \|_{\infty ; m , 2} \leq R$. So, in total, $\phi \in L_T^{\infty} H^m$ and $\| \phi \|_{\infty ; m , 2} \leq R$.

Furthermore, the sequence $(\v{A}^{\!n} , \partial_t \v{A}^{\!n})_{n \geq 1}$ is a Cauchy sequence in the Banach space $L_T^{\infty} H^1 \oplus L_T^{\infty} L^2$, and is bounded by the constant $R > 0$ in the space $L_T^{\infty} H^m \oplus L_T^{\infty} H^{m-1}$. By a similar argument used above for the sequence $(\phi^n)_{n \geq 1}$, we conclude existence of a $( \v{A} , \dot{\v{A}} ) \in L_T^{\infty} H^m \times L_T^{\infty} H^{m-1}$ so that
\begin{align*}
\lim_{n \rightarrow \infty} \| (\v{A}^{\!n} , \partial_t \v{A}^{\!n}) - (\v{A} , \dot{\v{A}}) \|_{\infty ; 1 , 2 \oplus \infty ; 2} = 0
\end{align*}
and a subsequence $(\v{A}^{\!n_k} , \partial_t \v{A}^{\!n_k})_{k \geq 1}$ so that
\begin{align*}
(\v{A}^{\!n_k} , \partial_t \v{A}^{\!n_k})_{k \geq 1} \xrightarrow{~~ \mr{w}^{\ast} ~~} (\v{A} , \dot{\v{A}}) \hspace{5mm} \mr{in} ~~ L_T^{\infty} H^m \oplus L_T^{\infty} H^{m-1} . 
\end{align*}
Let $\vartheta \in C_c^{\infty} ( (0,T) ; \R )$ and $f \in H^{1-m} (\R^3)$, and note that
\begin{align*}
\int_0^T \int_{\R^3} \v{A}^{\!n} (t, \v{x}) f (\v{x}) \vartheta' (t) \dd \v{x} \dd t = - \int_0^T \int_{\R^3} \partial_t \v{A}^{\!n} (t, \v{x}) f (\v{x}) \vartheta (t) \dd \v{x} \dd t  .
\end{align*}
Taking the $n \rightarrow \infty$ limit of the previous identity we get that $\dot{\v{A}}$ is the distributional time derivative of $\v{A}$ in $\fD' ( [0,T] ; H^{1-m} (\R^3 ; \R^3) )$. Lastly, to see that $\diver{\v{A} (t)} = 0$ for a.e.~$t$ simply note
\begin{align*}
\| \diver{\v{A}} \|_{\infty ; 2} \leq \| \v{A} - \v{A}^{\!n} \|_{\infty ;  1 , 2} \xrightarrow{n \rightarrow \infty} 0 . 
\end{align*}
Putting everything together, $\exists (\phi , \v{A}) \in (\X_T^m , d)$ so that $d ( (\phi , \v{A}) , (\phi^n , \v{A}^{\!n}) ) \rightarrow 0$ as $n \rightarrow \infty$.
\end{proof}

Let $(\phi_0 , \v{a}_0 , \dot{\v{a}}_0) \in \X^m_0$, where $\X_0^m$ is defined by (\ref{def:initconds}). For a given pair $(\phi , \v{A}) \in \X^m_{T }$, consider the initial value problems for the unknown functions $(\xi , \v{K})$:
\begin{align}\label{eq:linearized_epsilon_system_1}
\left\lbrace \begin{array}{l}
\partial_t \xi - (i + \varepsilon) \sum_{j = 1}^N \Delta_{\v{x}_j} \xi = f[\phi , \tilde{\v{A}}]  \\
\xi (0) = \phi_0 ,
\end{array} \right. 
\end{align}
where
\begin{align}\label{eq:inhomo_linearized_epsilon_system_1}
 f[\phi , \tilde{\v{A}}]  = \left[ - (i + \varepsilon) \left( \fL ( \tilde{\v{A}} ) + V(\ul{\v{R}} , \Z) \right) + \varepsilon \left( T_{\mr{P}} [ \phi , \tilde{\v{A}} ] + V [\phi ] \right) \right] \phi , 
\end{align}
and 
\begin{align}\label{eq:linearized_epsilon_system_2}
\left\lbrace \begin{array}{l}
\square \v{K} = 4 \pi \alpha \Lambda_{\varepsilon}^{-1} \Hproj{ \J_{\!\mr{P}}[\phi , \tilde{\v{A}}] } \\
(\v{K} (0) , \partial_t \v{K} (0)) = (\v{a}_0 , \dot{\v{a}}_0) ,
\end{array} \right. 
\end{align}
In (\ref{eq:linearized_epsilon_system_1}), (\ref{eq:inhomo_linearized_epsilon_system_1}), and (\ref{eq:linearized_epsilon_system_2}), we recall that $\tilde{\v{A}} = \Lambda^{-1}_{\varepsilon} \v{A}$ is the regularized magnetic vector potential, $T_{\mr{P}} [\phi , \tilde{\v{A}}]$ is given by (\ref{def:total_kinetic}), $V[\phi]$ by (\ref{def:tot_potential_energy}), $\fL (\tilde{\v{A}}) = \sum_{j = 1}^N \cL_j (\tilde{\v{A}})$ where $\cL_j ( \tilde{\v{A}} )$ is given by (\ref{def:D_j}), and $\J_{\!\mr{P}}[\phi , \tilde{\v{A}} ]$ is given by (\ref{def:prob_current_compact}). Note that (\ref{eq:linearized_epsilon_system_1}) and (\ref{eq:linearized_epsilon_system_2}) are just inhomogeneous versions of the equations in (\ref{eq:MBMP_epsilon}). 

Equation (\ref{eq:linearized_epsilon_system_1}) is solved by the Schr\"{o}dinger-heat semigroup:
\begin{align}\label{eq:CA1.1}
\xi (t) & = e^{(i + \varepsilon) t \Delta} \phi_0 + \int_0^t e^{(i + \varepsilon) (t - \tau) \Delta} f[\phi (\tau) , \tilde{\v{A}} (\tau)] \dd \tau ,
\end{align}
where the action of $e^{(i + \varepsilon) t \Delta}$ is defined as the convolution with the kernel 
\begin{align*}
\frac{1}{(4 \pi (i + \epsilon) t)^{3N/2}} \exp{ \left\lbrace - \frac{| \ul{\v{x}} |^2}{4 (i + \varepsilon ) t} \right\rbrace } .
\end{align*}
Equation (\ref{eq:linearized_epsilon_system_2}) is solved by (\ref{eq:wave_soln}), namely,
\begin{align}\label{eq:CA1.2}
\v{K} (t) = \dot{\s}( t / \alpha ) \v{a}_0 + \alpha ~ \s(t/\alpha) \dot{\v{a}}_0 + 4 \pi \int_0^t \s ( (t-\tau) / \alpha ) \Lambda_{\varepsilon}^{-1} \Hproj{ \J_{\!\mr{P}} [ \phi (\tau) , \tilde{\v{A}} (\tau) ] } \dd \tau  ,
\end{align}
where $\s (t)$ and $\dot{\s} (t)$ are defined in (\ref{eq:wave_soln}). 

Consider the solution map
\begin{align}\label{eq:solution_map}
\Psi : \X^m_{T} \ni (\phi , \v{A}) \longmapsto (\xi , \v{K})
\end{align}
where $\xi (t)$ is given by (\ref{eq:CA1.1}) and $\v{K} (t)$ by (\ref{eq:CA1.2}).  At this point we observe that a unique fixed point of $\Psi$ would give us the existence of a unique local solution to (\ref{eq:MBMP_epsilon}) (the first part of Theorem \ref{thm:local_exist_MBMP_epsilon}). Therefore, the ultimate goal will be to choose $T , R > 0$ so that $\Psi$ is a contraction map on $( \X^m_{T } , d )$, and thereby prove that $\Psi$ has a unique fixed point via the Banach fixed point theorem. Before proving that $\Psi$ can be made into a contraction map on $( \X^m_{T } , d )$, we need to first argue that we can choose $T , R > 0$ so that $\Psi$ maps $\X^m (T,R)$ into itself. We conclude this section with a Lemma making this precise.

\begin{lem}\label{lem:maps_to_itself}
Let $(\phi_0 , \v{a}_0 , \dot{\v{a}}_0) \in \X^m_0$, where $\X_0^m$ is defined by (\ref{def:initconds}). There exists an $R > 0$ and $T_{\ast} = T_{\ast} (R) > 0$, both depending on $\varepsilon , N , K , \alpha , \Z, \ul{\v{R}}$, and $\| (\phi_0 , \v{a}_0 , \dot{\v{a}}_0 ) \|_{m,2 \oplus m,2 \oplus m-1 , 2}$, such that $\Psi : \X^m_T \rightarrow \X^m_T$ for all $T \in (0 , T_{\ast}]$. 
\end{lem}
\begin{proof}
Let $(\phi , \v{A}) \in \X^m_{T }$ and $\Psi (\phi , \v{A}) = (\xi , \v{K})$. Observe that $\v{K}$ is divergence-free using the formula (\ref{eq:CA1.2}). Fix $j \in \{1, \cdots , N\}$ and note that
\begin{align*}
\| \Lambda_{\varepsilon}^{-1} \Hproj{ \J_{\!\mr{P}} [\phi , \tilde{\v{A}}] } \|_{m-1 , 2} \leq \frac{1}{\varepsilon} \| \Hproj{ \J_{\!\mr{P}} [\phi , \tilde{\v{A}}] } \|_{\dot{H}^{m-2}} \lesssim \frac{1}{\varepsilon} \| \J_{\!\mr{P}} [\phi , \tilde{\v{A}}] \|_{m-2 , 2} ,
\end{align*}
where we've used the boundedness of $\Hproj : H^{m - 2} \rightarrow H^{m - 2}$. Therefore estimate (\ref{eq:Estimate-KG-1}) of Lemma \ref{lem:Estimates_Prob_current} (Estimates for the Probability Current Density) gives us
\begin{align}\label{eq:PhiXtoX0}
\| \Lambda_{\varepsilon}^{-1} \Hproj{ \J_{\!\mr{P}} [\phi , \tilde{\v{A}}] (t) } \|_{m-1 , 2} \lesssim (1 + R) R^2 , \hspace{1cm} \forall t \in [0,T],
\end{align}
and thus $\Lambda_{\varepsilon}^{-1} \Hproj{ \J_{\!\mr{P}} [\phi , \tilde{\v{A}}] } \in L_T^1 H^{m-1}$. With the previous conclusion we've satisfied the hypotheses in Lemma \ref{lem:Strichartz_wave_general} (Strichartz Estimates for the Wave Equation) and, as a consequence, we have $\v{K} \in C_T H^m \cap C_T^1 H^{m-1}$ and
\begin{align}\label{eq:PhiXtoX0.5}
\max_{k \in \{0,1\}} \| \partial_t^k \v{K} \|_{\infty ; m-k ,2} \lesssim \| (\v{a}_0 , \dot{\v{a}}_0) \|_{m,2 \oplus m-1 , 2} + \| \Lambda_{\varepsilon}^{-1} \Hproj{ \J_{\!\mr{P}}[\phi , \tilde{\v{A}} ] } \|_{1; m-1 , 2} 
\end{align}
Combining (\ref{eq:PhiXtoX0}) with (\ref{eq:PhiXtoX0.5}), we conclude the existence of a constant $C_1 > 0$ depending on $\varepsilon , m , N$, and $\alpha$ such that
\begin{align}\label{eq:PhiXtoX1}
\max_{k \in \{0,1\}} \| \partial_t^k \v{K} \|_{\infty ; m-k ,2} \leq C_1 \left[ \| (\v{a}_0 , \dot{\v{a}}_0) \|_{m , 2 \oplus m-1 , 2} + T (1 + R ) R^2 \right] .
\end{align}

We turn to estimating $\| \xi (t) \|_{m,2}$. To estimate $\| \xi (t) \|_{m,2}$, we take the $H^m$-norm of the defining formula (\ref{eq:CA1.1}) for $\xi (t)$ and apply (\ref{eq:Estimate-Kinetic-Coulomb}), (\ref{eq:Estimate-Pauli-3}), (\ref{eq:Estimate-Coulomb}) of Lemmas \ref{lem:Estimates-Energy}, \ref{lem:Estimates-Pauli}, and \ref{lem:Estimates-Coulomb}, respectively. This yields 
\begin{align*}
\| \xi (t) \|_{m, 2} & \lesssim \| \phi_0 \|_{m, 2} + \int_0^t \left( \left| T_{\mr{P}} [\phi , \tilde{\v{A}}] + V [\phi] \right|  \| \phi \|_{m,2}  + \| e^{(i + \varepsilon) (t - \tau) \Delta} [ \fL ( \tilde{\v{A}} ) \phi ] \|_{m,2}  \right. \\
& \hspace{7cm} \left.  + \| e^{(i + \varepsilon) (t - \tau) \Delta}  V (\ul{\v{R}} , \Z) \phi \|_{m ,2} \right) \dd \tau \\
& \lesssim \| \phi_0 \|_{m, 2} + \int_0^t \left[ ( 1 + \| \tilde{\v{A}}  \|_{1,2} )^2 + 1  \right] \| \phi \|_{1,2}^2  \| \phi \|_{m,2} \dd \tau \\
& \hspace{1cm} + \int_0^t (t-\tau)^{- \frac{1}{4}} \left[ 1 + (t-\tau)^{-\frac{1}{2}} \right] \left( 1 + \| \tilde{\v{A}}  \|_{m,2} \right) \| \tilde{\v{A}}  \|_{m,2} \| \phi \|_{m,2} \dd \tau \\
&  \hspace{1cm} + \int_0^t \left\lbrace 1 +  \left(1 + (t- \tau)^{- \frac{1}{2} }\right) \left((t-\tau)^{- \frac{9}{20}} + (t-\tau)^{-\frac{1}{4}} \right) \right\rbrace \| \phi \|_{m,2} \dd \tau  \numberthis \label{eq:PhiXtoX2} . 
\end{align*}
The last estimate (\ref{eq:PhiXtoX2}) allow us to conclude the existence of a constant $C_2 > 0$, depending on $\varepsilon$, $m$, $N$, $K$, $\alpha$, $\ul{\v{R}}$, and $\Z$, such that
\begin{align*}
\| \xi \|_{\infty ; m, 2} & \leq C_2 \left[ \| \phi_0 \|_{m,2} + T \left( 2 +  2 R + R^2 \right) R^3 + \left( T^{\frac{3}{4}} + T^{\frac{1}{4}} \right) (1 + R) R^2 \right. \\
& \left. \hspace{5cm} + \left( T + T^{\frac{3}{4}} + T^{\frac{11}{20}} + T^{\frac{1}{4}} + T^{\frac{1}{20}} \right) R \right] , \numberthis \label{eq:PhiXtoX3}
\end{align*}

Considering estimates (\ref{eq:PhiXtoX1}) and (\ref{eq:PhiXtoX3}) choose $R > 0$ such that
\begin{align}
& \| (\phi_0 , \v{a}_0 , \dot{\v{a}}_0) \|_{m,2 \oplus m,2 \oplus m-1 ,2}  \leq \frac{R}{2 \max{\{C_1 , C_2\}}} , \label{eq:PhiXtoX4} 
\end{align}
and choose $T_* > 0$ such that
\begin{align*} 
&  T_* (1 + 3 R + 2 R^2 + R^3) R + ( T_*^{\frac{3}{4}} + T_*^{\frac{1}{4}} ) (1+R)R   \\
& \hspace{3cm} + ( T_* + T_*^{\frac{3}{4}} + T_*^{\frac{11}{20}} + T_*^{\frac{1}{4}} + T_*^{\frac{1}{20}} ) ] \leq \frac{1}{2 \max{\{C_1 , C_2\}}} . \numberthis \label{eq:PhiXtoX5} 
\end{align*}
Equations (\ref{eq:PhiXtoX4}) and (\ref{eq:PhiXtoX5}) ensure that $\Psi$ maps $\X^m_{T }$ into itself for each $T \in (0 , T_*]$. 
\end{proof}

\section{Local Well-Posedness of the $\varepsilon$-Modified System.}\label{sec:local}

In this section is devoted to completing the proof of Theorem \ref{thm:local_exist_MBMP_epsilon}. We begin with a Lemma stating that the solution map $\Psi$ defined by (\ref{eq:solution_map}) on $(\X^m_T , d)$ can be made into a contraction map. 
\begin{lem}\label{lem:contraction}
Let $(\phi_0 , \v{a}_0 , \dot{\v{a}}_0) \in \X^m_0$, where $\X_0^m$ is defined by (\ref{def:initconds}). There exists an $R > 0$ and $T_{**} = T_{**} (R) > 0$, both depending on $\varepsilon , N , K , \alpha , \Z, \ul{\v{R}}$, and $\| (\phi_0 , \v{a}_0 , \dot{\v{a}}_0 ) \|_{m,2 \oplus m,2 \oplus m-1 , 2}$, such that, for all $T \in (0 , T_{**} ]$, $\Psi : (\X_T^m , d) \rightarrow (\X^m_T , d)$ is a contraction.
\end{lem}
\begin{proof}
Using Lemma \ref{lem:maps_to_itself} choose $R , T_* > 0$ so that $\Psi$ maps $\X^m_{T }$ to itself for $T \in (0, T_*]$. Fix $T \in (0, T_*]$, consider two pairs $(\phi , \v{A}) , (\phi' , \v{A}') \in \X^m_{T }$, and write $\Psi (\phi , \v{A}) = (\xi , \v{K})$ and $\Psi (\phi' , \v{A}' ) = (\xi' , \v{K}')$. Noting (\ref{eq:CA1.1}), (\ref{eq:CA1.2}), $( \xi (0) , \v{K} (0) , \partial_t \v{K} (0) ) = ( \phi_0 ,  \v{a}_0 , \dot{\v{a}}_0 )$, we observe that the difference $\xi - \xi'$ satisfies
\begin{align}\label{eq:Phicontract1}
(\xi - \xi') (t) = \int_0^t e^{(i + \varepsilon) (t-\tau) \Delta} \left( f[\phi (\tau) , \tilde{\v{A}} (\tau)] - f[\phi' (\tau) , \tilde{\v{A}}' (\tau)] \right) \dd \tau
\end{align}
and that the difference $\v{K} - \v{K}'$ satisfies
\begin{align}\label{eq:Phicontract1.25}
(\v{K} - \v{K}') (t) =  4 \pi \int_0^t \s ( (t - \tau) / \alpha ) \Lambda_{\varepsilon}^{-1} \Hproj{ \left( \J_{\!\mr{P}}[ \phi , \tilde{\v{A}} ] - \J_{\!\mr{P}}[\phi' , \tilde{\v{A}}'] \right) } (\tau) \dd \tau .
\end{align}

We need to control $d((\xi , \v{K}) , (\xi' , \v{K}'))$ by $d((\phi , \v{A}) , (\phi' , \v{A}'))$ to ultimately argue that $\Psi$ can be turned into a contraction. Estimating $\| \v{K} - \v{K}' \|_{\infty ; 1,2}$ and $\| \partial_t ( \v{K} - \v{K}') \|_{\infty;2}$ is a straightforward application of the energy estimate of Lemma \ref{lem:Strichartz_wave_general} (Strichartz Estimates for the Wave Equation) and estimate (\ref{eq:Estimate-KG-2}) of Lemma \ref{lem:Estimates_Prob_current} (Estimates for the Probability Current Density). We find
\begin{align*}
\max_{k = 0 ,1} \| \partial_t^k \left( \v{K} - \v{K}' \right) \|_{\infty ; 1 - k , 2} & \lesssim \| \J_{\!\mr{P}}[\phi , \tilde{\v{A}}] - \J_{\!\mr{P}}[\phi', \tilde{\v{A}}'] \|_{1 ; -1 , 2} \\
& \lesssim T R \left[ 2 + 3 R \right] d ( (\phi , \v{A}) , (\phi' , \v{A}') ) . \numberthis \label{eq:K-diff} 
\end{align*}

To estimate $\| \xi - \xi' \|_{\infty ; 1 , 2}$ we start with the formula (\ref{eq:Phicontract1}) for $\xi - \xi'$ and use the triangle inequality to find 
\begin{align*}
& \| (\xi - \xi') (t) \|_{1, 2} \\
& \lesssim \int_0^t \left( \left| T_{\mr{P}} [\phi , \tilde{\v{A}}]  -  T_{\mr{P}} [\phi' , \tilde{\v{A}}'] + V [\phi] - V [\phi'] \right| \| \phi' \|_{1,2} +  \left| T_{\mr{P}} [\phi , \tilde{\v{A}}] + V [\phi] \right| \| (\phi - \phi') \|_{1,2} \right. \\
& \left. + \|  e^{(i + \varepsilon) (t-\tau) \Delta} \left( [ \fL ( \tilde{\v{A}} ) \phi ] - [ \fL ( \tilde{\v{A}}' ) \phi' ] \right) \|_{1,2} + \| e^{(i + \varepsilon) (t-\tau) \Delta} V(\ul{\v{R}} , \Z) \left( \phi - \phi' \right) \|_{1,2} \right) \dd \tau . \numberthis \label{eq:Phicontract1.5}
\end{align*}
Using the same strategy that yielded (\ref{eq:PhiXtoX2}) and then (\ref{eq:PhiXtoX3}) of Lemma \ref{lem:maps_to_itself}, we apply (\ref{eq:Estimate-Energy-diff}), (\ref{eq:Estimate-Pauli-3}), (\ref{eq:Estimate-Coulomb}) of Lemmas \ref{lem:Estimates-Energy}, \ref{lem:Estimates-Pauli}, and \ref{lem:Estimates-Coulomb}, respectively, to find
\begin{align*}
\| (\xi - \xi') \|_{\infty ; 1,2} & \lesssim \{ T (4 + 8 R + 6 R^2 + R^3) R + ( T^{\frac{3}{4}} + T^{\frac{1}{4}}) ( 2 + 3 R )R  \\
& \hspace{1.5cm} + T + T^{\frac{3}{4}} + T^{\frac{11}{20}} + T^{\frac{1}{4}} + T^{\frac{1}{20}} \} d ( (\phi , \v{A}) , (\phi' , \v{A}') ) , \numberthis \label{eq:xi-diff} 
\end{align*}

Combining estimates (\ref{eq:K-diff}) through (\ref{eq:xi-diff}) we find
\begin{align}\label{eq:Phicontract2}
d( (\xi , \v{K}) , (\xi', \v{K}') ) \leq C g(T,R) d( (\psi , \v{A}) , (\psi', \v{A}') ) ,
\end{align}
where $C > 0$ is a constant depending on $\varepsilon$, $N$, $K$, $\alpha$, $\ul{\v{R}}$, and $\Z$, and
\begin{align*}
g(T , R) & = T (6 + 11 R + 6 R^2 + R^3) R + ( T^{\frac{3}{4}} + T^{\frac{1}{4}} ) (2 + 3 R)R \\
& \hspace{5cm} + T + T^{\frac{3}{4}} + T^{\frac{11}{20}} + T^{\frac{1}{4}} + T^{\frac{1}{20}} . \numberthis \label{eq:PhiContrac-def-g} 
\end{align*}
Choosing $0 < T_{**} < T_*$ so that $g(T_{**},R) = \frac{1}{2 C}$ ensures that $\Psi$, for example, satisfies
\begin{align*} 
d( \Psi (\psi , \v{A})  , \Psi (\psi' , \v{A}')  ) \leq \frac{1}{2} d( (\psi , \v{A}) , (\psi', \v{A}') ) .
\end{align*} 
Consequently, $\Psi$ is a contraction mapping on $(\X^m_{T } , d)$ for each $T \in (0, T_{**}]$.
\end{proof}

With Lemma \ref{lem:contraction} at our disposal, we are in a position to prove Theorem \ref{thm:local_exist_MBMP_epsilon}. 
\begin{proof}[Proof of Theorem \ref{thm:local_exist_MBMP_epsilon}]
Fix $(\phi_0 , \v{a}_0 , \dot{\v{a}}_0) \in \X^m_0$, where $\X_0^m$ is the space of initial conditions defined by (\ref{def:initconds}). Using Lemma \ref{lem:contraction} choose $R , T_{**} > 0$ so that, for each $T \in (0, T_{**}]$, $\Psi : (\X^m_{T} , d) \rightarrow (\X^m_{T} , d)$ is a contraction. Then, for each $T \in (0 , T_{**}]$, the Banach fixed point theorem allows us to conclude the existence a unique $(\phi , \v{A}) \in \X^m_{T} (R)$ that satisfies $\Psi (\phi , \v{A}) = (\phi , \v{A})$. Using the same estimates at produced the estimate (\ref{eq:xi-diff}) in the proof of Lemma \ref{lem:contraction}, we can show $\phi \in C_{T} H^m$. Moreover, $\v{A} \in C_{T} H^m \cap C_{T}^1 H^{m-1}$ by Lemma \ref{lem:Strichartz_wave_general}. In other words, the pair $(\phi , \v{A}) \in C_{T} H^m \times [C_{T} H^m \cap C_{T}^1 H^{m-1}]$ satisfies the equations
\begin{align}\label{eq:proof_thm6_1}
\left\lbrace \begin{array}{l}
\partial_t \phi =  - (i + \varepsilon) \Ham (\v{A}) \phi + \varepsilon \phi \left( T_{\mr{P}} [\phi , \tilde{\v{A}} ] + V [\phi ] \right) , \\
\square \v{A} = 4 \pi \alpha \Lambda^{-1}_{\varepsilon} \Hproj{ \J_{\!\mr{P}}[ \phi , \tilde{\v{A}} ] } , \\
\diver{\v{A}} = 0 , \\
(\phi , \v{A} , \partial_t \v{A}) |_{t = 0} = (\phi_0 , \v{a}_0 , \dot{\v{a}}_0) . 
\end{array} \right. 
\end{align}
 
Suppose there exists another pair $(u , \v{B}) \in C_{T_{**}} H^m \times [ C_{T_{**}} H^m \cap C_{T_{**}}^1 H^{m-1} ]$ solving the initial value problem (\ref{eq:proof_thm6_1}). Choose $R' > 0$ so that
\begin{align*}
\max{ \left\lbrace \| u \|_{\infty ; m , 2} , \| \v{B} \|_{\infty ; m , 2} , \| \partial_t \v{B} \|_{\infty ; m-1 , 2} \right\rbrace } \leq R'  . 
\end{align*}
From (\ref{eq:Phicontract2}) in the proof of Lemma \ref{lem:contraction} we see that the smaller we choose $T \in (0 , T_{**}]$, the larger $R > 0$ can be choosen to get a unique fixed point $(\phi , \v{A}) \in \X^m_T (R)$ of $\Psi$. Choose $T' \in (0, T_{**}]$ small enough so that we obtained a unique fixed point $(\phi , \v{A}) \in \X^m_{T'} (R')$. By uniqueness, $(\phi , \v{A}) = (u , \v{B})$ on $[0,T']$. Define
\begin{align*}
T_0 = \sup{ \left\lbrace t \in [ 0 , T_{**} ] ~ : ~ (\psi , \v{A}) = (u,\v{B}) ~ \text{on} ~ [0,t] \right\rbrace } .
\end{align*}
Note $0 < T' \leq T_0$. Suppose, to the contrary, that $T_0 < T_{**}$. Then the functions $(\phi ( \cdot + T_0 ) , \v{A} ( \cdot + T_0 ) )$ and $( u ( \cdot + T_0 ) , \v{B} ( \cdot + T_0 ) )$ are both $C_{T_{**} - T_0} H^m \times [ C_{T_{**} - T_0} H^m \cap C_{T_{**} - T_0}^1 H^{m-1} ]$ solutions to (\ref{eq:proof_thm6_1}) except with initial value $(\phi (T_0) , \v{A} (T_0) , \partial_t \v{A} (T_0) )$ instead of $(\phi_0 , \v{a}_0 , \dot{\v{a}}_0 )$. However, by the same reasoning that showed $(\phi , \v{A}) = (u , \v{B})$ on $[0,T']$, we conclude the existence of a $T^{\prime \prime} > 0$ so that $(\phi ( \cdot + T_0 ) , \v{A} ( \cdot + T_0 ) ) = ( u ( \cdot + T_0 ) , \v{B} ( \cdot + T_0 ) )$ on $[0 , T^{\prime \prime}]$, whereby $(\phi , \v{A}) = (u , \v{B})$ on $[T_0 , T_0 + T^{\prime \prime}]$. This contradicts the definition of $T_0$, and we must have $T_0 = T_{**}$. 

Define $T_{\mr{max}} \in (0 , \infty ]$ as
\begin{align*}
T_{\mr{max}} = \sup{ \left\lbrace T > 0 ~ : ~ \exists! (\phi , \v{A}) \in C_T H^m \times [C_T H^m \cap C_T^1 H^{m-1}] ~ \text{solution to} ~ (\ref{eq:proof_thm6_1}) \right\rbrace } .
\end{align*}
From the preceding discussion we know that $T_{\mr{max}}$ is well-defined. Let $\I = [0 , T_{\mr{max}} )$. Suppose $T_{\mr{max}} < \infty$ and $(\phi , \v{A}) \in C_{\I} H^m \times [C_{\I} H^m \cap C_{\I}^1 H^{m-1}]$ be the corresponding unique solution. Assume, to the contrary, that there exists a sequence $\{ t_k \}_{k \geq 1} \subset \I$ converging to $T_{\mr{max}}$ and satisfying
\begin{align*}
\max{ \left\lbrace \| \phi (t_k) \|_{m , 2} , \| \v{B} (t_k) \|_{m , 2} , \| \partial_t \v{B} (t_k) \|_{m-1 , 2} \right\rbrace } \leq R 
\end{align*}
for some $R > 0$. Using Lemma \ref{lem:contraction} we can find a $T(R) > 0$ and a solution to (\ref{eq:proof_thm6_1}) with the initial value $(\phi (t_k) , \v{A} (t_k) , \partial_t \v{A} (t_k))$ instead of $(\phi_0 , \v{a}_0 , \dot{\v{a}}_0 )$. So we extend our solution $(\phi , \v{A})$ to the interval $[t_k , t_k + T(R)]$. By choosing $k$ large enough, we can make $t_k + T(R) > T_{\mr{max}}$. This contradicts the definition of $T_{\mr{max}}$. Therefore, if $T_{\mr{max}} < \infty$, then
\begin{align*}
\limsup_{t \rightarrow T_{\mr{max}}^-} \| (\phi (t) , \v{A} (t) , \partial_t \v{A} (t) ) \|_{m,2 \oplus m,2 \oplus m-1 , 2} = \infty . 
\end{align*}

So far we have the existence of a maximal time interval $\I = [0 , T_{\mr{max}})$ for which we have a unique solution
\begin{align*}
(\phi , \v{A}) \in C_{\I} [H^m (\R^{3N}) ]^{2^N} \times [ C_{\I} H^m (\R^3 ; \R^3) \cap C^1_{\I} H^{m-1} (\R^3 ; \R^3) ]
\end{align*}
to (\ref{eq:proof_thm6_1}), and such that the blow-up alternative holds. This gives us the first portion of Theorem \ref{thm:local_exist_MBMP_epsilon}. What is left to show is the approximation portion of Theorem \ref{thm:local_exist_MBMP_epsilon}.

Let $(\phi_0 , \v{a}_0 , \dot{\v{a}}_0) \in \X_0^1$. Choose $R , T > 0$ so that
\begin{align*}
d ( \Psi ( \phi , \v{A} ) , \Psi (\phi' , \v{A}') ) \leq \frac{1}{2} d ( (\phi , \v{A}) , (\phi , \v{A}) ) , 
\end{align*}
for all $(\phi , \v{A}) , (\phi' , \v{A}') \in \X_T^1 (R)$, and let $(\phi , \v{A}) \in \X_T^1 (R)$ denoted the corresponding unique fixed point of $\Psi$. Consider a sequence of initial data $\{ (\phi_{0}^j , \v{a}^{j}_0 , \v{a}^{j}_1) \}_{j \geq 1}\subset \X_0^m$ and let $\{ ( \phi^j , \v{A}^j ) \}_{j \geq 1} \subset C_{T} H^m \times [C_{T} H^m \cap C_{T}^1 H^{m-1}]$ denote the corresponding sequence of solutions. Suppose that
\begin{align*}
\| ( \phi_0 - \phi_{0}^j , \v{a}_0 - \v{a}^{j}_0 , \dot{\v{a}}_0 - \v{a}^{j}_1)  \|_{1,2 \oplus 1,2 \oplus 2} \xrightarrow{~~ j \rightarrow \infty ~~} 0 .
\end{align*}
Observe that if $j$ is sufficiently large then (\ref{eq:PhiXtoX4}) holds with $(\phi_0 , \v{a}_0 , \dot{\v{a}}_0)$ replaced by $(\phi_0^j , \v{a}_0^j , \dot{\v{a}}_0^j)$, and therefore $(\phi^j , \v{A}^j) \in \X^m_{T} (R)$ when $j$ is sufficiently large. Using identical estimates that yielded (\ref{eq:K-diff}) and (\ref{eq:xi-diff}) in the proof of Lemma \ref{lem:contraction} we have the estimate
\begin{align*}
d ( (\phi , \v{A}) , (\phi^j , \v{A}^j) ) \leq C_1 \| ( \phi_0 - \psi_{0}^j , \v{a}_0 - \v{a}^{j}_0 , \dot{\v{a}}_0 - \v{a}^{j}_1) \|_{1,2 \oplus 1,2 \oplus 2} + C_2 g(T , R) d ( (\phi , \v{A}) , (\phi^j , \v{A}^j) ) ,
\end{align*}
where the function $g$ is defined by (\ref{eq:PhiContrac-def-g}) and $C_2$ is the \textit{same} constant appearing in (\ref{eq:Phicontract2}). Since $T$ was chosen so that $g(T , R) = 1 / (2 C_2)$, we conclude $d ( (\phi , \v{A}) , (\phi^j , \v{A}^j) ) \rightarrow 0$ as $j \rightarrow \infty$ on the time interval $[0 , T]$.

Consider as initial data $(\phi (T) , \v{A} (T) , \partial_t \v{A} (T)) \in \X^1_0$ and $\{ (\phi^j (T) , \v{A}^j (T) , \partial_t \v{A}^j (T)) \}_{j \geq 1} \subset \X_0^m$. By the preceding arguments, 
\begin{align*}
\| ( (\phi - \phi^j) (T) , (\v{A} - \v{A}^j) (T) , (\partial_t \v{A} - \partial_t \v{A}^j) (T) ) \|_{1,2 \oplus 1,2 \oplus 2} \xrightarrow{~~ j \rightarrow \infty ~~} 0  .
\end{align*}
Choose $R' , T' > 0$ so that
\begin{align*}
d ( \Psi ( \phi , \v{A} ) , \Psi (\phi' , \v{A}') ) \leq \frac{1}{2} d ( (\phi , \v{A}) , (\phi , \v{A}) ) , 
\end{align*}
for all $(\phi , \v{A}) , (\phi' , \v{A}') \in \X_{T'}^1 (R')$. Using the same notation, let $(\phi , \v{A}) \in \X_{T'}^1 (R')$ denoted the corresponding unique fixed point of $\Psi$ and let $\{ ( \phi^j , \v{A}^j ) \}_{j \geq 1} \subset C_{T'} H^m \times [C_{T'} H^m \cap C_{T'}^1 H^{m-1}]$ denote the sequence of solutions corresponding to the initial data $\{ (\phi^j (T) , \v{A}^j (T) , \partial_t \v{A}^j (T)) \}_{j \geq 1}$ in $\X_0^m$. As before, if $j$ is sufficiently large, then $(\phi^j , \v{A}^j) \in \X^m_{T'} (R')$. By the same reasoning as before, we can conclude $d ( (\phi , \v{A}) , (\phi^j , \v{A}^j) ) \rightarrow 0$ as $j \rightarrow \infty$ on the time interval $[0 , T']$ with $T' > T$. We can repeat this argument ad infinitum and conlude the desired convergence at each $t \in \I = [0 , T_{\mr{max}} )$.
\end{proof}

\section{Bound on the Coulomb Energy and Global Existence}\label{sec:conserve}

In this section we prove the dissipation of energy and global existence in $\X^1_0$ for the $\varepsilon$-modified system (\ref{eq:MBMP_epsilon}) as stated in Theorem \ref{thm:MBMP_epsilon_dissipation-laws}. As discussed in detail in \S\ref{sec:proof_strat}, the crucial result that is needed for the proof of global existence is the uniform bound on the Coulomb energy functional $V[\phi] = \langle \phi , V (\ul{\v{R}} , \Z) \phi \rangle_{L^2}$ on the space $\fC_N$, defined by (\ref{def:function_space_C}), as expressed in Theorem \ref{lem:bound_on_coulomb}. It is helpful to recall the definition
\begin{align*}
E^{\mr{G}}_{\mr{P}} (\alpha) = E^{\mr{G}}_{\mr{P}}  (\alpha , \Z , N , K) = \inf{ \left\lbrace  [\phi , \v{A} , \v{0}] ~ : ~ (\phi , \v{A}) \in \fC_N \right\rbrace } , 
\end{align*} 
where $E_{\mr{P}} [\phi , \v{A} , \v{0}] = T_{\mr{P}} [\phi , \v{A}] + V[\phi] + F[\v{A} , \v{0}]$ where $T_{\mr{P}}$, $V$, and $F$ are defined by (\ref{def:total_kinetic}), (\ref{def:tot_potential_energy}), and (\ref{def:field_energy_potentials}), respectively. We begin with a proof of Theorem \ref{lem:bound_on_coulomb}. 

\begin{proof}[Proof of Theorem \ref{lem:bound_on_coulomb}]
Fix $(\phi , \v{A}) \in \fC_N$. Throughout we abuse notation and abbreviate $E_{\mr{P}} [\phi , \v{A}] = E_{\mr{P}} [\phi , \v{A} , \v{0}]$ and $F [ \v{A} , \v{0} ] = F [ \v{A} ]$. Obviously, we have the lower bound $E_{\mr{P}} [\phi, \v{A}] \geq E_{\mr{P}}^{\mr{G}} (\alpha)$. We claim such a uniform lower bound implies
\begin{align}\label{eq:lem10_1}
\left( V[\phi] + F [\v{A}] \right)^2 \leq 4 |E^{\mr{G}}_{\mr{P}} (\alpha)| T_{\mr{P}} [\phi , \v{A}] .
\end{align}
Indeed, for $\lambda > 0$, consider the scaling $\phi_{\lambda} (\ul{\v{z}}) = \lambda^{3N/2} \phi (\lambda \ul{\v{z}})$ and $\v{A}_{\lambda} (\v{y}) = \lambda \v{A} (\lambda \v{y})$. Under this scaling
\begin{align*}
T_{\mr{P}} [\phi_{\lambda} , \v{A}_{\lambda}] + V[\phi_{\lambda}] + F [\v{A}_{\lambda}] = \lambda^2 T_{\mr{P}} [\phi , \v{A}] + \lambda \left( V[\phi] + F [\v{A}] \right) \geq E^{\mr{G}}_{\mr{P}} (\alpha) 
\end{align*}  
Minimizing over $\lambda$ in the previous expression yields (\ref{eq:lem10_1}).

Let $\{ (\phi^{n} , \v{A}^{\!n}) \}_{n \geq 1} \subset \fC_N$ be a sequence such that $E_n = T_n + V_n +  F_n \leq C (\alpha)$ where $E_n \equiv E_{\mr{P}} [\phi^n , \v{A}^{\!n}]$, $T_n \equiv T_{\mr{P}} [\phi^{n} , \v{A}^{\!n}]$, $V_n \equiv V[\phi^{n}]$, and $F_n \equiv F[\v{A}^{\!n}]$. Suppose, to the contrary, that $|V_n| \rightarrow \infty$ as $n \rightarrow \infty$. The condition $E_n \leq C(\alpha)$ implies that we necessarily have $V_n \rightarrow - \infty$. Set $\lambda_n = 1/ |V_n|$ and note $\lambda_n \rightarrow 0$ as $n \rightarrow \infty$. Consider the scaling $\Phi^{n} (\ul{\v{z}}) = \lambda_n^{3N/2} \phi^{n} ( \lambda_n \ul{\v{z}} )$ and $\v{a}_n (\v{y}) = \lambda_n \v{A}^{\!n} (\lambda_n \v{y})$. Moreover, from $E_{\mr{P}}^{\mr{G}} (\alpha) \leq E_n \leq C(\alpha)$ we have
\begin{align}\label{eq:lem10_2}
E_{\mr{P}}^{\mr{G}} (\alpha) \lambda_n \leq \frac{t_n}{\lambda_n} - 1 + \frac{1}{\alpha^2} f_n \leq C (\alpha) \lambda_n 
\end{align}
where $t_n = T [\Phi^{n} , \v{a}_n] = \lambda_n^2 T_n$ and $f_n = \| \v{a}_n \|_2^2 / (8 \pi) = \lambda_n \alpha^2 F_n$.

Pick $\nu$ with $\alpha > \nu$ and note that we have $E_{\mr{P}}^{\mr{G}} (\nu) > - \infty$. As before,  
\begin{align}\label{eq:lem10_4}
E_{\mr{P}}^{\mr{G}} (\nu) \lambda_n \leq \frac{t_n}{\lambda_n} - 1 + \frac{1}{\nu^2} f_n \leq C (\nu) \lambda_n .
\end{align}
Subtracting (\ref{eq:lem10_4}) from (\ref{eq:lem10_2}) we conclude
\begin{align}
( E_{\mr{P}}^{\mr{G}} (\alpha) - C(\nu) ) \lambda_n \leq \left( \frac{1}{\nu^2} - \frac{1}{\alpha^2} \right) f_n \leq ( C (\alpha) - E^{\mr{G}}_{\mr{P}} (\nu) ) \lambda_n ,
\end{align} 
and thus $f_n \rightarrow 0$ as $n \rightarrow \infty$. Feeding this back into (\ref{eq:lem10_2}) we conclude $\lim_{n \rightarrow \infty} (t_n / \lambda_n)  = 1$. Moreover, (\ref{eq:lem10_1}) implies
\begin{align*}
\left( \frac{f_n}{\alpha^2}  - 1 \right)^2 \leq 4 C(\alpha) t_n ,
\end{align*}
and as a consequence
\begin{align}\label{eq:lem10_3}
\liminf_{n \rightarrow \infty} t_n \geq \frac{1}{4 C(\alpha)} .
\end{align} 
However, (\ref{eq:lem10_3}) implies $\lim_{n \rightarrow \infty} (t_n / \lambda_n) = \infty$. Contradiction.
\end{proof}

With Theorem \ref{lem:bound_on_coulomb} at our disposal we are in a position to give a complete proof of Theorem  \ref{thm:MBMP_epsilon_dissipation-laws}. For the proof of Theorem \ref{thm:MBMP_epsilon_dissipation-laws} it will be useful to recall that if $( \phi , \v{A} ) \in \fC_N$ then the kinetic energy $T_{\mr{P}} [ \phi , \v{A} ]$, defined by (\ref{def:total_kinetic}), reduces to $T[\phi , \v{A}] = N \| \gvsig_1 \cdot (\v{p}_1 + \v{A}_1) \phi \|_2^2$. This is a consequence of the complete antisymmetry of $\phi$. Likewise, the total probability current density $\J_{\!\mr{P}}[ \phi , \v{A} ]$, defined by (\ref{def:prob_current_compact}), will reduce to
\begin{align*}
\J_{\!\mr{P}}[\phi , \v{A}] = - 2 \alpha N \re{ \int \langle \gvsig_1 \psi_{\ul{\v{z}}_1'} , \gvsig_1 \cdot (\v{p}_1 + \v{A}_1 ) \psi_{\ul{\v{z}}_1'}  \rangle_{\C^2} \dd \ul{\v{z}}_1' } .
\end{align*} 
We will abuse notation and abbreviate $\gvsig_1$, $\v{p}_1$, and $\v{A}_1$ by $\gvsig$, $\v{p}$, and $\v{A}$, respectively. 

\begin{proof}[Proof of Theorem \ref{thm:MBMP_epsilon_dissipation-laws}]
Fix $\varepsilon > 0$ and $m \in [1 , 2]$. Let $(\phi_0 , \v{a}_0 , \dot{\v{a}}_0) \in \X^m_0$ with $\phi_0 \in \BWN{H^m(\R^3;\C^2)}$ and $\| \phi_0 \|_2 = 1$. Let $(\phi , \v{A})$ be the corresponding solution on $\I$ to (\ref{eq:MBMP_epsilon}) as given by Theorem \ref{thm:local_exist_MBMP_epsilon}. It is straightforward to verify that $\partial_t \phi (t) \in H^{-m}$ since $\Ham (\v{A} (t)) \psi (t) \in H^{-m}$ for each $t \in \I$. Therefore, it makes sense to compute
\begin{align}\label{eq:proof_thm_3(0)}
\frac{\dd }{ \dd t} \| \phi \|_{2}^2 = 2 \re{ \langle \partial_t \phi , \phi \rangle_{H^{-m} , H^m} } = 2 \varepsilon ( \| \phi \|_{2}^2 - 1 ) \langle \Ham (\v{A}) \phi , \phi \rangle_{H^{-m} , H^m} .
\end{align}
Since $\| \phi_0 \|_{2} = 1$, (\ref{eq:proof_thm_3(0)}) implies $\| \phi (t) \|_{2} = 1$.

Consider the case $m = 2$. In this case, $\Ham (\v{A} (t)) \phi (t) \in L^2$ for each $t \in \I$ and, hence, we may take the time-derivaitve of the total energy $\E = \E [ \phi , \v{A} , \partial_t \v{A}]$, as defined in (\ref{def:epsilon_tot_energy}), to find
\begin{align*}
\frac{\dd \E}{ \dd t} & = 2 \re{ \langle \partial_t \phi , \Ham (\v{A}) \phi \rangle_{L^2} } + 2 N \re{ \langle \gvsig \cdot (\v{p} + \tilde{\v{A}}) \phi , ( \gvsig \cdot \partial_t \tilde{\v{A}} ) \phi \rangle_{L^2} } + \partial_t F [\v{A} , \partial_t \v{A}] \\
& = - 2 \varepsilon (  \| \Ham (\v{A}) \phi \|_{2}^2 - \langle \phi  , \Ham (\v{A}) \phi \rangle^2_{L^2} ) \\
& \hspace{2cm} + 2 N \re{ \langle \gvsig \cdot (\v{p} + \tilde{\v{A}}) \phi , ( \gvsig \cdot \partial_t \tilde{\v{A}} ) \phi \rangle_{L^2} } + \partial_t F [\v{A} , \partial_t \v{A}] \numberthis \label{eq:proof_thm_3(1)} .
\end{align*}
Using that $\v{A}$ satisfies the wave equation (\ref{eq:MBMP_epsilon}, second equation) we can show that the last two terms in (\ref{eq:proof_thm_3(1)}) cancel each other. From (\ref{eq:MBMP_epsilon}),
\begin{align*}
\partial_t F [\v{A} , \partial_t \v{A}] & = 2 \frac{1}{8 \pi \alpha^2} \langle \square \v{A} , \partial_t \v{A} \rangle_{L^2} \\
& = \frac{1}{\alpha} \langle \Lambda_{\varepsilon}^{-1} \Hproj{ \J_{\!\mr{P}}[\phi , \tilde{\v{A}}] } , \partial_t \v{A} \rangle_{L^2} \\
& = - 2 N \langle \re{ \int \langle \gvsig \phi_{\ul{\v{z}}_1'} , \gvsig \cdot (\v{p} + \tilde{\v{A}}) \phi_{\ul{\v{z}}_1'} \rangle_{\C^2} \dd \ul{\v{z}}_1' } , \partial_t \tilde{\v{A}} \rangle_{L^2} \\
& = - 2 N \re{ \langle \gvsig \cdot (\v{p} + \tilde{\v{A}}) \phi , ( \gvsig \cdot \partial_t \tilde{\v{A}} ) \phi \rangle_{L^2} } \numberthis \label{eq:proof_thm_3(2)} .
\end{align*}
Plugging (\ref{eq:proof_thm_3(2)}) into (\ref{eq:proof_thm_3(1)}) we arrive at
\begin{align*}
\frac{\dd \E}{\dd t} = - 2 \varepsilon ( \| \Ham (\v{A}) \phi \|_{2}^2 - \langle \phi  , \Ham (\v{A}) \phi \rangle^2_{L^2} ) ,
\end{align*}
which upon integrating yields (\ref{eq:Diss_Energy}). 

Continue assuming $m = 2$. Suppose $\alpha$ and $\alpha^2 \max{\Z}$ are sufficiently small to ensure $E^{\mr{G}}_{\mr{P}} > - \infty$. To prove the bounds (\ref{eq:uniform-bounds}), we first verify that hypothesis of Lemma \ref{lem:bound_on_coulomb}. For the moment we include the $\varepsilon$ and $t$ dependence of $\phi$ and $\v{A}$ for clarity. By previous results $\| \phi^{\varepsilon} (t) \|_{2} = 1$ (this, in fact, holds for any $m \in [1 , 2]$). Moreover, we note that
\begin{align*}
F[\tilde{\v{A}}^{\!\varepsilon} , \v{0}] \leq F[\v{A}^{\!\varepsilon} , \v{0}] \leq F[ \v{A}^{\!\varepsilon} , \partial_t \v{A}^{\!\varepsilon}] ,
\end{align*} 
and $\langle \phi^{\varepsilon} , \Ham^{\varepsilon} (\v{A}^{\!\varepsilon}) \phi^{\varepsilon} \rangle^2_{L^2} \leq \| \Ham^{\varepsilon} (\v{A}^{\!\varepsilon}) \phi^{\varepsilon} \|_{2}^2$ by Cauchy-Schwartz. Therefore, from $E_{\mr{P}}^{\mr{G}} > - \infty$ and the dissipation of energy (\ref{eq:Diss_Energy}), we arrve at
\begin{align*}
E_{\mr{P}}^{\mr{G}} \leq T_{\mr{P}} [\phi^{\varepsilon} (t) , \tilde{\v{A}}^{\!\varepsilon} (t)] - V[\phi^{\varepsilon} (t)] + F[ \tilde{\v{A}}^{\!\varepsilon} (t) , \v{0} ] \leq \E [ \phi_0 , \v{a}_0 , \dot{\v{a}}_0 ] .
\end{align*}
Consequently, Lemma \ref{lem:bound_on_coulomb} tells us that
\begin{align}\label{eq:proof_thm_3_3}
|V[\phi^{\varepsilon} (t)]| = |\langle \phi^{\varepsilon} (t) , V (\ul{\v{R}} , \Z) \phi^{\varepsilon} (t) \rangle_{L^2} | \leq C
\end{align}
where $C$ is a finite constant depending on $\alpha$, $\Z$, $N$, $K$, and the initial data $(\phi_0 , \v{a}_0 , \dot{\v{a}}_0 )$, but \textit{independent} of $\varepsilon$ and $t$. Proceeding we will drop the $\varepsilon$ and $t$ dependence.

The bound (\ref{eq:proof_thm_3_3}) immediately gives us the second estimate in (\ref{eq:uniform-bounds}). Indeed, using the bound on the Coulomb energy we find
\begin{align*}
F [\v{A} , \partial_t \v{A}] \leq |\E [ \phi_0 , \v{a}_0 , \dot{\v{a}}_0 ]| + |V[\phi]| \leq C_2 ,
\end{align*}
where $C_2 = |\E [ \phi_0 , \v{a}_0 , \dot{\v{a}}_0 ]| + C$. This, in turn, yields the third estimate in (\ref{eq:uniform-bounds}) by differentiation:
\begin{align*}
\frac{\dd}{\dd t} \| \v{A} \|_{2}^2 = 2 \langle \v{A} , \partial_t \v{A} \rangle_{L^2} \leq 2 \| \v{A} \|_{2} \| \partial_t \v{A} \|_{2} \leq 2 \| \v{A} \|_{2} \sqrt{C_2} .
\end{align*}
Hence,
\begin{align*}
\| \v{A} \|_{2} \leq C_3 (1 + t) , 
\end{align*}
where $C_3 = \max{\{ \|\v{a}_0 \|_2 , \sqrt{C_2} \} }$. Deriving the first estimate in (\ref{eq:uniform-bounds}) requires a more careful analysis. Consider $\delta > 0$ to be specified later. First, note that
\begin{align}\label{eq:proof_thm_3_4}
\| \v{p} \phi \|_2 = \sqrt{N} \| \v{p}_1 \phi \|_2 = \sqrt{N} \| \gvsig_1 \cdot \v{p}_1 \phi \|_2 \leq \sqrt{N} \left( \| \gvsig_1 \cdot ( \v{p}_1 + \tilde{\v{A}}_1 ) \phi \|_2 + \| \tilde{\v{A}}_1 \phi \|_2 \right)  .
\end{align}
The first term on the right hand side of (\ref{eq:proof_thm_3_4}) can be bounded in the same way as the field energy $F [ \v{A} , \partial_t \v{A}]$. Indeed, using the dissipation of energy (\ref{eq:Diss_Energy}) and the fact that $|V[\phi]|$ is uniformly bounded, we have
\begin{align}\label{eq:proof_thm_3_5}
\sqrt{N} \| \gvsig_1 \cdot ( \v{p}_1 + \tilde{\v{A}}_1 ) \phi \|_2 = \sqrt{T_{\mr{P}} [\phi , \tilde{\v{A}}]} \leq \sqrt{ |\E [ \phi_0 , \v{a}_0 , \dot{\v{a}}_0 ]| + |V[\phi]| } \leq \sqrt{C_2} .
\end{align}
To estimate $\| \tilde{\v{A}}_1 \phi \|_2$ we rely on the Gagliardo–Nirenberg inequality, Lemma \ref{lem:GN_inequality}, in the particular case $d = r = 3$ and $p = q = 2$, which reads $\| f \|_{3} \leq c \| f \|_2^{1/2} \| \v{p} f \|_{2}^{1/2}$ where $c > 0$ is some universal constant. Using this inequality, together with the Young's inequality for products: $2 ab \leq (\delta^{-1} a)^2 + (\delta b)^2$, we find
\begin{align*}
\| \tilde{\v{A}}_1 \phi \|_2 & \leq \| \tilde{\v{A}} \|_6  \left[ \int \left( \int | \phi_{\ul{\v{z}}_1'} (\v{z}_1) |^3 \dd \v{z}_1  \right)^{2/3} \dd \ul{\v{z}}_1'  \right]^{1/2} \\
& \leq \sqrt{ 8 \pi \alpha^2 S_3^{-1} F[\v{A} , \partial_t \v{A}] } \left[ \int \left( \int | \phi_{\ul{\v{z}}_1'} (\v{z}_1) |^2 \dd \v{z}_1  \right)^{1/2} \left( \int | \v{p}_1 \phi_{\ul{\v{z}}_1'} (\v{z}_1) |^2 \dd \v{z}_1 \right)^{1/2} \dd \ul{\v{z}}_1'  \right]^{1/2} \\
& \leq \sqrt{ 4 \pi \alpha^2 S_3^{-1} C_2 }  \left[ \delta^{-1} \| \phi \|_2 + \delta \| \v{p}_1 \|_2 \right] , \numberthis \label{eq:proof_thm_3_6}
\end{align*}
where $S_3$ is the sharp constant in Sobolev's inequality on $\R^3$: $S_3 \| f \|_6^2 \leq \| \nabla f \|_2^2$. Choosing $\delta$ so that
\begin{align*}
\sqrt{ 16 \pi \alpha^2 S_3^{-1} C_2 } \delta = \frac{1}{2}
\end{align*}
we can feed (\ref{eq:proof_thm_3_5}) and (\ref{eq:proof_thm_3_6}) back into (\ref{eq:proof_thm_3_4}) and arrive at a uniform bound on $\| \v{p} \phi \|_2$. Summarizing, we've derived the bounds (\ref{eq:uniform-bounds}) for $m = 2$. That these uniform estimates in (\ref{eq:uniform-bounds}) hold for $1 \leq m < 2$ follows immediately from the convergence result in Theorem \ref{thm:local_exist_MBMP_epsilon}. The last claim of Theorem \ref{thm:MBMP_epsilon_dissipation-laws} follows immediately from the uniform estimates in the energy class (\ref{eq:uniform-bounds}) and the blow-up alternative in Theorem \ref{thm:local_exist_MBMP_epsilon}. 
\end{proof}

\section{Proof of the Main Result}\label{sec:proof_main_result}

This last section completes the proof of the main result in this thesis, namely Theorem \ref{thm:weak_solns_MBMP}. The proof will use the Banach-Alaoglu Theorem: for a given Banach space $X$, the unit ball in $X^{\ast}$ is compact in the weak$^{\ast}$ topology (if $X$ is reflexive, we may replace weak$^{\ast}$ topology with the weak topology) \cite[Theorem IV.21]{RS72}. The $\varepsilon$-independent bounds (\ref{eq:uniform-bounds}) together with the Banach Alaoglu Theorem then allow us to extract a weak$^{\ast}$ converging subsequence of solutions to the $\varepsilon$-modified system (\ref{eq:MBMP_epsilon}). Proving that this weak$^{\ast}$ limit satisfies the MBMP equations (\ref{eq:MBMP}) requires the Rellich-Kondrachov Compactness Theorem \ref{lem:rellich_kon} and Aubin-Lions Lemma \ref{lem:aubin_lions}.

\begin{proof}[Proof of Theorem \ref{thm:weak_solns_MBMP}]
Consider
\begin{align*}
(\psi_0 , \v{a}_0 , \dot{\v{a}}_0) \in \BWN{H^1(\R^3;\C^2)}  \times H^1 (\R^3 ; \R^3) \times L^2 (\R^3 ; \R^3) ,
\end{align*}
with $\| \psi_0 \|_2 = 1$ and $\diver{\v{a}_0} = \diver{\dot{\v{a}}_0} = 0$. Let $\{ \varepsilon_n \}_{n \geq 1} \subset \R_+$ with $\varepsilon_n \rightarrow 0$. Combining Theorem \ref{thm:local_exist_MBMP_epsilon} and \ref{thm:MBMP_epsilon_dissipation-laws}, there exists a sequence of solutions
\begin{align*}
\{ (\phi^{n} , \v{A}^{\!n}) \}_{n \geq 1} \subset C ( \R_+ ; \BWN{H^1(\R^3;\C^2)} ) \times [C (\R_+ ; H^1 (\R^3 ; \R^3)) \cap C^1 (\R_+ ; L^2 (\R^3 ; \R^3) ] 
\end{align*}
of the modified equations
\begin{align}\label{eq:weaksolns4}
\left\lbrace \begin{array}{l}
 \partial_t \phi^{n} - (i + \varepsilon_n) \sum_{j = 1}^N \Delta_{\v{x}_j} \phi^{n} \\
 \hspace{1cm} = \varepsilon_n (T_n + V_n) \phi^{n} - (i + \varepsilon_n) \left( \fL ( \tilde{\v{A}}^{\!n} ) - V(\ul{\v{R}} , \Z) \right) \phi^{n} , \\
 \square \v{A}^{\!n} = 4 \pi \alpha \Lambda_{\varepsilon_n}^{-1} \Hproj{ \J [ \phi^{n} , \tilde{\v{A}}^{\!n} ] } , \\
 \diver{\v{A}^{\!n}} = 0 , ~~~~  \tilde{\v{A}}^{\!n} = \Lambda_{\varepsilon_n}^{-1} \v{A}^{\!n}  , \\
 ( \phi^{n} (0) , \v{A}^{\!n} (0) , \partial_t \v{A}^{\!n} (0) ) = (\psi_0 , \v{a}_0 , \dot{\v{a}}_0) .
\end{array} \right.
\end{align} 
where $T_n = T_{\mr{P}} [\phi^{n} , \v{A}^{\!n}]$, $V_n = V[ \phi^{n} ]$, and $\fL (\tilde{\v{A}}^{\!n}) = \sum_{j=1}^N \cL_j ( \tilde{\v{A}}^{\!n} )$ is given by (\ref{def:D_j}). Moreover, the bounds 
\begin{align}\label{eq:weaksolns5}
 \| \nabla \phi^{n} (t) \|_2 \leq C_1 , \hspace{1cm} F [\v{A}^{\!n} , \partial_t \v{A}^{\!n}] (t) \leq C_2 , \hspace{1cm} \| \v{A}^{\!n} (t) \|_2 \leq C_3 ( 1 +  t )  
\end{align}
are satisfied. The estimates (\ref{eq:Estimate-Pauli-1}) and (\ref{eq:Estimate-Coulomb-appen_1}) of Lemmas \ref{lem:Estimates-Pauli} and \ref{lem:Estimates-Coulomb}, respectively, yield
\begin{align}\label{eq:weaksolns6}
\| [ \fL ( \tilde{\v{A}}^{\!n} ) - V(\ul{\v{R}} , \Z) ] \phi^n \|_{\frac{3}{2}} \lesssim ( 1 + \| \v{A}^{\!n} \|_{1,2} ) \| \v{A}^{\!n} \|_{1,2} \| \phi^{n} \|_{1,2} + \| \phi^{n} \|_{1,2}  .
\end{align}
Furthermore, in the same way we estimated (\ref{eq:Estimate-KG-3}), we have
\begin{align}\label{eq:weaksolns7}
\| \J [ \phi^{n} , \tilde{\v{A}}^{\!n} ]  \|_{\frac{3}{2}} \lesssim (1 + \| \v{A}^{\!n} \|_{1,2} ) \| \phi^{n} \|_{1,2} .
\end{align}

The bounds (\ref{eq:weaksolns5}) allow us to apply the Banach-Alaoglu Theorem, and, thus, we may extract a subsequence, still denoted by $\{ (\phi^{n} , \v{A}^{\!n}) \}_{n \geq 1}$, such that
\begin{align}
& \v{A}^{\!n} \xrightarrow{~~\mr{w}^{\ast}~~} \v{A} ~~~~ \text{in} ~~~ L^{\infty} ([0,T] ; H^1 ) , \label{eq:weaksolns8} \\
& \partial_t \v{A}^{\!n} \xrightarrow{~~\mr{w}^{\ast}~~} \partial_t \v{A} ~~~~ \text{in} ~~~ L^{\infty} (\R_+ ; L^2  ) \label{eq:weaksolns9} \\
& \phi^{n} \xrightarrow{~~\mr{w}^{\ast}~~} \psi ~~~~ \text{in} ~~~ L^{\infty} (\R_+ ; H^1) , \label{eq:weaksolns10} \\
& \J [\phi^{n} , \tilde{\v{A}}^{\!n} ] \xrightarrow{~~\mr{w}^{\ast}~~} \beta ~~~~ \text{in} ~~~ L^{\infty} ([0,T] ; L^{\frac{3}{2}} ) \label{eq:weaksolns11} \\
& [ \fL ( \tilde{\v{A}}^{\!n} ) - V(\ul{\v{R}} , \Z) ] \phi^{n} \xrightarrow{~~\mr{w}^{\ast}~~} \gamma ~~~~ \text{in} ~~~ L^{\infty} ([0,T] ; L^{\frac{3}{2}}  ), \label{eq:weaksolns12} 
\end{align}
for all $0 < T < \infty$. Passing to the limit in (\ref{eq:weaksolns4}), and using (\ref{eq:weaksolns8}) through (\ref{eq:weaksolns12}), we find
\begin{align}\label{eq:weaksolns13}
\left\lbrace \begin{array}{l}
\partial_t \psi - i \sum_{j = 1}^N \Delta_{\v{x}_j} \psi = - i \gamma ,  \\
\square \v{A} = 4 \pi \alpha \Hproj{ \beta } , \\
\diver{\v{A}} = 0 . 
\end{array} \right. 
\end{align}
as equations in $\fD' (\R_+ ; \BWN{H^{-1}(\R^3;\C^2)} \times \fD' (\R_+ ; H^{-1} (\R^3 ; \R^3))$. We note that in passing to the limit we've used Theorem \ref{lem:bound_on_coulomb} and the dissipation of energy (\ref{eq:Diss_Energy}) to ensure $|T_n + V_n| \nrightarrow \infty$ as $\varepsilon_n \rightarrow 0$. Now, $\partial_t \v{A} \in L^{\infty} (\R_+ ; L^2 (\R^3 ; \R^3))$, $\partial_t^2 \v{A} \in L^{\infty} (\R_+ ; H^{-1} ( \R^3 ; \R^3))$, and $\partial_t \psi \in L^{\infty} (\R_+ ; H^{-1} (\R^{3N} ; \C^{2^N}))$ by (\ref{eq:weaksolns13}). Thus
\begin{align*}
(\psi , \v{A} , \partial_t \v{A}) \in L^{\infty}_{\mr{loc}} (\R_+ ; H^1 \oplus H^1 \oplus L^2) \cap C (\R_+ ; H^{-1} \oplus L^2 \oplus H^{-1}) ,
\end{align*}
and this implies the weak continuity $(\psi , \v{A} , \partial_t \v{A}) \in C^{\mr{w}} (\R_+ ; H^1 \oplus H^1 \oplus L^2)$. Indeed, let $u \in L^{\infty} (\I ; X) \cap C (\I^{\circ} ; Y)$ for a closed interval $\I$, a reflexive Banach space $X$, and normed space $Y$ with continuous embedding $X \hookrightarrow Y$. Take $\{ t_n \}_{n \geq 1} \subset \I$, $t_n \rightarrow t$. Then $\{ u (t_n) \}_{n \geq 1}$ is bounded in $X$, hence contains a weakly converging subsequence $\{ u(t_{n_k}) \}_{k \geq 1}$ converging weakly to some $u_0 \in X$. Due to the continuity in $Y$, we have $u_0 = u(t)$. This means the weak subsequential limit does not depend on the subsequence. Moreover, each subsequence of $\{ u (t_n) \}_{n \geq 1}$ contains another subsequence converging to $u(t)$ weakly in $X$. This implies $u(t_n) \rightharpoonup u(t)$ in $X$, which is the weak continuity.

Next we show that $\gamma = [ \fL ( \v{A} ) - V(\ul{\v{R}} , \Z) ] \psi$ and $\beta = \J [\psi , \v{A}]$. Let $I \subset \R_+$ be a bounded interval and $\Omega \subset \R^3$, $S \subset \R^{3N}$ be bounded and open, and assume $\partial \Omega , \partial S$ are both $C^1$. It suffices to show that $\gamma$ and $\beta$ coincide with $[ \fL ( \v{A} ) - V(\ul{\v{R}} , \Z) ] \psi$ and $\J [\psi , \v{A}]$ on $I \times S$ and $I \times \Omega$, respectively. Now, by (\ref{eq:weaksolns5}), $\{ ( \v{A}^{\!n} , \partial_t \v{A}^{\!n} ) \}_{n \geq 1}$ is a bounded sequence in $L^4 (I ; H^1 (\Omega ; \R^3) \times L^2 (\Omega ; \R^3))$. From the Rellich-Kondrachov Theorem \ref{lem:rellich_kon} we have $H^1 (\Omega ; \R^3) \hookrightarrow L^4 (\Omega ; \R^3) \subset L^2 (\Omega ; \R^3)$ and, hence, Lemma \ref{lem:aubin_lions} guarantees that there is a subsequence of $\{ \v{A}^{\!n} \}_{n \geq 1}$, still denoted by $\{ \v{A}^{\!n} \}_{n \geq 1}$, such that
\begin{align}\label{eq:weaksolns15}
\v{A}^{\!n} \xrightarrow{~~ n \rightarrow \infty ~~} \v{A} ~~~~ \text{in} ~~~ L^4 (I \times \Omega)
\end{align}
Further, note that $\{ \partial_t \phi^{n} \}_{n \geq 1}$ is bounded in $L^{\infty} (I ; H^{-1} (S ; \C^{2^N}))$ using (\ref{eq:weaksolns4}). This implies that $\{ (\phi^{n} , \partial_t \phi^{n} ) \}_{n \geq 1}$ is bounded in 
\begin{align*}
L^2 (I ; H^1 (S ; \C^{2^N}) \times H^{-1} (S ; \C^{2^N})).
\end{align*}
Again applying the Aubin-Lions Lemma \ref{lem:aubin_lions}, we conclude
\begin{align}\label{eq:weaksolns16}
\phi^{n} \xrightarrow{~~ n \rightarrow \infty ~~} \psi ~~~~ \text{in} ~~~ L^4 (I \times S )
\end{align}
From (\ref{eq:weaksolns8}), (\ref{eq:weaksolns10}), (\ref{eq:weaksolns15}), and (\ref{eq:weaksolns16}) it is straightforward to show that
\begin{align*}
& \Lambda_{\varepsilon_n}^{-1} \J [\phi^{n} , \tilde{\v{A}}^{\!n}] \rightharpoonup \J [\psi , \v{A}] ~~~~ \text{in} ~~~ L^{\frac{4}{3}} (I \times \Omega) , \\
& [ \fL ( \tilde{\v{A}}^{\!n} ) - V(\ul{\v{R}} , \Z) ] \phi^{n} \rightharpoonup [ \fL ( \v{A} ) - V(\ul{\v{R}} , \Z) ] \psi ~~~~ \text{in} ~~~ L^{\frac{4}{3}} (I \times S )  .
\end{align*}
Moreover (\ref{eq:weaksolns11}) through (\ref{eq:weaksolns12}) imply 
\begin{align*}
& \Lambda_{\varepsilon_n}^{-1} \v{J} [\phi^{n} , \tilde{\v{A}}^{\!n}] \rightharpoonup \beta ~~~~ \text{in} ~~~ L^{\frac{4}{3}} (I \times \Omega) , \\
& [ \fL ( \tilde{\v{A}}^{\!n} ) - V(\ul{\v{R}} , \Z) ] \phi^{n} \rightharpoonup \gamma ~~~~ \text{in} ~~~ L^{\frac{4}{3}} (I \times S )  .
\end{align*}
Since weak limits are unique we conclude $\gamma = [ \fL ( \v{A} ) - V(\ul{\v{R}} , \Z) ] \psi$ and $\beta = \J [\psi , \v{A}]$ on $I \times \Omega$ and $I \times S$, respectively. 

It remains to show that $(\psi , \v{A} , \partial_t \v{A})$ satisfies the initial conditions in (\ref{eq:weaksolns4}). Since
\begin{align*}
( \v{A}^{\!n} , \partial_t \v{A}^{\!n} ) \in L^2 ([0,T] ; H^1 (\R^3 ; \R^3) \times L^2 (\R^3 ; \R^3)) ,
\end{align*}
we may integrate by parts to find
\begin{align*}
\int_0^T \langle \v{A}^{\!n} (s) \partial_t f(s) + \partial_t \v{A}^{\!n} (s) f(s) , \phi \rangle_{H^1 , H^{-1}} ds = - \langle \v{a}_0 , \phi \rangle_{H^1 , H^{-1}} 
\end{align*}
for all $\phi \in L^2$ and $f \in C^{\infty} (\R)$ with $f(0) = 1$ and $f(T) = 0$. Passing to the limit $\varepsilon_n \rightarrow 0$ and using (\ref{eq:weaksolns9}) and (\ref{eq:weaksolns10}) we find
\begin{align*}
\int_0^T \left\lbrace \v{A} (s) \partial_t f(s) + \partial_t \v{A} (s) f(s) \right\rbrace ds = - \v{a}_0 
\end{align*}
in $L^2 (\R^3)$, which implies that $\v{A} (0) = \v{a}_0$. Likewise,
\begin{align*}
& - \langle \dot{\v{a}}_0 , \eta
\rangle_{H^{-1} , H^1} = \\
& \int_0^T \langle \partial_t \v{A}^{\!n} (s) \partial_t f(s) + ( \Delta \v{A}^{\!n} (s) + 4 \pi \alpha \Lambda_{\varepsilon_n}^{-1} \Hproj{ \J [\phi^{n} (s) , \tilde{\v{A}}_n (s)] } ) f(s) , \eta \rangle_{H^{-1} , H^1} ds
\end{align*}
for all $\eta \in H^1$  and $f \in C^{\infty} (\R)$ with $f(0) = 1$ and $f(T) = 0$. Again, passing to the limit as $n \rightarrow \infty$ and using (\ref{eq:weaksolns10}) and (\ref{eq:weaksolns13}), we arrive at
\begin{align*}
\int_0^T \left\lbrace \partial_t \v{A} (s) \partial_t f(s) + \partial_t^2 \v{A} (s) f(s) \right\rbrace ds = - \dot{\v{a}}_0
\end{align*}
in $H^{-1}$, which implies $\partial_t \v{A} (0) = \dot{\v{a}}_0$. An identical argument implies that $\phi (0) = \phi_0$. 
\end{proof}

\chapter{Open Problems}\label{chap:open}

This final chapter is devoted to a discussion of some interesting problems and questions regarding the MBMP equations which we have not managed to resolve at the time of writing. Many other open problems regarding related systems were discussed in Chapter \ref{chap:othermodels}. Some parts of this section constitute our own personal speculations on possible avenues to resolving some of these questions. We make no attempt to make any precise conjectures. 

\section{Local and Global Well-Posedness}  

An obvious problem that remains open is the local well-posedness in, for example, $H^2 \times H^2 \times H^1$ to (\ref{eq:MP}), (\ref{eq:MPC}), and (\ref{eq:MBMP}). As discussed in Chapter \ref{chap:othermodels} (in particular, \S\ref{sec:MS}) we have tried with no success several strategies to prove the existence and uniqueness of a local solution to (\ref{eq:MP}) for initial data $(\psi_0 , \v{a}_0 , \dot{\v{a}}_0) \in H^s ( \R^3 ; \C^2) \oplus H^{r} (\R^3 ; \R^3) \oplus H^{r - 1} (\R^3 ; \R^3)$ for various choices of $s$ and $r$ (especially, $s = r = 2$). We do not doubt that such a local well-posedness result holds for (\ref{eq:MP}), (\ref{eq:MPC}), and (\ref{eq:MBMP}). However, we do not know how to proceed, and it seems as though a new idea is needed. 

One possible way to obtain higher regularity global (and, hence, local) solutions to (\ref{eq:MBMP}) using some of the methods introduced in this thesis is to prove $\varepsilon$-independent bounds on the solutions to $(\phi , \v{A})$ to the $\varepsilon$-modified system (\ref{eq:MBMP_epsilon}). It is in fact possible to prove that for all regularities $m \in [1,2]$ the solutions in Theorem \ref{thm:local_exist_MBMP_epsilon} exist for all time. This is done by using the dispersive estimates for the heat kernel together with the Strichartz estimates for the wave equation to prove $t$-independent bounds on $(\phi , \v{A} , \partial_t \v{A})$ in $H^m \oplus H^m \oplus H^{m-1}$. However, such bounds depend on $\varepsilon^{-1}$   and therefore a compactness argument cannot be used. One potential resolution to this problem is to use the Strichartz estimates for the Schr\"{o}dinger unitary group $e^{i t \Delta}$ to control $(\phi , \v{A} , \partial_t \v{A})$ in $H^m \oplus H^m \oplus H^{m-1}$ in such a way that a compactness argument can be applied to take the limit $\varepsilon \rightarrow 0$. 

If one manages to develop a local well-posedness theory for solutions to (\ref{eq:MP}), (\ref{eq:MPC}), and (\ref{eq:MBMP}), then the next obvious question is can we extend to global solutions. For (\ref{eq:MP}) a sufficiently low regularity local well-posedness theory together with the conservation of energy should suffice to prove global existence. For (\ref{eq:MPC}) we expect the additional assumption $Z < Z_c$ is necessary to obtain a global theory from a sufficiently low regularity local theory. Indeed, to use conservation of energy, which in this case is
\begin{align*}
E_{\mr{P}}  [\psi , \v{A} , \partial_t \v{A}] = \| \gvsig \cdot (\v{p} + \v{A}) \psi \|_{L^2}^2 - Z \langle \psi , | \v{x} |^{-1} \psi \rangle_{L^2} + F [ \v{A} , \partial_t \v{A} ] ,
\end{align*}
together with a blow-up alternative to prove global existence one needs a uniform (in time) bound on the Coulomb energy $\langle \psi , | \v{x} |^{-1} \psi \rangle_{L^2}$. However, as we show in Lemma \ref{lem:bound_on_coulomb}, this is only possible if $E_{\mr{P}}  [\psi , \v{A} , \v{0} ]$ is bounded below with lower bound independent of the magnetic field. We expect the situation to be similar in the full many-body case. 

\section{Blow-up and Virial Identities}

One question that we find particularly interesting is the possibility of finite-energy solutions to (\ref{eq:MBMP}) which blow-up in finite time when $\alpha$ or $\alpha^2 \max{\Z}$ are large. Here we focus on the special case of the MPC Equations (\ref{eq:MPC}) and ask whether it is possible to have initial data $(\psi_0 , \v{a}_0 , \dot{\v{a}}_0)$ with finite energy so that the corresponding solution $(\psi , \v{A} , \partial_t \v{A})$ blows up in $H^1 \oplus H^1 \oplus L^2$-norm as $t \rightarrow T_{\mr{max}}^-$ when $Z > Z_c$. 

One possible route to try and construct blow-up solutions is to consider the second time derivative of $\langle \psi , |\v{x}|^2 \psi \rangle_{L^2}$ in the spirit of Glassey \cite{G77}. In that article, Glassey constructed blow-up solutions for the defocusing nonlinear Schr\"{o}dinger (NLS) equation, which reads 
\begin{align}\label{eq:nonlinear_schr}
\left\lbrace  \begin{array}{l}
i \partial_t \psi = - \Delta \psi - | \psi|^{p-1} \psi \\
\psi (0 , \v{x}) = \psi_0 (\v{x}) \in H^1 (\R^n ; \C) .
\end{array} \right. 
\end{align}
For a collection of results concerning the well-posedness of the NLS equation (\ref{eq:nonlinear_schr}), see \cite[Section 3.3]{taodispersive2006}. The key idea is to derive a virial identity for the second time derivative of the expectation value of $|\v{x}|^2$ - it reads
\begin{align}\label{eq:glassey_virial}
\frac{1}{4} \frac{\dd^2}{\dd t^2} \langle \psi , |\v{x}|^2 \psi \rangle_{L^2} = 2 \| \nabla \psi \|_2^2 - n \frac{p-1}{p+1} \| \psi \|_{p+1}^{p+1}  . 
\end{align}
We see from (\ref{eq:glassey_virial}) that, when $p \geq 1 + 4/n$, 
\begin{align*}
\frac{1}{4} \frac{\dd^2}{\dd t^2} \langle \psi , |\v{x}|^2 \psi \rangle_{L^2} \leq 2 E_{\mr{NS}} ,
\end{align*}
where $E_{\mr{NS}} (t) = \| \nabla \psi (t) \|_2^2 - \frac{2}{p+1} \| \psi (t) \|_{p+1}^{p+1}$ is the conserved energy under the flow generated by (\ref{eq:nonlinear_schr}). Hence, if the energy $E(0)$ of the initial data $\psi_0 \in H^1 (\R^n ; \C)$ is negative, then $\langle \psi , |\v{x}|^2 \psi \rangle_{L^2}$ is a strictly concave, positive function for the during of the time evolution. The corresponding solution $\psi$ to (\ref{eq:nonlinear_schr}), therefore, cannot exist in $H^1 (\R^n ;\C)$ for arbitrarily long times, otherwise expectation value of $|\v{x}|^2$ would be zero at some positive time and, by the uncertainty principle $\| \psi \|_{2}^2 \leq 2 \| \v{x} \psi \|_{2} \| \nabla \psi \|_{2}$, the kinetic energy would become $+ \infty$ (note the $L^2$-norm of $\psi$ is also preserved under the NLS time evolution). 

If we attempt to prove the existence of blow-up solutions to the MPC equations using this method, we arrive at an interesting virial identity for $\langle \psi , |\v{x}|^2 \psi \rangle_{L^2}$ for a variant of the MPC equations we considered in this thesis:
\begin{align}\label{eq:MPC_variant}
\left\lbrace \begin{array}{l}
i \partial_t \psi = \left( [\gvsig \cdot (\v{p} + \v{A})]^2 - \dfrac{Z}{|\v{x}|} - \dfrac{1}{\alpha} \varphi \right) \psi  \\[1.2ex]
\diver{\v{E}} = 4 \pi \rho , \hspace{5mm} \diver{\v{B}} = 0  \\
\curl{\v{E}} = - \alpha \partial_t \v{B} , \hspace{5mm} \curl{\v{B}} = 4 \pi \alpha \v{J}_{\mr{P}} [\psi , \v{A}] + \alpha \partial_t \v{E} \\
\v{E} = - \nabla \varphi - \alpha \partial_t \v{A} , \hspace{5mm} \v{B} = \curl{\v{A}} .
\end{array} \right.
\end{align}
where $\v{J}_{\mr{P}} [\psi , \v{A}]$ is the Pauli probability current (\ref{def:Pauli_prob_current}). Note that the charge continuity equation (\ref{eq:charge_continuity}) is satisfied. Formally, $\|\psi \|_{L^2}$ and the total energy 
\begin{align*}
E [\psi , \v{A} , \partial_t \v{A}] = \| \gvsig \cdot (\v{p} + \v{A}) \psi \|_{2}^2 - Z \langle \psi , | \v{x} |^{-1} \psi \rangle_{L^2} + F[ \v{B} , \v{E} ] ,
\end{align*}
where $F[\v{B} , \v{E}] = (\|\v{B}\|_{2}^2 + \|\v{E}\|_{2}^2) / (8 \pi \alpha^2)$, are conserved. We note that this system of PDEs was argued to be unphysical in Chapter \ref{chap:intro}. However, the full symmetry of Maxwell's equations allows us to derive the following virial identity. For any sufficiently smooth solution $(\psi , \v{A} , \varphi)$ to (\ref{eq:MPC_variant}) we have
\begin{align*}
& \frac{1}{2} \frac{\dd^2 }{\dd t^2} \left( \frac{1}{2} \| |\v{x}| \psi \|_2^2 + \| |\v{x}| \v{B} \|_2^2 + \| | \v{x} | \v{E} \|_2^2 \right) \\
& = 2 \| \gvsig \cdot (\v{p} + \v{A} ) \psi \|_{L^2}^2 - Z \langle \psi , |\v{x}|^{-1} \psi \rangle_{L^2} + F [ \v{B} , \v{E} ] - \frac{1}{2} \frac{\dd }{\dd t} \int_{\R^3} |\v{x}|^2 \v{J}_{\mr{P}} [ \psi , \v{A} ] \cdot \v{E} ~ \dd \v{x} , \numberthis \label{eq:virial_identity}
\end{align*} 
where $\v{S} = (\v{E} \wedge \v{B}) / (4 \pi \alpha)$ is the Poynting vector. We now give a formal derivation of this identity.

Consider the function $f(t) = \langle \psi (t) , \mu (\v{x}) \psi (t) \rangle_{L^2}$ where $\mu : \R^3 \rightarrow \R$ is a radial multipler to be specified later. Using the commutator identity $[A^2 , B] = A [A,B] + [A,B] A$ and taking the time derivative we find 
\begin{align*}
\frac{\dd f}{\dd t} & = \langle i [ (\v{p} + \v{A})^2 , \mu ] \psi , \psi \rangle_{L^2} \\
& = \langle i ( (\v{p} + \v{A}) \cdot [ \v{p} , \mu ] + [ \v{p} , \mu ] \cdot (\v{p} + \v{A}) ) \psi , \psi \rangle_{L^2} \\
& = 2 \re{ \langle \nabla \mu \cdot (\v{p} + \v{A}) \psi , \psi \rangle_{L^2} } .
\end{align*} 
Choosing $\mu (\v{x}) = |\v{x}|^2$, we arrive at
\begin{align*}
\frac{1}{4} \frac{\dd}{\dd t} \langle \psi , |\v{x}|^2 \psi \rangle_{L^2} = \re{ \langle \v{x} \cdot (\v{p} + \v{A}) \psi , \psi \rangle_{L^2} } ,
\end{align*}
Rewriting the previous equation in terms of the \textit{dilation operator} $D = (\v{x} \cdot \v{p} + \v{p} \cdot \v{x} ) / 2$ we conclude for the first time derivative 
\begin{align}\label{eq:virial_identity_calc_1}
\frac{1}{4} \frac{\dd}{\dd t} \langle \psi , |\v{x}|^2 \psi \rangle_{L^2} = \langle D \psi , \psi \rangle_{L^2} + \langle \v{x} \cdot \v{A} \psi , \psi \rangle_{L^2} .
\end{align}

Taking another time derivative of both sides of (\ref{eq:virial_identity_calc_1}) we find
\begin{align}\label{eq:virial_identity_calc_2}
\frac{1}{4} \frac{\dd^2}{\dd t^2} \langle \psi , |\v{x}|^2 \psi \rangle_{L^2} = \langle i [ H , D ] \psi , \psi \rangle_{L^2} + \langle \v{x} \cdot \partial_t \v{A} \psi , \psi \rangle_{L^2} + \langle i [H , \v{x} \cdot \v{A}] \psi , \psi \rangle_{L^2} ,
\end{align}
where 
\begin{align*}
H = [\gvsig \cdot (\v{p} + \v{A})]^2 - Z | \v{x} |^{-1} - \frac{1}{\alpha} \varphi .
\end{align*}
To proceed we must evaluate the two commutators $[H,D]$ and $[H , \v{x} \cdot \v{A}]$. To compute $[H,D]$ we note the identity
\begin{align*}
\frac{\dd }{\dd \theta} e^{- i \theta D} H e^{i \theta D} \Big|_{\theta = 0} = i [ H , D ] ,
\end{align*}
and $(e^{i \theta D} f)(\v{x}) = e^{3 \theta /2} f(e^{\theta} \v{x})$. Hence, our task is to compute $e^{-i \theta D} H e^{i \theta D}$. First we compute $e^{-i \theta D} (\v{p} + \v{A}) e^{i \theta D} f (\v{x})$ and find
\begin{align*}
e^{-i \theta D} (\v{p} + \v{A}) e^{i \theta D} f (\v{x}) & = e^{3 \theta / 2} \left( e^{-i \theta D} \v{p} f (e^{\theta} \v{x})  + e^{-i \theta D} (\v{A} (e^{-\theta} \cdot ) f ) ( e^{\theta} \v{x} ) \right) \\
& = e^{3 \theta / 2} \left( e^{\theta} e^{-i \theta D} (\v{p} f) (e^{\theta} \v{x}) + e^{-3 \theta / 2} \v{A} (e^{-\theta} \v{x}) f (\v{x}) \right) \\
& = e^{\theta} (\v{p} f) (\v{x}) + \v{A} (e^{-\theta} \v{x}) f (\v{x}) \\
& = ( (e^{\theta} \v{p} + \v{A} (e^{-\theta} \cdot)) f ) (\v{x}) .
\end{align*}
Hence, $e^{-i \theta D} [\gvsig \cdot (\v{p} + \v{A})]^2 e^{i \theta D} = e^{2 \theta} [ \gvsig \cdot ( \v{p} + \v{A}_{\theta} )]^2$ where $\v{A}_{\theta} (\v{x}) = e^{-\theta} \v{A} (e^{-\theta} \v{x})$. In total,
\begin{align}\label{eq:virial_identity_calc_3}
e^{-i \theta D} H e^{i \theta D} & = e^{2 \theta} [\gvsig \cdot ( \v{p} + \v{A}_{\theta} )]^2 - e^{\theta} Z |\cdot |^{-1} - \frac{1}{\alpha} \varphi (e^{- \theta} \cdot ) .
\end{align} 
Differentiating the right hand side of (\ref{eq:virial_identity_calc_3}) with respect to $\theta$ and evaluating at $\theta = 0$ one finds
\begin{align}\label{eq:virial_identity_calc_4}
i [H , D] & = 2 [ \gvsig \cdot (\v{p} + \v{A})]^2 - Z |\v{x}|^{-1} + \frac{1}{\alpha} (\v{x} \cdot \nabla) \varphi + \frac{\dd}{\dd \theta} [ \gvsig \cdot  (\v{p} + \v{A}_{\theta}) ]^2 \Big|_{\theta = 0} .
\end{align}

Next we must evaluate the commutator $[H , \v{x} \cdot \v{A}]$. This commutator immediately reduces to $[(\v{p} + \v{A})^2 , \v{x} \cdot \v{A}]$. Using the vector calculus identity
\begin{align*}
\nabla (\v{u} \cdot \v{v}) = (\v{u} \cdot \nabla )\v{v} + (\v{v} \cdot \nabla )\v{u} + \v{u} \wedge \curl{\v{v}} + \v{v} \wedge \curl{\v{u}} ,
\end{align*}
for $\v{u} , \v{v} : \R^3 \rightarrow \R^3$, we compute that
\begin{align*}
\nabla (\v{x} \cdot \v{A}) = - \frac{\dd }{\dd \theta} \Big|_{\theta = 0} \v{A}_{\theta} + \v{x} \wedge \v{B} .
\end{align*}
Consequently,
\begin{align}\label{eq:virial_identity_calc_5}
\langle i [H , \v{x} \cdot \v{A}] \psi , \psi \rangle_{L^2} = 2 \re{ \langle \left( - \frac{\dd }{\dd \theta} \Big|_{\theta = 0} \v{A}_{\theta} + ( \v{x} \wedge \v{B} ) \right) \cdot (\v{p} + \v{A}) \psi , \psi \rangle_{L^2} } .
\end{align}
Combining (\ref{eq:virial_identity_calc_4}) and (\ref{eq:virial_identity_calc_5}) we arrive at
\begin{align*}
\frac{1}{4} \frac{\dd^2 }{\dd t^2} \langle \psi , |\v{x}|^2 \psi \rangle_{L^2} & = 2 \| \gvsig \cdot (\v{p} + \v{A}) \psi \|_{L^2}^2 - Z \langle \psi , |\v{x}|^{-1} \psi \rangle_{L^2} + \frac{\dd }{\dd \theta} \langle \gvsig \cdot \v{B}_{\theta} \psi , \psi \rangle_2 \Big|_{\theta = 0} \\
& \hspace{1cm} - \frac{1}{\alpha} \langle \v{x} \cdot \v{E} \psi , \psi \rangle + 2 \re{ \langle \v{x} \wedge \v{B} \psi , (\v{p} + \v{A}) \psi \rangle_{L^2} } . \numberthis \label{eq:virial_identity_calc_6}
\end{align*}
Similar virial identities to (\ref{eq:virial_identity_calc_6}) may be found in \cite{FV09, G11}. 

We may simplify (\ref{eq:virial_identity_calc_6}) further by noting that
\begin{align*}
2 \re{ \langle \v{x} \wedge \v{B} \psi , (\v{p} - q \v{A}) \psi \rangle_2 } = - \frac{1}{\alpha} \int (\v{x} \wedge \v{B}) \cdot \v{J}_{\mr{P}} [\psi , \v{A}] - \int (\v{x} \wedge \v{B}) \cdot \curl{\langle \psi , \gvsig \psi \rangle_{\C^2}}  ,
\end{align*}
and
\begin{align*}
\int (\v{x} \wedge \v{B}) \cdot \curl{\langle \psi , \gvsig \psi \rangle_{\C^2}} & = \int \curl{(\v{x} \wedge \v{B})} \cdot \langle \psi , \gvsig \psi \rangle_{\C^2} \\
& = \int (- 2 \v{B} - (\v{x} \cdot \nabla) \v{B}) \cdot \langle \psi , \gvsig \psi \rangle \\
& = \frac{\dd }{\dd \theta} \Big|_{\theta = 0} \langle \gvsig \cdot \v{B}_{\theta} \psi , \psi \rangle_{L^2} .
\end{align*}
Therefore, (\ref{eq:virial_identity_calc_6}) becomes
\begin{align*}
\frac{1}{4} \frac{\dd^2 }{\dd t^2} \langle \psi , |\v{x}|^2 \psi \rangle_{L^2} & = 2 \| \gvsig \cdot (\v{p} + \v{A}) \psi \|_{L^2}^2 - Z \langle \psi , |\v{x}|^{-1} \psi \rangle_{L^2} \\
& \hspace{1cm} - \frac{1}{\alpha} \langle \v{x} \cdot \v{E} \psi , \psi \rangle - \frac{1}{\alpha} \int (\v{x} \wedge \v{B}) \cdot \v{J} [\psi , \v{A}] . \numberthis \label{eq:virial_identity_calc_7}
\end{align*}

To further reduce (\ref{eq:virial_identity_calc_7}) we may use the Maxwell equations in (\ref{eq:MPC}). Note that 
\begin{align}\label{eq:virial_identity_calc_8}
- \frac{1}{\alpha} \int (\v{x} \wedge \v{B}) \cdot \v{J}_{\mr{P}} [\psi , \v{A}] = - \frac{1}{4 \pi \alpha^2} \left( \int (\v{x} \wedge \v{B}) \cdot \curl{\v{B}} - \alpha \int (\v{x} \wedge \v{B}) \cdot \partial_t \v{E} \right) .
\end{align}
For the first term on the right hand side of (\ref{eq:virial_identity_calc_8}) we find
\begin{align*}
\int (\v{x} \wedge \v{B}) \cdot \curl{\v{B}} & = \int \curl{(\v{x} \wedge \v{B})} \cdot \v{B} = \int (- 2 \v{B} - (\v{x} \cdot \nabla) \v{B}) \cdot \v{B} \\
& = \int (- 2 |\v{B}|^2 - \frac{1}{2} (\v{x} \cdot \nabla) |\v{B}|^2) = \int (- 2 |\v{B}|^2 + \frac{3}{2} |\v{B}|^2) = - \frac{1}{2} \int |\v{B}|^2 .
\end{align*}
For the second term  on the right hand side of (\ref{eq:virial_identity_calc_8}) we find
\begin{align*}
\int (\v{x} \wedge \v{B}) \cdot \partial_t \v{E} & = \partial_t \int (\v{x} \wedge \v{B}) \cdot \v{E} - \int (\v{E} \wedge \v{x}) \cdot \partial_t \v{B} \\
& = \partial_t \int (\v{x} \wedge \v{B}) \cdot \v{E} + \frac{1}{\alpha} \int (\v{E} \wedge \v{x}) \cdot \curl{\v{E}} \\
& = \partial_t \int (\v{B} \wedge \v{E}) \cdot \v{x} + \frac{1}{\alpha} \int \curl{(\v{E} \wedge \v{x})} \cdot \v{E} \\
& = - 4\pi \alpha \partial_t \int \v{S} \cdot \v{x} + \frac{1}{\alpha} \int ( 2 \v{E} - \v{x} \diver{\v{E}} + (\v{x} \cdot \nabla) \v{E})  \cdot \v{E} \\
& = \frac{1}{2 \alpha} \int |\v{E}|^2 - \frac{1}{\alpha} \int (\v{x} \cdot \v{E}) \diver{\v{E}} - 4\pi \alpha \int \partial_t \v{S} \cdot \v{x} ,
\end{align*}
Putting the previous two calculations together, and using $\diver{\v{E}} = - 4 \pi \alpha |\psi|^2$, we discover 
\begin{align*}
\frac{1}{4} \frac{\dd^2 }{\dd t^2} \| \v{x} \psi \|_2^2 = 2 \| \gvsig \cdot (\v{p} + \v{A}) \psi \|_{L^2}^2 - Z \langle \psi , |\v{x}|^{-1} \psi \rangle_{L^2} + F [\v{B} , \v{E}] - \frac{\dd }{\dd t} \int_{\R^3} \v{S} \cdot \v{x} ~ \dd \v{x} .
\end{align*}
Finally, using the energy density conservation law $\partial_t u_{\mr{EM}} + \diver{\v{S}} = - \v{J}_{\mr{P}} \cdot \v{E}$, where $u_{\mr{EM}} = (|\v{B}|^2 + |\v{E}|^2)/(8 \pi)$ is the field energy density, we arrive at (\ref{eq:virial_identity}). Whether the identity (\ref{eq:virial_identity}) can be utilized further to say something about the time interval of existence for solutions to (\ref{eq:MPC_variant}) remains to be seen.

\begin{singlespace}  %
	\setlength\bibitemsep{0pt}  %
	\printbibliography[title={References}]
\end{singlespace}

\addcontentsline{toc}{chapter}{References}  %

\chapter*{Vita}
\addcontentsline{toc}{chapter}{Vita}  %
Thomas Forrest Kieffer was born in the December of 1993 in Albany, GA. He spent most of his youth in Navarre, FL where he attended Navarre High School from August 2008 to May 2012. He moved to Atlanta, GA in August 2012 to begin undergraduate studies at Georgia Institute of Technology. He graduated with Highest Honors from the Georgia Tech in May 2015 with a B.S. in Applied Mathematics and a B.S. in Physics. He started graduate school as a Physics PhD student at Georgia Tech in August 2016 and, after earning an M.S. in Physics, he transferred in May 2017 to the Mathematics PhD program at Georgia Tech to work with Professor Michael Loss. He earned his PhD in Mathematics in May 2020.

\end{document}